\documentclass[12pt, final]{article}
\usepackage[utf8]{inputenc}
\usepackage{multirow}

\usepackage[utf8]{inputenc} % set input encoding (not needed with XeLaTeX)
\usepackage[T1]{fontenc}

\usepackage{silence}
\usepackage{csquotes}
\MakeOuterQuote{"} % activate straight " ... " as proper curly quotes
\usepackage{xpatch}
\usepackage{natbib}

\usepackage{bibunits}
\usepackage{multirow}
\usepackage{amsmath}
\usepackage{cases}
\usepackage{xfp}
\usepackage{sectsty}
\sectionfont{\Large}
\subsectionfont{\large}

\usepackage{titlesec}
\titlespacing*{\section}{0pt}{2ex plus 1ex minus .2ex}{1ex plus .2ex}
\titlespacing*{\subsection}{0pt}{2ex plus 1ex minus .2ex}{1ex plus .2ex}
\titlespacing*{\subsubsection}{0pt}{2ex plus 1ex minus .2ex}{1ex plus .2ex}

\titleformat*{\section}{\large\bfseries}
\titleformat*{\subsection}{\fontsize{12.5}{12.5}\bfseries}
\titleformat*{\subsubsection}{\normalsize\itshape}
\titleformat*{\paragraph}{\normalsize\bfseries}
\titlespacing*{\paragraph}{0pt}{1ex}{1em}
\titleformat*{\subparagraph}{\large\bfseries}

\usepackage[letterpaper,margin=1in]{geometry}

\usepackage{setspace}
\makeatletter
\g@addto@macro\normalsize{%
  \setlength{\abovedisplayskip}{9pt}
  \setlength{\belowdisplayskip}{9pt}
  \setlength{\abovedisplayshortskip}{9pt}
  \setlength{\belowdisplayshortskip}{9pt}

\thickmuskip=7mu plus 5mu
\medmuskip=5.5mu plus 2mu minus 3mu
}
\makeatother

\usepackage[marginal]{footmisc}
\usepackage{graphicx}
\graphicspath{{./}}
\usepackage{array}
\usepackage{paralist}
\usepackage{verbatim}
\usepackage{enumerate}
\usepackage{amsthm}
\usepackage{booktabs}
\usepackage{tabularx}
\usepackage{multirow}
\usepackage{bigstrut}
\usepackage{amssymb}
\usepackage{color}
\usepackage{bbm}
\usepackage{soul}
\usepackage{pdflscape}
\usepackage{multicol}
\usepackage{enumitem}
\usepackage{placeins} % \FloatBarrier keeps an appendix's exhibits inside that appendix
\usepackage[linewidth=1pt]{mdframed}
\usepackage{xspace}
\usepackage{tocloft}

\usepackage{etoc}
\usepackage{adjustbox}
\usepackage{anyfontsize}
\usepackage{ifdraft}
\usepackage{eucal}
\usepackage{rotating}

\titleformat*{\section}{\fontsize{14}{14}\bfseries}
\titleformat*{\subsection}{\fontsize{13}{12.5}\bfseries}

\usepackage[textsize=small,obeyFinal,colorinlistoftodos]{todonotes}

\usepackage[labelfont=bf,justification=centering]{caption}
\usepackage[labelfont=bf,justification=justified,labelformat=simple]{subcaption}
\usepackage{etoolbox,xstring,xspace}

\newcommand{\PreserveBackslash}[1]{\let\temp=\\#1\let\\=\temp}
\newcolumntype{C}[1]{>{\PreserveBackslash\centering}p{#1}}
\newcolumntype{R}[1]{>{\PreserveBackslash\raggedleft}p{#1}}
\newcolumntype{L}[1]{>{\PreserveBackslash\raggedright}p{#1}}

\usepackage{tikz}
\usetikzlibrary{matrix,shapes,arrows,fit,tikzmark}
\usetikzlibrary{shapes.misc}
\usetikzlibrary{backgrounds}
\usetikzlibrary{shapes.multipart}
\usetikzlibrary{arrows.meta}
\usetikzlibrary{shapes.geometric}

\usepackage{newpxtext}
\usepackage{textcomp}
\usepackage[upint]{newpxmath}
\usepackage{bm}

\definecolor{darkblue}{rgb}{0,0,0.5}
\definecolor{darkred}{rgb}{0.55,0,0}
\definecolor{ash}{RGB}{50,50,50}
\setulcolor{darkred}

\newtheorem{proposition}{Proposition}
\newtheorem{lemma}{Lemma}

\newtheorem{definition}{Definition}

\usepackage{float}
\usepackage{xcolor}
\usepackage{tcolorbox}

\definecolor{darkred}{rgb}{0.35,0.0,0.0}
\definecolor{darkblue}{rgb}{0.0,0.0,0.5}
\definecolor{navyblue}{rgb}{0.0,0.0,0.8}
\definecolor{cadmiumgreen}{rgb}{0.0,0.6,0.24}
\definecolor{red_tol}{HTML}{D92120}
\definecolor{blue_tol}{HTML}{404096}
\definecolor{blugreen_tol}{HTML}{529DB7}

\usepackage{hyperref}
\hypersetup{
    pdfstartview={FitH},
    pdftitle={Chaining Tasks, Redefining Work: A Theory of AI Automation},
    pdfauthor={Demirer, Mert and Horton, John J. and Immorlica, Nicole and Lucier, Brendan and Shahidi, Peyman},
    colorlinks=true,
    linkcolor=darkblue,
    citecolor=darkred,
    filecolor=magenta,
    urlcolor=darkred
}

\usepackage{titling}
\makeatletter
\AtBeginDocument{%
\let\origref\ref
\renewcommand*\ref[1]{%
  \origref{#1}\xlabel{#1}}
}
\newrobustcmd*\xlabel[1]{%
   \ifcsdef{siteref@doc@#1}{}{\csgdef{siteref@doc@#1}{,}}%
    \@bsphack%
    \begingroup
       \csxdef{siteref@doc@#1}{\csuse{siteref@doc@#1},\thepage}%
         \protected@write\@auxout{}%
        {\string\SiteRef{siteref@#1}{\csuse{siteref@doc@#1}}}%
     \endgroup
     \@esphack%
}
\newrobustcmd*\SiteRef[2]{\csgdef{#1}{#2}}
\newrobustcmd*\xref[1]{%
\ifcsundef{siteref@#1}{%
     \@latex@warning@no@line{Label `#1' not defined}
     }{%
    \begingroup
      \StrGobbleLeft{\csuse{siteref@#1}}{2}[\@tempa]\relax%
      \def\@tempb{}%
      \@tempcnta=0\relax%
      \@tempcntb=\@ne\relax%
      \def\do##1{\advance\@tempcnta\@ne}%
      \expandafter\docsvlist\expandafter{\@tempa}%
       \def\do##1{%
         \ifnum\@tempcntb=\@tempcnta\relax%
            \hyperpage{##1}%
         \else
            \hyperpage{##1},%
          \fi%
          \advance\@tempcntb\@ne
       }%
       [\expandafter\docsvlist\expandafter{\@tempa}]\xspace%
    \endgroup
   }%
}
\makeatother

\newlength{\fignotesskip}
\AtEndEnvironment{table}{\vspace{-0.5\baselineskip}}
\AtEndEnvironment{figure}{\vspace{-0.5\baselineskip}}

\usepackage{comment}

\newcolumntype{H}{>{\setbox0=\hbox\bgroup}c<{\egroup}@{}}

\newcounter{todocounter}

\newcounter{todoListItems}

\usepackage[normalem]{ulem}
\usepackage{makecell}

\usepackage{color-edits}
\addauthor[BL]{bl}{blue}
\addauthor[PS]{ps}{red}
\addauthor[NSI]{nsi}{orange}
\tcbuselibrary{skins,breakable}

\usetikzlibrary{positioning, decorations.pathreplacing, decorations.pathmorphing, calc}

\tikzstyle{task} = [rectangle, minimum width=1cm, minimum height=1cm, text centered, draw=black, fill=gray!20, rounded corners, text width=2.2cm, align=center]
\tikzstyle{decision} = [rectangle, minimum width=4cm, minimum height=2cm, text centered, draw=black, fill=green!20, rounded corners, text width=3cm, align=center]
\tikzstyle{process} = [rectangle, minimum width=2.5cm, minimum height=2cm, text centered, draw=black, fill=orange!40, rounded corners, text width=3cm, align=center]
\tikzstyle{process2} = [rectangle, minimum width=2.5cm, minimum height=2cm, text centered, draw=black, fill=gray!10, text width=3cm, align=center]
\tikzstyle{smallbox} = [rectangle, minimum width=1cm, minimum height=1cm, text centered, draw=black, fill=gray!20, text width=1.5cm, align=center]
\tikzstyle{arrow} = [thick,->,>=stealth]

\newtheorem{lem}[lemma]{Lemma}

\theoremstyle{definition} % upright (non-italic) body for examples and remarks
\newtheorem{example}{Example}
\theoremstyle{plain}

\newif\ifboxedresults
\boxedresultsfalse

\ifboxedresults
  \tcolorboxenvironment{example}{
    breakable, boxrule=0.5pt, arc=2pt,
    colback=red!5, colframe=red!50!black,
    left=10pt, right=10pt, top=8pt, bottom=8pt,
    before skip=12pt, after skip=12pt,
  }
  \tcolorboxenvironment{proposition}{
    breakable, boxrule=0.5pt, arc=2pt,
    colback=blue!5, colframe=blue!50!black,
    left=10pt, right=10pt, top=8pt, bottom=8pt,
    before skip=12pt, after skip=12pt,
  }
  \tcolorboxenvironment{theorem}{
    breakable, boxrule=0.5pt, arc=2pt,
    colback=blue!5, colframe=blue!50!black,
    left=10pt, right=10pt, top=8pt, bottom=8pt,
    before skip=12pt, after skip=12pt,
  }
  \tcolorboxenvironment{lem}{
    breakable, boxrule=0.5pt, arc=2pt,
    colback=blue!5, colframe=blue!50!black,
    left=10pt, right=10pt, top=8pt, bottom=8pt,
    before skip=12pt, after skip=12pt,
  }
  \tcolorboxenvironment{cor}{
    breakable, boxrule=0.5pt, arc=2pt,
    colback=blue!5, colframe=blue!50!black,
    left=10pt, right=10pt, top=8pt, bottom=8pt,
    before skip=12pt, after skip=12pt,
  }
  \tcolorboxenvironment{observation}{
    breakable, boxrule=0.5pt, arc=2pt,
    colback=blue!5, colframe=blue!50!black,
    left=10pt, right=10pt, top=8pt, bottom=8pt,
    before skip=12pt, after skip=12pt,
  }
\fi

\newcommand{\T}[0]{\mathcal{T}}
\newcommand{\J}[0]{\mathcal{J}}

\newcommand{\manualLetter}{M}
\newcommand{\AIletter}{A}
\newcommand{\handoffLetter}{H}
\newcommand{\skillCostLetter}{c}
\newcommand{\timeCostLetter}{t}
\newcommand{\laborLetter}{l}
\newcommand{\skillAdjustedTimeLetter}{\ensuremath{\tau}}
\newcommand{\aggregateManualLabor}{M}
\newcommand{\aggregateAIlabor}{A}

\newcommand{\timecost}[1]{\ensuremath{\timeCostLetter_{#1}}}
\newcommand{\skillcost}[1]{\ensuremath{\skillCostLetter_{#1}}}
\newcommand{\hccost}[1]{\skillcost{#1}}
\newcommand{\labor}[1]{\ensuremath{\laborLetter_{#1}}}
\newcommand{\handofftime}[1]{\ensuremath{\timeCostLetter^{\handoffLetter}_{#1}}}

\newcommand{\manualTime}[1]{\ensuremath{\timeCostLetter^{\manualLetter}_{#1}}}
\newcommand{\AItime}[1]{\ensuremath{\timeCostLetter^{\AIletter}_{#1}}}
\newcommand{\manualSkill}[1]{\ensuremath{\skillCostLetter^{\manualLetter}_{#1}}}
\newcommand{\AIskill}[1]{\ensuremath{\skillCostLetter^{\AIletter}_{#1}}}

\newcommand{\minTime}[1]{\ensuremath{\timeCostLetter^*_{#1}}}

\newcommand{\skillAdjustedTimeManual}[1]{\ensuremath{\skillAdjustedTimeLetter^{\manualLetter}_{#1}}}
\newcommand{\skillAdjustedTimeAI}[1]{\ensuremath{\skillAdjustedTimeLetter^{\AIletter}_{#1}}}
\newcommand{\skillAdjustedTimeHandoff}[1]{\ensuremath{\skillAdjustedTimeLetter^{\handoffLetter}_{#1}}}

\newcommand{\minitab}[2][1]{\begin{tabular}{#1}#2\end{tabular}}
\newcommand{\smalldisplay}{%
  \edef\restorebaselineskip{\baselineskip=\the\baselineskip\relax}%
  \small\restorebaselineskip
}

\usepackage{lscape}
\usepackage{geometry}
\usepackage{pdflscape}
\graphicspath{{./}}

\title{Chaining Tasks, Redefining Work:\\ A Theory of AI Automation\thanks{Demirer: \mbox{mdemirer@mit.edu}; Horton: \mbox{jjhorton@mit.edu}; Immorlica: \mbox{nicimm@microsoft.com}; Lucier: \mbox{brlucier@microsoft.com}; Shahidi: \mbox{peymansh@mit.edu}. We thank Martin Fleming, Seyed Mahdi Hosseini Maasoum, Guy Lichtinger, Anand Shah, Philip Trammell, and Michael Zhao for helpful discussions. We also thank seminar and conference participants at MIT's Economics of AI Reading Group, FutureTech Lab, Social Analytics Lab, and IDE BIG.AI conference, as well as at Wharton, Columbia Business School, ACM EC'26, and ESIF 2026 at Cornell.}}

\author{Mert Demirer \\ MIT \& NBER \and John J. Horton \\ MIT \& NBER \and Nicole Immorlica \\ Yale \& Microsoft \and Brendan Lucier \\ Microsoft \and Peyman Shahidi \\ MIT}

\date{\mbox{} \\
First version: February 22, 2026 \\
This version: \today}

\begin{document}
 \etocdepthtag.toc{main}
%\etocsettagdepth{main}{subsection}
%\etocsettagdepth{appendix}{none}

% "OA - 11" is wider than tocloft's default 1.55em page-number box, so widen it and
% the right margin to match, otherwise the dot leaders run into the number. This has to
% be set before the draft Contents / List of Figures / List of Tables below and not
% just before the appendix Contents: those three lists carry appendix page numbers too,
% and at the default width the number is printed on top of the entry title.
\cftsetpnumwidth{3.9em}
\cftsetrmarg{4.9em}

\ifoptionfinal
{}%
{
{
\small
\thispagestyle{empty}
% \listoftodos
% \newpage
\tableofcontents
\newpage
\listoffigures
\newpage
\listoftables
\newpage
}
}

%\pagenumbering{gobble} % no page numbers
\thispagestyle{empty}

\maketitle

% \vspace{-2.5em}
% \begin{center}
% {\textbf{Preliminary Draft: Please do not cite or distribute without consent of the authors.}}
% \end{center}

%\vspace{-3.5em}
\begin{abstract}
\noindent
We study the automation and productivity effects of AI when production is a sequence of interdependent steps with a Leontief structure.
Each step can be executed manually by labor, augmented with AI, or automated within contiguous AI-executed steps called ``chains.''
By characterizing the firm's cost-minimizing assignment of steps to humans and AI, we show that chaining can overturn comparative advantage logic, since human-advantaged steps may be optimally assigned to AI as part of a chain.
Chaining also makes gains from automation depend on how AI-exposed steps cluster in the workflow beyond just how many there are, and can make the returns to improving AI quality non-monotone in a pattern resembling the productivity J-curve.
Empirically, we show that existing automation patterns are consistent with the model's predictions: a step is more likely to be AI-executed when its neighbors are, and AI execution is more common where AI-exposed steps cluster in the workflow.

\vspace{0.5cm}


\end{abstract}
\vspace{1ex}

\setcounter{page}{1}

\section{Introduction}
A growing body of evidence shows that artificial intelligence (AI) tools raise productivity at the task level \citep{noy2023experimental, brynjolfsson2023generative, demirer2026writing}.
Although these gains may be substantial, the larger and likely longer run effects will come from reorganizing production, including shifts in the skill mix of tasks, human capital requirements, and the definition and design of jobs \citep{brynjolfsson2021productivity}.
The workhorse framework in economics for studying the impacts of automation technologies is the task-based model, which represents production as the completion of a set of independent tasks (e.g., \citealp{autor2003skill, acemoglu2011skills, acemoglu2018, acemoglu2019automation, acemoglu2022tasks}).
In these models, what matters is each task's amenability to substitution by, or complementarity with, automation technologies (e.g., computers, robots, or AI).
This approach, however, abstracts from the order in which tasks are performed and from their grouping into jobs.
We argue that this sequencing is economically consequential, because it changes equilibrium predictions about how jobs are reorganized and which units of work are automated when technologies such as AI enter production.

We model production as a Leontief technology over an ordered sequence of exogenously specified \emph{steps}, and treat the definition of tasks and their partitioning across jobs as endogenous outcomes.
A \emph{step} is the primitive unit of work in our framework, corresponding to what classic models call a ``task.''
We instead reserve the term \emph{task} for a contiguous block of steps that the firm endogenously designates for joint execution by a human worker, potentially with the aid of AI.
In the absence of AI, tasks optimally collapse to single-step blocks so that steps and tasks coincide.
With AI, some steps are chained together and automated, creating a meta-task whose output a human worker validates.

The firm seeks to minimize production costs by choosing which steps of the sequence to delegate to AI and which to assign to human workers.
As the AI technology improves, the firm might choose to automate many of these steps, hiring a single, potentially lower-skill worker to verify the output of that automation.
Whether this strategy is optimal depends not only on the effectiveness of the AI technology at accomplishing individual production steps, but also on the sequential relationship between steps, the relative ease of verification versus manual execution, wages, and coordination costs between workers.
These determinants are what set AI apart from earlier waves of automation.
Its capabilities are broad but jagged, varying sharply even across adjacent steps in a workflow \citep{dell2023navigating}, which makes the sequencing of work central to the deployment decision.
As a result, improvements in AI can shift deployment decisions at multiple points across the production chain at once, triggering discrete changes in overall production strategy.

We distinguish two horizons in our framework by how much of its organization the firm is free to adjust, and study how AI affects production in each, as summarized in Table~\ref{tab:short_long_run}.
In the short run, the firm takes its workers and their job boundaries as given, along with what those jobs pay, and deploys AI to minimize the time cost of production.
In the long run, in tandem with AI deployment, the firm can also redraw the boundaries of jobs, which changes the skills each job demands and hence the wage it commands.
\begin{table}[!t]
    \centering
    \vspace{0.1cm}
    \caption{Summary of AI's Impact in the Short and Long Run}
    \label{tab:short_long_run}
    \footnotesize
    \begin{adjustbox}{max width=0.98\textwidth}
    % Two horizons: short run and long run.
\begin{tabular}{|c|c|c|c|c|}
\hline
\textbf{Time Horizon} &
\textbf{AI Deployment} &
\textbf{Job Design} &
\textbf{Worker Wages} &
\textbf{Comments} \\
\hline
Short run &
Optimized &
Fixed &
\minitab[c]{Fixed} &
\minitab[c]{Productivity gains by\\automating existing workflows} \\
\hline
Long run &
\minitab[c]{Optimized} &
Optimized &
\minitab[c]{Adjust to AI deployment\\and job design} &
\minitab[c]{Joint optimization of\\AI deployment and job design\\to minimize total labor costs}\\
\hline
\end{tabular}

    \end{adjustbox}

    \begin{minipage}{1\linewidth}
\begin{spacing}{0.2}
{\fontsize{9.5pt}{10pt}\selectfont Notes: This table summarizes how the firm's problem differs across the two horizons studied in the paper. In the short run the firm takes job boundaries and wages as given and chooses only its AI deployment; in the long run job boundaries are chosen jointly with AI deployment, and wages adjust to the skills each job then demands.}
\end{spacing}
\end{minipage}
\end{table}

\paragraph{Automation versus Augmentation.}
Central to our framework is the distinction between full automation of a production step versus worker augmentation.
In our model, three modes of step completion are recognized: manual, augmented, and automated.
\emph{Manual} completion of a step means that a human carries out the work without any AI assistance.
\emph{Augmented} completion involves a human performing the step with the use of AI.
For example, a human describes the requirements of a step to an AI tool, the AI then executes the step, and its output is verified by the human.
In contrast, we say a step has been \emph{automated} if it is completed by AI end-to-end without any direct human intervention.
We view AI automation as a way of putting certain production steps ``under the hood'' of a logical meta-task that a worker is attempting to complete.

To fix ideas, think of the workflow of a data scientist whose job requires completing five steps: defining the business question (Step 1), finding and fetching data needed for answering the question (Step 2), building an analysis pipeline (Step 3), drafting a report (Step 4), and presenting the findings (Step 5).
As part of an AI-augmented request to complete one of these steps (such as Step 4: drafting a report), the AI might also be assigned to complete one or more preceding steps (such as Step 3: running the analysis, or Steps 2 and 3: finding data and running the analysis) as part of that single request.
Any such preceding step is said to be automated, and the set of all such automated steps, plus the final step that the worker is actively engaged with, can be thought of as a single worker task which we term an \emph{AI chain}.
By choosing which steps to automate, augment, or perform manually, the firm implicitly designs the set of tasks exposed to its workers.

Although augmented and automated production steps both involve AI, they differ in whether a human verifies the AI's output, which augmentation requires and automation does not.
Figure~\ref{fig:task_division} illustrates these concepts using the five-step data scientist workflow described earlier.
\begin{figure}[!t]
  \centering
  \caption{Illustrative Example for Division of Labor}
  \label{fig:task_division}
  \begin{tikzpicture}[scale=0.78, transform shape,
    node distance=0.5cm and 1.2cm,
    every node/.style={rectangle, rounded corners, draw, align=center, minimum width=2cm, minimum height=1cm},
    manual/.style={fill=gray!10},
    automated/.style={fill=green!30},
    augmented/.style={fill=orange!30},
    labelbox/.style={draw, rectangle, rounded corners, font=\bfseries, minimum width=2.25cm, minimum height=1cm},
    dashedbox/.style={draw, rectangle, rounded corners, dashed, inner sep=0.3cm, line width=1pt},
    >=latex
  ]

    % Top layer (Steps)
    \node[manual] (S1) {Step 1\\(Manual)};
    \node[automated, right=of S1] (S2) {Step 2\\(Automated)};
    \node[automated, right=of S2] (S3) {Step 3\\(Automated)};
    \node[augmented, right=of S3] (S4) {Step 4\\(Augmented)};
    \node[manual, right=of S4] (S5) {Step 5\\(Manual)};

    % Arrows between steps
    \draw[->, thick] (S1) -- (S2);
    \draw[->, thick] (S2) -- (S3);
    \draw[->, thick] (S3) -- (S4);
    \draw[->, thick] (S4) -- (S5);

    % Dashed boxes around steps
    \node[dashedbox, draw=olive, fit=(S1)] (DS1) {};
    \node[dashedbox, draw=blue, fit=(S2)(S3)(S4)] (DS2-4) {};
    \node[dashedbox, draw=olive, fit=(S5)] (DS5) {};

    % AI box top center aligned with Step 3
    \node[labelbox, fill=blue!20, above=1.25cm of S3] (AI) {AI};

    % Human box bottom center aligned with Step 3
    \node[labelbox, fill=olive!20, below=1.5cm of S3] (Human) {Human};

    % Lines from AI box to AI steps
    \draw[thin, blue!80] (AI.south) -- (S2.north);
    \draw[thin, blue!80] (AI.south) -- (S3.north);
    \draw[thin, blue] (AI.south) -- (S4.north);

    % Lines from Human box to Human steps
    \draw[thin, olive!80] (S1.south) -- (Human.north);
    \draw[thin, olive] (S4.south) -- (Human.north);
    \draw[thin, olive!80] (S5.south) -- (Human.north);
  \end{tikzpicture}
  \par\vspace{\fignotesskip}
\begin{minipage}{1\linewidth}
\begin{spacing}{0.2}
{\fontsize{9.5pt}{10pt}\selectfont Notes: This figure illustrates division of labor across five steps: Steps 1 and 5 are manual; Steps 2 and 3 are AI-automated (done without direct human intervention); Step 4 is AI-augmented (done with AI but requiring a human to verify its output).
Lines from the AI box indicate steps involving AI, and lines from the Human box indicate steps requiring human execution or verification.
Dashed boxes group steps according to their task assignments: Steps 1 and 5 form separate human tasks, while Steps 2–4 form an AI chain task.}
\end{spacing}
\end{minipage}
\end{figure}
Steps 1 and 5 are performed manually by a human.
Steps 2, 3, and 4 form an AI chain and are done using AI.
Steps 2 and 3 are automated because their outputs feed directly into subsequent AI steps without human review.
Step 4 is augmented because its output is verified by a human before proceeding.
The worker's job therefore consists of three tasks: the manual step~1, the AI chain spanning steps~2--4, whose only human input is verifying the output of step~4, and the manual step~5.

Automation is therefore available only inside a chain.
A lone AI step must have its output verified, so chaining is what lets AI work unsupervised.
This is potentially a source of large productivity gains, \emph{but} it requires AIs that can perform each step with a high probability of success.
In this sense, the model's basic setup is conceptually close to the O-ring model of production \citep{kremer1993oring, gans2026oring}.

\paragraph{Optimizing Automation and Fragmentation.}
The decision to deploy AI on a given production step compares the cost of manual execution to the unified AI-based execution cost, which is not the same as assigning each step to whichever factor of production has comparative advantage in it.
Chaining overturns the comparative advantage logic because a human verifies only the output of the chain's last step, so verification is a fixed cost of the chain rather than a marginal cost that scales with its length.
Appending a neighboring step to an existing chain therefore costs only a lower probability of end-to-end success, whereas assigning that step to a human terminates the chain and creates an additional checkpoint, entailing another verification or the cost of manual execution.
When AI is sufficiently reliable on the marginal step, avoiding this checkpoint can dominate, pulling the step into AI execution even if manual execution is preferred for it in isolation.
In the data scientist workflow, for instance, verifying a report produced after AI executes Steps~3--4 takes the same effort as verifying one produced after AI executes Step 4 alone; the AI is just less likely to succeed at the former.
What tips a step into automation is therefore not just its own characteristics but also those of its neighbors.

More generally, we should expect the gains from AI automation to be greater when AI-able steps co-occur in the workflow.
A job is more or less ``fragmented'' in its AI exposure depending on whether the steps where AI is most effective are clustered together or dispersed throughout the workflow.
We formalize this idea in a \emph{fragmentation index}, the expected cost of a simple execution plan endowed with foresight into AI's successes and failures, and show analytically that it approximately tracks the cost of optimal AI deployment.
We then provide empirical evidence that, in practitioner-documented business processes, workflows with higher fragmentation see a weaker translation from AI exposure to AI execution.

Because chaining makes the cost of automating a step depend on its neighbors, the firm cannot settle each step's execution mode in isolation, and its problem ranges over all possible configurations of production.
As the number of feasible arrangements grows exponentially in the number of steps, whether the firm can calculate its optimum is not obvious.
We show the firm can indeed solve this optimization, exactly in the short run and to any desired accuracy once wages and job boundaries are also chosen in the long run.

\paragraph{Non-Monotone Effects of AI Improvements.}
The AI chaining feature of the model implies that improvements in AI quality can generate non-monotone productivity effects.
For any fixed production arrangement, better AI reduces costs smoothly by raising AI success probabilities and lowering expected execution times.
However, firms can optimize AI deployment over a discrete set of options, and the cost-minimizing configuration can switch at particular AI quality thresholds.
As a result, marginal improvements in AI may have little effect when they do not change the optimal production arrangement, which is especially likely in the early days of adoption when AI quality is low.
Once AI quality crosses a threshold, longer AI chains become viable and the optimal deployment shifts abruptly, reallocating steps between humans and AI.
Our framework thus provides a micro-foundation for the productivity J-curve phenomenon \citep{brynjolfsson2021productivity} and yields testable predictions on when and how AI-driven productivity gains trigger structural reorganization in labor markets.

\paragraph{Long-Run Job Design.}
As the firm can freely reorganize its production around what AI can do over longer horizons, AI changes what workers do and not just how quickly they do it, and through that, wages and labor demand.
How broadly to define each job involves a trade-off: adding more tasks to a worker's job raises the wage that job must pay, since the worker has to acquire the skills all of those tasks require, while dividing them across more jobs adds coordination cost, since every job boundary forces a \emph{hand-off} of intermediate output among workers.

Which way improvements in AI move job boundaries is ambiguous, because it presses on both sides of this trade-off.
Absorbing steps into chains strips skill out of the tasks a worker retains, making broad jobs cheaper to sustain, while the same reduction, together with any direct saving in coordination costs, makes each boundary cheaper to draw.
Our framework thus supplies not a prediction about job breadth but an account of what determines it, and of what would have to be measured to settle the question empirically.

\paragraph{Aggregate Production.}
Although we develop these results at the level of a single firm, we show that, under additional assumptions, firm-level production aggregates across the economy to a constant elasticity of substitution (CES) representation.
Our framework can thus serve as a micro-foundation for economy-wide production functions and be used to study the aggregate effects of AI on labor markets and productivity.

\paragraph{Empirical Evidence.}
We complement our theoretical analysis with new empirical evidence that speaks directly to the core mechanisms emphasized by the model.
Taking the model to data requires three things about every step: AI's ability to perform it, the extent to which AI currently executes it, and where it sits in the production sequence.
The first two we draw from existing sources, linking O*NET tasks to human assessments of AI exposure \citep{eloundou2023gpts} and to realized AI execution outcomes from Anthropic's Economic Index \citep{handa2025economictasksperformedai}.
The third, the sequence in which tasks are performed, is not recorded in O*NET, so we take two approaches.
First, we use the Process Classification Framework (PCF) of the American Productivity \& Quality Center, whose member working groups document the activities of each business process in the order practitioners perform them.
Because the PCF covers only operational business processes rather than the full range of occupations, we also generate the ordering of each O*NET occupation's tasks ourselves, and we report our results using both.

We test, and find evidence consistent with, three key predictions of the model.
First, the chaining mechanism implies that AI-executed steps should appear in contiguous blocks rather than scattered across the workflow.
This is what we find in both datasets, where AI execution operates over consecutive steps far more often than chance would produce.
Second, chaining also creates strong local complementarities in AI automation decisions within a workflow, so that a step positioned next to AI-executed neighbors may be pulled into AI execution by the gains from automating it as part of a chain, even if human execution would be preferred for that step in isolation.
We find that when conceptually similar steps appear across different workflows, a given step is more likely to be executed by AI in the workflow where its neighboring steps are also AI-executed.
Third, the model emphasizes that fragmentation of AI-exposed steps plays a central role in determining the returns to AI automation.
Controlling for the share of steps exposed to AI, workflows whose AI-exposed steps are more dispersed across the production sequence execute a lower share of their steps with AI.

Because our O*NET results rest on a task ordering that we generate with a language model, we show that they are robust to how that sequence is constructed.
First, we show that the estimates we obtain are not artifacts of the particular prompt we use to elicit workflow sequences from an LLM (Appendix~\ref{app:gptPrompts_robustness}).
Second, they are not artifacts of the particular way O*NET lists occupation tasks (Appendix~\ref{app:frequency_robustness}).
Third, our sequencing procedure recovers to a substantial degree the orderings of two corpora that no language model ordered, the PCF's practitioner-documented processes and process-mining event logs recorded as the work was carried out (Appendix~\ref{app:external_validation}).

Taken together, these results show that the sequencing logic at the heart of the model is already visible in how firms use AI at this early stage of adoption.
What AI executes depends not only on comparative advantage and step-level characteristics, but also on where a step sits relative to its neighbors.
This is why the share of steps exposed to AI can be a misleading predictor of occupational impacts when production steps are technologically interdependent.
As adoption widens and AI becomes reliable enough to sustain longer chains, we expect these patterns to grow more pronounced.

\section{Related Literature}
\paragraph{Task Interdependence in Task-based Models of Automation.}
Our framework relates to the literature on task-based models of technical change, which generally view tasks as predetermined and technologically independent objects combined through a CES production function \citep{acemoglu2011skills, acemoglu2018, acemoglu2019automation, acemoglu2024tasks}.
We depart from this literature by modeling production as a sequence of technologically interdependent steps, where AI chaining creates complementarities across adjacent steps that can overturn comparative advantage as the driver of who does what.\footnote{Two contemporary works by \cite{trammell2026workflows} and \cite{farboodi2026data} study a different form of linkage across tasks, where performing a task raises productivity on others through learning, by workers in the former and by machines in the latter. Our emphasis is instead on the ordering of tasks in production.}
Chaining also lets step-level decisions aggregate into the system-level substitution that \cite{autor2003skill}, \cite{acemoglu2019automation}, and \cite{bresnahan2002information} emphasize, since entire chains can be automated in a single leap rather than step by step.

\paragraph{Complementarities, O-Ring Dynamics, and the Productivity J-Curve.}
Our work contributes to the literature on friction-laden technology adoption and organizational complementarities, in which information technologies raise productivity mainly alongside complementary organizational change \citep{bresnahan2002information, bloom2016management}.\footnote{\citet{kim2026mapping} provide recent experimental evidence of this pattern for AI, holding tools and training fixed across startups and varying only what founders learn about reorganizing production around them.}
We micro-found these complementarities through the O-ring theory of production \citep{kremer1993oring}. Chaining makes the payoff to automating a step depend on whether its neighbors can join the same AI block, so the returns to improving any one step turn on performance elsewhere and arrive only once AI quality crosses a threshold that makes longer chains worthwhile.
Our results align with \citet{gans2026oring}, who likewise emphasizes task interdependence in O-ring production and obtains lumpy automation through time reallocation rather than workflow adjacency.
Because these returns arrive only once AI quality passes such a threshold, the marginal benefit of improving AI is non-monotone, and in this sense our framework provides a micro-foundation for the ``productivity J-curve'' pattern of technology adoption \citep{brynjolfsson2021productivity, mcelheran2025rise, furman2019ai, tambe2012productivity, mcelheran2024adoption}.

\paragraph{Division of Labor and the Boundary of the Job.}
Our work also connects to the literature on the division of labor within a firm.
The Coase--Williamson logic, in which economic boundaries are shaped by coordination, transaction, and governance costs \citep{coase1937nature, williamson1971vertical, williamson1979transaction}, applies to the ``boundary of the job'' as well, in that the firm chooses how many and which tasks to bundle into a worker's role, making the degree of specialization endogenous \citep{murphy1986specialization, becker1992division, dessein2006adaptive}.
Our hand-off costs are the intra-firm analogue of the transaction costs that AI cannot fully subsume, capturing the tacit or social knowledge aspect of work \citep{deming2017growing}, and playing the role that the cost of unbundling tasks plays in \citet{garicano2026weak}. 
AI nonetheless redraws job boundaries by stripping skill out of the steps it absorbs, effectively changing the skill requirements of jobs \citep{autor2024new, autor2025expertise, autor2026makes}.\footnote{
\cite{ide2025ai} provides a distinct but related view, embedding AI as an autonomous problem solver in a knowledge hierarchy \citep{garicano2000hierarchies, garicano2006organization} and studying how it reshapes spans of control within the firm.
}

\paragraph{Measuring AI Exposure and Realized Execution.}
A rapidly growing literature has sought to quantify the impact of new automation technologies on the labor market.
Early work studies the susceptibility of occupations to computerization \citep{frey2017future}, and more recent work develops measures of occupational exposure to AI \citep{brynjolfsson2018can, webb2020impact, felten2021occupational, eloundou2023gpts, hampole2025artificial, tomlinson2025working}, typically by mapping technologies to specific tasks and then aggregating these task-level mappings to construct occupation-level indices.
We build on these measures by showing how, in addition to the share of tasks exposed to AI, the arrangement of those tasks in a workflow also governs the gains from automation.

\section{AI Chaining and Production in the Short Run} \label{sec:shortrun}
In this section and the next we consider the firm's production in the short run.
A firm's production consists of a sequence of steps $\mathcal{S}=(s_1,\ldots,s_m)$, every one of which must be completed to produce a unit of output.\footnote{
We hold the set of production steps fixed and abstract from the creation or disappearance of steps.
}
The firm can assign each step either to a human worker or to AI.
A step assigned to a human takes worker time and demands skill.
A step assigned to AI succeeds only with some probability, and its output may need to be reviewed by a human worker before it is accepted.

In the short run, job boundaries are already drawn, so each worker's job covers a fixed block of steps; workers have acquired the skills their roles require; and the firm has hired them at wages agreed in advance.
The firm's only margin of adjustment is AI deployment, that is, which steps it assigns to AI.
That choice sets how much time each worker spends completing steps, and hence the wage bill the firm pays.
Given these short-run assumptions, we assume without loss of generality that there is a single worker type whose wage is normalized to $1$.
Section~\ref{sec:longrun} relaxes these, letting the firm divide the steps among workers so that the skill a job demands, and hence the wage it commands, becomes endogenous.
Table~\ref{tab:notation} in Appendix~\ref{app:tables_and_figures} summarizes our notation.

\subsection{Steps}
\label{sec:model.steps}

Steps are the smallest unit of work in our framework, and each is executed in one of three modes: \emph{manually}, \emph{augmented}, or \emph{automated}.
If the firm assigns step $i$ to manual execution (denoted by \manualLetter), a human worker spends time \manualTime{i} completing the step without AI.
Alternatively, the firm may assign the step to AI (denoted by \AIletter), in which case a worker spends time \AItime{i} prompting and verifying the output of one AI attempt at it.\footnote{This could, for example, represent the time spent specifying what the step requires and then checking whether what the AI returns meets it.}$^{,}$\footnote{With highly capable AIs, the human role shrinks to verifying AI's outputs, as in \cite{agrawal2019exploring}.}
Such an attempt succeeds independently with some (step-dependent) probability $q_i=\alpha^{d_i}$, where $\alpha$ is the quality of the general-purpose AI technology and $d_i$ represents how ``difficult'' step $i$ is for AI, and it must be repeated until it succeeds. 
Since the worker repeats prompting and verifying AI's output until success, the expected time cost of verification across trials is $\AItime{i}/q_i$.
Definitions~\ref{def:manual_step} and \ref{def:augmented_step} formalize these two modes of execution.

\begin{definition}[Manual Step]
\label{def:manual_step}
A step $i$ is performed manually if it is executed entirely by a human worker without AI assistance.
The time cost of a manual step is denoted by \manualTime{i}.
\end{definition}

\begin{definition}[Augmented Step]
\label{def:augmented_step}
A step $i$ is augmented if it is executed by AI, after which a human worker verifies its output.
The expected time cost of verifying a single augmented step is $\AItime{i}/q_i$, where $q_i = \alpha^{d_i} \in (0,1]$ is the AI's success probability on the step, increasing in the quality $\alpha$ of the general-purpose AI technology and decreasing in the step's difficulty $d_i$ for AI.
\end{definition}

\begin{figure}[!t]
  \centering
  \caption{Bundling Steps into Tasks via AI Strategy}
  \label{fig:steps_tasks}
  % Steps-to-Tasks Visualization.
% This is the top two layers of plots/TikZ_visualization/hierarchical_production_illustration.tex,
% cut at the task layer and kept identical to it in scale, node distance, box
% styles, and colors, so that a step, a task, and an aggregation arrow read the
% same in both figures. The only omissions are the job layer and the red snaked
% hand-off arrows, neither of which exists before jobs are introduced.
\begin{adjustbox}{center,trim=0pt 0pt 0pt 0pt,clip,margin=-2cm 0cm}
  \begin{tikzpicture}[scale=0.73, transform shape,
    node distance=0.35cm and 0.91cm,
    every node/.style={rectangle, rounded corners, draw, align=center, minimum width=2cm, minimum height=1cm},
    manual/.style={fill=gray!10},
    automated/.style={fill=green!30},
    augmented/.style={fill=orange!30},
    dashedbox/.style={draw, rectangle, rounded corners, dashed, inner sep=0.25cm, line width=1.25pt},
    humanbox/.style={rectangle, rounded corners, draw, align=center, minimum width=2.25cm, minimum height=1cm, fill=olive!20},
    aibox/.style={rectangle, rounded corners, draw, align=center, minimum width=2.25cm, minimum height=1cm, fill=blue!20},
    >=latex
  ]

    % Top layer (Steps)
    \node[manual] (S1) {Step 1\\(Manual)};
    \node[augmented, right=of S1] (S2) {Step 2\\(Augmented)};
    \node[manual, right=of S2] (S3) {Step 3\\(Manual)};
    \node[automated, right=of S3] (S4) {Step 4\\(Automated)};
    \node[automated, right=of S4] (S5) {Step 5\\(Automated)};
    \node[augmented, right=of S5] (S6) {Step 6\\(Augmented)};
    \node[manual, right=of S6] (S7) {Step 7\\(Manual)};

    % Arrows between steps
    \draw[->, thick] (S1) -- (S2);
    \draw[->, thick] (S2) -- (S3);
    \draw[->, thick] (S3) -- (S4);
    \draw[->, thick] (S4) -- (S5);
    \draw[->, thick] (S5) -- (S6);
    \draw[->, thick] (S6) -- (S7);

    % Dashed boxes grouping consecutive steps into tasks
    \node[dashedbox, draw=olive, fit=(S1)] (DS1) {};
    \node[dashedbox, draw=blue, fit=(S2)] (DS2) {};
    \node[dashedbox, draw=olive, fit=(S3)] (DS3) {};
    \node[dashedbox, draw=blue, fit=(S4)(S5)(S6)] (DS4-6) {};
    \node[dashedbox, draw=olive, fit=(S7)] (DS7) {};

    \path (S1.west) -- (S7.east) coordinate[midway] (CenterTop);

    % Bottom layer (Tasks)
    \node[humanbox, below=2.4cm of CenterTop, xshift=-6cm] (T1) {Task 1\\(Human)};
    \node[aibox, right=of T1] (T2) {Task 2\\(AI Chain)};
    \node[humanbox, right=of T2] (T3) {Task 3\\(Human)};
    \node[aibox, right=of T3] (T4) {Task 4\\(AI Chain)};
    \node[humanbox, right=of T4] (T5) {Task 5\\(Human)};

    % Arrows between tasks
    \draw[->, thick] (T1) -- (T2);
    \draw[->, thick] (T2) -- (T3);
    \draw[->, thick] (T3) -- (T4);
    \draw[->, thick] (T4) -- (T5);

    % Aggregation arrows, colored to match the dashed box each task comes from
    \draw[->, thick, olive] (DS1.south) -- (T1.north);
    \draw[->, thick, blue]  (DS2.south) -- (T2.north);
    \draw[->, thick, olive] (DS3.south) -- (T3.north);
    \draw[->, thick, blue]  (DS4-6.south) -- (T4.north);
    \draw[->, thick, olive] (DS7.south) -- (T5.north);

  \end{tikzpicture}
\end{adjustbox}
  \par\vspace{\fignotesskip}
\begin{minipage}{1\linewidth}
\begin{spacing}{0.2}
{\fontsize{9.5pt}{10pt}\selectfont Notes: The top layer shows seven production steps, each completed in manual (gray), automated (green), or augmented (orange) mode; dashed boxes group contiguous steps into tasks. The bottom layer shows the resulting tasks: a run of automated steps ending in an augmented step forms an AI chain (blue; here steps 4--6 become Task~4), and a lone augmented step is an AI chain of length one (here step~2 becomes Task~2), while a step performed manually on its own is a human task (olive). Horizontal arrows show the production flow within each layer, and each colored vertical arrow shows which dashed block of steps becomes which task.}
\end{spacing}
\end{minipage}
\end{figure}
In the top layer of Figure~\ref{fig:steps_tasks}, Step~1, for example, is a manual step, whereas Step~2 is an augmented step performed by the AI and verified by a human.

Manual and augmented execution specify the mode of a step in isolation.
There is a third mode, which we call \emph{automation}, that the firm obtains by assigning multiple consecutive steps to AI.
In this case, the AI attempts to complete the entire run of steps without a human reviewing the intermediate outputs, and the worker verifies only the output of the last step.
Every step before the last is then automated, since no human is involved in it, while the last step is augmented.
Chaining steps in this way lowers the probability that an AI attempt succeeds, since the AI must complete every step in the run successfully, but it economizes on human time, since the output is verified only once.
We call such a contiguous run of AI-executed steps an \emph{AI chain}.

More precisely, an AI chain spans steps $(s_\ell, \dotsc, s_r)$, where the worker verifies only the output of the final step $s_r$.
One attempt at the chain therefore costs the same as verifying the last step $s_r$, namely time \AItime{r}.\footnote{
One could instead let the cost of verifying a chain scale with its length, for example by charging a fixed overhead for every step the chain spans on top of verifying its endpoint.
The channel we are describing would nonetheless remain the same, as extending a chain would still economize on execution cost, provided that this overhead is smaller than what the appended step would cost on its own.
We do not model this overhead verification cost to avoid complicating the notation.
}
Assuming independent failures across steps, an AI attempt succeeds with probability $\prod_{i=\ell}^{r} q_i$, and the chain is retried until it succeeds, for a total expected time of $\AItime{r}/\prod_{i=\ell}^{r} q_i$.
Since $q_i = \alpha^{d_i}$, the chain succeeds with probability $\alpha^{\sum_{i=\ell}^{r} d_i}$, so $\sum_{i=\ell}^{r} d_i$ is the chain's total difficulty, aggregating additively over its steps.

Note that a chain may consist of a single step, with $\ell = r$, in which case the chain is simply an augmented step, so an augmented step is equivalent to an AI chain of length one, a convention we adopt throughout the paper.
We now give the formal definitions of automated steps and AI chains.

\begin{definition}[Automated Step]
\label{def:automated_step}
A step is automated if it is executed by AI end-to-end, with no human verifying its output, so that its direct human time cost is zero.%
\footnote{That execution of automated steps does not incur any cost to the firm is consistent with the AI subscription business model, where the firm pays an upfront fixed cost to an AI provider and uses AI tools at zero marginal cost afterward. The forces described by the model remain unaffected even if we assumed that each time the firm deploys AI on a step it has to incur a cost.}
\end{definition}

\begin{definition}[AI Chain]
\label{def:ai_chain}
An AI chain is a contiguous block of one or more sequential steps executed by AI, in which all steps but the final one are automated and the final step is augmented.
An AI chain spanning steps $(s_\ell,\dots,s_r)$ has steps $(s_\ell,\dots,s_{r-1})$ automated and step $s_r$ augmented, with expected time cost $\AItime{r}/\prod_{i=\ell}^{r} q_i$.
\end{definition}
In Figure~\ref{fig:steps_tasks}, Steps~4--6 form an AI chain. The human prompts the AI for the entire chain, but verifies the output of only Step~6, the augmented step. Steps~4 and~5 are hence performed without any direct human intervention and are automated.
Step~2 forms an AI chain (of length one) on its own as well, containing only a single augmented step.

\subsection{Tasks}
\label{sec:model.tasks}

If steps are the smallest unit of work, tasks are the smallest unit of \emph{human} work.
The distinction follows from the execution modes of steps above. 
A worker is involved in a step only when they perform it manually or verify what the AI returns at the end of a chain, whereas an automated step runs inside the chain entirely beneath their awareness.
Every point at which a worker executes or verifies work therefore opens a task for them, and that task absorbs whatever the AI has run automatically since the previous such point.
The steps of a workflow map into tasks accordingly. A manual step is a task on its own, and an entire AI chain is a single \emph{composite} task, delegated to the AI at its start and received back at its augmented endpoint.

The firm's choice of where to place chains is therefore what creates the tasks its workers hold, and we record that choice as an \emph{AI deployment strategy}.

\begin{definition}[AI Deployment Strategy]
\label{def:ai_strategy}
An AI deployment strategy (or AI strategy for short) is a partition of the step sequence $\mathcal{S}=(s_1,\ldots,s_m)$ into contiguous blocks, each of which is either a single manual step or an AI chain (automated or augmented).
The blocks of the partition are the firm's tasks, forming the task sequence $\T = (T_1, \dotsc, T_n)$, and each task $T_b$ carries an expected time cost \timecost{b}: \manualTime{i} for a manual step $s_i$, and $\AItime{r}/\prod_{i=\ell}^{r} q_i$ for a chain spanning $(s_\ell, \dotsc, s_r)$.
\end{definition}

The AI deployment strategy in Figure~\ref{fig:steps_tasks} maps the seven steps in the top row into five tasks in the bottom row: two AI chains (a length-three chain of Steps 4--6 and a length-one chain of Step~2) and three manual tasks.

\subsection{Short-Run Production}
\label{sec:shortrun.production}

The firm employs human labor to perform tasks and produce the output, and labor is the only input used in production.\footnote{
Viewed as a production function, our sequential technology is Leontief at the firm level, because producing a unit of output requires completing every step in fixed proportions.
The Leontief form is natural for modeling a specific production process at a disaggregated level \citep{deloecker2021io}.
}
The worker carries out the entire step sequence $\mathcal{S} = (s_1, \dotsc, s_m)$, which the firm's AI strategy organizes into the task sequence $\T = (T_1, \dotsc, T_n)$, and is paid at the normalized rate of $w = 1$ per unit of time for the whole time production takes.
The cost of production is therefore the total expected time requirement of the workflow's tasks, and the firm chooses the AI deployment strategy that minimizes
\begin{align}
\label{eq:totalcost_shortterm}
\min_{\T} \ \sum_{T_b \in \T} \timecost{b}.
\end{align}

Problem~\eqref{eq:totalcost_shortterm} defines the firm's short-run cost minimization problem.
In the short run, AI is purely productivity-improving, compressing the time a fixed set of steps takes by letting the firm minimize cost through a combination of augmented and automated steps.

AI chaining makes solving this problem non-trivial.
Because a chain links several steps into a single unit of execution, the firm cannot solve the problem in \eqref{eq:totalcost_shortterm} separately for each step; the choice instead ranges over all feasible arrangements of the workflow.
This linkage gives rise to rich dynamics in how AI is deployed across production, and hence in how automation affects it.
We discuss these implications next.

\section{Implications of Chaining} \label{sec:results}
In this section we draw out three implications of AI chaining in the model.
First, we show that chaining can overturn comparative advantage logic, in that a step that is cheaper to perform manually in isolation may nonetheless be executed by AI in the optimal configuration of production (Section~\ref{sec:comparative_advantage}).
Second, we introduce a \emph{fragmentation index} that captures how dispersed a workflow's AI-able steps are, and show that it approximates the cost of the firm's optimal AI strategy, so that the gains from AI depend on the arrangement of AI-able steps (Section~\ref{sec:fragmentation}).
Third, we show that the marginal gains from improving AI quality are non-monotone, abruptly rising whenever improvements in AI quality make redesigning workflows worthwhile (Section~\ref{sec:nonlinear}).
All proofs are collected in Appendix~\ref{app:omitted_proofs}.

\subsection{Overturning of Comparative Advantage}
\label{sec:comparative_advantage}

For a single step, the decision whether to deploy AI depends simply on whether the manual execution cost exceeds the AI augmentation cost.
With more than one step, however, the optimal deployment of AI depends on more than each step's own characteristics, because of the benefits of automating multiple steps together in an AI chain.
As a result, the factor of production holding the comparative advantage on a step in isolation need not be the one that performs it in the optimal production configuration.
That is, a step that is cheaper to perform manually may still be executed by AI as part of a chain.

Consider a step $k$ on which the human holds the comparative advantage, so that $\manualTime{k} < \AItime{k}/q_k$, together with its immediate neighbors: a predecessor $k-1$, and a successor $k+1$ executed as a standalone augmented step.\footnote{The augmented step $k+1$ can be thought of as a longer AI chain terminating augmented at some step $r > k+1$. Such a chain can be treated as a single augmented step with composite cost parameters, namely verification cost $\AItime{r}$ and success probability $\prod_{i=k+1}^{r} q_i$, so what follows applies unchanged.}
Ask what it costs to admit $k-1$ and $k$ into the chain that ends at $k+1$, holding the rest of the workflow fixed.
The marginal cost of admitting the two is the chain's cost beyond verifying step $k+1$ alone,
\[
\underbrace{\frac{\AItime{k+1}}{q_{k-1}\,q_k\,q_{k+1}}}_{\text{cost of chain over } \{k-1,\,k,\,k+1\}}
\;-\;
\underbrace{\frac{\AItime{k+1}}{q_{k+1}}}_{\text{cost of augmenting } k+1 \text{ alone}}
\;=\;
\underbrace{\frac{\AItime{k+1}\,(1-q_{k-1}q_k)}{q_{k-1}\,q_k\,q_{k+1}}}_{\text{marginal cost of folding the pair $\{k-1, k\}$ in}}.
\]
Writing the joint failure risk as $1-q_{k-1}q_k = (1-q_{k-1}) + q_{k-1}(1-q_k)$ splits this marginal cost into a share for each admitted step separately,
\begin{equation}
\label{eq:ca_test}
\underbrace{\frac{\AItime{k+1}\,(1-q_{k-1}q_k)}{q_{k-1}\,q_k\,q_{k+1}}}_{\text{marginal cost of folding the pair $\{k-1, k\}$ in}}
\;=\;
\underbrace{\frac{\AItime{k+1}\,(1-q_{k-1})}{q_{k-1}\,q_k\,q_{k+1}}}_{\text{step $k-1$'s share}}
\;+\;
\underbrace{\frac{\AItime{k+1}\,(1-q_k)}{q_k\,q_{k+1}}}_{\text{step $k$'s share}}.
\end{equation}

There are two channels through which the human-advantaged step $k$ can end up automated.
The first is that step $k$'s share of the marginal cost, the second term in \eqref{eq:ca_test}, falls below what executing it alone would cost, which by assumption is manual execution $\manualTime{k}$.
This condition indicates whether chaining steps $k$ and $k+1$ alone costs less than the two being executed optimally on their own.
Note that this does not depend on the predecessor, and is easiest to see when we set $q_{k-1}=1$, so that step $k-1$ becomes ``free'' and its share in \eqref{eq:ca_test} vanishes.
Importantly, step $k$'s own verification cost $\AItime{k}$ never enters this comparison, since an automated step's output is never checked on its own.

Even when the comparison above does not favor folding the step into the chain, step $k$ may still end up automated.
The second channel runs through the predecessor.
Suppose the predecessor step $k-1$ is expensive to execute alone, whether manually or augmented, but reliable enough for AI that its contribution to the chain's failure risk is negligible (i.e., as $q_{k-1} \to 1$ the first term in \eqref{eq:ca_test} approaches zero).
Admitting the pair $\{k-1, k\}$ is then worthwhile as long as the time saved by folding in the predecessor more than covers the shortfall step $k$ leaves when it chains with its successor alone.
Step $k$ can therefore be automated to bridge the expensive but AI-reliable predecessor to the chain's endpoint, even though its own share of the marginal cost exceeds what executing it alone would cost and, without the predecessor, it would have stayed manual.
Proposition~\ref{prop:ca_local} formalizes the overturning of comparative advantage and makes a more general claim.

\begin{proposition}[Overturning of Comparative Advantage]
\label{prop:ca_local}
Hold the execution of the workflow outside the block $\{k-1,\, k,\, k+1\}$ fixed, and let no chain cross the block's boundary.
Let the human hold the comparative advantage on step $k$, so that $\manualTime{k} < \AItime{k}/q_k$.
Then:
\begin{enumerate}
\item[(i)] Whatever step $k$'s own parameters, there is a nonempty open set of neighbor parameters at which the cost-minimizing arrangement has step $k$ executed by AI (automated or augmented) no matter how wide the human's advantage on it is.
\item[(ii)] Fix the cost parameters of the three steps and let their success probabilities vary. If AI executes step $k$ at $(q_{k-1},\, q_k,\, q_{k+1})$, then it does so at every $(q'_{k-1},\, q'_k,\, q'_{k+1})$ with $q'_i \geq q_i$ for each $i$.
\end{enumerate}
\end{proposition}

Intuitively, part~(i) of Proposition~\ref{prop:ca_local} says that whether a step is executed by AI depends not just on its own cost parameters, as comparative advantage logic would predict, but also on those of its neighbors.
A step with identical characteristics can thus be automated in one workflow where it sits beside an AI-able run, yet remain manual in another where AI completes the neighboring steps only with low probability.
Where part~(i) says that being human-advantaged carries no guarantee of being executed one way or the other, part~(ii) says that whatever the answer at a given set of parameters, making a step and/or its neighbors friendlier to AI can only push it toward AI execution.
A step that AI executes thus never reverts to manual as its neighbors become more reliable, and the same holds when the step's own reliability improves alongside theirs.

To put this result into perspective, consider a workflow containing the following three steps: an expensive step that AI completes very reliably at position $k-1$, an \emph{AI-hard} step at $k$, and an \emph{AI-easy} step at $k+1$,\footnote{One can think of an AI-easy step as one that is ``exposed'' to AI, and an AI-hard step as ``unexposed.''}
\begingroup\smalldisplay
\[
\bigl(\manualTime{k-1},\, \AItime{k-1},\, q_{k-1}\bigr) = (25,\,25,\,0.95), \quad
\bigl(\manualTime{k},\, \AItime{k},\, q_k\bigr) = (6,\,3.5,\,0.3), \quad
\bigl(\manualTime{k+1},\, \AItime{k+1},\, q_{k+1}\bigr) = (8,\,4,\,0.9).
\]
\endgroup
In isolation, comparative advantage assigns steps $k-1$ and $k$ to the human and step $k+1$ to AI.
Yet the cost-minimizing strategy chains all three, augmented at step $k+1$, at a cost of $15.6$ against the $35.4$ the comparative advantage arrangement would cost.

Figure~\ref{fig:ca_regions} shows the implications of Proposition~\ref{prop:ca_local} in this setup.
\begin{figure}[!t]
\caption{When Is a Human-Advantaged Step Executed by AI?}
\label{fig:ca_regions}
\centering
\includegraphics[width=0.50\linewidth]{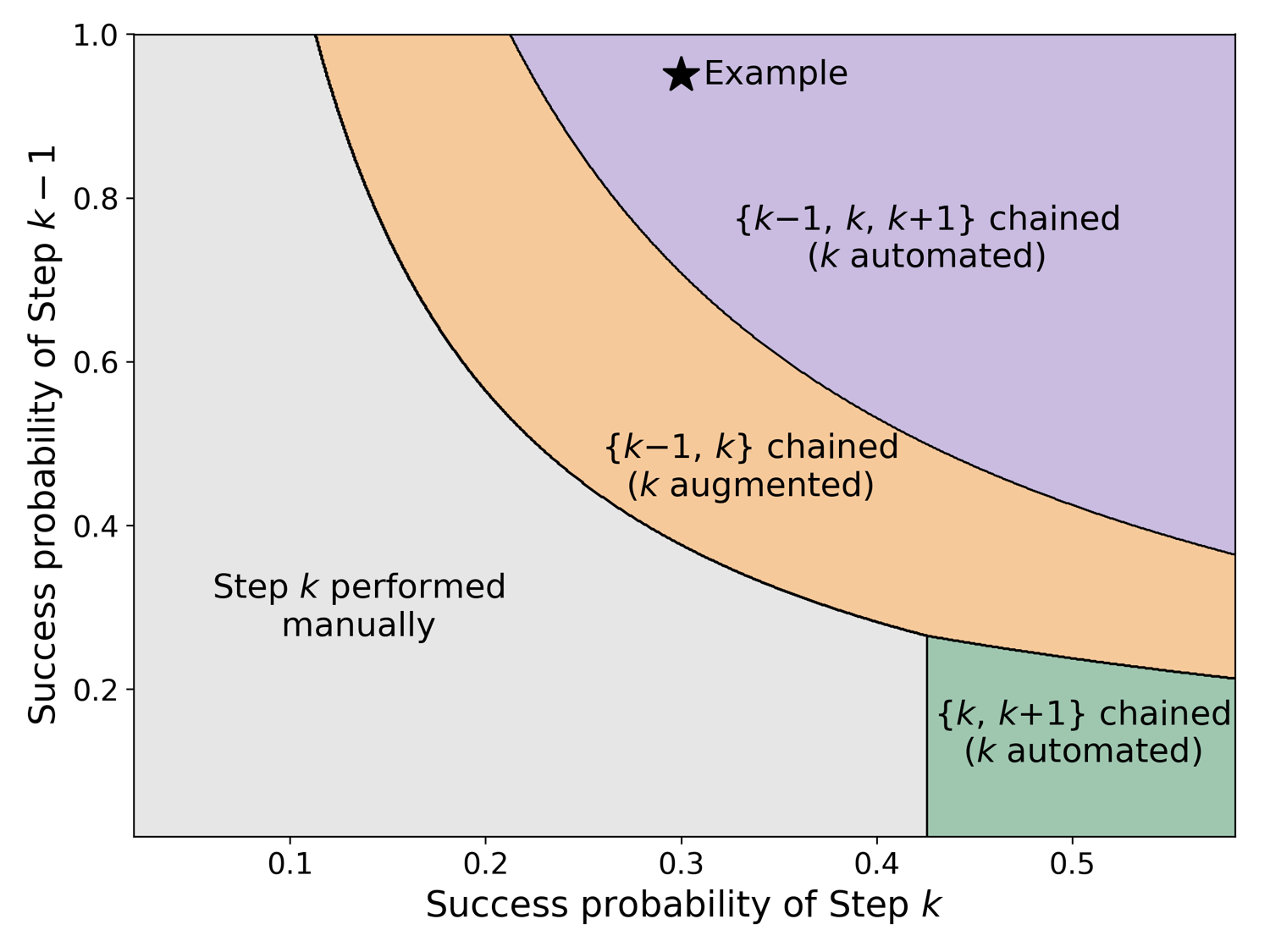}
\par\vspace{\fignotesskip}
\begin{minipage}{1\linewidth}
\begin{spacing}{0.2}
{\fontsize{9.5pt}{10pt}\selectfont Notes: The figure shows the four arrangements available for a human-advantaged Step~$k$ and its immediate neighbors, and shades which of them the firm chooses as the success probabilities of Steps~$k-1$ and~$k$ vary. Every other parameter is held at the value given in the text. The human holds the comparative advantage on Step~$k$ throughout the plotted range. Each shaded region is labelled by the arrangement the firm chooses in it, and the star marks the parameters of that example.}
\end{spacing}
\end{minipage}
\end{figure}
It holds all parameters fixed except the success probabilities of Steps~$k-1$ and~$k$, and shades how the cost-minimizing arrangement executes Step~$k$ as those two vary, with the example itself marked by a star.
The human keeps the comparative advantage on Step~$k$ throughout, since the horizontal axis stops where augmenting the step in isolation would break even against its manual cost of $6$.
Even so, Step~$k$ is performed manually in the lower left where AI is unreliable on both steps $k$ and $k-1$.
Moving to the right, the chain $\{k,\, k+1\}$ becomes more efficient at $q_k \approx 0.43$, where the failure risk Step~$k$ adds to it falls below the manual work it displaces; moving up instead leaves Step~$k$ untouched and reaches the three-step chain by way of the band in which the chain stops at Step~$k$.\footnote{The orange band, where Step $k$ is augmented, is not a general feature. When verifying AI's output at Step~$k$ is expensive enough relative to executing Step~$k+1$ on its own, stopping the chain at Step~$k$ never pays, and the firm crosses directly from manual execution to automation.}
Varying a step's own success probability and that of a single neighbor is thus enough to cover all four feasible ways of executing a human-advantaged step, minimal as that adjustment is.
Part~(ii) of Proposition~\ref{prop:ca_local} can be read off the same figure as well. Starting anywhere that AI executes Step~$k$ and moving up, to the right, or both, one never crosses back into the grey region, so improvements in AI reliability only ever add steps to those AI executes.

We must note that Proposition~\ref{prop:ca_local} is a \emph{local} statement. It compares rearrangements of a step and its neighbors while holding the AI strategy on the rest of the workflow fixed, and so it falls short of determining whether the overturning takes place in the overall cost-minimizing arrangement of the entire production sequence.
Nonetheless, the forces determining the execution mode of a single step in the global optimum are the ones outlined here.

\subsection{Workflow-level AI Exposure and the Fragmentation Index}
\label{sec:fragmentation}

The previous section examined interdependence locally, asking how a step's immediate neighbors shape its execution mode.
The same forces also operate at the level of the workflow as a whole. Whether, and how much, a workflow benefits from AI depends not only on how many of its steps AI can reliably perform, but on how those steps are arranged across the entire sequence.
The logic is again the complementarity that chaining creates. The more fragmented the AI-able steps are across the production sequence, the fewer of them can share a verification.
This section formalizes this logic and provides an approach to measurement.
We first start with an example that makes the point numerically.

\begin{example}
\label{ex:fragmentation}
Consider a workflow containing four steps built from the AI-hard step of Section~\ref{sec:comparative_advantage} and a version of its AI-easy step that is costlier to verify and that AI completes less reliably, with two steps of each kind.
Figure~\ref{fig:ex_fragmentation} shows the same four steps arranged in two different orders.
\begin{figure}[!t]
\caption{Role of Clustering of AI-able Steps in Facilitating Automation}
\label{fig:ex_fragmentation}
\centering
\begin{minipage}{\linewidth}
  \centering
  \subcaption{Dispersed AI-able Steps Configuration}
  % Fragmentation example (interleaved). Two layers:
%  - top: each step's difficulty for AI (easy=teal, hard=red), alternating;
%  - bottom: the optimal execution of each step (all manual=gray).
% No two easy steps are adjacent, so no chain can form and automation is never
% available. On its own neither step is worth giving to AI (easy: 6/0.7 ~ 8.57
% against 8 manual; hard: 3.5/0.3 ~ 11.67 against 6), so all four are manual,
% for a total of 2(8) + 2(6) = 28. Colors follow Figures 1-2.
\centering
\begin{adjustbox}{max width=\linewidth}
\begin{tikzpicture}[>=latex,
    box/.style={rectangle, rounded corners, draw, minimum width=1.8cm, minimum height=0.8cm, align=center, font=\small},
    easy/.style={fill=teal!40},
    hard/.style={fill=red!40},
    manual/.style={fill=gray!10},
    augmented/.style={fill=orange!30},
    idx/.style={draw=none, font=\scriptsize},
    rowlab/.style={draw=none, font=\footnotesize\itshape}]

  % ---- top layer: step difficulty ----
  \node[box, easy] (d1) {Easy};
  \node[box, hard, right=0.7cm of d1] (d2) {Hard};
  \node[box, easy, right=0.7cm of d2] (d3) {Easy};
  \node[box, hard, right=0.7cm of d3] (d4) {Hard};

  \node[idx, above=1pt of d1] {Step 1};
  \node[idx, above=1pt of d2] {Step 2};
  \node[idx, above=1pt of d3] {Step 3};
  \node[idx, above=1pt of d4] {Step 4};

  % ---- bottom layer: optimal execution (tasks) ----
  \node[box, manual, below=1.0cm of d1] (e1) {Manual};
  \node[box, manual, below=1.0cm of d2] (e2) {Manual};
  \node[box, manual, below=1.0cm of d3] (e3) {Manual};
  \node[box, manual, below=1.0cm of d4] (e4) {Manual};

  \foreach \a/\b in {d1/e1,d2/e2,d3/e3,d4/e4} \draw[->,thick] (\a)--(\b);

  % arrows between consecutive steps and between consecutive tasks
  \foreach \a/\b in {d1/d2,d2/d3,d3/d4} \draw[->,thick] (\a)--(\b);
  \foreach \a/\b in {e1/e2,e2/e3,e3/e4} \draw[->,thick] (\a)--(\b);

  \node[rowlab, left=1.4cm of d1] {AI Difficulty:};
  \node[rowlab, left=1.4cm of e1] {Execution:};

  % total cost of the depicted execution, for comparison with the clustered figure
  \node[font=\footnotesize, anchor=west] at ($(e4.east)+(0.45,0)$) {Total cost $= 28$};
\end{tikzpicture}
\end{adjustbox}
\end{minipage}
\par\vspace{1.6em}
\begin{minipage}{\linewidth}
  \centering
  \subcaption{Clustered AI-able Steps Configuration}
  % Fragmentation example (clustered). Two layers:
%  - top: each step's difficulty for AI (easy=teal, hard=red), the easy steps first;
%  - bottom: the optimal execution. The contiguous easy run becomes one AI chain
%    (automated up to the augmented endpoint, grouped by the blue dashed box); the
%    hard steps are performed manually. Colors follow Figures 1-2.
\centering
\begin{adjustbox}{max width=\linewidth}
\begin{tikzpicture}[>=latex,
    box/.style={rectangle, rounded corners, draw, minimum width=1.8cm, minimum height=0.8cm, align=center, font=\small},
    easy/.style={fill=teal!40},
    hard/.style={fill=red!40},
    manual/.style={fill=gray!10},
    automated/.style={fill=green!30},
    augmented/.style={fill=orange!30},
    idx/.style={draw=none, font=\scriptsize},
    rowlab/.style={draw=none, font=\footnotesize\itshape},
    chainbox/.style={draw=blue, rectangle, rounded corners, dashed, inner sep=0.20cm, line width=1pt}]

  % ---- top layer: step difficulty ----
  \node[box, easy] (d1) {Easy};
  \node[box, easy, right=0.7cm of d1] (d2) {Easy};
  \node[box, hard, right=0.7cm of d2] (d3) {Hard};
  \node[box, hard, right=0.7cm of d3] (d4) {Hard};

  \node[idx, above=1pt of d1] {Step 1};
  \node[idx, above=1pt of d2] {Step 2};
  \node[idx, above=1pt of d3] {Step 3};
  \node[idx, above=1pt of d4] {Step 4};

  % ---- bottom layer: optimal execution (tasks) ----
  \node[box, automated, below=1.0cm of d1] (e1) {Automated};
  \node[box, augmented, below=1.0cm of d2] (e2) {Augmented};
  \node[box, manual, below=1.0cm of d3] (e3) {Manual};
  \node[box, manual, below=1.0cm of d4] (e4) {Manual};

  \foreach \a/\b in {d1/e1,d2/e2,d3/e3,d4/e4} \draw[->,thick] (\a)--(\b);

  % arrows between consecutive steps and between consecutive tasks
  \foreach \a/\b in {d1/d2,d2/d3,d3/d4} \draw[->,thick] (\a)--(\b);
  \foreach \a/\b in {e1/e2,e2/e3,e3/e4} \draw[->,thick] (\a)--(\b);

  % the contiguous easy run is executed as a single AI chain
  \node[chainbox, fit=(e1)(e2)] (chain) {};
  \node[font=\footnotesize, text=blue!60!black] at ($(chain.south)+(0,-0.30)$) {single task (AI chain)};

  \node[rowlab, left=1.4cm of d1] {AI Difficulty:};
  \node[rowlab, left=1.4cm of e1] {Execution:};

  % total cost of the depicted execution, for comparison with the interleaved figure
  \node[font=\footnotesize, anchor=west] at ($(e4.east)+(0.45,0)$) {Total cost $\approx 24.24$};
\end{tikzpicture}
\end{adjustbox}
\end{minipage}
\par\vspace{\fignotesskip}
\begin{minipage}{1\linewidth}
\begin{spacing}{0.2}
{\fontsize{9.5pt}{10pt}\selectfont Notes: In each panel the top layer gives each step's difficulty for AI and the bottom layer its execution in the cost-minimizing strategy, colored as in Figure~\ref{fig:steps_tasks}. The AI-hard steps have parameters $(\manualTime{i},\, \AItime{i},\, q_i) = (6,\,3.5,\,0.3)$ and the AI-easy steps $(8,\,6,\,0.7)$. The two panels contain exactly the same four steps and differ only in the order in which they are performed. In panel~(a) every AI-easy step is flanked by AI-hard ones, so no two AI-able steps are adjacent, no chain can form, and every step is performed manually. In panel~(b) the AI-easy steps form a contiguous run, which the firm executes as a single AI chain: Step~1 is automated and Step~2 is augmented and verified. The total cost at the right of each execution row makes the two arrangements directly comparable.}
\end{spacing}
\end{minipage}
\end{figure}

Panel~(a) interleaves the two kinds, alternating AI-easy and AI-hard steps.
Because every easy step is flanked by hard ones, no two AI-easy steps are ever adjacent, and pairing an easy step with a hard one in a chain costs more than running both alone, so no step is automated.
The firm is left to treat each step on its own, and on its own neither kind of step is worth giving to AI.
The workflow therefore uses no AI at all, even though half of its steps are ones AI performs comparatively well.

Panel~(b) clusters the AI-easy steps instead, placing them side by side.
Chaining then becomes a cost-effective option, because AI automates the first easy step at zero direct cost and verifies only once at the second.
A single verification now covers two steps instead of none, which is enough to make AI worth deploying where it was not before.
\end{example}

Note that the two arrangements contain exactly the same set of steps and differ only in the order in which those steps are performed, yet one of them uses no AI while the other hands half of its steps to AI.
Dispersion of AI-able steps in the workflow can therefore stop AI exposure from converting into AI execution altogether, rather than merely making that execution dearer.
This is precisely the margin that we explore empirically later.
Example~\ref{ex:fragmentation} isolates the arrangement of steps in the case that matters for what follows, where the steps AI performs poorly are ones the firm would not fold into a chain under any arrangement.\footnote{In the normalized case introduced below ($\AItime{i} = 1$), this holds whenever $\manualTime{h} \leq 1/q_h - 2$ for each AI-hard step $s_h$. Extending a chain over $s_h$ multiplies its expected cost by $1/q_h$ while saving only the cost of executing $s_h$ on its own, which requires $1/q_h \geq 1 + \min\{\manualTime{h}, 1/q_h\}$. The binding deviation is the one that instead merges the chains on either side of $s_h$, since it saves one of the two verifications those chains pay separately. It requires $1/q_h \geq 2 + \min\{\manualTime{h}, 1/q_h\}$, and the minimum must then be $\manualTime{h}$, which gives the condition above.}

As gains from AI automation appear to depend on how clustered or dispersed a workflow's AI-able steps are, having a measure of that arrangement would be useful, one that stands in for them without solving the firm's optimization problem outright.
To build such a measure, return to the firm's AI deployment problem in \eqref{eq:totalcost_shortterm} and consider the special case in which AI verification costs are normalized, $\AItime{i} = 1$ for all $i$, so that every step takes the same time to prompt and verify.\footnote{The normalization puts every step on the same footing as far as the cost of using AI goes, leaving where the AI-able steps sit as the only thing that separates one workflow from another.}
Suppose each step $s_i$ succeeds independently with probability $q_i$, and say that any step that does not succeed \emph{fails}.
Write $F$ for the set of steps that fail and $\mathcal{C}$ for the random collection of maximal connected components of non-failed steps, and assign each $C \in \mathcal{C}$ the weight $\omega(C) = \min\bigl\{ 1, \sum_{s_i \in C} \manualTime{i} \bigr\}$, which is $1$ unless the manual costs of its steps sum to less than the normalized cost of verifying one AI attempt.
The realized fragmentation of the workflow is then
\begin{equation}
    \sum_{s_i \in F}\min\left\{\manualTime{i}, \frac{1}{q_i}\right\} + \sum_{C \in \mathcal{C}} \omega(C),\label{eq:fragmentation}
\end{equation}
and the \emph{fragmentation index} of the workflow is its expected value.
Proposition~\ref{prop:fragmentation} shows how the fragmentation index and optimal AI deployment cost are related.

\begin{proposition}
\label{prop:fragmentation}
   Fix a sequence of $m$ steps $\mathcal{S} = (s_1, \dotsc, s_m)$, each with normalized AI verification cost $\AItime{i} = 1$.
   Let $FI$ denote its fragmentation index and let $OPT$ denote the minimum cost over all AI strategies.
   Then $\tfrac{1}{8}OPT \leq FI \leq \tfrac{5}{4}OPT$.
   If we further assume $\manualTime{i} \geq 1$ for all $i$, then $\tfrac{1}{4}OPT \leq FI \leq \tfrac{5}{4}OPT$.
\end{proposition}

Proposition~\ref{prop:fragmentation} says that a single closed-form quantity tracks, up to constant factors, what the firm's whole optimization is worth, so a workflow with a high fragmentation index is costly to produce however cleverly the firm deploys AI.
What the quantity in \eqref{eq:fragmentation} measures is easiest to see through a prophet who can look ahead and tell which steps AI would complete on its first attempt.
Runs of steps that will all succeed are natural candidates for a chain, so the prophet chains each such run, unless running that whole run manually is cheaper, and executes every remaining step on its own in whichever of the two modes is cheaper for it; the fragmentation index is the expected cost of that strategy.
This is not the optimal strategy even with foresight, since a long chain with a good chance of success may be worth attempting even when the prophet knows it will fail on the first try.
What Proposition~\ref{prop:fragmentation} establishes is that its cost nonetheless approximates the cost of the optimal strategy, which enjoys no foresight at all.

Read this way, the index is lower when the steps AI performs reliably are clustered, since a long run of such steps is more likely to survive as a single connected component once failures are realized, and a single component is charged only once.
When every step takes at least as long manually as a single AI verification (i.e., $\manualTime{i} \geq 1$) every component is charged $\omega(C) = 1$ and the fragmentation index reduces to\footnote{
To see how Equation~\ref{eq:fragmentation_closed_form} is obtained, note that $\manualTime{i} \geq 1$ implies $\sum_{s_i \in C} \manualTime{i} \geq 1$ for every component $C$, so $\omega(C) = 1$ and the realized fragmentation in \eqref{eq:fragmentation} becomes $\sum_{s_i \in F}\min\{\manualTime{i}, 1/q_i\} + \lvert \mathcal{C} \rvert$.
Taking expectations, step $s_i$ fails with probability $1-q_i$, and a component starts at $s_i$ exactly when $s_i$ succeeds and either $i=1$ or $s_{i-1}$ fails, so $\mathbb{E}\lvert \mathcal{C} \rvert = q_1 + \sum_{i=2}^{m} q_i(1-q_{i-1}) = \sum_{i=1}^{m} q_i - \sum_{i=2}^{m} q_{i-1} q_i$.
}
\begin{equation}
FI \;=\; \underbrace{\sum_{i=1}^{m} (1-q_i)\min\left\{\manualTime{i}, \frac{1}{q_i}\right\} \;+\; \sum_{i=1}^{m} q_i}_{\text{invariant to the order of the steps}} \;-\; \sum_{i=2}^{m} q_{i-1}\,q_i .
\label{eq:fragmentation_closed_form}
\end{equation}
The first two terms depend on the workflow's steps but not on the order in which they are performed.
Across re-orderings of a given workflow the index is therefore decreasing in $\sum_{i=2}^{m} q_{i-1} q_i$, the total reliability of adjacent pairs of steps, and is lowest precisely when the steps AI performs reliably sit beside one another.\footnote{That only adjacent pairs enter is a consequence of the assumption rather than an approximation. Once every component is charged the same $1$ whatever its length, the arrangement matters only through how many components there are, and a component begins wherever a step that succeeds follows one that failed, which is an event involving two steps.
Absent the assumption, $\omega(C)$ depends on the whole of $C$ and no such expression in adjacent pairs exists.
Without this assumption the charge-once saving can disappear altogether. If the manual costs of a merged run sum to less than the cost of one verification, $\omega$ is additive over that run's steps and the index does not respond to the arrangement at all.
}
A further implication of the proof is an approximately optimal greedy algorithm for constructing AI strategies, where the firm groups steps into a chain as long as the probability of success stays sufficiently high, terminating the chain and starting a new one when it falls too low, and running a chain manually whenever its steps' total manual time is below the cost of one verification.

\subsection{Non-Monotonicity of Gains from AI Improvements}
\label{sec:nonlinear}

So far we have asked how AI's impact on a workflow depends on the arrangement of its steps, holding the capability of the AI technology fixed.
We now turn to what happens as that capability rises.
We show that the gains from marginal improvements in AI quality accrue non-monotonically, because improvements that are at first absorbed within the existing arrangement eventually trigger a discrete reorganization of work.

Recall that $\alpha$ measures the quality of the general-purpose AI technology, so that advances in that technology are represented by increases in its value.
When AI quality is so low that the firm deploys AI nowhere, marginal improvements offer no return at all, since every step is performed manually and production costs do not depend on $\alpha$.
Once $\alpha$ passes the level at which deploying AI somewhere in the workflow becomes worthwhile, further improvements do pay off, but at a diminishing rate, since a more reliable AI is being applied to a fixed set of AI-executed steps.
And when $\alpha$ reaches a level at which the firm re-optimizes its deployment strategy, extending a chain to absorb steps that were previously manual, the marginal return jumps up discontinuously, because improvements in quality now compound over a longer stretch of the workflow.
We call the values of $\alpha$ at which the optimal strategy changes reorganization thresholds. Between them the returns to AI progress diminish, and at each of them those returns jump upward as longer AI chains become worth deploying.\footnote{Stated precisely, the returns never fall at a threshold, and they jump upward at every threshold where the cost curves of the two adjacent strategies cross non-tangentially. That is the generic case, and it holds at both thresholds of Example~\ref{ex:nonmonotone}; Lemma~\ref{lem:reorg_threshold} in Appendix~\ref{app:nonmonotone} gives the condition and shows it cannot be dropped.}
Example~\ref{ex:nonmonotone} traces them in a two-step workflow.

\begin{example}
\label{ex:nonmonotone}
Take the AI-hard and AI-easy steps of Section~\ref{sec:comparative_advantage} once more, relabelled steps~1 and~2 and keeping their time costs, but now let their difficulties be $d_1 = 11$ and $d_2 = 1$, so that AI succeeds at them with probabilities $q_1 = \alpha^{11}$ and $q_2 = \alpha$.
At an AI quality of about $\alpha = 0.9$ these deliver the success probabilities used there, namely $q_1 \approx 0.3$ and $q_2 = 0.9$.
Figure~\ref{fig:example5_nonmonotone} traces the optimal cost and its marginal benefit as AI quality $\alpha$ rises.

\begin{figure}[!t]
\caption{Non-monotonicity of Gains from Improving AI Quality}
\label{fig:example5_nonmonotone}
\centering
\begin{minipage}[t]{0.49\linewidth}
  \centering
  \subcaption{Total cost by AI strategy}
  \label{fig:example5_costs}
  \includegraphics[width=\linewidth]{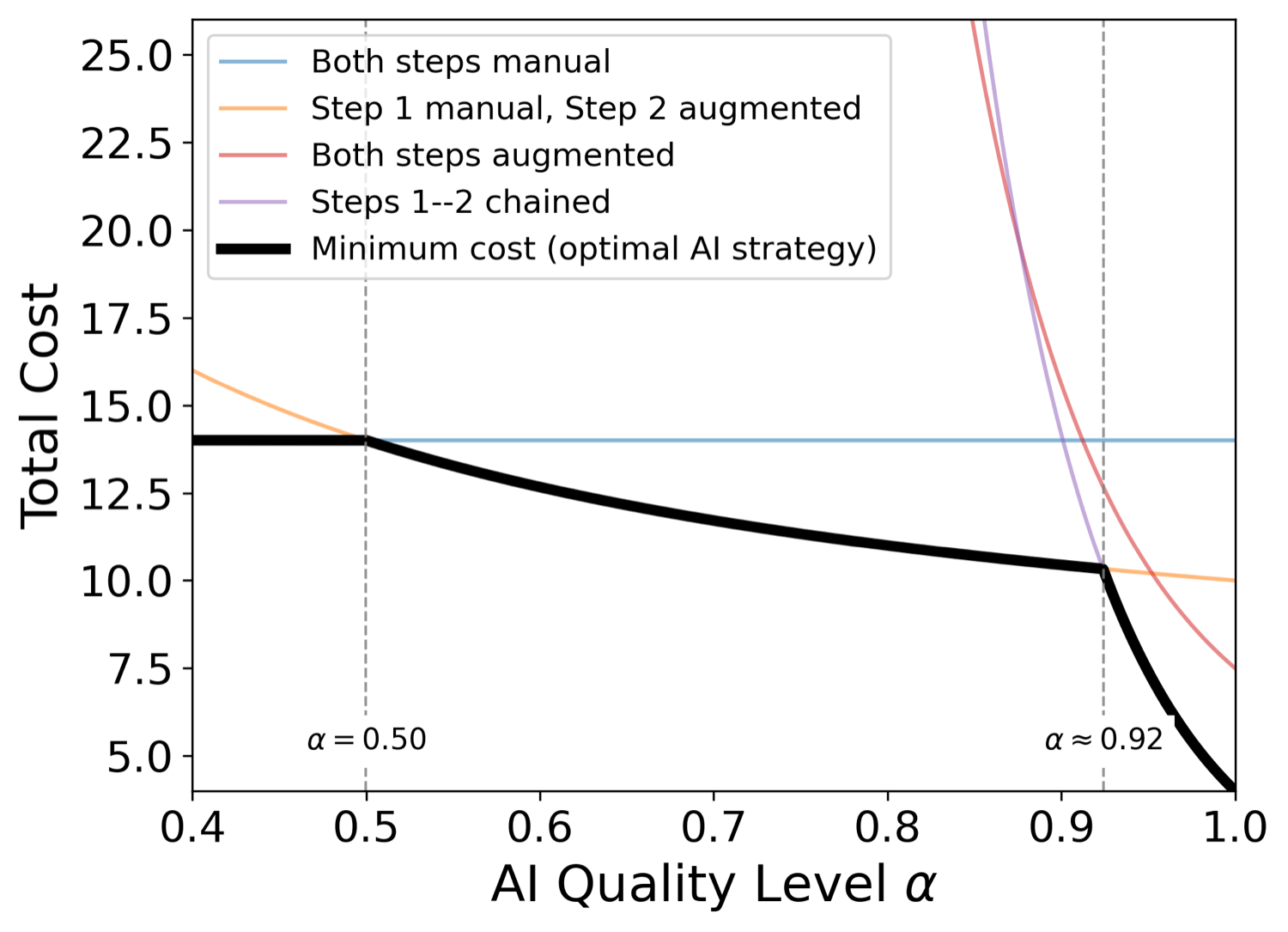}
\end{minipage}
\hfill
\begin{minipage}[t]{0.49\linewidth}
  \centering
  \subcaption{Marginal benefit of improving $\alpha$}
  \label{fig:example5_marginalBenefit}
  \includegraphics[width=\linewidth]{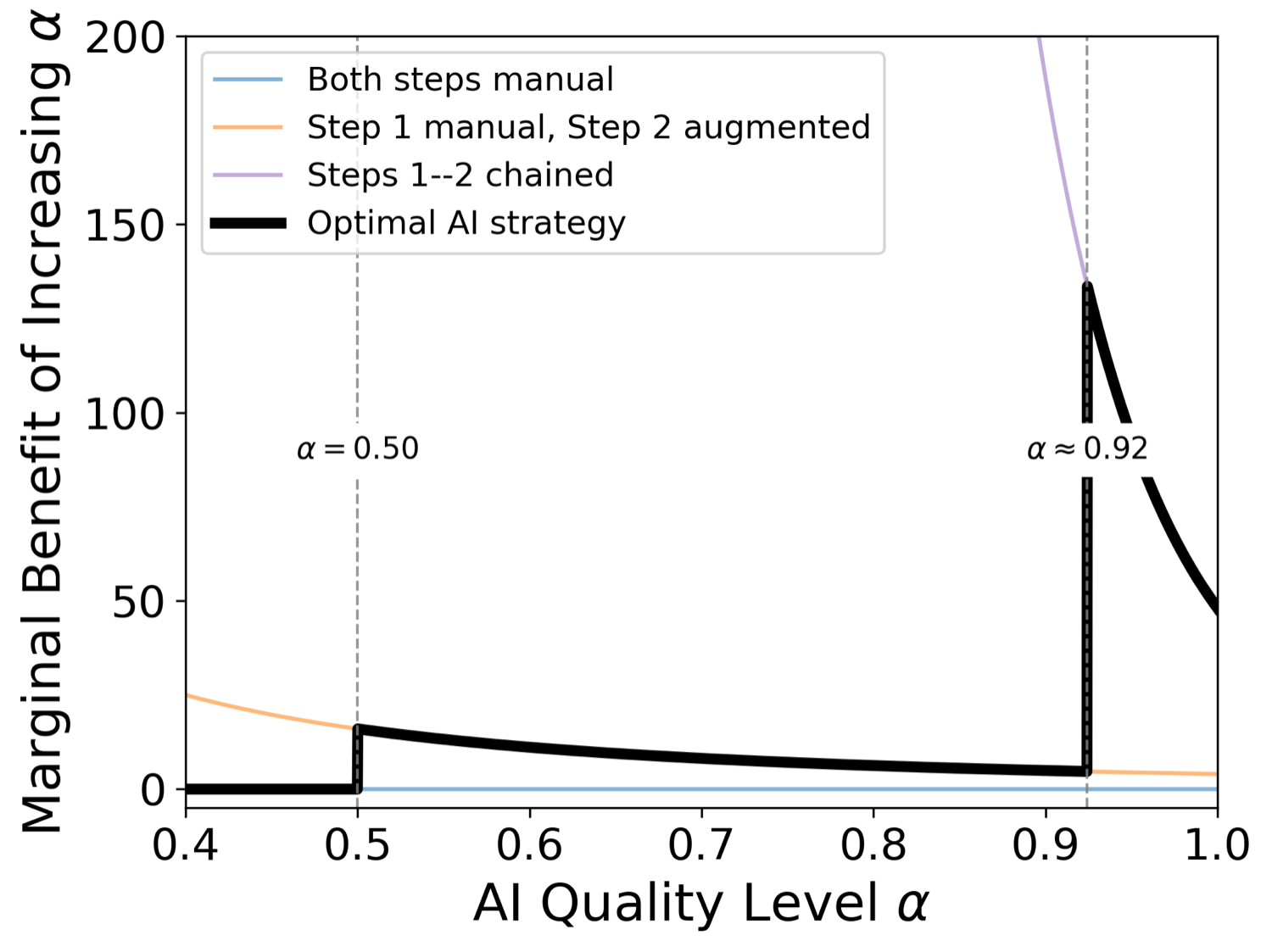}
\end{minipage}
\par\vspace{\fignotesskip}
\begin{minipage}{1\linewidth}
\begin{spacing}{0.2}
{\fontsize{9.5pt}{10pt}\selectfont Notes: This figure plots the two-step workflow of Example~\ref{ex:nonmonotone} as AI quality $\alpha$ rises. Panel (a) shows the cost of each AI strategy, with the optimal strategy given by their lower envelope (thick black line). Table~\ref{tab:nonmonotone_costs} in Appendix~\ref{app:tables_and_figures} gives each strategy's cost and the range of $\alpha$ over which it is optimal. Panel (b) shows the marginal benefit of improving AI quality, $-\mathrm{d}C/\mathrm{d}\alpha$. The marginal benefit jumps at each $\alpha$ where the optimal strategy switches, so gains from AI improvements are non-monotone in $\alpha$.}
\end{spacing}
\end{minipage}
\end{figure}

With $q_1 = \alpha^{11}$ and $q_2 = \alpha$, the firm chooses among five configurations, ranging from both steps manual to both chained.
Table~\ref{tab:nonmonotone_costs} in Appendix~\ref{app:tables_and_figures} collects them, together with the cost expression of each, the marginal benefit of improving AI quality, and the range of $\alpha$ over which the configuration is cost-minimizing.
Below $\alpha = 0.5$ the firm uses no AI; between $0.5$ and roughly $0.92$ it augments the easy step alone; above that it chains the two.

Panel~(b) shows that the marginal return to improving AI quality jumps up at each threshold where the optimal strategy switches, since a different strategy governs cost on either side of the threshold.
It is also worth noting that the jump at $\alpha \approx 0.92$ is larger than the one at $\alpha = 0.50$. Forming the chain pulls a second step under AI, so cost falls more steeply.
This ordering is specific to this example, however, since a later threshold does not necessarily introduce a longer chain than an earlier one.
\end{example}

We formalize the non-monotonicity in marginal returns to improving AI quality in Appendix~\ref{app:nonmonotone}.
There we show that the cost of any fixed AI strategy varies smoothly with $\alpha$, so that the firm's optimal cost is the lower envelope of finitely many smooth cost curves, and prove in Lemma~\ref{lem:reorg_threshold} that the marginal gains from improving AI diminish monotonically within each regime where a single strategy remains optimal, never fall at a reorganization threshold, and jump upward at every threshold at which the two adjacent cost curves cross non-tangentially, which is the generic case.

The non-monotone productivity gains described here are a familiar pattern for general-purpose technologies. 
The steam engine, electricity, and the computer each delivered only modest improvements in the early days of their adoption, and their full impact arrived once production had been reorganized around them, the shape \citet{brynjolfsson2021productivity} describe as the productivity J-curve.
What the model shows is that the same shape arises within a single workflow, because the returns to better AI wait on discrete reorganizations of it.
The reorganizations available here are nonetheless the narrowest kind, since the firm changes only which steps AI executes while job boundaries stay fixed.
Over the horizon on which AI quality improves, firms can redefine jobs, changing the skills each demands and the wages it commands, and it is that larger reorganization we study next.

\section{Job Design and Production in the Long Run} \label{sec:longrun}
Throughout Sections~\ref{sec:shortrun} and~\ref{sec:results} we held the set of steps assigned to the worker fixed, so that AI could improve only the time those steps took to execute.
Here we relax this assumption and allow the firm to restructure production fully around the capabilities of AI.

Over longer horizons the firm can readjust its organizational structure, redrawing the boundaries between jobs and with them the allocation of steps to workers, while optimizing its use of AI.
Moreover, adjustments in the labor market let the firm pay workers according to the actual skill requirements of the jobs it designs, effectively endogenizing wages.
Below we develop the firm's long-run production problem.

\subsection{Jobs and Wages}
\label{sec:longrun.jobdesign}

Just as an AI strategy partitions the steps into tasks, the firm's long-run choice partitions those tasks into jobs.
Definition~\ref{def:job_design} formalizes that second layer of bundling, which Figure~\ref{fig:hierarchical_production} adds to the example of Figure~\ref{fig:steps_tasks}.

\begin{definition}[Job Design]
\label{def:job_design}
A \emph{job} is a contiguous block of tasks $J = (T_b, \dotsc, T_{b'})$, with $b \leq b'$, assigned to a single worker.
A \emph{job design} is a partition $\mathcal{J} = (J_1, \dotsc, J_p)$ of the task sequence $\T$ into $p$ such jobs, one for each worker the firm hires.
\end{definition}
In the bottom two layers of Figure~\ref{fig:hierarchical_production}, Tasks~1 and~2 form Job~1, Task~3 forms Job~2, and Tasks~4 and~5 form Job~3.
Note that a worker's job can potentially contain both manual and AI chain tasks.

\begin{figure}[!t]
  \centering
  \caption{Illustrative Example of a Firm's Production in the Long Run}
  \label{fig:hierarchical_production}
  % Local override for draft_mert: step order updated (1 manual, 2 augmented, 3 manual,
% 4-5 automated, 6 augmented, 7 manual). The parent ../../../plots/TikZ_visualization/
% copy keeps the old order for the legacy writeup/main.tex and is left untouched.
% Hierarchical Production Visualization
\begin{adjustbox}{center,trim=0pt 0pt 0pt 0pt,clip,margin=-2cm 0cm}
  \begin{tikzpicture}[scale=0.73, transform shape,
    node distance=0.35cm and 0.91cm,
    every node/.style={rectangle, rounded corners, draw, align=center, minimum width=2cm, minimum height=1cm},
    manual/.style={fill=gray!10},
    automated/.style={fill=green!30},
    augmented/.style={fill=orange!30},
    dashedbox/.style={draw, rectangle, rounded corners, dashed, inner sep=0.25cm, line width=1.25pt},
    humanbox/.style={rectangle, rounded corners, draw, align=center, minimum width=2.25cm, minimum height=1cm, fill=olive!20}, % changed to olive
    aibox/.style={rectangle, rounded corners, draw, align=center, minimum width=2.25cm, minimum height=1cm, fill=blue!20},
    jobbox/.style={rectangle, rounded corners, draw, align=center, minimum width=2.25cm, minimum height=1cm, fill=cyan!20},
    industrybox/.style={rectangle, rounded corners, draw, align=center, minimum width=3cm, minimum height=1cm, fill=teal!60},
    >=latex
  ]

    % Top layer (Steps)
    \node[manual] (S1) {Step 1\\(Manual)};
    \node[augmented, right=of S1] (S2) {Step 2\\(Augmented)};
    \node[manual, right=of S2] (S3) {Step 3\\(Manual)};
    \node[automated, right=of S3] (S4) {Step 4\\(Automated)};
    \node[automated, right=of S4] (S5) {Step 5\\(Automated)};
    \node[augmented, right=of S5] (S6) {Step 6\\(Augmented)};
    \node[manual, right=of S6] (S7) {Step 7\\(Manual)};

    % Arrows between steps (specific arrows as curly lines)
    \draw[->, thick] (S1) -- (S2);
    \draw[->, thick, red, decorate, decoration={snake, amplitude=.5mm, segment length=3.05mm}] (S2) -- (S3);
    \draw[->, thick, red, decorate, decoration={snake, amplitude=.5mm, segment length=3.05mm}] (S3) -- (S4);
    \draw[->, thick] (S4) -- (S5);
    \draw[->, thick] (S5) -- (S6);
    \draw[->, thick] (S6) -- (S7);

    % Dashed boxes (top layer, colors updated)
    \node[dashedbox, draw=olive, fit=(S1)] (DS1) {};
    \node[dashedbox, draw=blue, fit=(S2)] (DS2) {};
    \node[dashedbox, draw=olive, fit=(S3)] (DS3) {};
    \node[dashedbox, draw=blue, fit=(S4)(S5)(S6)] (DS4-6) {};
    \node[dashedbox, draw=olive, fit=(S7)] (DS7) {};

    \path (S1.west) -- (S7.east) coordinate[midway] (CenterTop);

    % Middle layer (Tasks, human task color updated)
    \node[humanbox, below=2.4cm of CenterTop, xshift=-6cm] (T1) {Task 1\\(Human)};
    \node[aibox, right=of T1] (T2) {Task 2\\(AI Chain)};
    \node[humanbox, right=of T2] (T3) {Task 3\\(Human)};
    \node[aibox, right=of T3] (T4) {Task 4\\(AI Chain)};
    \node[humanbox, right=of T4] (T5) {Task 5\\(Human)};

    % Arrows between tasks (specific arrows as curly lines)
    \draw[->, thick] (T1) -- (T2);
    \draw[->, thick, red, decorate, decoration={snake, amplitude=.5mm, segment length=3.05mm}] (T2) -- (T3);
    \draw[->, thick, red, decorate, decoration={snake, amplitude=.5mm, segment length=3.05mm}] (T3) -- (T4);
    \draw[->, thick] (T4) -- (T5);

    % Dashed boxes (middle layer)
    \node[dashedbox, draw=cyan, fit=(T1)(T2)] (DT1-2) {};
    \node[dashedbox, draw=cyan, fit=(T3)] (DT3) {};
    \node[dashedbox, draw=cyan, fit=(T4)(T5)] (DT4-5) {};

    % Job layer nodes
    \node[jobbox, below=1.7cm of T3] (J2) {Job 2};
    \node[jobbox, left=of J2] (J1) {Job 1};
    \node[jobbox, right=of J2] (J3) {Job 3};

    % Arrows between jobs (specific arrows as curly lines)
    \draw[->, thick, red, decorate, decoration={snake, amplitude=.5mm, segment length=3.05mm}] (J1) -- (J2);
    \draw[->, thick, red, decorate, decoration={snake, amplitude=.5mm, segment length=3.05mm}] (J2) -- (J3);

    % Arrows connecting layers (unchanged)
    \draw[->, thick, cyan] (DT1-2.south) -- (J1.north);
    \draw[->, thick, cyan] (DT3.south) -- (J2.north);
    \draw[->, thick, cyan] (DT4-5.south) -- (J3.north);

    \draw[->, thick, olive] (DS1.south) -- (T1.north);
    \draw[->, thick, blue] (DS2.south) -- (T2.north);
    \draw[->, thick, olive] (DS3.south) -- (T3.north);
    \draw[->, thick, blue] (DS4-6.south) -- (T4.north);
    \draw[->, thick, olive] (DS7.south) -- (T5.north);

  \end{tikzpicture}
\end{adjustbox}

  \par\vspace{\fignotesskip}
\begin{minipage}{1\linewidth}
\begin{spacing}{0.2}
{\fontsize{9.5pt}{10pt}\selectfont Notes: The top two layers reproduce Figure~\ref{fig:steps_tasks}, with $m=7$ steps grouped by dashed boxes into the $n=5$ tasks below them.
This figure adds the third layer, in which cyan dashed boxes group tasks into $p=3$ jobs, each assigned to a separate worker.
Red curly arrows mark the hand-offs incurred at job boundaries, and can be traced from the bottom layer up to the tasks and steps they separate.}
\end{spacing}
\end{minipage}
\end{figure}

Completing a step demands not only time but also skill, by which we mean the body of knowledge and the specific know-how a worker must acquire in order to perform it.
The short run had no need of this notion, as we assumed workers had acquired the skills necessary to perform all activities in their job.
In the long run, however, the firm chooses which steps to assign to which worker, effectively changing the skill requirement of their jobs.
We thus introduce skills to capture this other work dimension.

Accordingly, we extend Definitions~\ref{def:manual_step}--\ref{def:ai_chain} to incorporate the skill requirement of steps as well, so that 
the skill and time requirement of performing step $i$ manually is $(\manualSkill{i},\, \manualTime{i})$, while verifying the output of one AI attempt at it costs $(\AIskill{i},\, \AItime{i})$.
The two values need not coincide, since doing the work oneself calls for knowing how to do it, whereas verifying what an AI returns may call for knowledge of a different kind.
As with time, a task inherits the skill of whichever mode executes it, so each task $T_b$ in the task sequence $\T = (T_1, \dotsc, T_n)$ carries a pair of skill and time requirements $(\skillcost{b}, \timecost{b})$, which is $(\manualSkill{i},\, \manualTime{i})$ if $T_b$ is the manual step $s_i$, and $\bigl(\AIskill{r},\, \AItime{r}/\prod_{i=\ell}^{r} q_i\bigr)$ if $T_b$ is an AI chain spanning $(s_\ell, \dotsc, s_r)$.
Both cost components of a chain's pair (for a single verification attempt) are therefore inherited from its augmented endpoint.
In particular, the chain demands only the skill required to verify the augmented step's output, and an automated step inside the chain demands no skill of anyone, since no human ever handles it directly.

The firm hires one worker for each job in its job design $\mathcal{J}$.
The worker for job $J_j$ completes all tasks associated with their job.
The total time needed to complete all tasks in job $J_j$ is $\sum_{T_b \in J_j} \timecost{b}$.
The total skill required to complete the tasks in job $J_j$ is $\sum_{T_b \in J_j} \skillcost{b}$.\footnote{
We assume that the skills different steps require are non-overlapping, and thus add up linearly in jobs.
Human capital requirements of steps do overlap in practice, but modeling that overlap would burden the results this section derives with complexities that obscure the main channels we aim to describe.
That said, the results rest on a considerably weaker version of this assumption, namely that each step demand some increment of skill no other step supplies, so that a worker who takes on more steps in their job must acquire more skill.
Any strictly increasing functional form delivers this; linear addition, which we adopt here, is simply the most tractable.
}
We then assume that a worker's wage is proportional to the skill level of their job.\footnote{
One motivation for this formulation is workers paying a human capital investment to acquire the skills necessary for a job.
This investment must be offset by wages, and we express skill levels in units of their corresponding requisite wage.
Rather than formalizing such a wage model here, we impose this as an assumption and provide a more detailed wage formulation in Appendix~\ref{sec:aggregation}.
}
That is,
\begin{align}
\label{eq:wage}
\text{Wage}_j = \sum_{T_b \in J_j} \skillcost{b}.
\end{align}

The firm must pay workers their wage per unit of time regardless of which tasks they are assigned to complete at any given moment.
%\footnote{
%This derivation implicitly assumes a common base wage rate; i.e., all wage differentiation is driven by the human capital cost of acquiring skills.
%More generally, different tasks $T_b$ within job $J_j$ could have different base wage rates $w_b$ that could be influenced, for example, by the supply of and demand for labor of the corresponding type.
%One special case, explicitly motivated and discussed in Appendix~\ref{sec:aggregation}, assumes manual tasks are executed by manual labor at base wage rate $w_{\manualLetter}$, and AI chains are executed by AI management labor at base wage rate $w_{\AIletter}$.
%The formulation in \eqref{eq:wage} implicitly normalizes these base wage rates to $1$, incorporating them into skill costs $\skillcost{b}$ for notational simplicity.
%We maintain this normalization throughout subsequent sections until we explicitly distinguish human labor and AI management labor base wage rates in Appendix~\ref{sec:aggregation}.
%\label{foot:wage}
%}
Taking into account both time and skill requirements, the total wage bill paid to a worker assigned to job $J_j$, per unit of output, is
\begin{align}
\label{eq:jobcost}
\text{WageBill}_j = \Biggl(\sum_{T_b \in J_j} \skillcost{b} \Biggr)\Biggl(\sum_{T_b \in J_j} \timecost{b} \Biggr).
\end{align}
Note what Equation~\eqref{eq:jobcost} implies for a job.
The worker is compensated for every skill the job demands over its whole duration, even though at any given moment they are performing a single task and drawing on only the portion of their knowledge that task requires.
The remainder of their human capital sits idle at that moment and is paid for regardless.

A university professor is a concrete example.
The job of a professor bundles research and teaching together with the small chores that come with them, and the wage it commands reflects the skill the bundle as a whole demands.
Wiping the board demands almost none of that skill, yet the minutes spent on it are paid at the professor's rate rather than at what the chore would command on its own.
For those minutes the expertise that makes the professor worth hiring sits unused, and the university pays for it all the same.

\subsection{Coordination and Hand-Off Cost}
\label{sec:longrun.handoff}

Absent any frictions, the firm would minimize production costs by specializing as far as it can.
Equation~\eqref{eq:jobcost} implies dividing a job in two splits its skill requirement between two workers while leaving the total time to be worked unchanged, so it strictly lowers the wage bill.
Carried to its conclusion, this logic would have the firm fully specialize work and assign every task to its own worker, each of whom acquires exactly the knowledge their task requires and none of whom is ever paid for skill they are not using.

In reality, work is not organized this way, because dividing production among many workers is costly.
Each division creates a point at which work must pass from one person to the next, and coordinating that intermediate output takes time, because whoever receives it has to be told what was done, why, and what state the output is now in, and whoever hands it over has to put it in a form that can be received.\footnote{\cite{lee2023task} formalize such interdependencies and construct an occupation-level index of coordination intensity from the O*NET dataset, which they relate to wages and to exposure to AI.}
We model this coordination friction with an additional \emph{hand-off} cost, in time, that a worker spends passing their work on to the next worker, incurred whenever $J_j$ is not the final job in the job design $\mathcal{J}$.

Because production in our model exhibits an order, we can be more specific about what a hand-off costs.
A hand-off happens at a particular point in production, and what it costs depends on which point.
Some places are natural seams, where the work has reached a self-contained intermediate output, a finished design or a signed-off specification, that can be described briefly and passed on.
Others fall in the middle of something, where whoever receives the work has to reconstruct the context the previous worker had built up.
We therefore attach the hand-off cost to the step at which the break occurs rather than to the worker or to the number of jobs, so that \emph{where} the firm divides the workflow matters and not merely how often it divides.
The red curly arrows in Figure~\ref{fig:hierarchical_production} illustrate that hand-offs occur at job boundaries and that their cost reflects the characteristics of the last step in a worker's job.

Formally, given step $s_i$, we write $\handofftime{i}$ for the additional hand-off time spent by a worker for whom the last step of their job is $s_i$, and $\handofftime{}(J_j)$ for the hand-off time of job $J_j$, which is equal to $\handofftime{i}$ if $s_i$ is the final step in $J_j$.
The final job hands off to no one, so $\handofftime{m} = 0$.
Conditional on where the hand-off occurs in the production process, its cost does not otherwise depend on previous hand-off events.
Thus the total time needed to complete job $J_j$, including hand-off costs, is
\[
\handofftime{}(J_j) + \sum_{T_b \in J_j} \timecost{b}.
\]

Because the worker is paid for as long as their job takes, this time enters the wage bill exactly as execution time does, so Equation~\eqref{eq:jobcost} becomes
\begin{align}
\label{eq:wagebill}
\text{WageBill}_j = \Biggl(\sum_{T_b \in J_j} \skillcost{b} \Biggr) \Biggl(\handofftime{}(J_j) + \sum_{T_b \in J_j} \timecost{b} \Biggr).
\end{align}
Equation~\eqref{eq:wagebill} highlights the two opposing forces that shape a job's boundaries.
Adding more tasks to a worker's job increases the cumulative skill requirement and thus the worker's wage.
However, if tasks are kept separate, workers must incur hand-off costs at each job boundary.
These two effects jointly determine the optimal degree of task aggregation.

Going back to the university professor example, this is perhaps why the chores fall to the professor in the first place. On wages alone the university should hand the cleaning of the board to a cheaper worker. What stops it is that the work would have to change hands mid-lecture, which means the professor must find someone, the lecture waits, and the worker must be told what is wanted. That costs more than the professorial minutes it saves, so the professor's job absorbs the chore.

We explore this trade-off in greater detail in Section~\ref{sec:longrun.specialization}, after introducing the problem the firm solves in choosing its jobs.

\subsection{Long-Run Production}
\label{sec:longrun.production}

Given the sequence of steps $\mathcal{S} = (s_1, \dotsc, s_m)$, the firm can design both the partition $\T$ of steps into tasks and the partition $\mathcal{J}$ of tasks into jobs.
The choice of $\T$ determines the time and skill requirements of each task.
Given $\T$, the choice of $\mathcal{J}$ then determines worker wage rates and time needed per unit of output (including hand-off costs).
Formally, we can write $\mathcal{P}(X)$ for the set of partitions of a sequence $X$ into contiguous subsequences.
Then the full optimization problem faced by the firm can be expressed as
\begin{align}
\label{eq:totalcost_with_handoff}
\min_{\T \in \mathcal{P}(\mathcal{S})} \
\min_{\mathcal{J} \in \mathcal{P}(\T)} \
\text{TotalCost}(\mathcal{J}; \T)
= \sum_{J_j \in \mathcal{J}} \text{WageBill}_j =
\sum_{J_j \in \mathcal{J}} \Biggl[\Biggl(\sum_{T_b \in J_j} \skillcost{b} \Biggr) \Biggl(\handofftime{}(J_j) + \sum_{T_b \in J_j} \timecost{b} \Biggr)\Biggr]
\end{align}
where recall that, for each $T_b \in \T$, $\timecost{b}$ is as described in Section~\ref{sec:model.tasks} and $\skillcost{b}$ as above, both depending on the selected mode of operation for individual steps.

We refer to \eqref{eq:totalcost_with_handoff} as the firm's long-run optimization problem, as it anticipates the adjustment of worker wages to fit the skill requirements of each job and sets job responsibilities accordingly.
It nests the short-run problem in \eqref{eq:totalcost_shortterm} of Section~\ref{sec:shortrun.production} as the case of a single job spanning the whole production, with wage being fixed rather than chosen.
%In Section~\ref{sec:extensions.computation} we discuss how the solution to Problem~\eqref{eq:totalcost_with_handoff} can be approximated to any desired accuracy in polynomial time.

Next we discuss the implications of this optimization problem.

\subsection{AI and the Limits of Worker Specialization}
\label{sec:longrun.specialization}

Here we discuss how AI affects the boundaries of jobs.
Because the firm chooses those boundaries when it solves Problem~\eqref{eq:totalcost_with_handoff}, the answer runs through the trade-off that fixes them, between the skill a broad job demands and the coordination a narrow one requires.
We set that trade-off out first, and then ask how AI moves it.

The wage bill in \eqref{eq:wagebill} is a product of a skill and a time, so it reads naturally as an area, which is what makes the trade-off visible.
Draw each task as a rectangle whose height is its skill cost $\skillcost{b}$ and whose width its time cost $\timecost{b}$; a job's wage bill is then the area of the rectangle that spans the tasks it contains, as tall as the total skill it demands and as wide as the total time it takes.
Figure~\ref{fig:lr_job_design} draws the two-task case, the smallest in which the firm has a choice to make.
\begin{figure}[!t]
    \centering
    \caption{The Trade-off between Worker Specialization and Coordination}
    \label{fig:lr_job_design}
    \includegraphics[width=0.92\linewidth]{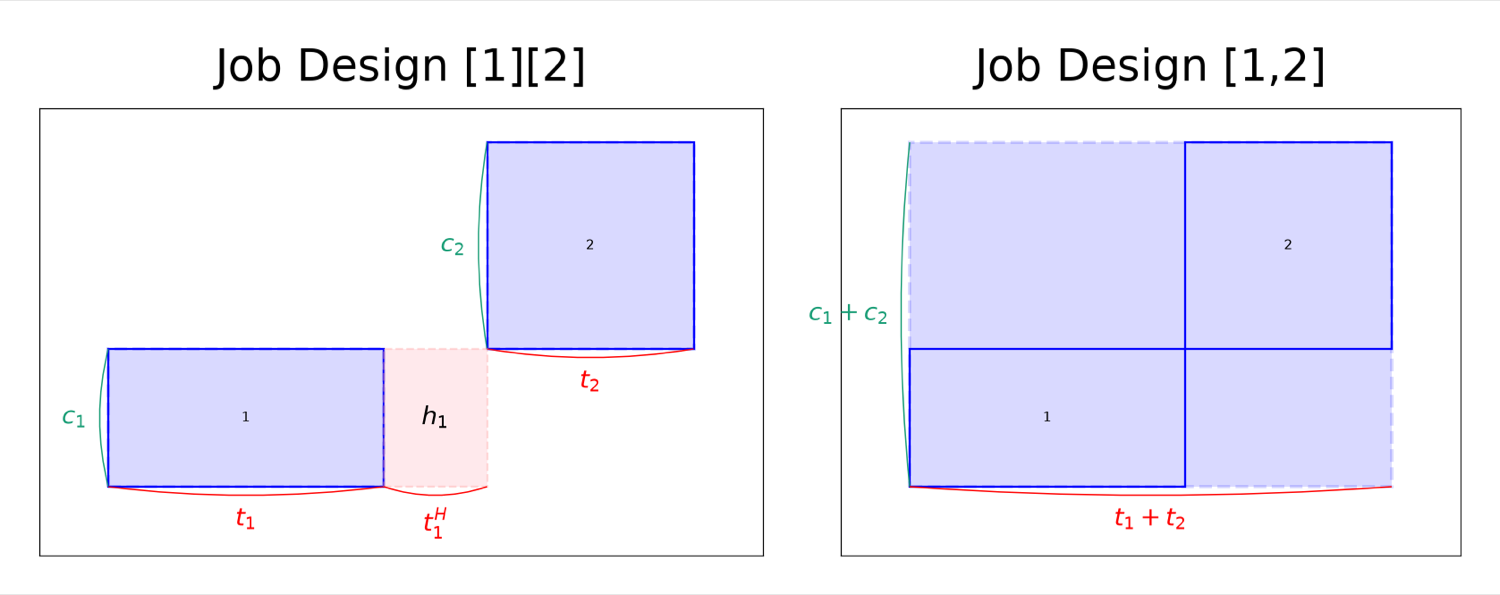}
    \par\vspace{\fignotesskip}
\begin{minipage}{1\linewidth}
\begin{spacing}{0.2}
{\fontsize{9.5pt}{10pt}\selectfont Notes: The blue bounded rectangles are tasks, with height the task's skill requirement ($\skillcost{}$) and width its time requirement ($\timecost{}$).
The shaded areas are the wage bills of jobs.
In the left panel the two tasks are assigned to two specialized workers, and the pink rectangle is the hand-off cost incurred between them.
In the right panel they are bundled into a single job, which eliminates the hand-off but requires a worker skilled enough to perform both tasks and paid at the corresponding rate throughout.}
\end{spacing}
\end{minipage}
\end{figure}
In the left panel the tasks go to two specialized workers, each of whom acquires only the skill their own task needs, so no one is paid for knowledge they are not using; but the work must pass between them, and the pink rectangle is what that hand-off costs.
Notably, the hand-off shares the skill requirements of the worker doing it for an extra amount of time added to the last task in their job.
In the right panel one worker performs both tasks, which removes the hand-off but obliges that worker to hold both skills throughout, the dead area sitting beside whichever task they are not currently performing.
Bundling therefore buys the removal of a hand-off at the price of idle skill, and where those two prices balance determines the firm's degree of specialization.

With more than two tasks the optimum need not sit at either extreme.
Consider a workflow that, under a fixed AI deployment strategy $\T$, has $n=3$ tasks with the cost parameters
\[ \bigl(\skillcost{b},\, \timecost{b},\, \handofftime{b}\bigr)_{b=1,2,3} = (3,\,1,\,3),\quad (1,\,2,\,0.5),\quad (2,\,2,\,0). \]
With three tasks there are four ways the firm can design jobs,\footnote{
More generally, a production sequence with $n$ tasks has $2^{(n-1)}$ unique job designs.
} and Figure~\ref{fig:job_design} draws all four in the same terms, once with hand-off costs switched off and once with them on; Table~\ref{tab:job_design} in Appendix~\ref{app:tables_and_figures} reports the cost of each job design.
\begin{figure}[!t]
\begin{center}
\caption{Geometric Interpretation of Optimal Job Design} \label{fig:job_design}
\begin{subfigure}[b]{0.49\textwidth}
    \caption{Job Design Without Hand-off Costs}
    \includegraphics[width=\textwidth]{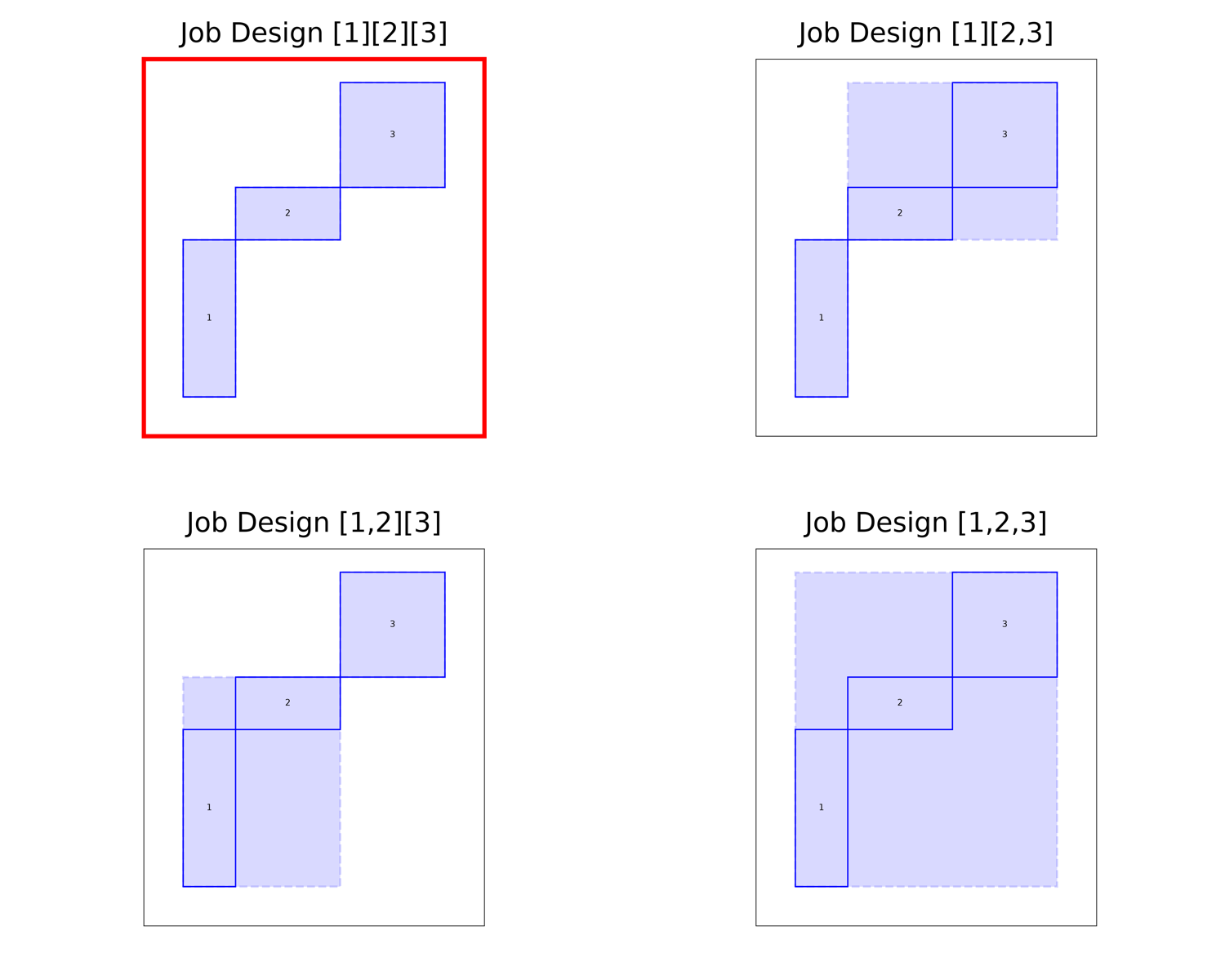}
\end{subfigure}
\begin{subfigure}[b]{0.49\textwidth}
    \caption{Job Design With Hand-off Costs}
    \includegraphics[width=\textwidth]{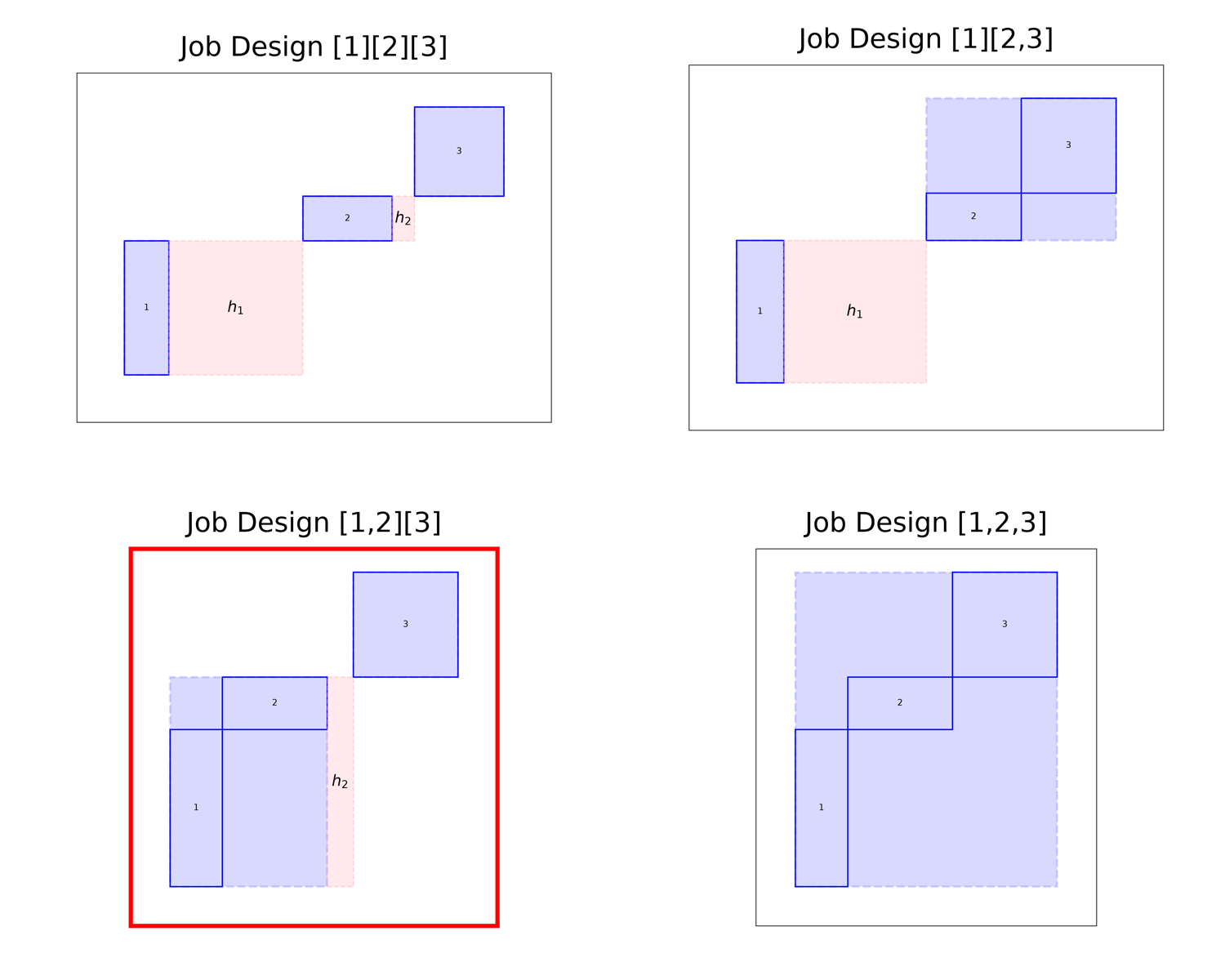}
\end{subfigure}
\end{center}
\par\vspace{\fignotesskip}
\begin{minipage}{1\linewidth}
\begin{spacing}{0.2}
{\fontsize{9.5pt}{10pt}\selectfont Notes: Each panel draws all four job designs of the three-task workflow in the convention of Figure~\ref{fig:lr_job_design}, with the area of each job's bounding box its wage bill. The cost-minimizing design is outlined in red. Panel~(a) sets hand-off costs to zero, in which case full specialization is optimal. Panel~(b) introduces them, and the optimum moves to $[1,2][3]$. The skill, time, and hand-off costs of the three tasks are given in the text, and Table~\ref{tab:job_design} in Appendix~\ref{app:tables_and_figures} reports the cost of each design.}
\end{spacing}
\end{minipage}
\end{figure}
Panel~(a) sets hand-off costs to zero and recovers the frictionless benchmark. The cost-minimizing design, outlined in red, is full specialization, with each task its own separate job.
Panel~(b) reinstates hand-offs and the optimum moves to $[1,2][3]$, in which tasks 1 and 2 share a job and task 3 stands alone.
The reason is the large hand-off time $\handofftime{1}=3$ that a boundary between tasks 1 and 2 would trigger, since bundling them avoids that hand-off at the price of a worker who must hold both skills and be paid for both throughout.

This is the machinery AI acts on, and it does so through two channels.
The first is skill. An AI chain leaves the worker holding only the skill needed to verify its endpoint, so the tasks AI absorbs become short rectangles in the picture above.
That cuts the idle skill a bundled job pays for, which makes broader jobs cheaper to sustain and pushes the firm toward consolidation.
But the same reduction works in the opposite direction as well, because what a hand-off costs is also a skill multiplied by a time. A worker whose job AI has made less skill-intensive is cheaper to leave idle at a boundary, so the hand-off the firm was bundling to avoid becomes less worth avoiding.
In the workflow above, for instance, tasks 1 and 2 share a job because a skill-$3$ worker would otherwise sit through a hand-off of length $3$; were AI to take over task 1 and leave only verification skill behind, that reason would weaken along with the wage rate it applies to.

The second channel is time.
For many steps, checking what AI produces takes less time (in expectation) than performing them, so the tasks AI absorbs become narrow rectangles as well as short ones.
Narrower tasks cut the idle skill a bundled job pays for, since what bundling wastes is one task's skill sitting through another task's time, and that pushes toward consolidation as the skill reduction did.
The hand-off answers to the same force.
We have treated $\handofftime{i}$ as a primitive to capture the tacit and social knowledge of transferring work, but AI plausibly lowers it as well, by absorbing the peripheral coordination a transfer requires, such as summarizing what was done or drafting the documentation that accompanies it.
Cheaper hand-offs make every boundary cheaper to draw, which pushes toward narrower jobs and more specialization.
Panel~(a) of Figure~\ref{fig:job_design} is the limiting case, with hand-offs free and full specialization optimal.
Time therefore divides as skill does, in that a shorter task makes bundling cheaper, a shorter hand-off makes splitting cheaper, and the two work against each other.

What this framework delivers is therefore not a prediction about job breadth but an account of what determines it.
AI presses on both sides of the specialization trade-off at once, and which way jobs move depends on the relative size of the skill and time it strips out and the coordination it leaves behind.
That is an empirical question, and one the framework poses in a form that says what would have to be measured to answer it.

\section{Extensions and Other Results} \label{sec:extensions}
This section briefly provides two additional results that the framework delivers, but defers their full derivations and proofs to the appendix.
Section~\ref{sec:extensions.aggregation} shows that, under additional assumptions, the firms' individual Leontief production functions can be aggregated into a single CES production function at the economy level.
This connects our firm-level, sequential model of production to models that study the implications of automation from a macro perspective.
Section~\ref{sec:extensions.computation} discusses how the cost minimization problems we have posed for the firm's short- and long-run optimization can actually be solved.

\subsection{CES Representation}
\label{sec:extensions.aggregation}

So far we have focused on a single firm's production: which steps it chains together, where it draws the boundaries of jobs, and what those jobs pay workers.
These adjustments do not stay inside the firm.
Repeated across firms, they are what determines how much labor the economy employs, how output responds to capital and to labor, and what becomes of both as AI quality improves.
These are the terms the automation literature works in, writing production functions directly at the level of the aggregate economy \citep{acemoglu2018, acemoglu2022tasks, acemoglu2025simple}.

Our framework can be expressed in the same terms.
We distinguish two types of labor: \emph{manual labor}, which performs the manual portion of work, and \emph{AI management labor}, by which we mean the labor that verifies the AI's output at the end of a chain.\footnote{
The distinction between these two types of labor is deliberate rather than necessary.
One could carry the aggregation through with a single type of labor, pooling the two requirements into one input.
We keep them apart because the model already distinguishes the skill a step demands when a worker performs it, $\manualSkill{i}$, from the skill it demands when a worker verifies the output of an AI attempt at it, $\AIskill{i}$, so two types of labor are the natural counterpart of a distinction we have already drawn.
}
Appendix~\ref{sec:aggregation} shows that the firm-level technology developed earlier, in which the content of tasks and of jobs is chosen rather than given, can be aggregated to
\begin{align}
\label{eq:lr_macro_ces}
Y = \left(\theta_{\AIletter}\,\aggregateAIlabor^{\rho} + \theta_{\manualLetter}\,\aggregateManualLabor^{\rho} + (1-\theta_{\AIletter}-\theta_{\manualLetter})\,K^{\rho}\right)^{\frac{1}{\rho}},
\end{align}
over economy-wide AI management labor $\aggregateAIlabor$, manual labor $\aggregateManualLabor$, and capital $K$, with $Y$ denoting aggregate output and $\rho < 0$ assumed, so that the elasticity of substitution $\sigma = 1/(1-\rho)$ lies below one.
The aggregation supports Equation~\eqref{eq:lr_macro_ces} for any such $\rho$ rather than pinning a particular one.
This is the familiar constant elasticity of substitution (CES) form \citep{deloecker2021io}, with which the macroeconomic implications of automation are usually studied.

Arriving at Equation~\eqref{eq:lr_macro_ces} is not straightforward, and it is worth stating what makes it work.
No individual firm in our model substitutes AI management labor for manual labor, since producing a unit of output requires completing smaller units of production in fixed proportions, so each firm's technology is Leontief.
Aggregate substitution comes instead from composition, in the manner of \citet{houthakker1955pareto} and \citet{levhari1968note}.
Firms have access to the same AI technology but differ in how effectively they are able to deploy it, and when AI management labor becomes relatively cheaper the economy's input mix moves because the set of firms for which using AI is worthwhile has moved, not because any firm has changed its own mix.
Appendix~\ref{sec:aggregation} obtains that dispersion from the order in which a firm has to commit, settling on its AI strategy and job design before it learns how effectively it can put AI to use, and derives in closed form the distribution of output across realized AI effectiveness under which Equation~\eqref{eq:lr_macro_ces} holds.\footnote{
The aggregate representation asks more of the economy than the firm-level argument does: capital productivity is common across firms, every firm organizes identically, realized effective AI quality is the single dimension along which firms differ, and the CES weights are tied to the skill-adjusted time requirements implied by the firm's chosen AI strategy and job design. See Appendix~\ref{sec:aggregation} for the details.
}

The micro and macro representations of production are therefore not in competition.
Our account of how chaining redraws tasks and jobs is not an alternative to the aggregate view of automation but a foundation for it, and the implications people draw from improving AI quality in a CES economy carry over along the margin the representation covers, the substitution between AI management labor and the rest of production.
Importantly, the organization of work does not disappear in aggregation; it is absorbed into the parameters of the CES function.
However, because that absorption is not invertible, aggregate data cannot reveal workflow structure, so objects such as the fragmentation index must be measured where the arrangement of work is still visible, which is what our empirical exercise in Section~\ref{sec:empirics} does.

\subsection{Computation of Firm's Short- and Long-run Optimums}
\label{sec:extensions.computation}

Throughout the paper we have assumed the firm can obtain its cost-minimizing arrangement without discussing how it is found.
This is not computationally trivial, as chaining prevents the firm from deciding each step's mode in isolation. 
Rather, the choice ranges over all feasible arrangements, of which there are exponentially many.

That exponential count overstates the difficulty, however, because many of the configurations overlap heavily.
Any two arrangements that treat a fixed stretch of the workflow alike cost the same over that stretch, so the firm needs to calculate its cost just once.
Both the short- and long-run problems can therefore be formulated recursively and solved by dynamic programming, building the workflow up from the beginning one step at a time and reusing the work spent on an early stretch for every arrangement of what follows.
The two optimization problems differ only in how much a settled stretch has to remember about itself.
In the short run the configuration's time cost is all that matters, and the optimum can be calculated exactly.
In the long run the configuration's cost carries not just the time but also the skill of the worker whose job is still not terminated, and exactness gives way to an approximation the firm can make as tight as it likes.
We take the cost minimization of the two horizons in turn, collecting the proofs in Appendix~\ref{app:omitted_proofs}.

\paragraph{The Short Run Optimization.}
With the job and its wage fixed, the firm minimizes total expected time in Problem~\eqref{eq:totalcost_shortterm}.
For $k \leq m$, let $C[k]$ denote the minimum time needed to complete a hypothetical workflow containing only steps $1$ through $k$, so that $C[0] = 0$ and $C[m]$ is the optimum.
Step $k$ is either performed manually, or augmented as the endpoint of an AI chain reaching back to some earlier step $\ell$; in either case the steps before that point are optimized separately, which gives
\begin{align}
\label{eq:shortrun_recursion}
C[k] = \min\Biggl\{ \underbrace{C[k-1] + \manualTime{k}}_{\text{step $k$ manual}}, \ \underbrace{\min_{\ell < k}\ \Bigl( C[\ell] + \frac{\AItime{k}}{\prod_{i=\ell+1}^{k} q_i} \Bigr)}_{\text{step $k$ ends a chain begun at } \ell+1} \Biggr\}.
\end{align}
Because time costs enter Problem~\eqref{eq:totalcost_shortterm} additively, a stretch's time cost is all that the recursion in \eqref{eq:shortrun_recursion} needs from it.
Each of the $m$ values costs $O(m)$ to evaluate, since the chain ending at step $k$ may begin anywhere before it.
This makes the computational complexity of the algorithm $O(m^2)$.
Proposition~\ref{prop:shortrun_dp} formalizes this.
\begin{proposition}[Short-run AI Deployment]
\label{prop:shortrun_dp}
Given a workflow with $m$ steps, the AI strategy that solves Problem~\eqref{eq:totalcost_shortterm} can be calculated in time $O(m^2)$ via dynamic programming.
\end{proposition}

\paragraph{The Long Run Optimization.}
The firm now chooses job boundaries as well, and the objective in \eqref{eq:totalcost_with_handoff} multiplies each job's total skill by its total time.
A task's contribution to cost therefore depends on how much skill and how much time its job already carries, so the position alone no longer summarizes a stretch and we carry two further coordinates.
Let $R(i,\, \skillcost{},\, \timecost{})$ be the minimum cost of completing steps $1$ through $i$ when the worker currently accumulating tasks already holds total skill $\skillcost{}$ and total time $\timecost{}$ (including their hand-off) and let $V(i)$ be the corresponding cost when that worker's job closes at step $i$, so that $V(i) = R(i,\, 0,\, \handofftime{i})$ with the job closure option removed.
The base case is $R(0,\, \skillcost{},\, \timecost{}) = \skillcost{}\,\timecost{}$, and the firm faces three options at each step rather than two,
\begin{align}
\label{eq:longrun_recursion}
R(i,\, \skillcost{},\, \timecost{}) = \min\Biggl\{ \underbrace{\skillcost{}\,\timecost{} + V(i)}_{\text{close the job}}, \ \underbrace{R\bigl(i-1,\, \skillcost{} + \manualSkill{i},\, \timecost{} + \manualTime{i}\bigr)}_{\text{step $i$ manual}}, \ \underbrace{\min_{0 \le r < i}\ R\Bigl(r,\, \skillcost{} + \AIskill{i},\, \timecost{} + \tfrac{\AItime{i}}{\prod_{i'=r+1}^{i} q_{i'}}\Bigr)}_{\text{step $i$ ends a chain begun at } r+1} \Biggr\},
\end{align}
with the optimum at $V(m)$, recalling that the final job hands off at no cost.\footnote{
The first option is available only when the active worker already holds a task, $\skillcost{} > 0$; the states with $\skillcost{} = 0$ are the job-opening states, and taking their value to be $V(i)$ is what keeps the recursion well founded.
}
Closing a job is where the accumulated skill and time are finally multiplied into a wage bill, which is why the recursion has to carry both alongside the position $i$.
That is also what makes the long run more complicated to solve.

Skill and time costs are continuous, so the recursion must distinguish between infinitely many situations rather than the $m$ of the short run, and it cannot be carried out to completion.
Proposition~\ref{prop:totalcost_optimization_dp} states that the optimal cost can nonetheless be approximated up to an error term that can be made arbitrarily small in polynomial time.
\begin{proposition}[Long-run AI Deployment and Job Design]
\label{prop:totalcost_optimization_dp}
Fix a workflow with $m$ steps whose skill and time costs, and whose non-zero hand-off costs, lie in $[1/B, B]$ for some $B > 0$.
For any $\epsilon > 0$, a pair of AI strategy $\T$ and job design $\J$ whose long-run cost in \eqref{eq:totalcost_with_handoff} is within a factor $(1+\epsilon)$ of the minimum can be computed in time $O(m^4 \epsilon^{-2} \log^2(mB))$ via dynamic programming.
\end{proposition}

Together these results also suggest a use for the framework beyond what we do here.
Paired with detailed measurement of the time and skill content of individual tasks, the algorithms above could be run on real workflows to ask how a given advance in AI capability would redraw them.

\section{Empirical Evaluation} \label{sec:empirics}
In this section we empirically test the implications of our model.
The predictions we take to the data are those of the short-run problem of Section~\ref{sec:shortrun}, in which the boundaries of jobs are fixed and AI acts only on how the steps within them are executed.
This is a deliberate restriction rather than a limitation of the theory.
The long-run margins of Section~\ref{sec:longrun} require firms to reorganize and workers to retrain, and not enough time has passed since AI tools became widely available for these adjustments to appear in the data we observe.
What we can observe is where AI is being deployed inside workflows as they are currently arranged, which is precisely what the short-run problem speaks to.
The three predictions we test are: 
(1) AI-executed steps tend to appear next to each other in the production sequence, forming AI chains; 
(2) a step is more likely to be executed by AI when its neighbors are AI-executed; and 
(3) workflows in which AI-exposed steps are more dispersed have fewer of those steps actually executed by AI.

Our analysis requires three attributes of each step: its AI exposure status (exposed or unexposed to AI), which is how we measure whether a step is AI-able in the sense of Section~\ref{sec:results}, its realized mode of execution (manual or AI-executed), and its position in the production sequence.
We draw on five sources:
\begin{enumerate}[itemsep=0pt, parsep=0pt]
    \item The O*NET 27.3 Database \citep{onet27_3},
    \item Human-generated AI exposure labels for O*NET tasks from \cite{eloundou2023gpts},
    \item Realized AI execution labels for O*NET tasks from the Anthropic Economic Index \citep{handa2025economictasksperformedai},
    \item A GPT-generated task ordering for each O*NET occupation,
    \item The APQC Process Classification Framework (PCF) \citep{apqc_pcf}.
\end{enumerate}
The first four build our O*NET sample, in which a workflow is an occupation, and the fifth builds our PCF sample, in which a workflow is a process group.
All sources are publicly available datasets apart from the O*NET task ordering, which we generate ourselves.
Below we briefly discuss both samples before testing the predictions. 
Appendices~\ref{app:data_prep} and~\ref{app:external_validation} respectively give the construction details of the O*NET and PCF samples.

The O*NET database is the standard decomposition of work in the United States.
It breaks each of more than 850 occupations into a list of constituent tasks, roughly 18,000 in total, and has served as the backbone of a large empirical literature on the task content of jobs and their exposure to new technologies \citep[e.g.,][]{autor2013putting, deming2017growing, brynjolfsson2018can, webb2020impact, felten2021occupational}.
\cite{eloundou2023gpts} map these O*NET tasks to AI exposure by developing a rubric asking whether access to an LLM would reduce the time required to perform a task by at least half while preserving output quality, and applying it to every task using both human annotators and GPT-4.
We use their human-generated labels and treat tasks in their E1 and E2 categories as exposed to AI and all other tasks as unexposed.
Whereas these labels capture what AI could in principle do, the Anthropic Economic Index \citep{handa2025economictasksperformedai} captures what AI actually does.
It maps millions of conversations between users and Claude, Anthropic's AI chatbot, to O*NET tasks and classifies each conversation by whether the interaction augments or automates the task, providing a task-level measure of realized AI usage.\footnote{One feature of the data shapes how we use it throughout. The model's automated versus augmented distinction is defined by production sequencing, whereas the Anthropic labels are assigned on task-level feasibility and are agnostic to workflow position, so the data occasionally shows an automated task followed by a manual one, which the model's definition of a chain does not allow. We therefore do not distinguish the two labels in what follows and treat both as indicating AI execution. Appendix~\ref{app:data_prep} documents the discrepancy and provides examples.}

The PCF catalogues business processes such as procurement, order management, after-sales service, inventory, and regulatory compliance, with APQC's member working groups recording the activities that compose each process in the order they are performed.
Pooling APQC's Cross-Industry framework with its industry-specific ones provides $13{,}482$ steps in $525$ unique process groups.
Since AI exposure and execution are defined on O*NET tasks, we carry those labels onto PCF steps by matching each PCF step to its nearest O*NET task using natural language embeddings, preserving unmatched steps in the workflows and coding them as neither AI-exposed nor AI-executed.
The lossy label transfer leaves $12\%$ ($4\%$) of PCF steps AI-exposed (AI-executed), against $44\%$ ($13\%$) in O*NET.

The two corpora should be viewed as complements in our analysis.
The PCF's workflow orderings are recorded by practitioners and are thus native to the dataset.
However, the PCF covers only operational business processes, and its AI exposure and execution content is thin because of the lossy PCF--O*NET match.
The O*NET corpus, on the other hand, has broader coverage of economic activity but lists an occupation's tasks without providing the order in which they are performed.
To make the O*NET dataset suitable for our analysis, we construct workflow orderings for it by prompting GPT-5-mini \citep[through Expected Parrot;][]{Horton2024EDSL} with each occupation's complete set of tasks and asking for the most reasonable order in which they would be performed in practice.

Because our O*NET results rest on a task ordering that we generate with a language model we probe their validity using three robustness tests.
First, we show that the task orderings, and by extension our conclusions, remain unchanged across ten deliberately varied alternative prompts (Appendix~\ref{app:gptPrompts_robustness}).
Second, we test our task sequencing procedure against the PCF corpus as well as eight process-mining event logs from the 4TU.ResearchData database \citep{roadfines_log, sepsis_log, bpi2017_log, bpi2020_log} whose workflow orderings are known and show that our approach restores the correct workflow order to a substantial degree (Appendix~\ref{app:external_validation}).
Third, we restrict the O*NET sample to progressively more frequently performed tasks by dropping rarely performed ``filler'' steps that can separate steps in an actual workflow and find that our conclusions remain unchanged (Appendix~\ref{app:frequency_robustness}).
Together, findings of these tests imply that our O*NET results are unlikely to be artifacts of the prompt, the task sequencing procedure, or the way O*NET catalogues tasks.

Finally, the structure of our theoretical framework closely aligns with the structure of both datasets.
Each production step in the model corresponds to a PCF step or an O*NET task, and each workflow corresponds to the model's step sequence.
In the absence of AI, a step (that is, a PCF step or an O*NET task) is executed manually and maps directly into a task under our model's definition.
Thus, without AI, a PCF step or an O*NET task can be viewed as both a step and a task, which matches how both datasets list their units of work.
Under AI execution, however, multiple steps (corresponding to a subset of PCF steps or O*NET tasks) may be chained together and executed jointly, forming a composite task as defined in Section~\ref{sec:shortrun}.
In the rest of this section, we use the terms step, PCF step, O*NET task, and task interchangeably and expect the reader to keep these distinctions in mind.
We similarly use the term workflow for both an O*NET occupation and a PCF process group without restating the mapping each time.

\subsection{Prediction \#1: Tendency to Have Runs of Consecutive AI-executed Tasks}
\label{sec:chainLength_prediction}

The first prediction of the model that we test is that AI-executed steps appear in contiguous blocks rather than being randomly scattered in the workflow.
This prediction follows directly from Definition~\ref{def:ai_chain} on how AI chains are formed in the model.
If firms deploy AI, chaining can lower costs because the worker verifies the output once only at the end of the chain rather than after every step.

We begin with the O*NET sample.
We compute the observed average AI chain length and benchmark it against two randomization placebos.
First, we permute task positions within each occupation's workflow.
Second, we reassign tasks, each carrying its AI-execution label, randomly across occupations within the same major SOC occupation group.
The two placebos isolate different margins: the first holds each occupation's task composition fixed and randomizes the ordering in its workflow, while the second holds each major occupation group's AI intensity and each occupation's size fixed and randomizes which tasks, and hence which execution labels, an occupation contains.

Panel~(a) of Figure~\ref{fig:aiChains_graphs_def1} overlays the two placebo distributions of average AI chain length, one across 1{,}000 task position permutations and one across 1{,}000 task AI-execution label reassignment permutations, and compares both with the observed mean AI chain length, marked by the dashed vertical line.
\begin{figure}[!t]
\caption{Average Length of AI Chains: Actual versus Placebo Datasets}
\label{fig:aiChains_graphs_def1}

\centering
\begin{subfigure}[b]{0.49\textwidth}
    \centering
    \subcaption{O*NET Occupations}
    \label{fig:aiChains_onet}
    \includegraphics[width=\textwidth]{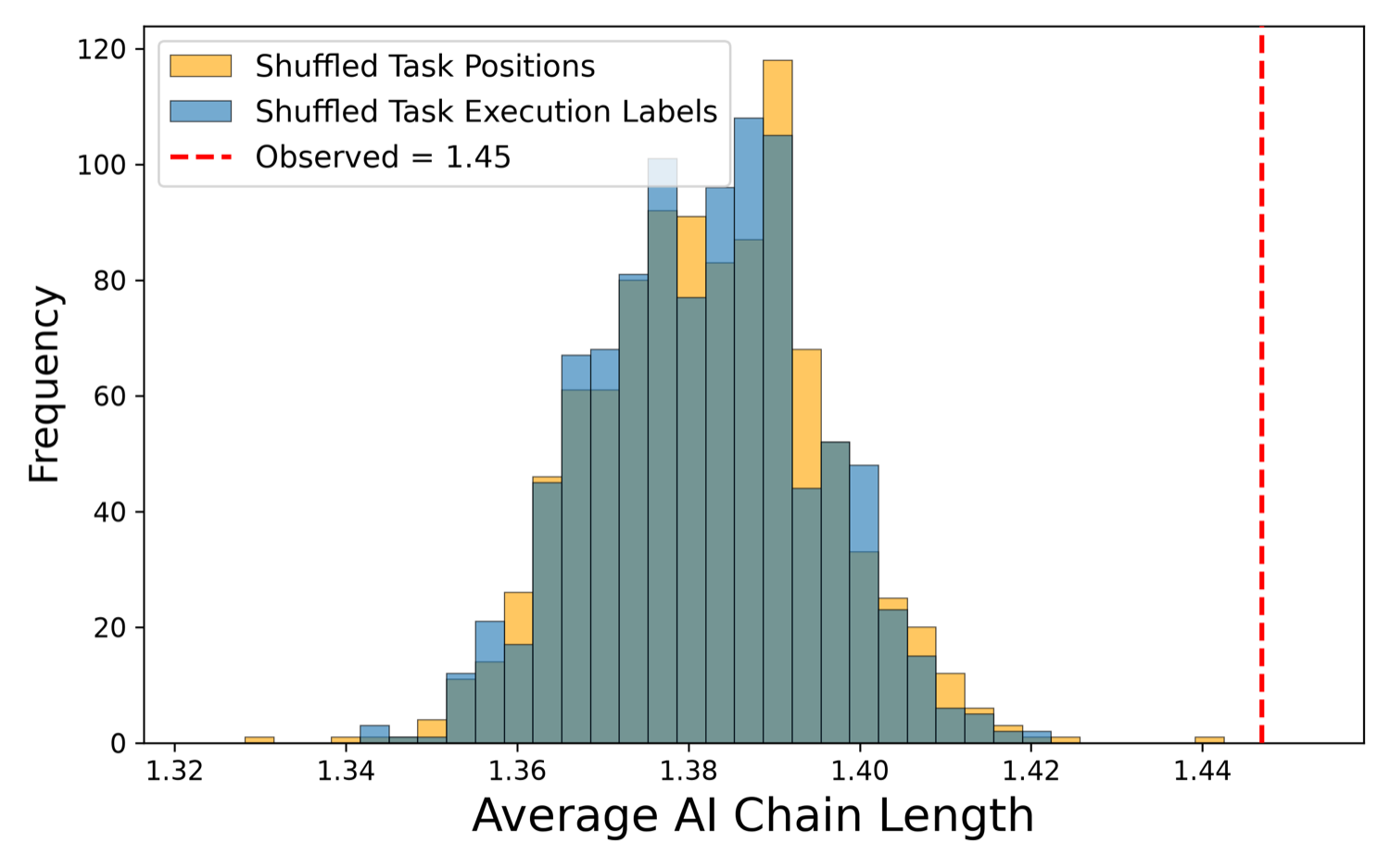}
\end{subfigure}
\hfill
\begin{subfigure}[b]{0.49\textwidth}
    \centering
    \subcaption{PCF Process Groups}
    \label{fig:aiChains_pcf}
    \includegraphics[width=\textwidth]{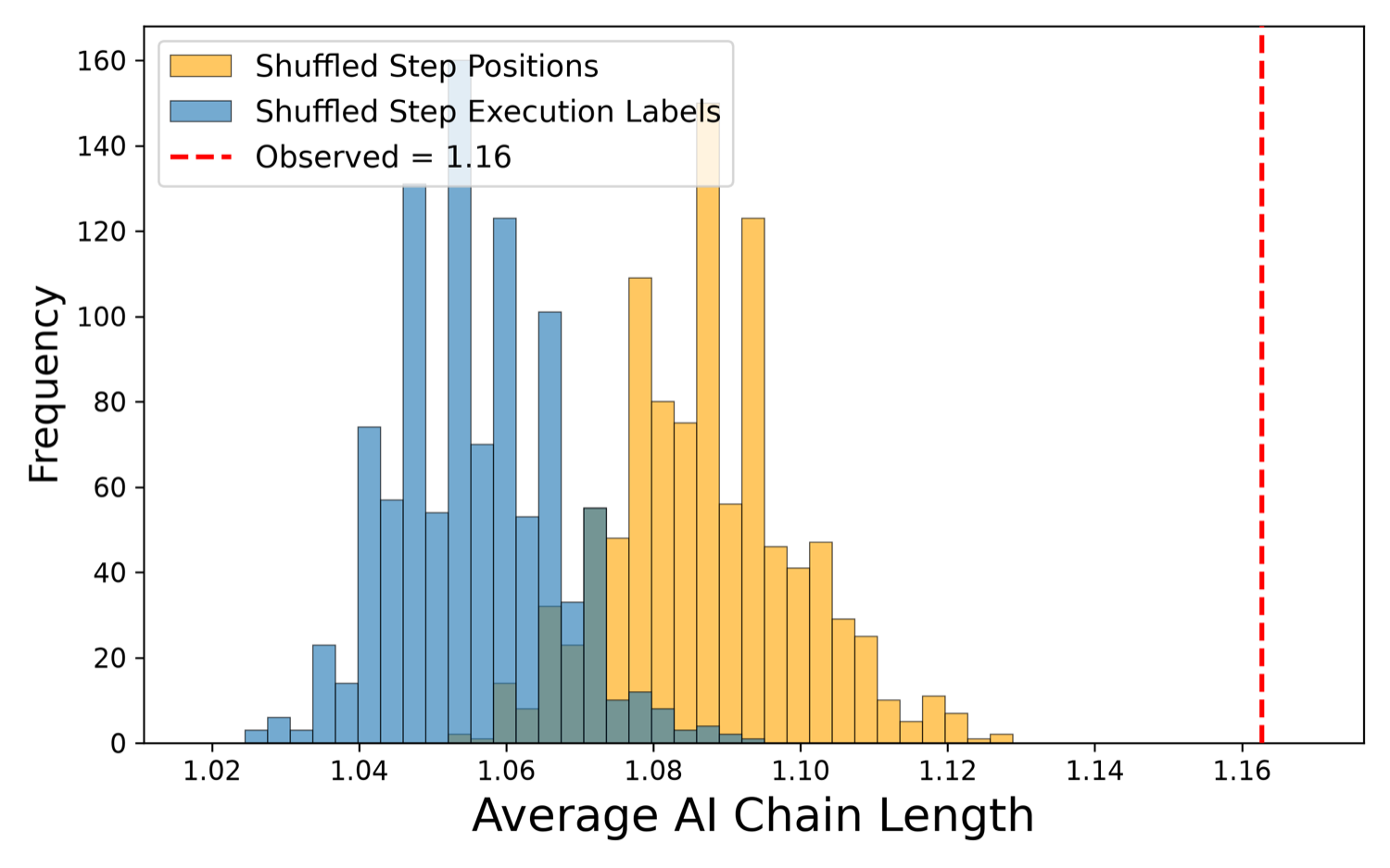}
\end{subfigure}

\par\vspace{\fignotesskip}
\begin{minipage}{1\linewidth}
\begin{spacing}{0.2}
{\fontsize{9.5pt}{10pt}\selectfont Notes: Red dashed lines show the observed average AI chain length in the actual dataset.
Panel~(a) uses the O*NET sample, in which a workflow is an occupation and the ordering is the GPT-imputed task sequence; Panel~(b) uses the PCF sample, in which a workflow is a process group and the ordering is the one documented by practitioners, with exposure and execution labels transferred from each step's nearest O*NET task.
Each panel overlays both placebo distributions, each over $1{,}000$ draws.
The orange histogram randomizes the positions of steps within each workflow, and the blue histogram reassigns steps, each carrying its execution label, across workflows within the same group, preserving workflow sizes; the grouping is the major SOC group in Panel~(a) and the PCF Category in Panel~(b).
The horizontal scale differs between the panels, since only $4\%$ of PCF steps carry an AI-execution label after the transfer against $13\%$ in O*NET, which lowers the attainable chain length throughout the PCF sample.}
\end{spacing}
\end{minipage}
\end{figure}
We observe that the average AI chain length in the data is 1.45 and is noticeably larger than in both placebo distributions.
This indicates that AI-executed tasks tend to cluster within occupations, forming longer AI chains than would arise under either type of randomization.
At the same time, the modest magnitude of 1.45 indicates that long AI chains remain rare, perhaps reflecting the early stage of AI adoption; our model predicts that chains will lengthen as AI quality improves.

Next we repeat the exercise on the PCF sample, substituting its own units throughout, with a process group standing in for an occupation, the activity sequence practitioners documented for the GPT-imputed task sequence, and the PCF Category for the major SOC occupation group within which the second placebo reassigns.
Nothing else about the test changes, and Panel~(b) of Figure~\ref{fig:aiChains_graphs_def1} reports the same verdict.
The average AI chain there spans $1.16$ steps, against $1.45$ in O*NET.
The smaller magnitude is mechanical, since AI execution in the PCF is thinner than in O*NET, which leaves far less opportunity for long runs to form and puts the ceiling on the statistic well below the O*NET sample's.
What the PCF sample can still ask is whether the runs that do form are longer than the same labels arranged at random along the same sequences, and they are.
Because no language model placed these steps in order, the contiguity here is a property of the documented workflow and not of our sequencing procedure.

\subsection{Prediction \#2: AI Execution Is More Likely Next to AI-executed Steps}
\label{sec:DWA_prediction}

The second prediction of the model that we test is that when the same step appears in two workflows, in one sitting between strong AI performers and in the other located between manual ones, it is more likely to be executed by AI in the first.
This is the empirical content of Proposition~\ref{prop:ca_local}.
Intuitively, when a step is surrounded by AI-executed steps, there are local gains from grouping it with its neighbors as part of a long AI chain that are absent when the step is situated next to manual steps.
For example, step X is more likely to be executed by AI in the first worklfow below than in the second:

\vspace{0.6em}
\begin{adjustbox}{center}
\begin{tikzpicture}[scale=0.93, transform shape,
    node distance=1cm,
    every node/.style={rectangle, rounded corners, draw, minimum width=1.9cm, minimum height=0.9cm, align=center},
    human/.style={fill=olive!20},
    ai/.style={fill=blue!20},
    xnode/.style={fill=teal!20},
    >=latex
]

% ======================
% OCCUPATION 1
% ======================

\node[draw=none, anchor=west] (L1) at (0,0) {\textbf{Workflow 1:}};

% Place chain on its own clean baseline
\begin{scope}[xshift=3.5cm, yshift=0cm]

\node[draw=none] (Dots1L) at (0,0) {\ \ $\cdots \cdots$};
\node[ai, right=1 of Dots1L] (A1a) {AI};
\node[xnode, right=1 of A1a] (X1) {X};
\node[ai, right=1 of X1] (A1b) {AI};
\node[draw=none, right=1 of A1b] (Dots1R) {$\cdots \cdots$};

\draw[->, thick] (Dots1L.east) -- (A1a);
\draw[->, thick] (A1a) -- (X1);
\draw[->, thick] (X1) -- (A1b);
\draw[->, thick] (A1b) -- (Dots1R.west);

\end{scope}

% ======================
% OCCUPATION 2
% ======================

\node[draw=none, anchor=west] (L2) at (0,-1.4) {\textbf{Workflow 2:}};

% Place chain on its own line (same centering)
\begin{scope}[xshift=3.5cm, yshift=-1.4cm]

\node[draw=none] (Dots2L) at (0,0) {\ \ $\cdots \cdots$};
\node[human, right=1 of Dots2L] (H2a) {Manual};
\node[xnode, right=1 of H2a] (X2) {X};
\node[human, right=1 of X2] (H2b) {Manual};
\node[draw=none, right=1 of H2b] (Dots2R) {$\cdots \cdots$};

\draw[->, thick] (Dots2L.east) -- (H2a);
\draw[->, thick] (H2a) -- (X2);
\draw[->, thick] (X2) -- (H2b);
\draw[->, thick] (H2b) -- (Dots2R.west);

\end{scope}

\end{tikzpicture}
\end{adjustbox}
\vspace{-1.4em}

Testing this prediction requires identifying steps that appear across different workflows.
Because O*NET tasks are occupation-specific, no single task is shared across occupations in the data.
To work around this, we rely on O*NET's broader categories of Detailed Work Activities (DWAs), which group conceptually similar tasks across occupations.
In our baseline O*NET sample, we drop tasks that are mapped to multiple DWAs.
We further restrict attention to DWAs that contain tasks appearing in more than one occupation.
Together, these restrictions reduce the sample to 10{,}708 tasks spread across 1{,}748 DWAs.
We treat the remaining tasks as instances of the same step appearing across different occupations and estimate the following regression\footnote{In a robustness exercise we refine our definition of ``similar'' tasks.
The resulting estimates closely mirror the patterns observed in the full sample in both sign and magnitude.
See Appendix~\ref{app:additional_robustness} for details.}:
% Broken after the \beta_2 term: on one line the regression is 523pt wide against a
% 470pt measure and the closing \Big) and comma ran past the right margin. \small is
% not enough on its own (477pt), and the two rows come to 312pt and 205pt at full size.
% \Big( and \Big) are explicit-size delimiters, not \left ... \right, so they can sit on
% different rows of the alignment.
\begin{equation}
\begin{aligned}
    \Pr\!\left( \text{is\_ai}_{k} = 1 \mid X_k \right)
    ={}
    &\Lambda\!\Big(
        \beta_0
        + \beta_1 \, \text{prev2\_is\_ai}_{k}
        + \beta_2 \, \text{prev\_is\_ai}_{k} \\
    &\qquad
        + \beta_3 \, \text{next\_is\_ai}_{k}
        + \beta_4 \, \text{next2\_is\_ai}_{k}
    \Big),
\end{aligned}
    \label{eq:DWA_regression_ai}
\end{equation}
where $\Lambda$ is the logistic CDF and \emph{$\text{is\_ai}_{k}$} indicates whether step $k$ is executed by AI.
The variables \emph{$\text{prev2\_is\_ai}_{k}$} and \emph{$\text{prev\_is\_ai}_{k}$} equal 1 if the step two positions before or immediately before step $k$ is AI-executed, respectively, and \emph{$\text{next\_is\_ai}_{k}$} and \emph{$\text{next2\_is\_ai}_{k}$} are defined analogously for the steps immediately after and two positions after step $k$.
In implementations we also control for step $k$'s AI exposure status and the number of steps in its workflow.

Table~\ref{tab:DWA_regression_aiExecution_mainSample} reports average marginal effects from estimating Equation~\eqref{eq:DWA_regression_ai} for a series of specifications.
\begin{table}[!t]
\centering
\caption{Neighboring Tasks' AI Execution Status and a Task's AI Execution}
\label{tab:DWA_regression_aiExecution_mainSample}
    \footnotesize
\resizebox{\textwidth}{!}{
% --- LaTeX Table for full_0 ---
\begin{tabular}{lcccccc}
\toprule
 & \multicolumn{6}{c}{Probability that Focal Task ($k$) is AI-executed} \\
\cmidrule(lr){2-7}
 & (1) & (2) & (3) & (4) & (5) & (6) \\
\midrule
Task ($k-2$) is AI-executed & 0.05*** & 0.01 & 0.00 & -0.01 & 0.06* & -0.01 \\
 & (0.02) & (0.01) & (0.01) & (0.03) & (0.03) & (0.03) \\
\addlinespace
Task ($k-1$) is AI-executed & 0.11*** & 0.06*** & 0.05*** & 0.05*** & 0.13*** & 0.05*** \\
 & (0.02) & (0.01) & (0.01) & (0.02) & (0.03) & (0.02) \\
\addlinespace
Task ($k+1$) is AI-executed & 0.11*** & 0.05*** & 0.05*** & 0.04** & 0.10*** & 0.04** \\
 & (0.02) & (0.01) & (0.01) & (0.02) & (0.02) & (0.02) \\
\addlinespace
Task ($k+2$) is AI-executed & 0.04** & -0.00 & -0.01 & 0.00 & 0.04 & 0.00 \\
 & (0.02) & (0.01) & (0.01) & (0.02) & (0.03) & (0.02) \\
\addlinespace
\midrule
Pseudo $R^2$ & 0.130 & 0.183 & 0.180 & 0.197 & 0.030 & 0.197 \\
Observations & 10,708 & 10,708 & 9,861 & 4,096 & 4,096 & 4,096 \\
SOC Group FE &  & Major & Minor &  &  &  \\
DWA FE &  &  &  & \checkmark &  & \checkmark \\
NumTasks in DWA-Occupation Control &  &  &  &  & \checkmark & \checkmark \\
\bottomrule
\end{tabular}

}
\begin{minipage}{1\linewidth}
\begin{spacing}{0.2}
{\fontsize{9.5pt}{10pt}\selectfont Notes: Clustered standard errors in parentheses. *** p$<$0.01, ** p$<$0.05, * p$<$0.1.
This table reports average marginal effects from estimating Equation~\eqref{eq:DWA_regression_ai}.
The dependent variable in all specifications is an indicator for whether task $k$ is AI-executed ($\text{is\_ai}_{k}$).
The sample is restricted to similar tasks across occupations, identified as tasks belonging to the same DWA in the O*NET dataset.
All specifications control for the AI exposure status of task $k$ and for the total number of tasks in task $k$'s occupation.
Standard errors are bootstrapped using $B=200$ replications and clustered at the DWA level.
The estimating sample varies across columns. Minor-group fixed effects in column~(3) drop groups with no within-group variation in the outcome, removing 847 tasks. The DWA fixed effects in columns~(4) and~(6) are identified only from DWAs containing both an AI-executed and a non-AI-executed task, a condition met by 534 of the 1{,}748 DWAs in the sample, leaving 4{,}096 tasks; column~(5) carries no DWA fixed effects but is estimated on that same restricted sample so that it is directly comparable with column~(4). The number of DWA clusters is 1{,}748 in columns~(1)--(2), 1{,}705 in column~(3), and 534 in columns~(4)--(6).}
\end{spacing}
\end{minipage}\end{table}
The baseline in column~(1) shows that AI execution anywhere in the focal task's local context is associated with a higher probability that the task is AI-executed, at one and at two positions away alike, with the effects of the adjacent steps roughly twice the size of those two positions out.
The remaining columns ask how much of this survives more demanding comparisons.
Columns~(2)--(3) absorb occupation family differences with SOC major and minor group fixed effects, and column~(4) restricts the variation to within detailed work activities with DWA fixed effects, addressing differences in AI-ability across types of task.
Columns~(5)--(6) repeat the exercise controlling for the number of same-DWA tasks in the focal task's occupation, which rules out mechanical proximity inflating the estimates when several same-DWA tasks sit in one workflow; columns~(4)--(6) are all estimated on the same restricted sample, so they are comparable with one another rather than with column~(1).
Across the specifications that absorb occupation-family or DWA heterogeneity, columns~(2)--(4) and~(6), the pattern is the same.
The immediate neighbor effect attenuates to between $0.04$ and $0.06$, relative to the baseline for columns~(2)--(3) and to the restricted-sample estimate of column~(5) for columns~(4) and~(6), but stays statistically and economically meaningful, while the effects two positions away fall to between $-0.01$ and $0.01$ and lose significance.
Column~(5) carries no fixed effects and, estimated on the restricted sample of column~(4), reproduces the baseline pattern rather than the attenuated one.
That the baseline estimates attenuate tells us they partly reflect systematic differences across occupation families; that the immediate neighbor effect persists tells us the local pattern is not merely such a difference.

Because we have no external benchmark for what constitutes a large effect here, the magnitudes are hard to read on their own.
We therefore re-estimate Equation~\eqref{eq:DWA_regression_ai} on 1,000 placebo datasets that randomize task positions within occupations, for the four specifications in columns~(1)--(4) of Table~\ref{tab:DWA_regression_aiExecution_mainSample}, and report the result in Figure~\ref{fig:DWA_regression_aiExecution_mainSample} of Appendix~\ref{app:tables_and_figures}.
Randomizing positions strips the ordering of any information, so the placebo distributions show what these coefficients would look like when the position of a task in the workflow does not matter.
In each case, the actual orderings deliver stronger immediate-neighbor effects and weaker distant-neighbor effects than the placebos.
This implies that local work context and proximity do work that the reshuffled orderings cannot reproduce by chance.

Overall, the results indicate that AI execution of immediate neighbors of the focal task is associated with a higher likelihood that the task is executed by AI, whereas the execution status of more distant neighbors matters much less, especially after controlling for a range of relevant factors.
This aligns with the relatively short average AI chain length observed in the data (mean length 1.45 reported in Subsection~\ref{sec:chainLength_prediction}), implying that tasks two positions away rarely fall in the same chain as the focal task.
It is possible that with improvements in AI quality, chains become longer and execution status of more distant neighbors becomes more relevant for the execution mode of the focal task.

For this prediction, we do not attempt a counterpart on the PCF sample, where too few similar steps are observed with both execution outcomes to support the comparison. See Appendix~\ref{app:external_validation} for details.

\subsection{Prediction \#3: AI Execution Is Lower Where AI-able Steps Are Dispersed}
\label{sec:fragmentation_prediction}

The two tests so far were local, in that they operated over a small window around a step.
The last prediction of the model that we test is how the positioning of steps over the whole workflow matters for AI execution.
Specifically, we test whether realized AI execution in a production sequence varies with how clustered or dispersed its AI-exposed steps are, and not merely with how many of its steps are exposed.
This is the empirical content of Proposition~\ref{prop:fragmentation}.

To measure the fragmentation of AI-exposed steps in the data, we adapt the fragmentation index of Section~\ref{sec:fragmentation} to our empirical setting.
The \emph{empirical fragmentation index} (EFI) is the ratio of the number of \emph{potential tasks}, as defined by our theoretical model, to the total number of steps, which is fixed and is measured by the raw count of steps in the workflow.
Consider a workflow with five steps.
If none of the five steps are exposed to AI, then the number of tasks in the theoretical sense is five, and the EFI equals $5/5 = 1$.
If instead two consecutive steps are exposed to AI (and thus can form an AI chain of length two when all AI-exposed steps turn into AI-executed steps), the number of tasks in the model falls to four and the EFI becomes $4/5 = 0.8$.
Thus, the empirical fragmentation index decreases as more AI-exposed steps can potentially be consolidated into longer AI chains.\footnote{The EFI is, after normalizing by the number of steps, the special case of the realized fragmentation expressed in \eqref{eq:fragmentation} in which every step takes one unit of time in either mode, $\manualTime{i} = \AItime{i} = 1$, and AI either always or never succeeds at a step, $q_i \in \{0,1\}$.
It therefore captures the arrangement of a workflow's AI-able steps, which is what we observe in the data, and not the differences in time cost and reliability that the theoretical index also prices.
Since $\manualTime{i} = 1$ satisfies the hypothesis of \eqref{eq:fragmentation_closed_form}, the EFI is exactly one minus the share of adjacent AI-able pairs in the workflow, so it inherits the ordering that expression establishes.}

Two components govern the empirical fragmentation index, and only one is the subject of Prediction~\#3.
Write $m_w$ for the number of steps in a workflow, $k_w$ for how many of them are AI-exposed, and $r_w$ for the number of maximal blocks of consecutive AI-exposed steps they form, each block being one potential AI chain.
EFI decomposes exactly as
\begin{equation}
    \text{EFI}_w \;=\; \frac{m_w - (k_w - r_w)}{m_w} \;=\; 1 - \underbrace{\frac{k_w}{m_w}}_{\text{exposure level}} \;+\; \underbrace{\frac{r_w}{m_w}}_{\text{arrangement}},
    \label{eq:efi_decomposition}
\end{equation}
where the last term is the object of interest.
Holding $k_w$ fixed, $r_w$ runs from $1$ when the AI-exposed steps are perfectly clustered to $k_w$ when no two of them are adjacent.
To read the index as a measure of dispersion, we must control for the share of AI-exposed steps, which nets out the level term and leaves the arrangement component.
We therefore estimate the following regression:
\begin{equation}
\begin{aligned}
    \text{ai\_execution}_{w} = {} & \beta_0 + \beta_1 \, \text{ai\_exposure}_{w} + \beta_2 \, \text{empirical\_fragmentation\_index}_{w} \\
    & + \beta_3 \, \text{num\_ai\_able\_steps}_{w} + \varepsilon_w,
\end{aligned}
    \label{eq:fragmentation_index_regression}
\end{equation}
where \emph{$\text{ai\_execution}_{w}$} is the share of steps executed by AI in workflow $w$, \emph{$\text{ai\_exposure}_{w}$} is the share of AI-exposed steps, \emph{$\text{empirical\_fragmentation\_index}_{w}$} is the index defined above, and \emph{$\text{num\_ai\_able\_steps}_{w}$} is the number of AI-able steps in workflow $w$.\footnote{Controlling for workflow length $m_w$ instead, or dropping the count control altogether, moves the fragmentation coefficient by at most $0.04$ standard deviations and changes no verdict.}
The model predicts that workflows with higher AI exposure should have a higher share of steps executed by AI ($\beta_1 > 0$), and that, conditional on exposure, AI execution should be lower where AI-able steps are more dispersed ($\beta_2 < 0$).

Table~\ref{tab:fragmentation_index_regression_exposure} reports the results.
\begin{table}[!t]
    \caption{AI-exposed Step Fragmentation and AI Execution Outcomes}
    \label{tab:fragmentation_index_regression_exposure}
    \footnotesize
    \centering
    \adjustbox{max width=\textwidth}{\setlength{\tabcolsep}{7pt}
\begin{tabular}{lcccccc}
\toprule
 & \multicolumn{6}{c}{Share of AI-executed Steps} \\
\cmidrule(lr){2-7}
 & \multicolumn{3}{c}{O*NET Occupations} & \multicolumn{3}{c}{PCF Process Groups} \\
\cmidrule(lr){2-4}\cmidrule(lr){5-7}
 & (1) & (2) & (3) & (4) & (5) & (6) \\
\midrule
\addlinespace
Share of AI-exposed Steps & 0.49*** & 0.48*** & 0.39*** & 0.28*** & 0.42*** & 0.28*** \\
 & (0.10) & (0.09) & (0.09) & (0.10) & (0.10) & (0.10) \\
\addlinespace
Empirical Fragmentation Index & -0.01 & -0.09 & -0.04 & -0.35*** & -0.26** & -0.34*** \\
 & (0.10) & (0.09) & (0.09) & (0.10) & (0.11) & (0.11) \\
\hline\\[-1.25em]
Fixed Effect & & SOC Major & SOC Minor & & PCF Category & Framework \\
Number of AI-able Steps Control & \checkmark & \checkmark & \checkmark & \checkmark & \checkmark & \checkmark \\
R-squared & 0.35 & 0.64 & 0.72 & 0.37 & 0.49 & 0.38 \\
Adj. R-squared & 0.34 & 0.63 & 0.67 & 0.36 & 0.47 & 0.36 \\
Observations & 872 & 872 & 872 & 525 & 525 & 525 \\
\bottomrule
\end{tabular}
}
    \begin{minipage}{1\linewidth}
\begin{spacing}{0.2}
{\fontsize{9.5pt}{10pt}\selectfont Notes: Standardized coefficients. Heteroskedasticity-robust standard errors in parentheses. *** p$<$0.01, ** p$<$0.05, * p$<$0.1.
This table reports results from estimating Equation~\eqref{eq:fragmentation_index_regression}.
The dependent variable in all specifications is the share of steps executed by AI in the workflow.
``Share of AI-exposed Steps'' is the share of steps carrying an E1 or E2 label, the same label set from which the Empirical Fragmentation Index is constructed; the index itself captures how dispersed AI-exposed steps are across the workflow.
All specifications control for the number of AI-able steps in the workflow. 
Columns~(1)--(3) use the O*NET sample, where a workflow is an occupation and the ordering is the GPT-imputed task sequence; columns~(4)--(6) use the PCF sample, where a workflow is a process group and the ordering is the one documented by practitioners.
In the PCF sample, exposure and execution labels are transferred from each step's nearest O*NET task.
All variables are standardized to unit standard deviation within their own sample, so the reported coefficients are in standard-deviation units and are not directly comparable in levels across the two samples.}
\end{spacing}
\end{minipage}\end{table}
In each sample we estimate three specifications: no fixed effects, and then two groupings of workflows (SOC major and minor occupation groups in O*NET, PCF Category and Framework in the PCF).
The fixed effects ensure the results are not driven by a narrow subset of occupation or process families.
The exposure coefficient behaves as the model predicts in every column of both samples.
A one-standard-deviation increase in the share of AI-exposed steps is associated with a $0.39$ to $0.49$ standard-deviation increase in the AI-executed share in O*NET and a $0.28$ to $0.42$ standard-deviation increase in the PCF, significant at the 1\% level throughout.

The fragmentation coefficient is where the two samples part ways.
In the O*NET sample it carries the predicted negative sign in all three specifications, but the point estimates are small and none is statistically distinguishable from zero.
On O*NET occupations therefore we do not detect the workflow-level channel, but the estimates are not precise enough to rule it out either.
A modest fragmentation effect would thus go undetected in this sample.
On the PCF's documented sequences, however, the coefficient is negative and statistically significant in all three columns, at $-0.35$, $-0.26$ and $-0.34$, so that process groups whose AI-exposed steps are more scattered convert less of their exposure into execution.

That the channel is detected in the PCF and not in O*NET is what a test built on arrangement would lead one to expect.
The fragmentation index reads the ordering of a workflow as a whole, so it stands its best chance of detection where that ordering is recorded rather than inferred, which is the PCF's advantage over the imputed sequences in O*NET.
The PCF process groups cover a smaller part of economic activity than O*NET's occupational structure, and they inherit their labels through a lossy match, so the channel is documented on a narrower set of production workflows rather than across the economy.
We read the results together as saying that the arrangement channel is present where the arrangement is measured well enough to see it.

We must emphasize that what we document in these exercises are associations that speak to the mechanism the model describes, and not causal effects.
Measuring the latter would require exogenous variation in how much of a workflow is handed to AI and/or in how that workflow is arranged, which is not what we work with.

\section{Conclusion}
This paper develops a task-based framework for understanding how advances in AI reshape automation and production.
By modeling production as a sequence of steps that the firm aggregates into tasks, we show how firms optimally allocate work between humans and AI, whether in augmented or fully automated form.
A central feature of the model is the possibility of AI chains, in which multiple steps are executed by AI and only the final output is verified by a human.
We show that chaining can overturn comparative advantage logic in factor assignment by amplifying the importance of local work context.
That is, whether a unit of work is assigned to a human or to AI depends not only on the characteristics of the step in question but also on those of its neighbors in production.
By the same logic, how much a workflow stands to gain from automation depends on how it is arranged and not only on how much of it is exposed to AI. 
When AI-able steps are clustered, one verification at the end of a chain covers them all; when they are scattered across the workflow, each must be verified separately, and the gains from automation shrink.
We capture this in a fragmentation index, a measure of how dispersed a workflow's AI-able steps are, which approximates the cost of the firm's optimal AI deployment.

We characterize the firm's optimal AI deployment decision across two horizons: a short run, in which the boundaries of jobs and what they pay are both fixed, so AI only improves the time it takes to produce a unit of output; and a long run, in which the firm redraws those boundaries jointly with its AI deployment, trading the lower wages of narrower jobs against the coordination costs of maintaining more of them.

We also show that improvements in AI quality can induce discrete reorganizations of work.
A key implication of chaining is that marginal increases in AI quality may generate little or no cost savings until a threshold is crossed, after which the optimal production configuration changes discontinuously, so that the marginal returns to improving AI are non-monotone.
Our framework thus helps explain why firms invest so heavily in AI capabilities even when short-run returns appear limited, and it is consistent with the J-curve pattern of technology adoption \citep{brynjolfsson2021productivity}.

We complement our theoretical analysis with empirical evidence consistent with the core mechanisms implied by the model.
Using a task-level dataset of AI exposure and AI execution outcomes, we find that AI-executed steps tend to appear in contiguous AI chains, that when comparable steps appear across occupations those sitting in a more AI-executed local context are significantly more likely to be executed by AI themselves, and that workflows whose AI-exposed steps are more fragmented exhibit lower realized AI execution conditional on exposure.
These patterns align with the central economic forces highlighted by the model: the importance of positional relationships among steps, the local complementarities created by AI chains, and the role of AI-exposed step dispersion in shaping execution outcomes.

\clearpage

\bibliographystyle{chicago}
\setlength{\bibsep}{0ex}
\begin{spacing}{0.98}
\bibliography{rubin}
\end{spacing}

\newpage
\appendix
\newpage

\clearpage
%\renewcommand\thesection{Appendix \Alph{section}}
%\renewcommand\thesubsection{\Alph{section}.\arabic{subsection}}

% Appendix exhibit numbering.
%
% Every numbered object in the Online Appendix carries the letter of the appendix it
% sits in and restarts at 1 there: Table A.1, Figure D.2, Equation (C.15),
% Example B.3. \appendix has already redefined \thesection to \Alph{section}, so
% \thesection expands to that letter at the point the object is numbered.
%
% \counterwithin* installs the "reset at every \section" rule without touching
% \the<counter>, which the \renewcommand above it has already set.
%
% It also repairs the hyperlinks, which is why it has to be \counterwithin* and not a
% bare \@addtoreset-free reset. hyperref keeps a second, hidden name per counter,
% \theH<counter>, and builds PDF anchors out of that rather than out of the printed
% number. Under the previous OA- scheme nothing ever set it, so appendix Equation
% (OA.1) and body Equation (1) both anchored at "equation.1", the appendix anchors
% were dropped as duplicates, and every appendix link landed in the main text: 56
% "destination with the same identifier" warnings in the log. hyperref hooks
% \@addtoreset, so the \counterwithin* below redefines \theH<counter> to
% \theHsection.\arabic{<counter>} as a side effect, and the anchors come out as
% "equation.C.15", distinct from the body's "equation.15".
\newcommand{\oaNumberWithin}[1]{%
  \expandafter\renewcommand\csname the#1\endcsname{\thesection.\arabic{#1}}%
  \counterwithin*{#1}{section}%
}
\oaNumberWithin{figure}
\oaNumberWithin{table}
\oaNumberWithin{equation}
\oaNumberWithin{definition}
\oaNumberWithin{lemma}
\oaNumberWithin{assumption}
\oaNumberWithin{constraint}
\oaNumberWithin{proposition}
\oaNumberWithin{corollary}
\oaNumberWithin{theorem}
\oaNumberWithin{observation}
\oaNumberWithin{example}

\newcommand{\appendixletter}{\Alph{section}}

% Each of the two appendices is a self-contained document: its own title page, its own
% contents, its own reference list, its own page sequence. \startappendixpart does the
% bookkeeping for one of them. #1 is the short marker that prefixes both the page
% numbers and the section numbers (OA or SA), #2 the name on the title page.
%
% \thesection carries the marker so that a bare \ref is unambiguous across the two:
% Appendix OA.A and Appendix SA.A are different appendices and now read that way, and
% the exhibits below them come out as Table OA.A.1 and Table SA.A.1 since \thefigure is
% \thesection.\arabic{figure}. To go back to bare letters, drop the marker from
% \thesection here; nothing else depends on it.
%
% \theHsection is the hyperref anchor and has to carry the marker too. Without it the
% two Appendix A's would both anchor at "A", every exhibit under them would collide,
% and the duplicates would be dropped, which is the bug the OA- scheme used to have.
% The optional argument adds a small line under the name, e.g. (Not for Publication).
\newcommand{\startappendixpart}[3][]{%
  \clearpage
  \gdef\thesection{#2.\Alph{section}}%
  \gdef\theHsection{#2.\Alph{section}}%
  \expandafter\gdef\csname thepage\endcsname{#2 - \arabic{page}}%
  \setcounter{section}{0}%
  \setcounter{figure}{0}\setcounter{table}{0}\setcounter{equation}{0}%
  \setcounter{definition}{0}\setcounter{lemma}{0}\setcounter{assumption}{0}%
  \setcounter{constraint}{0}\setcounter{proposition}{0}\setcounter{corollary}{0}%
  \setcounter{theorem}{0}\setcounter{observation}{0}\setcounter{example}{0}%
  \setcounter{page}{1}%
  \pagestyle{empty}%
  \begin{center}
    {\LARGE Chaining Tasks, Redefining Work: \\ A Theory of AI Automation}\\[0.7em]
    \large Mert Demirer, John J. Horton, Nicole Immorlica, Brendan Lucier, and Peyman Shahidi \\[0.7em]
    {\LARGE #3}\\[0.6em]
    \if\relax\detokenize{#1}\relax\else{\large (#1)}\\[0.6em]\fi
  \end{center}
}

% ============================================================ ONLINE APPENDIX
% Sections OA.A to OA.C: the exhibits the body cites, the proofs, and the macro
% aggregation. Kept short on purpose; the data construction and the further tests
% live in the Supplementary Appendix that follows.
\startappendixpart{OA}{Online Appendix}
% Slightly tighter spacing around displayed equations in the Online Appendix only:
% 7pt instead of the 9pt the preamble sets for the rest of the paper. \normalsize
% resets these skips, so they are re-applied through it; the Supplementary Appendix
% restores the 9pt.
\makeatletter
\let\OA@normalsize\normalsize
\renewcommand\normalsize{\OA@normalsize
  \setlength\abovedisplayskip{7pt}\setlength\belowdisplayskip{7pt}%
  \setlength\abovedisplayshortskip{7pt}\setlength\belowdisplayshortskip{7pt}}
\makeatother
\normalsize
% Online Appendix sections are numbered with bare letters (A, B, C), not OA.A.
% \theHsection keeps the OA marker so hyperref anchors stay distinct from SA.A.
\gdef\thesection{\Alph{section}}

\setcounter{tocdepth}{2}
% The section numbers in these two contents lists are "OA.A" and "SA.B.1", not "A" and
% "B.1", so tocloft's default number columns (1.5em for sections, 2.3em for subsections)
% are too narrow and the label runs into the title. Widened once here, after the draft
% front matter at the top of this file has already been typeset with plain numbers.
\cftsetindents{section}{0em}{2.9em}
\cftsetindents{subsection}{2.9em}{4.0em}
% Initialize etoc's depth-tag store: tocloft/hyperref can bypass etoc's \tableofcontents,
% leaving \Etoc@savedtocdepth undefined when the .toc depth tags are processed.
\makeatletter\xdef\Etoc@savedtocdepth{\number\c@tocdepth}\makeatother
\tocloftpagestyle{empty}
% Three tags now, not two: the body, and one per appendix. Each contents list shows its
% own tag and hides the other two, which is what makes the two appendices separable.
\etocsettagdepth{main}{none}
\etocsettagdepth{onlineappendix}{subsection}
\etocsettagdepth{suppappendix}{none}
% No contents list: Appendix A starts on the title page, as in the Supplementary
% Appendix, to keep the Online Appendix within REStud's 30-page limit.
\thispagestyle{plain}
\pagestyle{plain}

\begin{bibunit}[chicago]
\etocdepthtag.toc{onlineappendix}

% A goes first because it is the only cross-cutting appendix: it holds exhibits cited
% from Sections 3, 4, 5 and 7, so it belongs to neither the theory block nor the
% empirical block and reads as an interruption anywhere between them. Of the two ends,
% the front wins because its most-consulted item is a lookup aid rather than a result:
% Table OA.A.1 is the notation summary, cited from the fourteenth line of Section 3.
\section{Additional Tables and Figures}
\label{app:tables_and_figures}
% Belt and braces for the first exhibit. The house format sets every float [!t], and on
% the first page of an appendix that puts the leading exhibit at the top of the page
% carrying the appendix heading, above the heading itself, so it reads as though it
% belonged to the previous appendix. The exhibits here are placed by hand with [h!] and
% \newpage, which already avoids that, and \suppressfloats[t] keeps the top of this
% page clear should any of them go back to floating.

\begingroup
\centering
\captionof{table}{Notation Used in the Theory Sections}
\label{tab:notation}
\footnotesize
\begin{tabularx}{\textwidth}{l X}
\toprule
\multicolumn{2}{l}{\textbf{Panel A: Steps, Tasks, and AI Deployment (Section~\ref{sec:shortrun}: Short Run)}} \\
\midrule
\textbf{Symbol} & \textbf{Description} \\
\midrule
$\alpha \in (0,1]$ & Quality of the general-purpose AI technology \\
$s_i$ & Production step $i$ \\
$\mathcal{S} = (s_1, \dotsc, s_m)$ & Sequence of the firm's $m$ production steps \\
$\manualTime{i}$ & Time cost of completing step $i$ manually \\
$\AItime{i}$ & Time cost of verifying the output of one AI attempt at step $i$ \\
$d_i$ & Difficulty of step $i$ for AI \\
$q_i = \alpha^{d_i}$ & Probability that AI completes step $i$ successfully \\
$T_b$ & Task $b$: a manual step or an AI chain (with a standalone augmented step considered a length-one AI chain) \\
$\mathcal{T} = (T_1, \dotsc, T_n)$ & AI deployment strategy; a partition of the steps into $n$ tasks \\
$\timecost{b}$ & Expected time cost of task $b$: $\manualTime{i}$ if a manual step, $\AItime{r}/\prod_{i=\ell}^{r} q_i$ if an AI chain over steps $\ell, \dotsc, r$ \\
\midrule
\multicolumn{2}{l}{\textbf{Panel B: Skills, Jobs, Wages, Hand-offs, and Job Design (Section~\ref{sec:longrun}: Long Run)}} \\
\midrule
\textbf{Symbol} & \textbf{Description} \\
\midrule
$\manualSkill{i}$ & Skill required to complete step $i$ manually \\
$\AIskill{i}$ & Skill required to verify the output of one AI attempt at step $i$ \\
$\skillcost{b}$ & Skill required by task $b$: $\manualSkill{i}$ if a manual step, $\AIskill{r}$ if an AI chain augmented at step $r$; an automated step requires none \\
$J_j$ & Job $j$: a contiguous block of tasks assigned to a single worker \\
$\mathcal{J} = (J_1, \dotsc, J_p)$ & Job design; a partition of the tasks into $p$ jobs \\
$\handofftime{i}$ & Hand-off time incurred when a worker's final step is $s_i$ \\
$\handofftime{}(J_j)$ & Hand-off time of job $J_j$, equal to $\handofftime{i}$ for its final step $s_i$ \\
$\text{Wage}_j$ & Wage rate of job $J_j$'s worker, equal to the total skill it requires, $\sum_{T_b \in J_j} \skillcost{b}$ \\
$\text{WageBill}_j$ & Per-unit-output labor cost of job $J_j$: $\bigl(\sum_{T_b \in J_j} \skillcost{b}\bigr)\bigl(\handofftime{}(J_j) + \sum_{T_b \in J_j} \timecost{b}\bigr)$ \\
\bottomrule
\end{tabularx}

\begin{minipage}{1\linewidth}
\begin{spacing}{0.2}
{\fontsize{9.5pt}{10pt}\selectfont Notes: Steps are indexed by $i$, tasks by $b$, and jobs by $j$; superscripts $M$ and $A$ denote the manual and AI modes of executing a step, and $H$ a hand-off between consecutive workers. The long-run problem uses both panels.}
\end{spacing}
\end{minipage}
\par\endgroup
% Keep Table A.1 first among the exhibits of this appendix.
\FloatBarrier

\begin{table}[p]
    \caption{Role of Hand-off Costs in Organizational Structure}
    \vspace{-0.3cm}
    \centering
    \resizebox{0.72\textwidth}{!}{%
    \begin{tabular}{c}
    \begin{minipage}{\textwidth}
    \centering
    Panel (a): Without Hand-off Costs \\[1ex]
    \begin{tabular}{|
  >{\centering\arraybackslash}m{1.7cm}|
  >{\centering\arraybackslash}m{0.9cm}|
  >{\centering\arraybackslash}m{1.8cm}|
  >{\centering\arraybackslash}m{1.8cm}|
  >{\centering\arraybackslash}m{1.8cm}|
  >{\centering\arraybackslash}m{1.6cm}|
  >{\centering\arraybackslash}m{1.7cm}|
  >{\centering\arraybackslash}m{1.7cm}|
  }
\hline
\textbf{Job Design} & \textbf{Job} & \textbf{Job Tasks} & $\sum\hccost{b}$ & $\sum\timecost{b}$ & \textbf{Job Cost} & \textbf{Total Cost} & \textbf{Optimal Design} \\
\hline
\multirow{3}{*}{[1][2][3]} 
  & 1 & $\{1\}$   & 3 & 1 & 3 & \multirow{3}{*}{9} & \multirow{3}{*}{\checkmark} \\
\cline{2-6}
  & 2 & $\{2\}$   & 1 & 2 & 2 &  &  \\
\cline{2-6}
  & 3 & $\{3\}$   & 2 & 2 & 4 &  &  \\
\hline
\multirow{2}{*}{[1,2][3]}
  & 1 & $\{1,2\}$ & 4 & 3 & 12 & \multirow{2}{*}{16} &  \\
\cline{2-6}
  & 2 & $\{3\}$   & 2 & 2 & 4 &  &  \\
\hline
\multirow{2}{*}{[1][2,3]}
  & 1 & $\{1\}$   & 3 & 1 & 3 & \multirow{2}{*}{15} &  \\
\cline{2-6}
  & 2 & $\{2,3\}$ & 3 & 4 & 12 &  &  \\
\hline
[1,2,3] & 1 & $\{1,2,3\}$ & 6 & 5 & 30 & 30 &  \\
\hline
\end{tabular}
    \end{minipage} \\[2ex]

    \begin{minipage}{\textwidth}
    \centering
    Panel (b): Including Hand-off Costs \\[1ex]
    \begin{tabular}{|
  >{\centering\arraybackslash}m{1.7cm}|
  >{\centering\arraybackslash}m{0.9cm}|
  >{\centering\arraybackslash}m{1.8cm}|
  >{\centering\arraybackslash}m{1.8cm}|
  >{\centering\arraybackslash}m{1.8cm}|
  >{\centering\arraybackslash}m{1.6cm}|
  >{\centering\arraybackslash}m{1.7cm}|
  >{\centering\arraybackslash}m{1.7cm}|
  }
\hline
\textbf{Job Design} & \textbf{Job} & \textbf{Job Tasks} & $\sum\hccost{b}$ & $\handofftime{} + \sum\timecost{b}$ & \textbf{Job Cost} & \textbf{Total Cost} & \textbf{Optimal Design} \\
\hline
\multirow{3}{*}{[1][2][3]} 
  & 1 & $\{1\}$   & 3 & 4 & 12 & \multirow{3}{*}{18.5} &  \\
\cline{2-6}
  & 2 & $\{2\}$   & 1 & 2.5 & 2.5 &  &  \\
\cline{2-6}
  & 3 & $\{3\}$   & 2 & 2 & 4 &  &  \\
\hline
\multirow{2}{*}{[1,2][3]}
  & 1 & $\{1,2\}$ & 4 & 3.5 & 14 & \multirow{2}{*}{18} & \multirow{2}{*}{\checkmark} \\
\cline{2-6}
  & 2 & $\{3\}$   & 2 & 2 & 4 &  &  \\
\hline
\multirow{2}{*}{[1][2,3]}
  & 1 & $\{1\}$   & 3 & 4 & 12 & \multirow{2}{*}{24} &  \\
\cline{2-6}
  & 2 & $\{2,3\}$ & 3 & 4 & 12 &  &  \\
\hline
[1,2,3] & 1 & $\{1,2,3\}$ & 6 & 5 & 30 & 30 &  \\
\hline
\end{tabular}
    \end{minipage}
    \end{tabular}
    }
    \label{tab:job_design}
    \footnotesize
    \begin{minipage}{1\linewidth}
\begin{spacing}{0.2}
{\fontsize{9.5pt}{10pt}\selectfont Notes: This table reports total production costs for each possible job design in an example workflow with a fixed automation strategy and $n = 3$ tasks.
The four job designs are drawn in Figure~\ref{fig:job_design}, and the cost parameters of the tasks are given in Section~\ref{sec:longrun.specialization}.
A design's total is the sum over its jobs of the job's total skill times the total time it takes, hand-off included, which is Equation~\eqref{eq:wagebill} applied job by job.
In Panel (a), in the absence of hand-off frictions, full specialization, corresponding to job design [1][2][3], is cost-minimizing.
Once hand-off costs are taken into account, in Panel (b) the optimal organizational structure changes to job design [1,2][3], which avoids the large coordination cost between tasks 1 and 2 at the expense of employing a higher-cost (more skilled) worker to complete those tasks.}
\end{spacing}
\end{minipage}

\vspace{0.6cm}

\centering
\caption{Configurations, Costs, and Marginal Gains in Example~\ref{ex:nonmonotone}}
\label{tab:nonmonotone_costs}
\singlespacing
% \footnotesize like every other table float in the paper. At full size this tabular
% is 519pt wide against a 470pt measure and overhangs the right margin by 49pt, which
% is the document's one long-standing "Overfull \hbox (48.87pt too wide)" warning.
\footnotesize
\setlength{\tabcolsep}{10pt}
\begin{tabular}{lccc}
\toprule
Configuration & Time cost ($C$) & $-\mathrm{d}C/\mathrm{d}\alpha$ & Optimal for \\
\midrule
Both steps manual & $\manualTime{1} + \manualTime{2} = 14$ & $0$ & $\alpha \leq 0.50$ \\
\midrule
Step 1 augmented & $\AItime{1}/q_1 + \manualTime{2} = 3.5\alpha^{-11} + 8$ & $38.5/\alpha^{12}$ & never \\
\midrule
Step 2 augmented & $\manualTime{1} + \AItime{2}/q_2 = 6 + 4/\alpha$ & $4/\alpha^{2}$ & $0.50 \leq \alpha \lesssim 0.92$ \\
\midrule
Both steps augmented & $\AItime{1}/q_1 + \AItime{2}/q_2 = 3.5\alpha^{-11} + 4/\alpha$ & $38.5/\alpha^{12} + 4/\alpha^{2}$ & never \\
\midrule
Steps 1--2 chained & $\AItime{2}/(q_1 q_2) = 4\alpha^{-12}$ & $48/\alpha^{13}$ & $\alpha \gtrsim 0.92$ \\
\bottomrule
\end{tabular}
\par\vspace{\fignotesskip}
\begin{minipage}{1\linewidth}
\begin{spacing}{0.2}
{\fontsize{9.5pt}{10pt}\selectfont Notes: The two-step workflow of Example~\ref{ex:nonmonotone}, the AI-hard and AI-easy steps of Section~\ref{sec:comparative_advantage}, in which $\manualTime{1}=6$, $\manualTime{2}=8$, $\AItime{1}=3.5$, $\AItime{2}=4$, and the AI succeeds at step $i$ with probability $q_i = \alpha^{d_i}$ for $d_1 = 11$ and $d_2 = 1$, so that $q_1 = \alpha^{11}$ and $q_2 = \alpha$.
Costs follow from the task time costs of Section~\ref{sec:model.tasks}, and $-\mathrm{d}C/\mathrm{d}\alpha$ is the marginal benefit of improving AI quality.
The last column gives the range of $\alpha$ over which each configuration is cost-minimizing.
Appendix~\ref{app:nonmonotone} states and proves the general result the example illustrates.
Neither configuration in which step~1 is augmented ever minimizes cost.
Augmenting it on its own costs $3.5\alpha^{-11}$, which is prohibitive while AI is unreliable, and by the time AI is reliable enough for that to be affordable, folding step~1 into a chain is cheaper still, because an automated step is never prompted or verified on its own, so its management cost is never paid.}
\end{spacing}
\end{minipage}
\end{table}

\begin{figure}[p]
  \centering
  \caption{Neighboring Tasks' AI Execution Status and a Task's AI Execution: Placebo Distributions}
  \label{fig:DWA_regression_aiExecution_mainSample}
  \begin{subfigure}[b]{\textwidth}
    \centering
    \subcaption{No Fixed Effects}
    \label{fig:ame_nofe}
    \includegraphics[width=\textwidth]{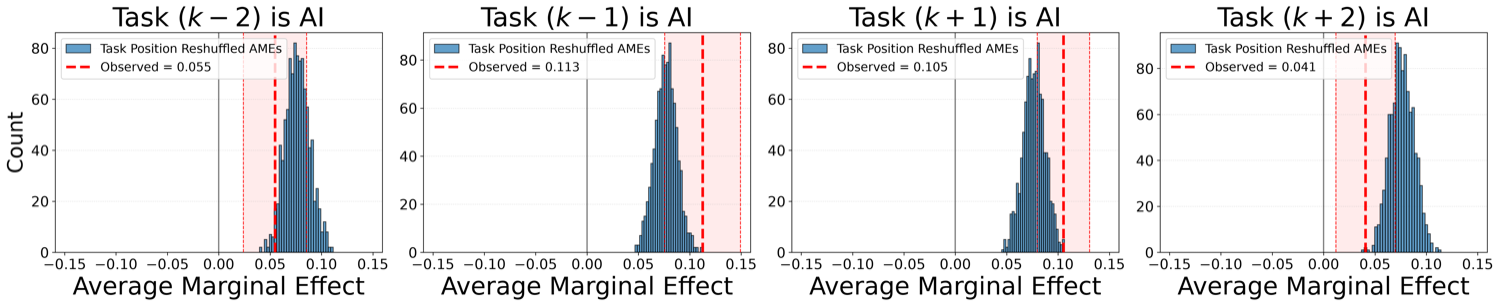}
  \end{subfigure}

  \vspace{1em}

  \begin{subfigure}[b]{\textwidth}
    \centering
    \subcaption{SOC Major Group Fixed Effects}
    \label{fig:ame_major}
    \includegraphics[width=\textwidth]{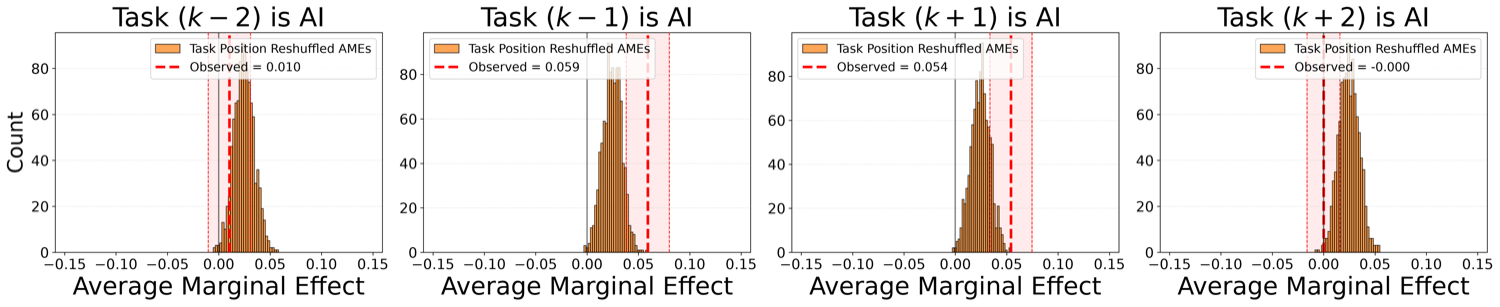}
  \end{subfigure}

  \vspace{1em}

  \begin{subfigure}[b]{\textwidth}
    \centering
    \subcaption{SOC Minor Group Fixed Effects}
    \label{fig:ame_minor}
    \includegraphics[width=\textwidth]{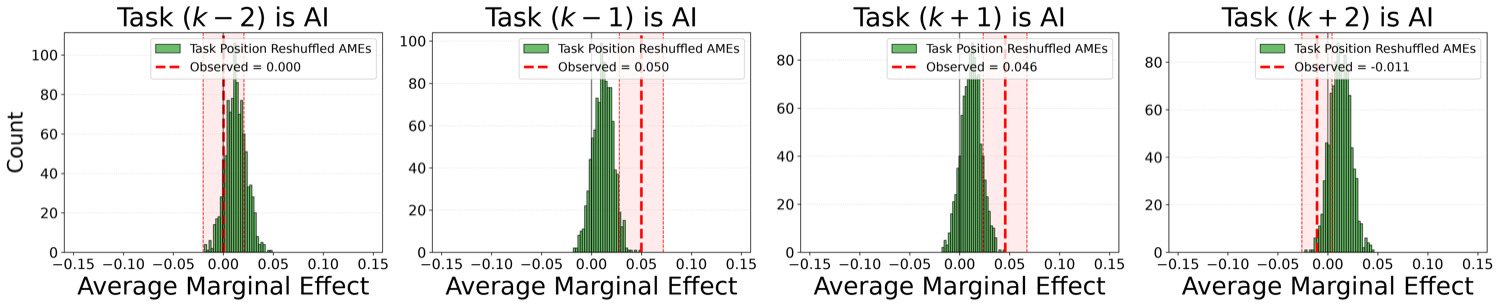}
  \end{subfigure}

  \vspace{1em}

  \begin{subfigure}[b]{\textwidth}
    \centering
    \subcaption{Detailed Work Activity Fixed Effects}
    \label{fig:ame_dwa}
    \includegraphics[width=\textwidth]{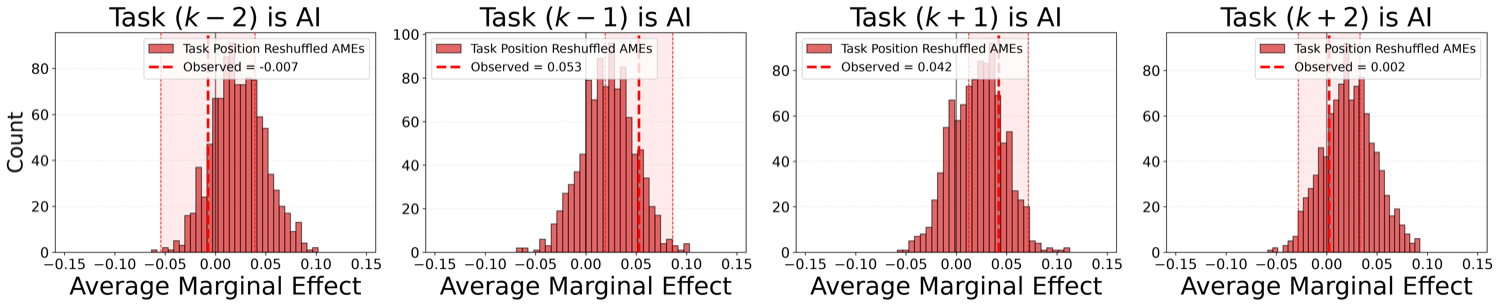}
  \end{subfigure}
  \par\vspace{\fignotesskip}
\begin{minipage}{1\linewidth}
\begin{spacing}{0.2}
{\fontsize{9.5pt}{10pt}\selectfont Notes: These graphs show that, among similar tasks appearing in multiple occupations, those sitting in a more AI-executed local context exhibit higher probabilities of being executed by AI, with the association concentrated on the steps immediately adjacent to the task.
The red dashed line in each graph indicates the observed average marginal effect of the corresponding variable in the original dataset reported in Table~\ref{tab:DWA_regression_aiExecution_mainSample}.
The red shaded area marks the 90\% confidence interval around the observed point estimates in the O*NET sample.
The histograms plot the distribution of average marginal effects across 1{,}000 randomized reshuffles of task positions within occupations.
Panel~(a) reports results from Equation~\eqref{eq:DWA_regression_ai}, while Panels~(b), (c), and (d) augment the baseline specification with SOC major group, SOC minor group, and DWA fixed effects, respectively.}
\end{spacing}
\end{minipage}
\end{figure}

% Print every exhibit of this appendix before Appendix B begins.
\FloatBarrier

\section{Omitted Proofs}
\label{app:omitted_proofs}

This appendix collects the proofs omitted from the main text, in the order the results appear there.

\subsection{Proof of Proposition~\ref{prop:ca_local}}
\label{app:ca_local_proof}

Throughout, write $\minTime{i} = \min\left\{\manualTime{i}, \frac{\AItime{i}}{q_i}\right\}$ for the cheapest way of executing step $i$ on its own, whether manually or as a standalone augmented step, and take $q_{k-1}, q_k, q_{k+1} \in (0,1)$.
Because chains are contiguous and none crosses the boundary of the block $\{k-1,\, k,\, k+1\}$, an arrangement of the block is a partition of it into consecutive runs, each run executed either as a single step at cost $\minTime{i}$ or as an AI chain augmented at its last step $r$, at cost $\AItime{r}/\prod_i q_i$ taken over the run's steps.

Before comparing them, note which of these arrangements can have AI execute step $k$ at all.
Executed as a singleton, step $k$ costs $\minTime{k} = \manualTime{k}$, since $\manualTime{k} < \AItime{k}/q_k$ by hypothesis, so augmenting it on its own is strictly dominated by performing it manually and never occurs in an optimum.
AI therefore executes step $k$ only if step $k$ belongs to a chain of length at least two, and since chains are contiguous and confined to the block, that chain is $\{k,\, k+1\}$, $\{k-1,\, k,\, k+1\}$, or $\{k-1,\, k\}$.
In the first two, step $k$ is not the chain's last step, so it is automated, while in the third it is the last step, hence augmented.

Pricing a singleton $\{k\}$ at $\manualTime{k}$ accordingly, the four partitions cost
\begin{align*}
V_0 &= \minTime{k-1} + \manualTime{k} + \minTime{k+1}, && \text{(no chain covers step $k$)} \\
V_1 &= \minTime{k-1} + \frac{\AItime{k+1}}{q_k\, q_{k+1}}, && \text{(chain $\{k,\, k+1\}$)} \\
V_2 &= \frac{\AItime{k+1}}{q_{k-1}\, q_k\, q_{k+1}}, && \text{(chain $\{k-1,\, k,\, k+1\}$)} \\
V_3 &= \frac{\AItime{k}}{q_{k-1}\, q_k} + \minTime{k+1}. && \text{(chain $\{k-1,\, k\}$)}
\end{align*}
AI executes step $k$ exactly when $\min\{V_1, V_2, V_3\} < V_0$, ties going to manual execution.

\begin{proof}[Proof of part (i)]
Fix $(\manualTime{k}, \AItime{k}, q_k)$ with $\manualTime{k} < \AItime{k}/q_k$ and fix $q_{k-1}, q_{k+1} \in (0,1)$.
Restrict attention to the two-parameter family of neighbor costs given by $\manualTime{k-1} = \AItime{k-1} = \mu > 0$ and $\manualTime{k+1} = 1$, $\AItime{k+1} = \tau > 0$, so that $\minTime{k-1} = \mu$ and $\minTime{k+1} = \min\{1, \tau/q_{k+1}\}$.
Since $V_0, \dotsc, V_3$ are continuous in all of the neighbor parameters, it suffices to exhibit, for each configuration, values of $(\mu, \tau)$ at which its cost is strictly below the other three; strictness then extends the conclusion to an open neighborhood.

\begin{itemize}
\item \emph{Configuration $\{k,\, k+1\}$, which automates step $k$.}
Set $\tau = q_k\, q_{k+1} \manualTime{k}/2$, so that $V_1 = \mu + \manualTime{k}/2$, and take
\[ 0 < \mu < \min\left\{ \frac{\manualTime{k}}{2}\cdot\frac{1-q_{k-1}}{q_{k-1}},\ \frac{\AItime{k}}{q_{k-1}\, q_k} - \frac{\manualTime{k}}{2} \right\}, \]
a nonempty interval because $q_{k-1} < 1$ and $\AItime{k}/(q_{k-1} q_k) > \AItime{k}/q_k > \manualTime{k}$.
The three comparisons are then
\[
V_0 - V_1 = \frac{\manualTime{k}}{2} + \minTime{k+1}, \quad
V_2 - V_1 = \frac{\manualTime{k}}{2}\cdot\frac{1-q_{k-1}}{q_{k-1}} - \mu, \quad
V_3 - V_1 = \frac{\AItime{k}}{q_{k-1}\, q_k} - \frac{\manualTime{k}}{2} - \mu + \minTime{k+1},
\]
and all three are positive: the first because $\minTime{k+1} > 0$, and the second and third by the first and second bounds on $\mu$ respectively.

\item \emph{Configuration $\{k-1,\, k,\, k+1\}$, which also automates step $k$.}
Set $\mu = 1$ and take $0 < \tau < q_{k+1} \min\{q_{k-1}\, q_k,\ \AItime{k}\}$, so that the long chain is cheap.
The three comparisons are
\begin{align*}
V_0 - V_2 &= 1 + \manualTime{k} + \minTime{k+1} - \frac{\tau}{q_{k-1}\, q_k\, q_{k+1}}, \\
V_1 - V_2 &= 1 + \frac{\tau}{q_k\, q_{k+1}} - \frac{\tau}{q_{k-1}\, q_k\, q_{k+1}}, \\
V_3 - V_2 &= \frac{\AItime{k}}{q_{k-1}\, q_k} + \minTime{k+1} - \frac{\tau}{q_{k-1}\, q_k\, q_{k+1}}.
\end{align*}
The bound on $\tau$ puts $\tau/(q_{k-1} q_k q_{k+1})$ below $1$ and $\AItime{k}/(q_{k-1} q_k)$, so all three are positive.

\item \emph{Configuration $\{k-1,\, k\}$, which augments step $k$.}
Set $\mu = \AItime{k}/(q_{k-1} q_k)$ and take $\tau > \max\{q_{k+1},\ q_k q_{k+1},\ q_{k-1} q_k q_{k+1} (\mu+1)\}$.
The first of these bounds puts $\tau/q_{k+1}$ above $1$, so that $\minTime{k+1} = 1$ and $V_3 = \mu + 1$, and the comparisons are
\[
V_0 - V_3 = \manualTime{k}, \qquad
V_1 - V_3 = \frac{\tau}{q_k\, q_{k+1}} - 1, \qquad
V_2 - V_3 = \frac{\tau}{q_{k-1}\, q_k\, q_{k+1}} - (\mu + 1),
\]
positive because $\manualTime{k} > 0$ and by the second and third bounds on $\tau$ respectively.
\end{itemize}

The margin of the human's advantage on step $k$, $\AItime{k}/q_k - \manualTime{k}$, was unrestricted throughout, which gives the second claim of part~(i).
\end{proof}

\begin{proof}[Proof of part (ii)]
Fix the six cost parameters of the three steps and let their success probabilities vary over $Q = (0,1)^3$.
Write
\[ \mathcal{A} \;=\; \bigl\{ (q_{k-1},\, q_k,\, q_{k+1}) \in Q \; : \; \text{AI executes step } k \bigr\} \]
for the set of reliabilities at which the cost-minimizing arrangement does not leave step $k$ to the human.
Part~(ii) is precisely the statement that $\mathcal{A}$ is \emph{upward closed} in the coordinatewise order, meaning that whenever a triple belongs to $\mathcal{A}$, so does every triple at least as large in all three coordinates.

Before the main argument, dispose of the points at which the standing hypothesis lapses.
Raising $q_k$ can push $\AItime{k}/q_k$ below $\manualTime{k}$, and at any such point step $k$ is cheaper augmented than performed manually even on its own, so AI executes it in every arrangement of the block and the conclusion holds there regardless of the other two coordinates.
It is therefore enough to prove the claim when the human still holds the comparative advantage at the larger triple, $\manualTime{k} \leq \AItime{k}/q'_k$.
In that case $q_k \leq q'_k$ gives $\manualTime{k} \leq \AItime{k}/q'_k \leq \AItime{k}/q_k$, so the advantage is in force at the smaller triple and everywhere between, and the four arrangements priced above describe the block throughout.
The proof turns the upward-closure statement into one about the three differences
\[ D_j \;=\; V_j - V_0, \qquad j = 1,\, 2,\, 3, \]
each of which compares one arrangement that has AI execute step $k$ against the arrangement that does not.
Three reductions get us there.

% sloppypar for this paragraph only. Its first line is a long run-in italic head followed
% by two inline formulas with no break points of their own, and at the house \tolerance
% TeX's least-bad option overran the right margin by 6pt. \sloppy's \emergencystretch
% lets it break before the \min instead, and its \hfuzz=0.5pt means an overfull line here
% would actually be reported: the preamble sets \hfuzz=100pt document-wide, so the log is
% silent about margin overruns (check them with the glyph-box walk in README).
% The environment has to close after the display, which does not end this paragraph.
\begin{sloppypar}
\emph{Reduction 1: from $\mathcal{A}$ to one chain at a time.}
AI executes step $k$ exactly when $\min\{V_1, V_2, V_3\} < V_0$, that is, when $D_j < 0$ for at least one $j$.
Writing $\mathcal{A}_j = \{(q_{k-1},\, q_k,\, q_{k+1}) \in Q : D_j < 0\}$ for the set on which chain $j$ by itself beats manual execution of step $k$, this says
\[ \mathcal{A} \;=\; \mathcal{A}_1 \cup \mathcal{A}_2 \cup \mathcal{A}_3 . \]
A union of upward-closed sets is upward closed, because a triple in the union lies in one of them and every larger triple then lies in that same one.
It is therefore enough to show that each $\mathcal{A}_j$ is upward closed, and the three can be handled separately.
We need not know which chain is optimal, only whether each alone beats leaving the step to the human.
\end{sloppypar}

\emph{Reduction 2: from all three coordinates at once to one at a time.}
Suppose each $\mathcal{A}_j$ is upward closed in each single coordinate, at every fixed value of the other two.
Given $(q_{k-1},\, q_k,\, q_{k+1}) \in \mathcal{A}_j$ and a larger triple, raise the coordinates in turn:
\[ (q_{k-1},\, q_k,\, q_{k+1}) \;\implies\; (q'_{k-1},\, q_k,\, q_{k+1}) \;\implies\; (q'_{k-1},\, q'_k,\, q_{k+1}) \;\implies\; (q'_{k-1},\, q'_k,\, q'_{k+1}) , \]
each step by closure in the coordinate being raised, holding the other two at the values they have reached.
The one-coordinate statements deliver the joint one, provided each holds at every value of the coordinates being held fixed, as they are proved below.

\emph{Reduction 3: from upward closure to a property of $D_j$.}
With the other two coordinates fixed, $D_j$ is a function of a single success probability $q$, and $\mathcal{A}_j$ is the set on which that function is negative.
That set is upward closed as soon as $D_j$, once negative, stays negative as $q$ rises.
It suffices for $D_j$ to be nonincreasing in $q$, which is what we verify in every case but one; in the exceptional case $D_j$ increases, and we show instead that it is negative on the entire range where the increase can occur.

Writing the three differences out,
\[
D_1 = \frac{\AItime{k+1}}{q_k\, q_{k+1}} - \manualTime{k} - \minTime{k+1}, \quad
D_2 = \frac{\AItime{k+1}}{q_{k-1}\, q_k\, q_{k+1}} - \minTime{k-1} - \manualTime{k} - \minTime{k+1}, \quad
D_3 = \frac{\AItime{k}}{q_{k-1}\, q_k} - \minTime{k-1} - \manualTime{k},
\]
the terms that move with the neighbor reliabilities are the chain costs and the two standalone costs
\[ \minTime{k-1} = \min\left\{ \manualTime{k-1},\, \frac{\AItime{k-1}}{q_{k-1}} \right\}, \qquad \minTime{k+1} = \min\left\{ \manualTime{k+1},\, \frac{\AItime{k+1}}{q_{k+1}} \right\}, \]
of which the first depends on $q_{k-1}$ alone and the second on $q_{k+1}$ alone.

\emph{A claim covering both coordinates.}
Let $m,\, A,\, B,\, c > 0$ and define, for $q \in (0,1)$,
\[ D(q) \;=\; \underbrace{\frac{A}{q}}_{\text{a chain through the neighbor}} \;-\; \underbrace{\min\left\{ m,\, \frac{B}{q} \right\}}_{\text{the neighbor on its own}} \;-\; \underbrace{c}_{\text{fixed}} . \]
Then $\{ q \in (0,1) : D(q) < 0 \}$ is upward closed.
To see this, note that the neighbor's standalone cost is the smaller of a manual cost $m$ that does not move with $q$ and an augmented cost $B/q$ that falls as $q$ rises, the two crossing at $q^{*} = B/m$.
Splitting there,
\[
D(q) \;=\;
\begin{cases}
\dfrac{A}{q} - m - c, & 0 < q \leq q^{*}, \\[2.2ex]
\dfrac{A - B}{q} - c, & q^{*} < q < 1,
\end{cases}
\]
the second line obtained by substituting $B/q$ for the minimum.
The two expressions agree at $q = q^{*}$, where $m = B/q^{*}$, so $D$ is continuous; and if $q^{*} \geq 1$ the second branch is empty and only the first arises.
On the first branch $D$ is strictly decreasing, since
\[ \frac{\mathrm{d}}{\mathrm{d}q}\left( \frac{A}{q} - m - c \right) \;=\; -\frac{A}{q^{2}} \;<\; 0 , \]
so raising $q$ cheapens the chain while leaving untouched the manual cost that the chain displaces.
On the second branch the behavior depends on the sign of $A - B$:
\begin{itemize}
\item if $A \geq B$, then $(A-B)/q$ is nonincreasing in $q$, so $D$ is nonincreasing on the branch;
\item if $A < B$, then $D$ increases with $q$, but $(A-B)/q < 0$ and $-c < 0$, so $D(q) < 0$ at \emph{every} point of the branch.
\end{itemize}
Either way, if $D$ is negative at some point of the second branch it is negative everywhere to the right of it.
Now take $q < q'$ in $(0,1)$ with $D(q) < 0$, and consider the three possible positions of the pair relative to $q^{*}$:
\begin{itemize}[itemsep=0pt, parsep=0pt]
\item if $q' \leq q^{*}$, both lie on the first branch and $D(q') < D(q) < 0$;
\item if $q^{*} < q$, both lie on the second branch, which was just handled;
\item if $q \leq q^{*} < q'$, then $D(q^{*}) \leq D(q) < 0$ by the first branch, and since the second-branch expression agrees with $D$ at $q^{*}$ and the argument above applies to it on $[q^{*},\, 1)$, we get $D(q') < 0$.
\end{itemize}
This establishes the claim.
Its content is that raising a neighbor's reliability always cheapens a chain that runs through that neighbor, and cheapens the neighbor's own execution only once that execution is an augmented one, in which case both costs carry the same factor $1/q$ and fall in the same proportion, so the comparison between them cannot reverse.

\emph{Closure in $q_{k-1}$.}
Fix $q_k$ and $q_{k+1}$ and let $q = q_{k-1}$ vary.
The difference $D_1$ does not depend on $q$ at all, because the two arrangements it compares, the chain $\{k,\, k+1\}$ and manual execution of step $k$, both leave the predecessor to be executed on its own, so the term $\minTime{k-1}$ appears in $V_1$ and $V_0$ alike and cancels.
A constant function is nonincreasing, so $\mathcal{A}_1$ is upward closed in this coordinate.
The other two differences are instances of the claim, with $B = \AItime{k-1}$ and $m = \manualTime{k-1}$ throughout and
\[
D_2\!: \ A = \frac{\AItime{k+1}}{q_k\, q_{k+1}}, \ c = \manualTime{k} + \minTime{k+1}, \qquad
D_3\!: \ A = \frac{\AItime{k}}{q_k}, \ c = \manualTime{k},
\]
neither of which moves with $q$.
The claim therefore makes $\mathcal{A}_2$ and $\mathcal{A}_3$ upward closed in $q_{k-1}$.
Nothing here ties $A$ to $B$, so both of the claim's second-branch cases can occur.
A predecessor whose own augmented cost falls faster than the cost of the chain it would anchor is exactly the case $A < B$, and it is settled not by monotonicity but by the chain's being cheaper throughout that branch to begin with.

\emph{Closure in $q_{k+1}$.}
The successor coordinate is the counterpart of the predecessor one, and the milder half of it, so the claim disposes of it directly.
The difference $D_3$ now plays the role $D_1$ played above and drops out for the mirror-image reason, namely that the chain $\{k-1,\, k\}$ and manual execution of step $k$ both leave the successor to be executed on its own, so $\minTime{k+1}$ cancels between $V_3$ and $V_0$.
The remaining two differences are again instances of the claim, read in $q = q_{k+1}$, with $B = \AItime{k+1}$ and $m = \manualTime{k+1}$ throughout and
\[
D_1\!: \ A = \frac{\AItime{k+1}}{q_k}, \ c = \manualTime{k}, \qquad
D_2\!: \ A = \frac{\AItime{k+1}}{q_{k-1}\, q_k}, \ c = \minTime{k-1} + \manualTime{k}.
\]
The one difference from the predecessor case is that $A$ and $B$ are no longer free of each other.
A chain terminating at the successor pays that same verification cost $\AItime{k+1}$, inflated by the failure risk of the steps ahead of it, so $A/B$ equals $1/q_k$ or $1/(q_{k-1} q_k)$ and is at least $1$.
Only the case $A \geq B$ of the claim can arise, and $\mathcal{A}_1$ and $\mathcal{A}_2$ are upward closed in $q_{k+1}$, with the differences nonincreasing outright.

\emph{Closure in $q_k$.}
This coordinate needs no work beyond the pricing already in hand.
With the human's advantage in force, the no-chain arrangement performs step $k$ manually and prices the two neighbors on their own, so $V_0$ does not involve $q_k$ at all, while each of $V_1$, $V_2$, and $V_3$ carries a factor $1/q_k$ in the cost of the chain covering step $k$ and is therefore strictly decreasing in it.
Every $D_j$ is thus strictly decreasing in $q_k$, and each $\mathcal{A}_j$ is upward closed in that coordinate.
This is the sense in which the step's own reliability and its neighbors' enter differently, in that a better neighbor cheapens the chain and the alternative to it at once, and the argument above is what settles the race between them, whereas a better step $k$ cheapens only the chain.

Each $\mathcal{A}_j$ is therefore upward closed in each coordinate separately, hence upward closed in the coordinatewise order by Reduction~2, and $\mathcal{A}$ is upward closed by Reduction~1.
\end{proof}

\subsection{Proof of Proposition~\ref{prop:fragmentation}}
\label{app:fragmentation_proof}

Throughout we maintain the normalization $\AItime{i} = 1$ for all $s_i$ and use the notation of Section~\ref{sec:fragmentation}: $F$ for the set of steps that fail, $\mathcal{C}$ for the maximal connected components of non-failed steps, $\omega(C)$ for their weights, and \eqref{eq:fragmentation} for the realized fragmentation whose expectation is the fragmentation index.
We additionally write $\minTime{i} = \min\left\{\manualTime{i}, \frac{1}{q_i}\right\}$ for the minimum cost of implementing step $s_i$ on its own, whether manually or through AI augmentation.
For a fixed sequence $\mathcal{S} = (s_1, \dotsc, s_m)$ we write $FI$ for its fragmentation index and $OPT$ for the cost of the optimal AI deployment strategy for its steps.

Both bounds compare $FI$ against the cost of a \emph{particular} AI deployment strategy rather than against $OPT$ directly, since $OPT$ is itself a minimum over exponentially many arrangements; the two directions differ in which strategy they use.
For the upper bound we bound the realized fragmentation task by task against an arbitrary strategy and then evaluate that bound at the cost-minimizing one, so the quantity being compared to is exactly $OPT$.
Manual tasks contribute the same amount to both sides, and the entire loss is incurred on AI chains, where it is a factor of $\tfrac{5}{4}$.
For the lower bound we instead construct a greedy strategy that extends a chain for as long as its success probability stays at least $1/2$; because that strategy is feasible its cost $ALG$ is at least $OPT$, so it suffices to bound $FI$ below by a fraction of $ALG$.
A charging argument supplies that fraction, assigning each of the greedy strategy's tasks a share of the realized fragmentation worth at least a quarter of what the task costs, and Lemma~\ref{lem:FI.lower.bound.8} repeats the argument without the assumption $\manualTime{i} \geq 1$ at the cost of a further factor of two.
Proposition~\ref{prop:fragmentation} then combines the three lemmas.

We begin with the upper bound.

\begin{lem}
\label{lem:FI.upper.bound}
    $FI \leq \tfrac{5}{4}OPT$.
\end{lem}
\begin{proof}
Fix an arbitrary task sequence $\T = (T_1, \dotsc, T_n)$.
Note that $FI$ depends only on the step sequence $\mathcal{S}$.
Every step carries its own success probability $q_i$ whatever $\T$ does with it, so the realization of $F$ is unaffected by $\T$.
We use $\T$ only to organize the terms of $FI$, grouping steps so that they can be compared with the cost $\T$ incurs on them.
Partition this into two sets of tasks: $\T_{\manualLetter}$ containing all (singleton) tasks that are completed manually, and $\T_{\AIletter}$ containing AI chains.
We first claim that
\[ 
FI 
\leq 
\sum_{(s_i) \in \T_{\manualLetter}}\manualTime{i}
+
\sum_{T_b \in \T_{\AIletter}}\left(1 + \sum_{s_i \in T_b}(1-\alpha^{d_i})(1+\minTime{i})\right). 
\]

To see why, first consider a singleton $(s_i) \in \T_{\manualLetter}$, which $\T$ executes manually at cost $\manualTime{i}$.
If $s_i$ lies in $F$, its own term in the realized fragmentation is $\minTime{i} \leq \manualTime{i}$; otherwise it lies in some connected component $C$ and raises that component's weight $\omega(C)$ by at most $\manualTime{i}$.
Either way it accounts for at most $\manualTime{i}$ of the realized fragmentation, matching what it contributes to the cost of $\T$.

Next consider the AI chains.
For any $T_b \in \T_{\AIletter}$, consider the set of steps in $T_b$ that fail in any given realization of $F$.
If no steps in $T_b$ fail, then $T_b$ collectively contributes at most $1$ to the realized fragmentation (since all $s_i \in T_b$ lie in the same connected component of non-failed steps).
Otherwise, each failed step $s_i \in T_b$ contributes $\minTime{i}$ to the realized fragmentation for itself plus an additional value of at most $1$ for the weight of a potential new connected component of non-failed steps.
As each step $s_i \in T_b$ fails independently with probability $(1-\alpha^{d_i})$, we get that the contribution of the steps in $T_b$ to the fragmentation index is at most $1 + \sum_{s_i \in T_b}(1-\alpha^{d_i})(1+\minTime{i})$ as claimed.

On the other hand, recall that if we take $\T$ to be the cost-minimizing task sequence, then
\[ 
OPT 
= 
\sum_{(s_i) \in \T_{\manualLetter}} \manualTime{i} 
+
\sum_{T_b \in \T_{\AIletter}} \alpha^{-d(T_b)}
\]
by definition, where we write $d(T_b) = \sum_{s_i \in T_b}d_i$ for the total difficulty of steps in task $T_b$.

Comparing our expression for $OPT$ with our bound on $FI$, we see that each task in $\T_{\manualLetter}$ contributes the same amount to each, whereas each $T_b \in \T_{\AIletter}$ contributes $1+\sum_{s_i \in T_b}(1-\alpha^{d_i})(1+\minTime{i})$ to the latter and $\alpha^{-d(T_b)}$ to the former.
We will show that, for each $T_b \in \T_{\AIletter}$, $1 + \sum_{s_i \in T_b}(1-\alpha^{d_i})(1+\minTime{i}) \leq \frac{5}{4} \alpha^{-d(T_b)}$, which will complete the proof.

First, since $\minTime{i} \leq \alpha^{-d_i}$ for each $s_i \in \mathcal{S}$, we have $1 + \sum_{s_i \in T_b}(1-\alpha^{d_i})(1+\minTime{i}) \leq 1 + \sum_{s_i \in T_b}(1-\alpha^{d_i})(1+\alpha^{-d_i})$.

Next note that for any $x,y\in [0,1]$, $(1-x)(1+\tfrac{1}{x}) + (1-y)(1+\tfrac{1}{y}) \leq (1-xy)(1+\tfrac{1}{xy}) $.
Repeatedly applying this fact, we get that $1 + \sum_{s_i \in T_b}(1-\alpha^{d_i})(1+\alpha^{-d_i}) \leq 1 + (1-\alpha^{d(T_b)})(1+\alpha^{-d(T_b)}) = 1 + \alpha^{-d(T_b)} - \alpha^{d(T_b)}$ for any $T_b \in \T_{\AIletter}$.
The ratio between this expression and $\alpha^{-d(T_b)}$ is maximized at $\alpha^{d(T_b)} = 1/2$ (equivalently $\alpha^{-d(T_b)} = 2$), achieving a value of $5/4$ as claimed.
\end{proof}

The bound of $\tfrac{5}{4}$ in Lemma~\ref{lem:FI.upper.bound} is tight, as the following example shows.

\begin{example}
\label{ex:FI.tight}
Consider a three-step workflow in which the first and last steps always succeed, while the second succeeds with probability $1/2$ and is twice as costly to perform manually:
\[ \bigl(\manualTime{i},\, \AItime{i},\, q_i\bigr)_{i=1,2,3} = (1,\,1,\,1),\quad (2,\,1,\,\tfrac12),\quad (1,\,1,\,1). \]
The optimal policy chains all three together for an expected cost of $\AItime{3}/(q_1 q_2 q_3) = 1/(1 \cdot \tfrac12 \cdot 1) = 2$, whereas the fragmentation index is $\tfrac12 \times 1 + \tfrac12 \times (1+2+1) = \tfrac52$.
The ratio $FI/OPT$ is therefore $\tfrac54$.
\end{example}

We next argue that $FI$ is not much smaller than $OPT$.
We begin with an analysis under the assumption that $\manualTime{i} \geq 1$ for all $s_i \in \mathcal{S}$.
That is, for each step $s_i$, manual completion costs at least as much as verifying one AI attempt at it (with guaranteed success).

\begin{lem}
\label{lem:FI.lower.bound.4}
    Suppose $\manualTime{i} \geq 1$ for all $s_i \in \mathcal{S}$.
    Then $FI \geq \tfrac{1}{4}OPT$.
\end{lem}
\begin{proof}

    Consider the following task sequence.
    Beginning with the first step $s_1$, consider the maximal contiguous sequence of steps $T$ whose success probability $\alpha^{d(T)}$ is greater than or equal to $1/2$.
    If that set is empty (i.e., the first step has success probability strictly less than $1/2$) then the greedy gives up on chaining and executes $s_1$ on its own, either manually or augmented, whichever is cheaper.
    We call the resulting singleton task an \emph{individual} task.
    Otherwise, the sequence $T = (s_1, \dotsc, s_\ell)$ is added to the task sequence as an AI chain, and we call it a \emph{non-individual} task.
    Note that a non-individual task may itself be a singleton, namely when $\ell = 1$; what separates the two labels is not the number of steps a task holds but whether its run passed the $1/2$ test.
    This procedure is then repeated beginning with the next step (which is $s_2$ in the former case, or $s_{\ell+1}$ in the latter case), until all steps have been added to a task.
    Call the resulting task sequence $\T$.

    \medskip

    Let $ALG$ denote the cost of the task sequence $\T$.
    Note that we must have $ALG \geq OPT$.
    We will argue that $FI \geq ALG/4$, which will prove the claim.

    \medskip

    We first define some notation.
    Write $\T_I$ for the set of individual tasks in $\T$ and $\T_{NI}$ for the set of non-individual tasks.
    Note that $\alpha^{d_i} < 1/2$ for all $(s_i) \in \T_I$, by construction.
    Note also that $\T_{\manualLetter} \subseteq \T_I$ (recalling that $\T_{\manualLetter}$ is the set of all singleton tasks executed manually), and this containment may be strict if $\T_I$ contains singleton tasks that are augmented.
    We emphasize that $\T_{NI}$ may still include singleton tasks, but any singleton task $(s_i) \in \T_{NI}$ must have the property that $\alpha^{d_i} \geq 1/2$.

    If the final step $s_m$ from $\mathcal{S}$ is contained in some $T_b \in \T_{NI}$, we say that $T_b$ is the \emph{terminal} task.
    All other tasks in $\T_{NI}$ are non-terminal tasks.
Note that if the final step is a singleton task in $\T_I$ then no task is terminal.
    For each non-terminal task $T_b \in \T_{NI}$, we'll write $\overline{T}_b$ to denote $T_b$ plus the first step that follows $T_b$.
    For example, if $T_b = (s_i, \dotsc, s_\ell)$, then $\overline{T}_b = (s_i, \dotsc, s_\ell, s_{\ell+1})$.
    The set $\T_{NI}$ then has the property that for each task $T_b \in \T_{NI}$ we have $\alpha^{d(T_b)} \geq 1/2$, and for each non-terminal $T_b \in \T_{NI}$ we have $\alpha^{(d(\overline{T}_b))} < 1/2$.

    \medskip

    We are now ready to relate $FI$ and $ALG$.
    Let $N = |\T_{NI}|$ be the number of non-individual tasks.
    Note first that
    \[ 
    ALG 
    = 
    \sum_{(s_i) \in \T_I} \minTime{i} 
    + 
    \sum_{T_b \in \T_{NI}} \alpha^{-d(T_b)} 
    \;\leq\;
    \sum_{(s_i) \in \T_I} \minTime{i}
    +
    2N 
    \]
    This is because $\alpha^{d(T_b)} \geq 1/2$ for each $T_b \in \T_{NI}$, and hence $\alpha^{-d(T_b)} \leq 2$.

    \medskip

    Next we bound $FI$.
    Note that the assumption $\manualTime{i} \geq 1$ implies that, in any realization of the connected components of non-failed steps $\mathcal{C}$, $\omega(C) = 1$ for all $C \in \mathcal{C}$.
    This implies that the realized fragmentation is equal to $\mathbf{1}\{s_m \notin F\}$ plus, for each $s_i \in F$ (i.e., each step $s_i$ that fails), a contribution of $\minTime{i}$ plus an additional $1$ in the event that $i > 1$ and the step immediately preceding $s_i$ did not fail.
    (This additional cost of $1$ accounts for creating a new connected component of non-failed steps ending with $s_{i-1}$.)

    We claim that $FI \geq N/2 + \frac{1}{2}\sum_{(s_i) \in \T_I} \minTime{i}$.
    We show this by charging the realized fragmentation to tasks in $\T_{NI}$ and $\T_I$, so that each $T_b \in \T_{NI}$ is charged at least $1$ with probability at least $1/2$, and each $(s_i) \in \T_I$ is charged at least $\minTime{i}$ with probability at least $1/2$.

    To see this, let $E_b$ denote the event that some step in $\overline{T}_b$ fails, for each non-terminal $T_b \in \T_{NI}$.
    Note that $E_b$ occurs with probability at least $1/2$, since $\alpha^{d(\overline{T}_b)} < 1/2$.
    Also, if $E_b$ occurs, then either some step $s_i \in T_b$ fails, or no step in $T_b$ fails but the step $s_j$ immediately following $T_b$ does.
    In the former case, we will charge the $\minTime{i} \geq 1$ contribution of $s_i$'s failure to $T_b$.
    In the latter case the step immediately preceding $s_j$ is the last step of $T_b$, which did not fail, so $s_j$'s failure creates a new connected component and we charge the additional contribution of $1$ that it carries (as described above) to $T_b$.
    Additionally, if there is a terminal $T_b$ in $\T_{NI}$, then we charge the baseline $\mathbf{1}\{s_m \notin F\}$ from the calculation of $FI$ to that terminal $T_b$.
    That baseline is charged with probability $q_m = \alpha^{d_m}$, and since $d(T_b) \geq d_m$ we have $q_m \geq \alpha^{d(T_b)} \geq 1/2$, so the terminal task too is charged at least $1$ with probability at least $1/2$.
    Aggregating over all these cases, we have that at least $1$ is charged to each task $T_b \in \T_{NI}$ with probability at least $1/2$.
    Finally, for each $(s_i) \in \T_I$ that fails (which occurs with probability at least $1/2$, by construction), we charge the $\minTime{i}$ portion of its contribution to task $(s_i)$.
    We conclude that $FI \geq N/2 + \frac{1}{2}\sum_{(s_i) \in \T_I}\minTime{i}$ as claimed.

    But now since $FI \geq N/2 + \frac{1}{2}\sum_{(s_i) \in \T_I}\minTime{i}$ and $ALG \leq 2N + \sum_{(s_i) \in \T_I}\minTime{i}$, we conclude $FI \geq ALG/4$ as claimed.
\end{proof}

The approximation factor in Lemma~\ref{lem:FI.lower.bound.4} cannot be improved below $\tfrac{4\sqrt{2}}{9(\sqrt{2}-1)} \approx 1.52$, although we do not have a matching lower bound of $4$.

\begin{example}
\label{ex:FI.lower.gap}
Consider a sequence of $m$ identical steps:
\[ 
\bigl(\manualTime{i},\, \AItime{i},\, q_i\bigr) 
= 
\bigl(\sqrt{2},\, 1,\, 1/\sqrt{2}\bigr) \ \text{ for all $i$}. 
\]
The optimal task sequence chains the steps three at a time, since a chain of $k$ steps costs $2^{k/2}$, or $2^{k/2}/k$ per step, which is minimized at $k=3$ with value $\tfrac{2\sqrt{2}}{3} \approx 0.9428$, below both the manual cost $\sqrt{2}$ and the per-step cost of any other chain length, so $OPT \to \tfrac{2\sqrt{2}}{3}\,m$.
The fragmentation index is such that each step contributes $\sqrt{2}$ if it fails, plus $1$ if the preceding step did not fail, for a total contribution of $(1-1/\sqrt{2})(\sqrt{2} + 1/\sqrt{2}) = \tfrac{3}{2}(\sqrt{2}-1) \approx 0.6213$.
As $m$ grows large, the ratio between $\tfrac{2\sqrt{2}}{3}m$ and $0.5 + m \times 0.6213$ approaches $\tfrac{4\sqrt{2}}{9(\sqrt{2}-1)} \approx 1.52$.
\end{example}

We can relax the assumption in Lemma~\ref{lem:FI.lower.bound.4} that $\manualTime{i} \geq 1$ for all $s_i$, but at the cost of a worse approximation factor.

\begin{lem}
\label{lem:FI.lower.bound.8}
    For general $\manualTime{i}$ (i.e., allowing $\manualTime{i} < 1$ for some steps $s_i$), $FI \geq \frac{1}{8}OPT$.
\end{lem}
\begin{proof}
    The argument follows the proof of Lemma~\ref{lem:FI.lower.bound.4}, with the following changes.

    First we make a slight change in the definition of the task sequence $\T$ that defines $ALG$.
    If a constructed task $T_b$ has the property that $\sum_{s_i \in T_b} \manualTime{i} < 1$ (which implies it is cheaper to run all steps of $T_b$ manually than as an AI chain that is guaranteed to succeed), we switch to running the steps of $T_b$ manually in the task sequence.
    In our analysis, we still consider $T_b$ to be a task in $\T_{NI}$; but its cost is taken to be $\sum_{s_i \in T_b} \manualTime{i}$.

    This modification changes our upper bound on the total cost of $ALG$ to the following:
    \[ 
    ALG 
    \leq 
    \sum_{ (s_i) \in \T_I} \minTime{i}
    + 
    2 \sum_{T_b \in \T_{NI}} \min\left\{1, \sum_{s_i \in T_b} \manualTime{i} \right\} . 
    \]

    Then, to bound $FI$, we claim that
    \[ 
    FI 
    \geq 
    \frac{1}{2}\sum_{(s_i) \in \T_I} \minTime{i} 
    +
    \frac{1}{4}\sum_{T_b \in \T_{NI}} \min\left\{1, \sum_{s_i \in T_b} \manualTime{i} \right\} 
    \]
    which would prove the claim.

    To see why this bound on $FI$ holds, we employ a charging argument as before.
    Recall that for each non-terminal $T_b \in \T_{NI}$, the event $E_b$ occurs with probability at least $1/2$.
    When it does, we charge to $T_b$ the contribution $\minTime{i}$ to $FI$ from any $s_i \in F$ that lies in $T_b$.
    Furthermore, when $E_b$ occurs, for each connected component $C$ that intersects $T_b$ we additionally charge the contribution $\omega(C)$ from $FI$ to $T_b$.
    Crucially, the contribution from each $C$ can be charged to at most two different tasks $T_b$ in this way: once for the leftmost $T_b$ that it intersects, and once for the rightmost $T_b$ that it intersects.
    (The reason is that any $T_b$ meeting $C$ other than these two has its successor step inside $C$ as well, so no step of $\overline{T}_b$ fails and $E_b$ does not occur.
    Only the leftmost such task, which may extend past a failed step to the left of $C$, and the rightmost, whose successor is the failed step that ends $C$, can trigger $E_b$.)
    Moreover, this charging to $T_b$ adds up to a  total value of at least $\min\left\{1, \sum_{s_i \in T_b} \manualTime{i} \right\}$.
    The terminal task, for which $E_b$ is not defined, is charged separately: either some step of it fails, in which case the charging above applies verbatim, or no step of it fails, in which case $T_b \subseteq C$ for some connected component $C$ and $\omega(C) \geq \min\{1, \sum_{s_i \in T_b} \manualTime{i}\}$.
    Either way the terminal task is charged at least $\min\{1, \sum_{s_i \in T_b} \manualTime{i}\}$, and with probability $1$ rather than $1/2$.
    Since the charging occurs with probability $1/2$ for each $T_b \in \T_{NI}$, and we are at most double-counting each $\omega(C)$, we conclude that the expected sum of all $\omega(C)$, plus $\minTime{i}$ for all failed steps $s_i$ that lie in some $T_b \in \T_{NI}$, is at least $1/4$ of $\sum_{T_b \in \T_{NI}}\min\left\{1, \sum_{s_i \in T_b} \manualTime{i} \right\}$.

    The charging for $(s_i) \in \T_I$ remains unchanged.
    The value $\minTime{i}$ is charged in the event that step $s_i$ fails, which occurs with probability at least $1/2$.

    Taken together, we obtain our desired $8$ approximation.
\end{proof}

We suspect the factor of $8$ in Lemma~\ref{lem:FI.lower.bound.8} is not tight; the following example shows only that it cannot be improved below $\tfrac{2}{3}(2+\sqrt{2}) \approx 2.276$.
We do not know whether the assumption $\manualTime{i} \geq 1$ is needed for the sharper factor of $4$ in Lemma~\ref{lem:FI.lower.bound.4}, as we have not found an instance violating it for which $OPT$ exceeds $4\,FI$.

\begin{example}
\label{ex:FI.necessity}
Consider a sequence that alternates between two step types, with $K > 1$ such pairs (so $2K$ steps in total), where $\epsilon$ is vanishingly small:
\[ \bigl(\manualTime{i},\, \AItime{i},\, q_i\bigr) = \bigl(0,\, 1,\, 1/\sqrt{2}-\epsilon\bigr) \ \text{ for odd $i$}, \qquad (1,\,1,\,1) \ \text{ for even $i$}. \]
The optimum performs manually, at zero cost, the odd steps that fall between blocks, and chains blocks of five consecutive steps that begin and end on an even step.
Such a chain succeeds with probability $(1/\sqrt{2})^{2} = 1/2$ and so costs $2$, while covering three even steps that would otherwise cost $1$ each, giving a total cost of $\tfrac{2}{3}K$ (ignoring $\epsilon$).
The fragmentation index is $1$ plus $1$ for each of the remaining $K-1$ odd steps that fails, which is $1 + (K-1)(1 - 1/\sqrt{2})$ (again ignoring the impact of $\epsilon$).
As $K$ grows large, the ratio tends to $\tfrac{2}{3}(2+\sqrt{2}) \approx 2.276$.
\end{example}

\begin{proof}[Proof of Proposition~\ref{prop:fragmentation}]
Combining Lemmas~\ref{lem:FI.upper.bound}, \ref{lem:FI.lower.bound.4}, and~\ref{lem:FI.lower.bound.8} yields Proposition~\ref{prop:fragmentation}.
In general $\tfrac{1}{8}OPT \leq FI \leq \tfrac{5}{4}OPT$, and under the additional assumption $\manualTime{i} \geq 1$ for all $i$, $\tfrac{1}{4}OPT \leq FI \leq \tfrac{5}{4}OPT$.
\end{proof}

\subsection{Non-Monotone Gains from AI Improvements: Formal Statement and Proof}
\label{app:nonmonotone}

This appendix states and proves the result that Example~\ref{ex:nonmonotone} of Section~\ref{sec:nonlinear} illustrates; Table~\ref{tab:nonmonotone_costs} in Appendix~\ref{app:tables_and_figures} reports that example's configurations and costs.

To make the argument of Section~\ref{sec:nonlinear} precise, consider a workflow with any number of steps and chains.
For any AI strategy $\T$, total cost is a fixed manual term plus a polynomial in $1/\alpha$,
\begin{equation}
\label{eq:strategy_cost}
C_\T(\alpha) = K_\T + \sum_{c \in \mathcal{K}(\T)} \AItime{r_c}\,\alpha^{-D_c},
\end{equation}
where $K_\T = \sum_{\text{manual } i} \manualTime{i}$ is the fixed cost of manually-executed steps and the sum runs over the AI chains of $\T$, denoted $\mathcal{K}(\T)$, with each chain $c$ (with augmented endpoint $r_c$) contributing a term of degree equal to its total difficulty $D_c = \sum_{i \in c} d_i \ge 1$.
The cost of $\T$ thus shifts in a continuous and smooth manner as $\alpha$ increases, with monotone marginal gains to technology improvements.

The firm optimizes costs over many possible AI strategies, expressed as the deployment problem in \eqref{eq:totalcost_shortterm}.
At each level of AI quality $\alpha$ the firm adopts the cheapest available strategy, so its optimal cost is the lower envelope of the strategy-specific cost curves, $C^*(\alpha) = \min_{\T} C_\T(\alpha)$.
We write $g_\T(\alpha) = -\,\mathrm{d}C_\T/\mathrm{d}\alpha$ for the marginal benefit of improving AI quality under a fixed strategy $\T$, and $g^*(\alpha) = -\,\mathrm{d}C^*/\mathrm{d}\alpha \ge 0$ for the firm's marginal benefit at the optimum; the function $g^*(\alpha)$ equals $g_\T(\alpha)$ on the regime of $\alpha$ values where $\T$ is optimal.
As $\alpha$ rises, the optimal strategy changes at finitely many values of $\alpha$, the reorganization thresholds of Section~\ref{sec:nonlinear}, which separate one regime from the next.

Write $\phi = C_{\T'} - C_\T$ for the cost difference between the strategies that are optimal just above and just below a threshold $\alpha_0$, and call the threshold \emph{transversal} if the two curves cross there non-tangentially, that is, if $\phi'(\alpha_0) \neq 0$.
Lemma~\ref{lem:reorg_threshold} shows that $g^*$ diminishes within each regime and never falls at a reorganization threshold, and that it jumps upward at every transversal threshold, so the returns to improving AI quality are non-monotone.

\begin{lem}[Non-monotone Marginal Returns to AI Quality Improvements]
\label{lem:reorg_threshold}
Let $C^*(\alpha) = \min_{\T} C_\T(\alpha)$ be the firm's optimal cost over all feasible AI strategies and $g^*(\alpha) = -\,\mathrm{d}C^*/\mathrm{d}\alpha$ its marginal benefit of improving AI quality.
Within each regime (an interval of $\alpha$ on which a single AI strategy is optimal), $g^*$ is non-increasing.
At every reorganization threshold (where the optimal strategy changes), $g^*$ never falls, and at every transversal threshold it jumps up discontinuously.
\end{lem}

\begin{proof}
Each $C_\T$ is a constant plus terms $\AItime{r_c}\,\alpha^{-D_c}$ with $D_c \ge 1$, with derivatives
\[
\begin{aligned}
C_\T'(\alpha) &= -\sum_{c \in \mathcal{K}(\T)} \AItime{r_c}\,D_c\,\alpha^{-D_c-1} \;\le\; 0, \\
C_\T''(\alpha) &= \sum_{c \in \mathcal{K}(\T)} \AItime{r_c}\,D_c(D_c+1)\,\alpha^{-D_c-2} \;\ge\; 0,
\end{aligned}
\]
so each $C_\T$ is non-increasing and convex.

Within a regime where $\T$ is optimal, $g^*(\alpha) = g_\T(\alpha) = -C_\T'(\alpha) \ge 0$, and the convexity of $C_\T$ makes it non-increasing in $\alpha$ (when $\T$ uses no chain $g^*(\alpha) = 0$ but is otherwise strictly decreasing).
The returns to a fixed configuration in a regime thus weakly diminish.
This proves the first claim of the lemma.

Next, let $\alpha_0 \in (0,1)$ be a threshold where the optimum switches from $\T$ (optimal just below $\alpha_0$) to $\T'$ (just above $\alpha_0$).\footnote{We assume this threshold exists. The only scenario where no such threshold exists is the case in which all steps are executed manually for all $\alpha \in (0,1)$. This is the trivial, uninteresting case and is assumed away.}
Define
\[
\phi(\alpha) = C_{\T'}(\alpha) - C_\T(\alpha).
\]
The function $\phi(\alpha)$ is the difference of two finite sums of powers of $\alpha$ and hence continuous and differentiable on $(0,1)$. 
The two costs are equal at $\alpha_0$, with $C_{\T'}$ the cheaper just above and $C_\T$ the cheaper just below.
Therefore, by definition $\phi$ satisfies:
\[
\phi(\alpha_0^{-}) \ge 0,
\qquad
\phi(\alpha_0) = 0,
\qquad
\phi(\alpha_0^{+}) \le 0.
\]
A differentiable function with such a sign change has nonpositive derivative at $\alpha_0$.
So $\phi'(\alpha_0) \le 0$, or equivalently, $C_{\T'}'(\alpha_0) \le C_\T'(\alpha_0)$.
By the definition $g_\T(\alpha) = -\,\mathrm{d}C_\T/\mathrm{d}\alpha$, this is
\[
g_\T(\alpha_0) \;=\; -C_\T'(\alpha_0)
\;\le\;
-C_{\T'}'(\alpha_0) \;=\; g_{\T'}(\alpha_0).
\]
Since $g^*(\alpha)$ coincides with $g_\T(\alpha)$ just below $\alpha_0$ and with $g_{\T'}(\alpha)$ just above, and both are continuous at $\alpha_0$, the firm's marginal benefit function $g^*(\alpha)$ does not fall as $\alpha$ crosses $\alpha_0$.
If in addition $\alpha_0$ is transversal, then $\phi'(\alpha_0) \neq 0$, which together with $\phi'(\alpha_0) \le 0$ gives $\phi'(\alpha_0) < 0$ and hence $g_\T(\alpha_0) < g_{\T'}(\alpha_0)$, so the marginal benefit jumps upward discontinuously.
Transversality cannot be dispensed with, because the sign change only forces $\alpha_0$ to be a root of $\phi$ of odd order, and at a root of order three or higher the envelope's marginal benefit passes through the threshold without a jump.\footnote{With $m = 4$, $d = (2,1,2,1)$, $\manualTime{} = (1000, 1000, 56, 1000)$ and $\AItime{} = (102, 65, 1000, 12)$, the strategy $\T$ that chains steps~1--2, performs step~3 manually, and augments step~4 is uniquely optimal just below $\alpha_0 = 1/2$, and the strategy $\T'$ that augments step~1 and chains steps~2--4 is uniquely optimal just above it. Writing $x = 1/\alpha$, their cost difference is $\phi = 12x^{4} - 65x^{3} + 102x^{2} - 12x - 56 = (x-2)^{3}(12x+7)$, which has a triple root at $x = 2$, and $g^*$ falls through $\alpha_0$ from $3170.5$ to $3165.5$ without jumping.}
The condition is nonetheless the generic case.
A threshold already imposes $\phi(\alpha_0) = 0$, so requiring $\phi'(\alpha_0) = 0$ as well is a second equation in the same unknown, and the time cost vectors $(\manualTime{}, \AItime{})$ for which it holds form a Lebesgue-null set.
Both thresholds of Example~\ref{ex:nonmonotone} are transversal, with $g^*$ rising from $0$ to $16.0$ at $\alpha = 0.50$ and from $4.7$ to $134.1$ at $\alpha \approx 0.92$.
\end{proof}

\subsection{Proof of Proposition~\ref{prop:shortrun_dp}}
\label{app:shortrun_proof}

\begin{proof}
We first note the following recursive formulation of the optimization problem.
For all $k \leq m$, let $C[k]$ denote the minimum cost of completing a hypothetical workflow that only includes steps $1$ through $k$.
Note that, because of the way we defined $C[k]$, step $k$ cannot be automated in any minimum-cost solution that determines $C[k]$; it can only be augmented or performed manually.
Note that $C[m]$ is the cost of the optimal solution.

We now show how to calculate $C[k]$ recursively.
As a base case we have $C[0] = 0$, as the empty set of steps requires no work to complete.
For $k \geq 1$, $C[k]$ is the lesser of
\[ C[k-1] + \manualTime{k} \]
and
\[ \min_{\ell < k} \Biggl\{ C[\ell] + \frac{\AItime{k}}{\prod_{i=\ell+1}^k q_i} \Biggr\}.
\]
In other words, we either complete step $k$ manually (at time $\manualTime{k}$, separately optimizing over steps $1$ through $k-1$) or we augment step $k$, at expected time $\AItime{k}/\prod_{i=\ell+1}^k q_i$.
In the latter case, we then optimize over the length of the AI chain that ends with the newly-augmented step $k$.
This could be a singleton chain with no automation (corresponding to $\ell = k-1$), or a longer chain that automates one or more steps before step $k$ (corresponding to a choice of $\ell < k-1$).
For any such choice of $\ell$, we then separately optimize the production strategy for steps $1$ through $\ell$.

Using this formulation, given the values $(C[0], \dotsc, C[k-1])$, we can calculate $C[k]$ in time $O(k)$ by considering each potential choice of $\ell$.
Doing so for each $k = 1, 2, \dotsc, m$ yields the value of $C[m]$ (and the corresponding optimal AI strategy) in total time $O(m^2)$.
\end{proof}

\subsection{Proof of Proposition~\ref{prop:totalcost_optimization_dp}}
\label{app:jobdesign.longrun}

In the long-run optimization problem, the firm anticipates the adjustment of worker wages to fit the skill requirements of each job and sets job responsibilities accordingly.
This contrasts with the AI deployment problem of Section~\ref{sec:shortrun.production}, which minimizes the time production takes while holding the assignment of steps to a single worker fixed. 
The long-run problem additionally optimizes how tasks are bundled into jobs.

The firm's long-run problem is stated in \eqref{eq:totalcost_with_handoff} of Section~\ref{sec:longrun.production}.
It designs both the partition $\T$ of steps into tasks and the partition $\mathcal{J}$ of tasks into jobs, where the choice of $\T$ determines the time and skill requirements of each task and, given $\T$, the choice of $\mathcal{J}$ determines worker wage rates and the time needed per unit of output, including hand-off costs.

What makes this harder than the short-run problem of Proposition~\ref{prop:shortrun_dp} is that steps must be assigned to tasks and the resulting tasks to workers at the same time, with hand-off costs at job boundaries counted along the way.
Proposition~\ref{prop:totalcost_optimization_dp} states that the problem nonetheless admits a polynomial-time approximation scheme, which we prove here.

\begin{proof}
We first state the cost minimization problem as a recursive problem and then use dynamic programming to approximate the optimal solution.

\emph{Recursive formulation.}
Consider a workflow with $m$ steps $\mathcal{S} = (s_1, s_2, \dotsc, s_m)$.
Because the wage bill of a job multiplies its total skill by its total time, the running value of either alone never reveals how an added task changes it, so the program carries both and builds the job design up one step at a time.
For $0 \leq i \leq m$, let $V(i)$ denote the minimum cost required to complete steps $1$ through $i$ using one or more jobs, including any hand-off costs to a subsequent worker.
Note then that $V(m)$ is the solution to the desired optimization problem, recalling that the hand-off time for the final job will always be $0$.
Note also that it is implicitly assumed in the definition of $V(i)$ that it optimizes over job designs for steps $1$ through $i$ in which the final job terminates after the completion of step $i$.

For job designs whose final job does not necessarily terminate after step $i$ but rather continues onward to subsequent steps, we need an auxiliary function.
For $0 \leq i \leq m$ and skill and time levels $\skillcost{}$ and $\timecost{}$, let $R(i,\; \skillcost{},\; \timecost{})$ denote the minimum cost of completing steps $1$ through $i$ when the worker currently accumulating work, whom we call the ``active worker,'' has already been assigned tasks with total skill cost $\skillcost{}$ and total time cost $\timecost{}$, where $\timecost{}$ also includes any hand-off cost for the active worker.
The minimum is taken over AI deployment strategies and job designs together.
The recursion may either expand the active worker's job or close it and assign the remaining steps to other workers.
Throughout, step $i$ is the highest-indexed step that has not yet been assigned.
In the base case no steps remain to assign, the active worker's totals are exactly $(\skillcost{}, \timecost{})$, and their wage bill is the product of the two,
\[ R(0,\; \skillcost{},\; \timecost{}) = \skillcost{}\,\timecost{}. \]

The two functions are linked as follows.
In $V(i)$, all job designs must end with a job that completes at step $i$ and hands off to the following worker.
Incurring this hand-off cost is equivalent to opening a job for the active worker whose accumulated time already includes $\handofftime{i}$ and whose accumulated skill is still $0$, which is the state $(i,\; 0,\; \handofftime{i})$, and then assigning step $i$ to that worker either manually or as the endpoint of an AI chain.
Since a job must contain at least one task by Definition~\ref{def:job_design}, the active worker's job cannot be closed while it is empty, so $V(0) = 0$ and, for $1 \leq i \leq m$,
\[
V(i) = \min\left\{ R\bigl(i-1,\; \manualSkill{i},\; \handofftime{i} + \manualTime{i}\bigr), \;\;
\min_{0 \le r < i}\ R\!\left(r,\; \AIskill{i},\; \handofftime{i} + \frac{\AItime{i}}{\prod_{i'=r+1}^{i} q_{i'}}\right) \right\}.
\]

For each $i \geq 1$, the minimum cost function $R(i,\; \skillcost{},\; \timecost{})$ satisfies the recursive relation
% \small (via \smalldisplay): at full size the three-option min is 481pt wide on a
% 470pt measure and its closing brace and period sat in the right margin. \small brings
% it to 440pt, which keeps the \Biggl\{ ... \Biggr\} spanning one line rather than
% forcing a break between the second and third option.
\begingroup\smalldisplay
\begin{align*}
R(i,\; \skillcost{},\; \timecost{}) = \min\Biggl\{
& \ \underbrace{\skillcost{}\,\timecost{} + V(i)}_{\text{job closed}}, \ \underbrace{R\bigl(i-1,\; \skillcost{} + \manualSkill{i},\; \timecost{} + \manualTime{i}\bigr)}_{\text{Step $i$ is manual}}, \ \underbrace{\min_{0 \le r < i}\ R\!\left(r,\; \skillcost{} + \AIskill{i},\; \timecost{} + \frac{\AItime{i}}{\prod_{i'=r+1}^{i} q_{i'}}\right)}_{\text{Step $i$ optimally chained}}
\Biggr\}.
\end{align*}
\endgroup
In words, when assigning step $i$, the three options the firm considers are:
\begin{itemize}
    \item \textbf{Job closure.}
    The term $\skillcost{}\,\timecost{} \; + \; V(i)$ corresponds to not adding any further steps to the job of the active worker.
    All steps $1$ through $i$ will be completed by other workers.
    The active worker's total cost is then $\skillcost{}\,\timecost{}$, and $V(i)$ is the total cost of completing steps $1$ through $i$ (including hand-off costs to the active worker).
    This option is the margin new to the long run.
    It is available only when the active worker already holds at least one task, that is when $\skillcost{} > 0$.
    The states with $\skillcost{} = 0$ are exactly the job-opening states $(i,\; 0,\; \handofftime{i})$, whose value is the $V(i)$ defined above; without this restriction the recursion evaluated at such a state would read $V(i) = \min\{V(i),\, \cdot\,\}$, which any value below the other two branches satisfies, so it determines nothing.

    \item \textbf{Manual step.}
    The term $R(i-1,\; \skillcost{} + \manualSkill{i},\; \timecost{} + \manualTime{i})$ corresponds to designating step $i$ for manual completion, and adding the corresponding (singleton) task.
    The manual execution costs $(\manualSkill{i}, \manualTime{i})$ of step $i$ are added to the accumulated costs of tasks already assigned to the active worker, extending their job to include step $i$ as well, and steps $1$ through $i-1$ are left to be assigned optimally.

    \item \textbf{AI chain.}
    The term $\min_{0 \le r < i} \ R\left(r,\; \skillcost{} + \AIskill{i},\; \timecost{} + \frac{\AItime{i}}{\prod_{i'=r+1}^{i} q_{i'}}\right)$ corresponds to creating an AI chain task for steps $r+1$ through $i$, with step $i$ being augmented and steps $r+1$ through $i-1$ automated, and assigning the resulting task to the job of the active worker.
    Step $i$ is thus chained with zero or more preceding (automated) steps, and the associated AI verification costs are added to the active worker's skill and time costs.
    The cut $r$ is the highest-indexed step that is not added to this AI chain, which can be any value between $0$ and $i-1$.
\end{itemize}
Note that in evaluating the AI chain option, step $i$ can only be augmented (rather than automated), as each step of the recursion adds a full AI chain to the active worker's job; AI strategies in which step $i$ is automated are considered when performing the calculations $R(r,\; \skillcost{},\; \timecost{})$ with $r > i$.
The two functions are filled in increasing order of $i$: first $V(i)$, which calls $R(r,\; \cdot,\; \cdot)$ only at indices $r < i$, and then $R(i,\; \skillcost{},\; \timecost{})$ for $\skillcost{} > 0$, which may call $V(i)$.
The cost of the optimal joint AI strategy and job design for the full sequence of $m$ steps is then $V(m)$, recalling that $\handofftime{m} = 0$.

\emph{Discretization and Computation Time.}
We solve the recursion by dynamic programming, filling in the values of $R(i,\; \skillcost{},\; \timecost{})$ for all $i$, $\skillcost{}$, and $\timecost{}$.
Since $\skillcost{}$ and $\timecost{}$ take continuous values, we discretize both to powers of $(1+\epsilon')$, rounding every running total \emph{up} to the next such power, where $\epsilon>0$ is the target accuracy and
\[
\epsilon' \;=\; \frac{\ln(1+\epsilon)}{2m+1} \;=\; \Theta(\epsilon/m).
\]
We show below why $\epsilon'$ must be smaller than $\epsilon$ by a factor of order $m$.

Suppose all skill and time costs, and all non-zero hand-off costs, lie in $[1/B, B]$ for some $B > 0$; recall that the final job's hand-off cost $\handofftime{m}$ is zero.
Completing every step manually with a separate worker incurs, for each step, a time and hand-off cost of at most $B$ each and a skill cost of at most $B$, for a total cost of at most $2B^2$ per step, so no solution containing a job with total cost above $2mB^2$ need be considered.
Since skill costs combine additively and every non-zero skill or time cost is at least $1/B$, a job's total skill lies in $[1/B,\, mB]$, and its total time is at most $2mB^2/(1/B) = 2mB^3$.
It therefore suffices to tabulate skill costs in $[1/B,\, mB]$ and time costs in $[1/B,\, 2mB^3]$.

Each of these ranges contains $O\bigl(\epsilon'^{-1}\log(mB)\bigr) = O\bigl(m\,\epsilon^{-1}\log(mB)\bigr)$ powers of $(1+\epsilon')$; rounding upward can carry a total past the bounds above by at most one factor of $(1+\epsilon')$, which changes these counts only by an additive constant.
The table is indexed by the step $i$, which takes $m+1$ values, together with the discretized skill and time levels, each taking $O\bigl(m\,\epsilon^{-1}\log(mB)\bigr)$ values, so it holds
\[
(m+1)\,\cdot\,O\bigl(m\,\epsilon^{-1}\log(mB)\bigr)\,\cdot\,O\bigl(m\,\epsilon^{-1}\log(mB)\bigr) \;=\; O\bigl(m^3\,\epsilon^{-2}\log^2(mB)\bigr)
\]
entries in total.
Each entry is filled in $O(m)$ time by taking the minimum over the job-closure option, the manual option, and the at most $m$ chain cuts $r$ in the recursion above, so filling the whole table takes
\[
O\bigl(m^3\,\epsilon^{-2}\log^2(mB)\bigr)\,\cdot\,O(m) \;=\; O\bigl(m^4\,\epsilon^{-2}\log^2(mB)\bigr)
\]
time.
Recording which option attains each minimum lets the firm read the optimal AI strategy and job design off $V(m) = R(m,\; 0,\; 0)$.

It remains to bound the discretization error.
The table is indexed by the rounded totals, so each transition adds an increment to a value that has already been rounded and then rounds again; the distortions therefore compound along a job rather than being incurred once.
A job holding $k$ tasks rounds its skill coordinate $k$ times, once for each task added, and its time coordinate $k+1$ times, once when the job opens at the state $(i,\, 0,\, \handofftime{i})$ and once for each task thereafter.
Since we round upward, an induction on the number of accumulations shows that the recorded skill exceeds the exact total by a factor of at most $(1+\epsilon')^{k}$ and the recorded time by a factor of at most $(1+\epsilon')^{k+1}$.
Each job's wage bill is formed by the single multiplication $\skillcost{}\,\timecost{}$ at the point the job-closure option closes it, so it exceeds the true bill by at most $(1+\epsilon')^{2k+1} \leq (1+\epsilon')^{2m+1}$, which our choice of $\epsilon'$ makes at most $e^{\ln(1+\epsilon)} = 1+\epsilon$.\footnote{Using $1+x \leq e^{x}$, so that $(1+\epsilon')^{2m+1} \leq e^{\epsilon'(2m+1)} = e^{\ln(1+\epsilon)}$.}
Every job's bill is therefore overstated by at most this factor and understated by none, so the pair $(\T,\, \J)$ the table returns has true cost at most $(1+\epsilon)$ times the minimum.
\end{proof}

\section{CES Representation at Macro Level}
\label{sec:aggregation}

Sections~\ref{sec:shortrun}--\ref{sec:longrun} study a single firm's choices, whereas the automation literature models production directly at the level of the aggregate economy \citep{acemoglu2018, acemoglu2022tasks, acemoglu2025simple}.
This section shows that, under additional assumptions, the two views connect.
A firm's many-input, task-level Leontief technology reduces to a two-input Leontief in skill-adjusted AI management labor and skill-adjusted manual labor (Appendix~\ref{sec:agg_within}), and aggregating across firms that differ only in how effectively they deploy the same AI technology yields an economy-wide CES production function.
This tractably links micro-level organization decisions to the macroeconomic implications of improving AI quality.

Before the aggregation analysis, let us first describe the setup of production.
A firm uses two production inputs: (1) skill-adjusted AI management labor (to complete AI chains), and (2) skill-adjusted manual labor (to perform manual tasks).
These two types of labor demand different compensations, in that executing an AI chain with skill and time cost of one demands $w_{\AIletter}$, while completing a manual task with similar skill and time requirements commands a compensation of $w_{\manualLetter}$.
We call $w_{\AIletter}$ and $w_{\manualLetter}$ the \emph{base wage rates} for AI management and manual labor, respectively.

Firms use these two inputs to complete tasks.
The skill and time requirements of tasks are set by the AI strategy $\T$, and the skill requirement of a job by $\T$ and the job design $\J$ together: a worker's compensation per unit of time depends on all tasks in the job, so the more tasks a job includes, the higher the compensation required for every task in it.
The total compensation of a job is the sum of the skill costs of its tasks, each weighted by the base wage rate of the labor that performs it.
Specifically, the compensation for job $J$ will be:
\begin{equation}
w_{\manualLetter} \left(\sum_{T_b^{\manualLetter} \in J} \manualSkill{b}\right) + w_{\AIletter} \left(\sum_{T_b^{\AIletter} \in J} \AIskill{b}\right),\label{eq:new_wage}
\end{equation}
which is composed of the sum of contributions of skill costs from its manual tasks denoted by set $T_b^{\manualLetter}$ (each weighted by $w_{\manualLetter}$) and its AI chains denoted by set $T_b^{\AIletter}$ (each weighted by $w_{\AIletter}$).\footnote{
This formulation is essentially equivalent to the original definition of wage in \eqref{eq:wage}.
If we multiply each skill cost in Equation~\eqref{eq:new_wage} by its respective base wage rate and redefine the task's skill cost as the resulting product, we arrive at the original formulation given by \eqref{eq:wage}.
The formulation in \eqref{eq:new_wage} allows tasks to contribute differently to the job's wage depending on their execution mode and the type of labor performing them.}

Next, we define skill-adjusted labor, which is the smallest unit of work in our analysis.
Let $J(b)$ denote the job to which task $b$ belongs, and $E(b)$ be the mode of execution of task $b$ (i.e., $E(b) = \manualLetter$ if $b$ is a manual task, or $E(b) = \AIletter$ if $b$ is an AI chain).
Write $d_b$ for the total difficulty of the steps task $b$ hands to AI, so that $d_b = \sum_{i=\ell}^{r} d_i$ for a chain spanning $(s_\ell, \dotsc, s_r)$ and $d_b = 0$ for a manual step.
A task's time requirement for a single attempt inherits its mode from $E(b)$ and is written $\timeCostLetter^{E(b)}_{b}$, which is $\manualTime{b}$ when task $b$ is a manual step, and $\AItime{b}$, the cost of verifying the chain's augmented endpoint, when it is an AI chain.
Recall from Definition~\ref{def:ai_strategy} that the \emph{expected} time cost of a task is $\timecost{b} = \timeCostLetter^{E(b)}_{b}\,\alpha^{-d_b}$, so that below we write the failure inflation $\alpha^{-d_b}$ explicitly wherever an expected quantity is formed rather than folding it into $\timeCostLetter^{E(b)}_{b}$.
Define the skill-adjusted time requirement of task $b$ as
\[
\skillAdjustedTimeLetter_b = \frac{w_{\manualLetter} \left(\sum_{T_{\ell}^{\manualLetter} \in J(b)} \manualSkill{\ell}\right) + w_{\AIletter} \left(\sum_{T_{\ell}^{\AIletter} \in J(b)} \AIskill{\ell}\right)}{w_{E(b)}} \, \timeCostLetter^{E(b)}_{b}.
\]
Note that the fraction appearing behind $\timeCostLetter^{E(b)}_{b}$ represents an effective skill adjustment factor, as the numerator captures the total compensation of the job task $b$ belongs to, while the denominator normalizes by the base wage rate for the specific mode of execution of task $b$.
The resulting fraction thus has units of skill intensity.
This characterization becomes useful later as it allows expressing the (expected) wage bill paid for task $b$ simply as $w_{\AIletter}\,\skillAdjustedTimeLetter_{b}\,\alpha^{-d_b}$ if it is an AI chain or as $w_{\manualLetter}\,\skillAdjustedTimeLetter_{b}$ if it is executed manually.

\subsection{Within-Firm Aggregation: Leontief to Leontief}
\label{sec:agg_within}

We now consider the firm's production.
To produce output, each firm must complete requirements of all \( m \) production steps \( s_1, \ldots, s_m \).
This naturally results in a Leontief production structure with multiple inputs.\footnote{
Different AI deployment strategies and job designs lead to different input proportions but share the same fundamental Leontief form.
}
We show that the firm's production function can be represented by a Leontief function with two firm-level aggregate inputs: AI management labor and manual labor.

Suppose the firm solves the long-run cost minimization problem in \eqref{eq:totalcost_with_handoff} for a given AI quality level $\alpha$, and obtains the optimal AI strategy \( \T \) and job design \( \J \).
Let \( |\T| = n \), implying the original \( m \) production steps are organized into \( n \) distinct tasks, and \( |\J| = p \), indicating these \( n \) tasks are grouped into \( p \) jobs.
With $\T$ and $\J$ fixed, each step's execution mode, the job boundaries, and each job's total skill requirement are given, so we can work directly with tasks rather than steps.

We assume hand-offs are performed manually by the worker who completes the final task of each job.
This allows us to simplify the analysis in two ways:
(1) Since job design is fixed, we know exactly which steps occur at the job boundaries and thus incur hand-off time costs. Therefore, hand-off of job $J_j$ can be treated as a standalone, human-executed task with time requirement \( \handofftime{}(J_j) \) and skill requirement equal to the total skill requirement of job $J_j$.\footnote{
Denote the skill-adjusted time cost of hand-off task of job $J_j$ with $\skillAdjustedTimeHandoff{}(J_j)$.
}
(2) The order of tasks, including the newly defined hand-off tasks, can be rearranged within the production sequence.
Specifically, we relabel tasks such that tasks \( 1 \) to \( k \) are AI chains and thus require AI management labor, while tasks \( k+1 \) to \( n \), along with the hand-off tasks \( n+1 \) to \( n+p-1 \), are executed manually and require manual labor.\footnote{
The final job \( p \) requires no hand-off by definition, so the last hand-off occurs at the end of job \( p-1 \).
}

The firm's task-level production can be expressed in the following Leontief form:
\begin{equation}
y = \min\left\{\frac{\labor{1}}{\skillAdjustedTimeAI{1}\,\alpha^{-d_{1}}},\,\cdots,\,\frac{\labor{k}}{\skillAdjustedTimeAI{k}\,\alpha^{-d_{k}}},\,\frac{\labor{k+1}}{\skillAdjustedTimeManual{k+1}},\,\cdots,\,\frac{\labor{n}}{\skillAdjustedTimeManual{n}},\,\frac{\labor{n+1}}{\skillAdjustedTimeHandoff{}(J_1)},\,\cdots,\,\frac{\labor{n+p-1}}{\skillAdjustedTimeHandoff{}(J_{p-1})}\right\}.\label{eq:task_level_prod}
\end{equation}
The \( y \) on the left-hand side represents the firm's output while $\labor{b}$ is the amount of labor assigned to task $b$. All other variables remain as previously defined.

Since the AI strategy and job design are fixed, the required amount of skill-adjusted time to spend on each task in the denominators of Equation~\eqref{eq:task_level_prod} are fixed, and the firm takes them as given.
In equilibrium, allocation of labor to tasks satisfies
\[
y=\frac{\labor{1}}{\skillAdjustedTimeAI{1}\,\alpha^{-d_{1}}}=\cdots=\frac{\labor{k}}{\skillAdjustedTimeAI{k}\,\alpha^{-d_{k}}}=\frac{\labor{k+1}}{\skillAdjustedTimeManual{k+1}}=\cdots=\frac{\labor{n}}{\skillAdjustedTimeManual{n}}=\frac{\labor{n+1}}{\skillAdjustedTimeHandoff{}(J_1)}=\cdots=\frac{\labor{n+p-1}}{\skillAdjustedTimeHandoff{}(J_{p-1})}.
\]
Notice that the ratio of labor allocated to any two tasks $a$ and $b$ is independent of output level and wage rates (regardless of whether tasks $a$ and $b$ are executed manually or by the help of AI), and solely depends on their relative skill-adjusted time requirements, which are fixed given $\T$ and $\J$.
The fixed ratio of labor inputs implies that the rate of substitution between any pair of tasks is independent of the labor allocated to other tasks:
\begin{equation}
\frac{\partial}{\partial \labor{z}}\left(\frac{\frac{\Delta y}{\Delta \labor{a}}}{\frac{\Delta y}{\Delta \labor{b}}}\right)=0,\qquad\forall z\neq a,b.\label{eq:task_subst_rate_indep}
\end{equation}

A necessary and sufficient condition to aggregate the firm's production function from a Leontief with one input per task into a Leontief with (firm-level) aggregate AI management labor and manual labor is that the rate of substitution between any two tasks within the same aggregate input type is independent of all tasks in the other aggregate input type \citep{leontief1947introduction, fisher1965embodied, felipe2003aggregation}.

The condition expressed in \eqref{eq:task_subst_rate_indep} satisfies the aggregation requirement above.
Thus, the firm's production function can be represented as:
\begin{equation}
y=\min\,\left\{\frac{\bar{\alpha}\,\labor{\AIletter}}{\skillAdjustedTimeLetter_{\AIletter}},\,\frac{\labor{\manualLetter}}{\skillAdjustedTimeLetter_{\manualLetter}}\right\}.\label{eq:micro_agg_prod}
\end{equation}
Here, $\labor{\AIletter}$ and $\labor{\manualLetter}$ represent respectively the firm-level aggregate AI management labor and manual labor.
Moreover, the aggregate skill-adjusted AI and manual time requirements, $\skillAdjustedTimeLetter_{\AIletter}$ and $\skillAdjustedTimeLetter_{\manualLetter}$, are defined as the sum of skill-adjusted time requirements of tasks in their respective groups:
\begin{equation}
\skillAdjustedTimeLetter_{\AIletter} = \sum_{b=1}^{k}\skillAdjustedTimeAI{b}, \qquad
\skillAdjustedTimeLetter_{\manualLetter} = \sum_{j=1}^{p-1}\skillAdjustedTimeHandoff{}(J_j) + \sum_{b=k+1}^{n}\skillAdjustedTimeManual{b}. \label{eq:def_t_AI}
\end{equation}
Finally, $\bar{\alpha}$ is the firm's effective AI quality level defined as
\begin{equation}
\bar{\alpha}=\frac{\sum_{b=1}^{k}\skillAdjustedTimeAI{b}}{\sum_{b=1}^{k}\skillAdjustedTimeAI{b}\alpha^{-d_{b}}},\label{eq:alpha_bar}
\end{equation}
so that the following relationship holds:
\[
\frac{\skillAdjustedTimeLetter_{\AIletter}}{\bar{\alpha}}=\sum_{b=1}^{k}\skillAdjustedTimeAI{b}\alpha^{-d_{b}}.
\]

In short, Equation~\eqref{eq:micro_agg_prod} allows us to analyze labor allocation using the simpler two-input Leontief production function for a given AI strategy \( \T \) and job design \( \J \).
This firm-level aggregate Leontief function implies the equilibrium condition:
\begin{equation}
y=\frac{\bar{\alpha}\,\labor{\AIletter}}{\skillAdjustedTimeLetter_{\AIletter}}
=\,\frac{\labor{\manualLetter}}{\skillAdjustedTimeLetter_{\manualLetter}}.
\label{eq:micro_fixed_proportions}
\end{equation}

\subsection{Cross-Firm Aggregation: Leontief to CES}

Next, we ask whether the Leontief production functions of many firms can be represented by a single economy-wide production function whose inputs aggregate the corresponding firm-level inputs.
The aggregation literature shows that suitable heterogeneity across firms can smooth micro-level production functions into a well-behaved aggregate one (see \cite{sato1975production}, Chapters 2 and 4).\footnote{
For instance, \cite{houthakker1955pareto} studies aggregation from micro Leontief functions to a macro Cobb-Douglas function; \cite{levhari1968note} investigates aggregation from micro Leontief functions to a macro CES function; and \cite{sato1969micro} explores aggregation from micro CES to macro CES functions. 
See \cite{baqaee2019foundations} for a broader formulation of this problem.}
Most of its existence results do not yield a closed form; we show that under certain conditions on firm heterogeneity, firm-level Leontief production functions aggregate to an economy-wide CES production function.

Consider a unit mass of firms in the economy, each producing output with AI management labor, human labor, and capital.
We assume capital is a fixed input whose productivity is common across firms, so that firms make their labor decisions conditional on a given capital stock and capital plays no part in the allocation in \eqref{eq:micro_agg_prod}.\footnote{Since we introduce firm-level heterogeneity through variations in effective AI quality alone, we abstract away from additional sources of heterogeneity across firms.}
We normalize the economy's aggregate capital stock to $1$.
Firms do not know their individual effective AI quality level before production occurs; instead, they share a common belief about the distribution from which these quality levels are drawn.

Production unfolds in two stages.
In the first stage, firms solve the long-run cost minimization problem in \eqref{eq:totalcost_with_handoff}, committing to an AI deployment strategy $\T$ and a job design $\J$ based on their \emph{expected} AI quality levels.
Since firms hold the same expectations about the distribution of effective AI quality, they choose identical AI strategies and job designs.
In the second stage, firms learn their realized effective AI quality levels, hire the required labor and capital, and begin production.\footnote{As production depends on the realized effective AI quality level, $\bar{\alpha}$ also determines firms' sizes.}
The resulting Leontief production function, determined by the previously selected AI strategy, job design, and the firm's realized effective AI quality, takes a form similar to Equation~\eqref{eq:micro_agg_prod}.

In what follows, we demonstrate that the heterogeneity in effective AI quality level can be structured so that micro-level Leontief production functions in \eqref{eq:micro_agg_prod} aggregate into a macro-level CES production function, in which the CES inputs are aggregates of the micro-level inputs \citep{levhari1968note}.
Suppose specifically that the macro-level production function takes the form:
\begin{equation}
Y = \left(\theta_{\AIletter}\,\aggregateAIlabor^{\rho}+\theta_{\manualLetter}\,\aggregateManualLabor^{\rho}+(1-\theta_{\AIletter}-\theta_{\manualLetter})\,K^\rho\right)^{\frac{1}{\rho}},\label{eq:macro_agg_prod}
\end{equation}
where $Y$ represents the (economy-wide) aggregate output, $\aggregateAIlabor$ denotes aggregate AI management labor, $\aggregateManualLabor$ denotes aggregate manual labor, $K$ is aggregate capital, and parameters $\theta_{\AIletter}$ and $\theta_{\manualLetter}$ represent the weights of corresponding inputs.
Aggregate capital is normalized to $1$ as described above, so the third term is constant at $1-\theta_{\AIletter}-\theta_{\manualLetter}$; its exponent is part of the CES form we posit rather than something the aggregation derives.

To link the macro production function above to the micro production functions, we must relate aggregated firm-level inputs \( \labor{\AIletter}\) and \( \labor{\manualLetter} \) from Equation~\eqref{eq:micro_agg_prod} to their corresponding economy-wide aggregates \( \aggregateAIlabor\) and \( \aggregateManualLabor \) in Equation~\eqref{eq:macro_agg_prod}.
This requires either determining the macro production function parameters from the distribution of heterogeneity, or specifying the form of heterogeneity given a set of macro production function parameters.
We adopt the latter approach and derive the output density function in terms of effective AI quality, denoted by \( \phi(\bar{\alpha}) \), as a function of \( \theta_{\AIletter}, \theta_{\manualLetter}, \rho \).
Throughout, we assume $\rho<0$, so the elasticity of substitution is $\sigma<1$, as is usual in this literature.

Normalize the output price to $p=1$, so that $w_{\AIletter}$ and $w_{\manualLetter}$ can be interpreted as real wage rates for a unit of skill-adjusted AI management and manual labor, respectively.
A firm produces only if it earns nonnegative profits:
\[
w_{\AIletter}\,\labor{\AIletter}+w_{\manualLetter}\,\labor{\manualLetter}\leq y.
\]
Substituting for $\labor{\manualLetter}$ and $y$ from Condition~\eqref{eq:micro_fixed_proportions} into the profitability condition yields:
\[
w_{\AIletter}\,\labor{\AIletter}+w_{\manualLetter}\,\frac{\skillAdjustedTimeLetter_{\manualLetter}\,\bar{\alpha}\,\labor{\AIletter}}{\skillAdjustedTimeLetter_{\AIletter}} \leq \frac{\bar{\alpha}\,\labor{\AIletter}}{\skillAdjustedTimeLetter_{\AIletter}}.
\]
Rearranging terms and canceling $\labor{\AIletter}$ gives the lower bound of firms' effective AI quality level:\footnote{We assume the parameters $w_{\AIletter},\,w_{\manualLetter},\,\skillAdjustedTimeLetter_{\AIletter},\,\skillAdjustedTimeLetter_{\manualLetter}$ are such that $0 < ({w_{\AIletter}\,\skillAdjustedTimeLetter_{\AIletter}})/(1-w_{\manualLetter}\,\skillAdjustedTimeLetter_{\manualLetter}) < 1$.}
\[
\bar{\alpha} \geq \frac{w_{\AIletter}\,\skillAdjustedTimeLetter_{\AIletter}}{1 - w_{\manualLetter}\,\skillAdjustedTimeLetter_{\manualLetter}}.
\]
This lower bound implies that only firms with realized effective AI quality level above this threshold will produce in equilibrium and others exit the market.
Moreover, observe from \eqref{eq:alpha_bar} that for $d_b \geq 0$ the upper bound of $\bar{\alpha}$ is 1 and is achieved when $\alpha \rightarrow 1^{-}$.
Therefore:
\begin{equation}
\underline{\alpha} \equiv \frac{w_{\AIletter}\,\skillAdjustedTimeLetter_{\AIletter}}{1 - w_{\manualLetter}\,\skillAdjustedTimeLetter_{\manualLetter}} \leq \bar{\alpha} \leq 1.
\label{eq:alpha_threshold}
\end{equation}

Let $\phi(\bar{\alpha})$ denote output across firms by effective AI quality, normalizing the density of firms over $\bar{\alpha}$ to one, so that a firm of type $\bar{\alpha}$ produces $y = \phi(\bar{\alpha})$ and aggregate output is the integral of $\phi$ over the producing firms.
We refer to $\phi$ as the distribution of effective AI quality in this sense.
From Condition~\eqref{eq:micro_fixed_proportions}, it directly follows that the AI management labor and manual labor used by this firm are: $\labor{\AIletter} = (\skillAdjustedTimeLetter_{\AIletter}/\bar{\alpha})\,\phi(\bar{\alpha}),$ and $\labor{\manualLetter} = \skillAdjustedTimeLetter_{\manualLetter}\,\phi(\bar{\alpha}).$
Consequently, aggregate output $Y$, aggregate AI management labor $\aggregateAIlabor$, and aggregate manual labor $\aggregateManualLabor$ can be expressed as:
\begin{equation}
Y = \int_{\underline{\alpha}}^{1}\phi(\bar{\alpha})\,d\bar{\alpha}, \qquad
\aggregateAIlabor = \int_{\underline{\alpha}}^{1}\skillAdjustedTimeLetter_{\AIletter}\frac{\phi(\bar{\alpha})}{\bar{\alpha}}\,d\bar{\alpha}, \qquad
\aggregateManualLabor = \int_{\underline{\alpha}}^{1}\skillAdjustedTimeLetter_{\manualLetter}\,\phi(\bar{\alpha})\,d\bar{\alpha}.
\label{eq:micro_Y}
\end{equation}

% sloppypar: four \eqref tags in one sentence, and under the OA./SA. numbering they are
% "(OA.C.11)" rather than "(C.11)", which is enough extra width that TeX could not
% break the line within the house \tolerance and ran 18pt into the right margin.
\begin{sloppypar}
Substituting the aggregated variables from the micro-level Equation~\eqref{eq:micro_Y} into the aggregate production function in \eqref{eq:macro_agg_prod}, we obtain:
\end{sloppypar}
% \small for the same reason as the display below it: at full size this one is 505pt
% wide against a 470pt measure and its closing exponent ran into the right margin.
% It needs no line break, since \small brings it to 464pt on its own.
\begingroup\smalldisplay
\begin{equation}
\int_{\underline{\alpha}}^{1}\,\phi(\bar{\alpha})\,d\bar{\alpha}
=
\left(
\theta_{\AIletter}\left(\skillAdjustedTimeLetter_{\AIletter}\int_{\underline{\alpha}}^{1}\,\frac{\phi(\bar{\alpha})}{\bar{\alpha}}\,d\bar{\alpha}\right)^{\rho}
+
\theta_{\manualLetter} \left(\skillAdjustedTimeLetter_{\manualLetter}\int_{\underline{\alpha}}^{1}\,\phi(\bar{\alpha})\,d\bar{\alpha}\right)^{\rho}
+
\left(1-\theta_{\AIletter}-\theta_{\manualLetter}\right)
\right)^{\frac{1}{\rho}},
\label{eq:macro_agg_prod_with_micro_aggregates}
\end{equation}
\endgroup
provided that the CES weights and the skill-adjusted time requirements satisfy
\begin{equation}
\theta_{\AIletter}\,\skillAdjustedTimeLetter_{\AIletter}^{\rho}+\theta_{\manualLetter}\,\skillAdjustedTimeLetter_{\manualLetter}^{\rho}=1,
\label{eq:ces_share_restriction}
\end{equation}
which is what the boundary condition at $\bar{\alpha} \to 1$ imposes.

Solving for $\phi(\bar{\alpha})$ in terms of $\theta_{\AIletter}$, $\theta_{\manualLetter}$, and $\rho$ then gives the following distribution for effective AI quality:
\begin{equation}
\phi(\bar{\alpha})=\frac{1}{1-\rho}\left(\frac{1-\theta_{\AIletter}-\theta_{\manualLetter}}{\theta_{\AIletter}\,\skillAdjustedTimeLetter_{\AIletter}^{\rho}}\right)^{\frac{1}{\rho}}(\bar{\alpha})^{\frac{1}{\rho-1}}\left[1-(\bar{\alpha})^{\frac{\rho}{\rho-1}}\right]^{-\frac{1+\rho}{\rho}}.
\label{eq:phi}
\end{equation}
Subsection~\ref{sec:aggregationDist} below derives this distribution.

This completes the production function aggregation procedure.

Two features of this construction are worth making explicit.
First, because $\skillAdjustedTimeLetter_{\manualLetter}$ is common to all firms, Equation~\eqref{eq:micro_Y} gives $\aggregateManualLabor = \skillAdjustedTimeLetter_{\manualLetter}\,Y$ at every wage vector, while capital is normalized to $K=1$.
As the participation threshold varies, the economy therefore traces a one-dimensional locus along which $\aggregateManualLabor/Y$ and $K$ are both fixed, rather than filling out the space of input combinations.
Equation~\eqref{eq:macro_agg_prod} should accordingly be read as a representation valid along that locus rather than as an identification of the economy's technology over arbitrary $(\aggregateAIlabor, \aggregateManualLabor, K)$.
Second, recovering a genuine three-input CES requires firms to differ in at least two of their input requirements, as in \citet{houthakker1955pareto} and \citet{levhari1968note}.
Our two-stage timing rules this out by construction, since every firm commits to the same $\T$ and $\J$ before learning $\bar{\alpha}$.
Allowing firms that anticipate deploying AI more effectively to organize work differently would make $\skillAdjustedTimeLetter_{\manualLetter}$ vary with $\bar{\alpha}$ and restore the missing dimension, at the cost of a joint distribution over organizational choices and realized quality that would be intractable without imposing further assumptions on the type of heterogeneity and firm behavior.

\subsection{Derivation of the Effective AI Quality Heterogeneity Distribution}
\label{sec:aggregationDist}

We derive \eqref{eq:phi} by extending the method of \cite{levhari1968note}.

Define $u=({w_{\AIletter}\,\skillAdjustedTimeLetter_{\AIletter}})/(1-w_{\manualLetter}\,\skillAdjustedTimeLetter_{\manualLetter})$ and rewrite Equation~\eqref{eq:macro_agg_prod_with_micro_aggregates} as:
\begin{equation}
\left(\int_{u}^{1}\,\phi(\bar{\alpha})\,d\bar{\alpha}\right)^{\rho}
=
\theta_{\AIletter}\left(\skillAdjustedTimeLetter_{\AIletter}\int_{u}^{1}\,\frac{\phi(\bar{\alpha})}{\bar{\alpha}}\,d\bar{\alpha}\right)^{\rho}
+
\theta_{\manualLetter} \left(\skillAdjustedTimeLetter_{\manualLetter}\int_{u}^{1}\,\phi(\bar{\alpha})\,d\bar{\alpha}\right)^{\rho}
+
\left(1-\theta_{\AIletter}-\theta_{\manualLetter}\right).
\label{eq:macro_agg_prod_with_hetDist}
\end{equation}
Define:
\begin{equation}
\Gamma(u)=\int_{u}^{1}\,\phi(\bar{\alpha})\,d\bar{\alpha},\label{eq:gamma}
\end{equation}
and
\begin{equation}
\Psi(u)=\int_{u}^{1}\,\frac{\phi(\bar{\alpha})}{\bar{\alpha}}\,d\bar{\alpha}.\label{eq:psi}
\end{equation}
Note that $\Gamma^{'}(u)=-\phi(u)$ and $\Psi^{'}(u)=-\phi(u)/u$.
Rewrite Equation~\eqref{eq:macro_agg_prod_with_hetDist} as:
\[
\Gamma^{\rho}(u)=\theta_{\AIletter}\,\skillAdjustedTimeLetter_{\AIletter}^{\rho}\,\Psi^{\rho}(u)+\theta_{\manualLetter}\,\skillAdjustedTimeLetter_{\manualLetter}^{\rho}\,\Gamma^{\rho}(u)+(1-\theta_{\AIletter}-\theta_{\manualLetter}).
\]
Rearranging terms, we obtain:
\begin{equation}
\left(1-\theta_{\manualLetter}\,\skillAdjustedTimeLetter_{\manualLetter}^{\rho}\right)\Gamma^{\rho}(u)=\theta_{\AIletter}\,\skillAdjustedTimeLetter_{\AIletter}^{\rho}\,\Psi^{\rho}(u)+(1-\theta_{\AIletter}-\theta_{\manualLetter}).\label{eq:macro_agg_pro_with_gamma}
\end{equation}

Differentiating Equation~\eqref{eq:macro_agg_pro_with_gamma} with respect to $u$ yields:
\[
\left(1-\theta_{\manualLetter}\,\skillAdjustedTimeLetter_{\manualLetter}^{\rho}\right)\Gamma^{\rho-1}(u)=\theta_{\AIletter}\,\skillAdjustedTimeLetter_{\AIletter}^{\rho}\,\Psi^{\rho-1}(u)\,u^{-1}.
\]
Solving for $\Psi^{\rho-1}(u)$ as a function of $u$ and $\Gamma^{\rho-1}(u)$ we get:
\[
\Psi^{\rho-1}(u)=\frac{\left(1-\theta_{\manualLetter}\,\skillAdjustedTimeLetter_{\manualLetter}^{\rho}\right)}{\theta_{\AIletter}\,\skillAdjustedTimeLetter_{\AIletter}^{\rho}}\Gamma^{\rho-1}(u)\,u.
\]
Finally, express $\Psi(u)$ in terms of $\Gamma(u)$:
\begin{equation}
\Psi(u)=\left(1-\theta_{\manualLetter}\,\skillAdjustedTimeLetter_{\manualLetter}^{\rho}\right)^{\frac{1}{\rho-1}}\left(\theta_{\AIletter}\,\skillAdjustedTimeLetter_{\AIletter}^{\rho}\right)^{\frac{1}{1-\rho}}u^{\frac{1}{\rho-1}}\,\Gamma(u).\label{eq:psi_solved}
\end{equation}

Substituting \( \Psi(u) \) from Equation~\eqref{eq:psi_solved} into Equation~\eqref{eq:macro_agg_pro_with_gamma}, we have:
\[
\left(1-\theta_{\manualLetter}\,\skillAdjustedTimeLetter_{\manualLetter}^{\rho}\right)\Gamma^{\rho}(u)=\left(1-\theta_{\manualLetter}\,\skillAdjustedTimeLetter_{\manualLetter}^{\rho}\right)^{\frac{\rho}{\rho-1}}\left(\theta_{\AIletter}\,\skillAdjustedTimeLetter_{\AIletter}^{\rho}\right)^{\frac{1}{1-\rho}}u^{\frac{\rho}{\rho-1}}\,\Gamma^{\rho}(u)+\left(1-\theta_{\AIletter}-\theta_{\manualLetter}\right).
\]
Rearranging terms to solve for \( \Gamma(u) \), we get:
\[
\Gamma(u)=\left(1-\theta_{\AIletter}-\theta_{\manualLetter}\right)^{\frac{1}{\rho}}\left[1-\theta_{\manualLetter}\,\skillAdjustedTimeLetter_{\manualLetter}^{\rho}-\left(1-\theta_{\manualLetter}\,\skillAdjustedTimeLetter_{\manualLetter}^{\rho}\right)^{\frac{\rho}{\rho-1}}\left(\theta_{\AIletter}\,\skillAdjustedTimeLetter_{\AIletter}^{\rho}\right)^{\frac{1}{1-\rho}}u^{\frac{\rho}{\rho-1}}\right]^{-\frac{1}{\rho}}.
\]

The relation $\Gamma'=-\phi$ pins down $\Gamma$ only up to an additive constant, so the expression above must still be checked against the boundary condition $\Gamma(1)=0$ that \eqref{eq:gamma} imposes by construction.
Evaluating the bracket at $u=1$ leaves
\[
1-\theta_{\manualLetter}\,\skillAdjustedTimeLetter_{\manualLetter}^{\rho}-\left(1-\theta_{\manualLetter}\,\skillAdjustedTimeLetter_{\manualLetter}^{\rho}\right)^{\frac{\rho}{\rho-1}}\left(\theta_{\AIletter}\,\skillAdjustedTimeLetter_{\AIletter}^{\rho}\right)^{\frac{1}{1-\rho}},
\]
which vanishes if and only if $1-\theta_{\manualLetter}\,\skillAdjustedTimeLetter_{\manualLetter}^{\rho}=\theta_{\AIletter}\,\skillAdjustedTimeLetter_{\AIletter}^{\rho}$, that is, if and only if Restriction~\eqref{eq:ces_share_restriction} holds.
Under that restriction the bracket collapses to $\theta_{\AIletter}\,\skillAdjustedTimeLetter_{\AIletter}^{\rho}\bigl[1-u^{\rho/(\rho-1)}\bigr]$, so that
\begin{equation}
\Gamma(u)=\left(\frac{1-\theta_{\AIletter}-\theta_{\manualLetter}}{\theta_{\AIletter}\,\skillAdjustedTimeLetter_{\AIletter}^{\rho}}\right)^{\frac{1}{\rho}}\left[1-u^{\frac{\rho}{\rho-1}}\right]^{-\frac{1}{\rho}},
\label{eq:gamma_solved}
\end{equation}
which satisfies $\Gamma(1)=0$ because $-1/\rho>0$ for $\rho<0$.
Restriction~\eqref{eq:ces_share_restriction} is therefore not an auxiliary assumption but the condition under which a density consistent with \eqref{eq:gamma} exists at all.
It ties the CES weights to $\skillAdjustedTimeLetter_{\AIletter}$ and $\skillAdjustedTimeLetter_{\manualLetter}$, and hence to the firm's AI strategy and job design.
It also delivers $1-\theta_{\manualLetter}\,\skillAdjustedTimeLetter_{\manualLetter}^{\rho}=\theta_{\AIletter}\,\skillAdjustedTimeLetter_{\AIletter}^{\rho}>0$, so the real powers appearing throughout are well defined.

Recall that \( \Gamma'(u)=-\phi(u) \).
Differentiating \eqref{eq:gamma_solved}, we obtain:
\[
\Gamma'(u)=\frac{1}{\rho-1}\left(\frac{1-\theta_{\AIletter}-\theta_{\manualLetter}}{\theta_{\AIletter}\,\skillAdjustedTimeLetter_{\AIletter}^{\rho}}\right)^{\frac{1}{\rho}}u^{\frac{1}{\rho-1}}\left[1-u^{\frac{\rho}{\rho-1}}\right]^{-\frac{1+\rho}{\rho}}.
\]
Thus, we have:
\[
\phi(\bar{\alpha})=\frac{1}{1-\rho}\left(\frac{1-\theta_{\AIletter}-\theta_{\manualLetter}}{\theta_{\AIletter}\,\skillAdjustedTimeLetter_{\AIletter}^{\rho}}\right)^{\frac{1}{\rho}}(\bar{\alpha})^{\frac{1}{\rho-1}}\left[1-(\bar{\alpha})^{\frac{\rho}{\rho-1}}\right]^{-\frac{1+\rho}{\rho}},
\]
which is the formula provided in \eqref{eq:phi} above.

\renewcommand{\refname}{\large References for Online Appendix}
\setlength{\bibsep}{0ex}
\begin{spacing}{0.98}
\putbib[rubin]
\end{spacing}
\end{bibunit}

% ====================================================== SUPPLEMENTARY APPENDIX
% Sections SA.A to SA.F: how the sample is built, and the tests that probe it.
% Back to the preamble's 9pt display spacing after the Online Appendix.
\makeatletter\let\normalsize\OA@normalsize\makeatother
\normalsize
\startappendixpart[Not for Publication]{SA}{Supplementary Appendix}
% No contents list here: as in the REStud supplement, the first section starts on the
% title page, which carries a page number like every other page.
\thispagestyle{plain}
\pagestyle{plain}

\begin{bibunit}[chicago]
\etocdepthtag.toc{suppappendix}

\section{Construction Details of the O*NET Sample}
\label{app:data_prep}

This appendix describes how the dataset used in our analyses is constructed from the four sources listed in Section~\ref{sec:empirics}:
\begin{enumerate}
    \item The O*NET 27.3 Database \citep{onet27_3},
    \item Human-generated AI exposure labels for O*NET tasks from \cite{eloundou2023gpts},
    \item Realized AI execution labels for O*NET tasks from the Anthropic Economic Index \citep{handa2025economictasksperformedai},
    \item A GPT-generated task ordering for each O*NET occupation.
\end{enumerate}
In what follows, we use the terms step, task, and O*NET task interchangeably, as in the main text; likewise, we use job and occupation interchangeably.

We use the May 2023 release of O*NET. 
This version contains roughly 18{,}000 tasks assigned to more than 850 U.S. occupations.
For AI exposure measures, we use human-generated labels from \cite{eloundou2023gpts}.
Throughout the paper, our measure of AI exposure combines their E1 and E2 labels.
We treat tasks with a human-assigned E1 or E2 label as exposed to AI and all other tasks as unexposed, which is the same label set from which the empirical fragmentation index of Section~\ref{sec:fragmentation_prediction} is constructed.\footnote{
\cite{eloundou2023gpts} label a task E1 if access to an LLM would reduce the time a human needs to complete it by at least half while preserving output quality.
E2 tasks reach the same time saving and quality only with additional software or tools built on top of the LLM.
Any task that is neither E1 nor E2 is labeled E0.
}

To obtain AI execution labels, we draw on Anthropic's Economic Index dataset \citep{handa2025economictasksperformedai}, which classifies millions of Claude conversations into six categories: \textit{``Validation'', ``Task Iteration'', ``Learning'', ``Directive'', ``Feedback Loop''}, and \textit{``Filtered.''}
Conversations in the first three categories (Validation, Task Iteration, Learning) are treated by Anthropic as AI-augmenting activities, whereas those in Directive and Feedback Loop are recognized as AI-automating activities.
The sixth category, Filtered, corresponds to conversations that either could not be classified or whose actual category could not be disclosed for privacy reasons.

Anthropic reports the share of matched conversations falling into each of the six categories for a subset of O*NET tasks.
In total, conversations are linked to 3{,}364 O*NET task statements, of which 1{,}066 have 100\% of their conversations filtered, leaving 2{,}298 with at least one non-filtered conversation.
Following Anthropic's definitions of AI-augmenting and AI-automating activities, we assign each of these 2{,}298 task statements an AI execution label based on the majority share of its non-filtered conversations across the augmenting versus automating categories.\footnote{For 23 of the 2{,}298 statements the two shares are exactly equal, and we break the tie toward automation. These statements account for 28 of the 721 automation records in the final sample. Because our analyses treat the two labels alike as indicators of AI execution, the tie-breaking convention does not affect any reported result.}
Of the 2{,}298, 1{,}830 match a task statement in our O*NET sample; since the same statement can be listed under several occupations, these labels attach to 2{,}347 of the 17{,}925 occupation-task records in the final sample (1{,}626 augmentation and 721 automation).
Table~\ref{tab:example_anthropic_AI_tasks} provides several example tasks labeled using this procedure. 

Finally, we treat tasks that do not appear in the Anthropic dataset, as well as those with 100\% filtered conversations, as manual tasks.
Figure \ref{fig:task_execution_dist} shows the distribution of task execution modes in our dataset. 
Of the 872 occupations, 809 (93\%) contain at least one AI-exposed task and 555 (64\%) contain at least one AI-executed task.
Figure \ref{fig:occupation_ai_share} shows the distribution of occupations by the fraction of their tasks that are exposed to AI (left) and executed by AI (right).
\begin{figure}[!t]
    \caption{Distribution of Modes of Task Execution in the Dataset}
    \label{fig:task_execution_dist}
    \centering
    \resizebox{0.6\textwidth}{!}{
    \includegraphics[width=.8\textwidth]{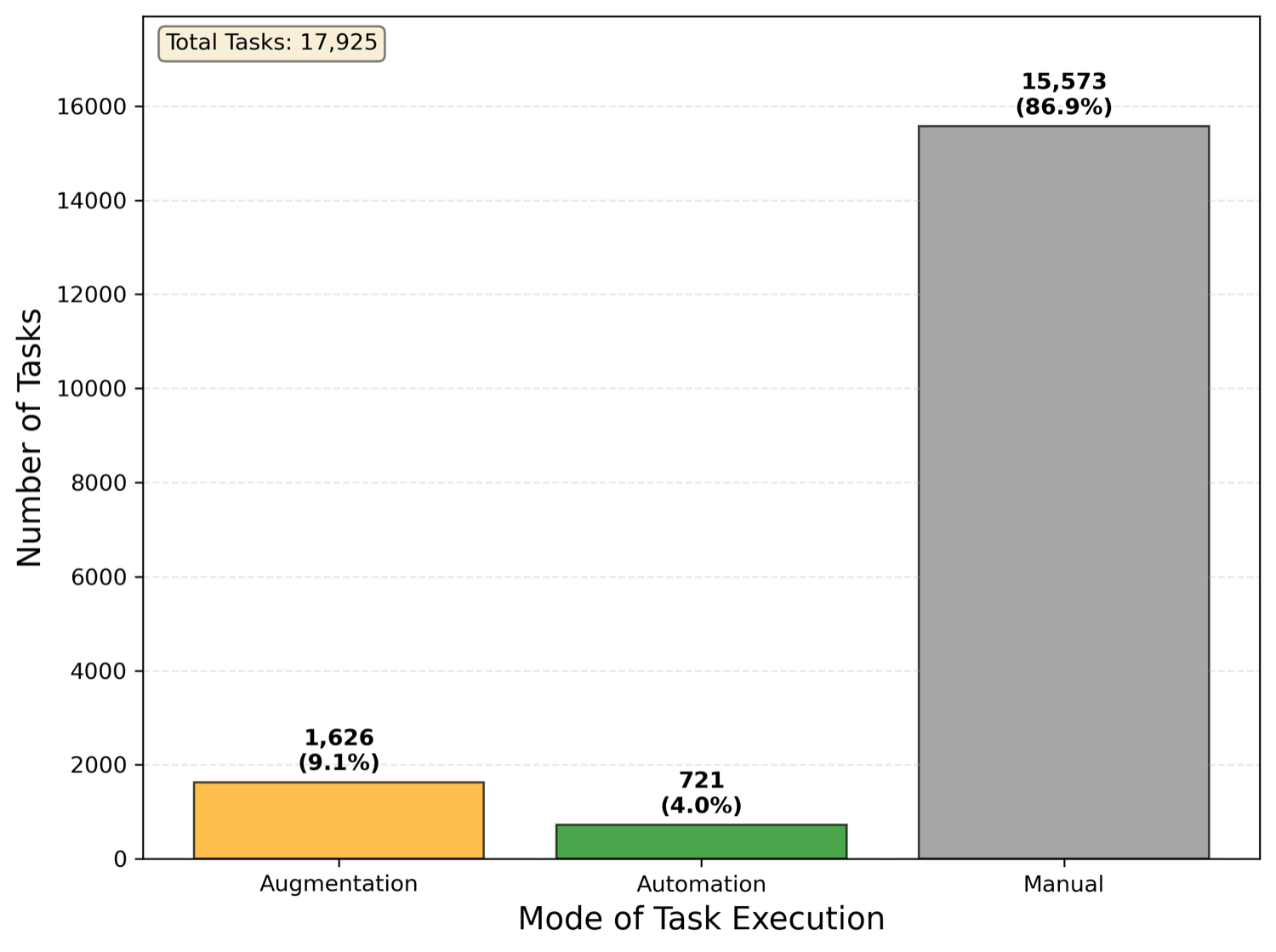}
    }
    \par\vspace{\fignotesskip}
\begin{minipage}{1\linewidth}
\begin{spacing}{0.2}
{\fontsize{9.5pt}{10pt}\selectfont Notes: The automation and augmentation labels are drawn from Anthropic's Economic Index dataset, and the universe of tasks comes from the May~2023 O*NET release.}
\end{spacing}
\end{minipage}
\end{figure}
\begin{figure}[!t]
    \caption{Distribution of Share of Occupation Tasks Exposed to and Executed by AI}
    \label{fig:occupation_ai_share}
    \centering
    \begin{subfigure}[b]{0.49\textwidth}
        
        \subcaption{AI Exposure}
        \includegraphics[width=\textwidth]{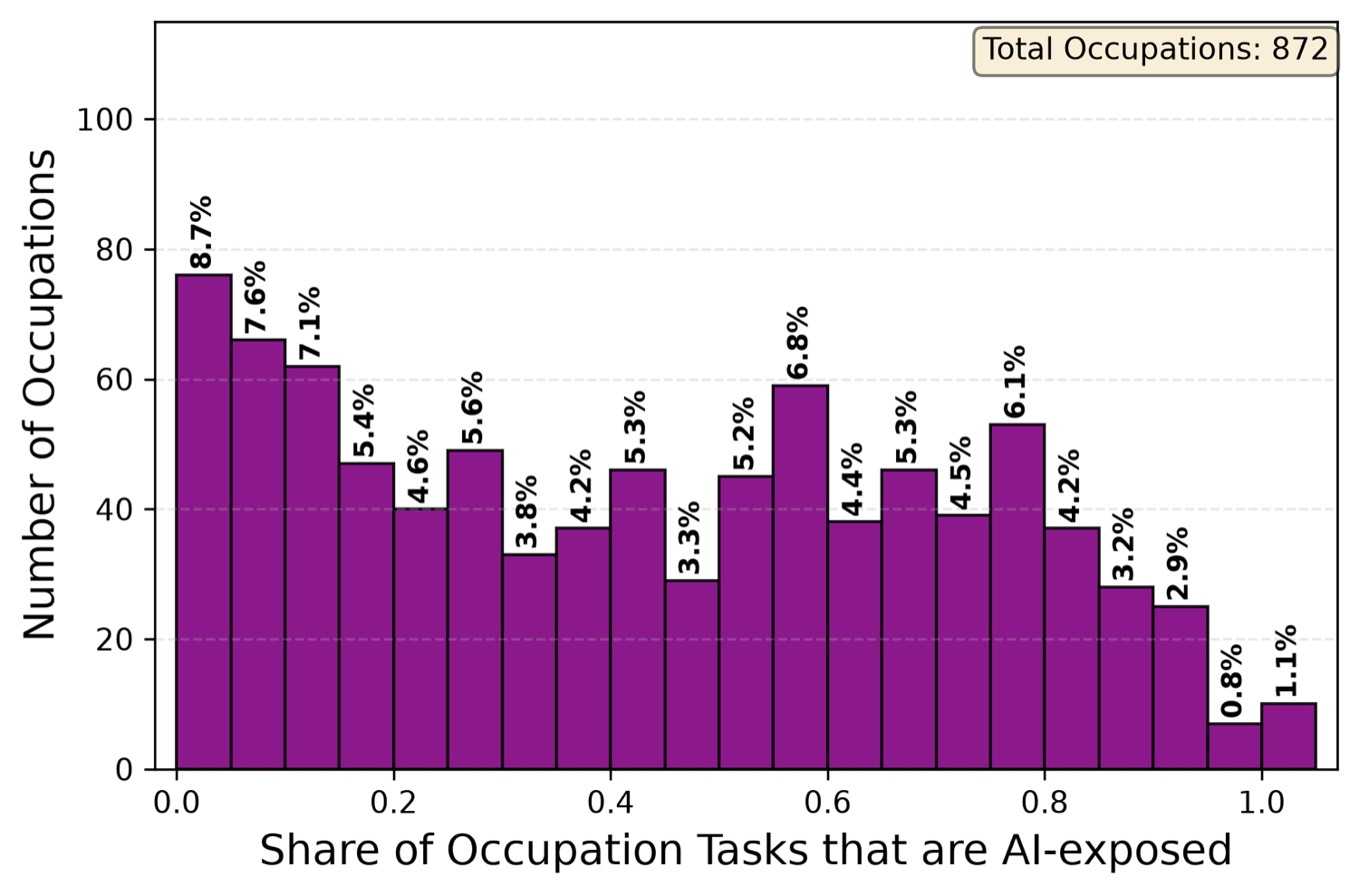}
    \end{subfigure}
    \begin{subfigure}[b]{0.49\textwidth}
        
        \subcaption{AI Execution}
        \includegraphics[width=\textwidth]{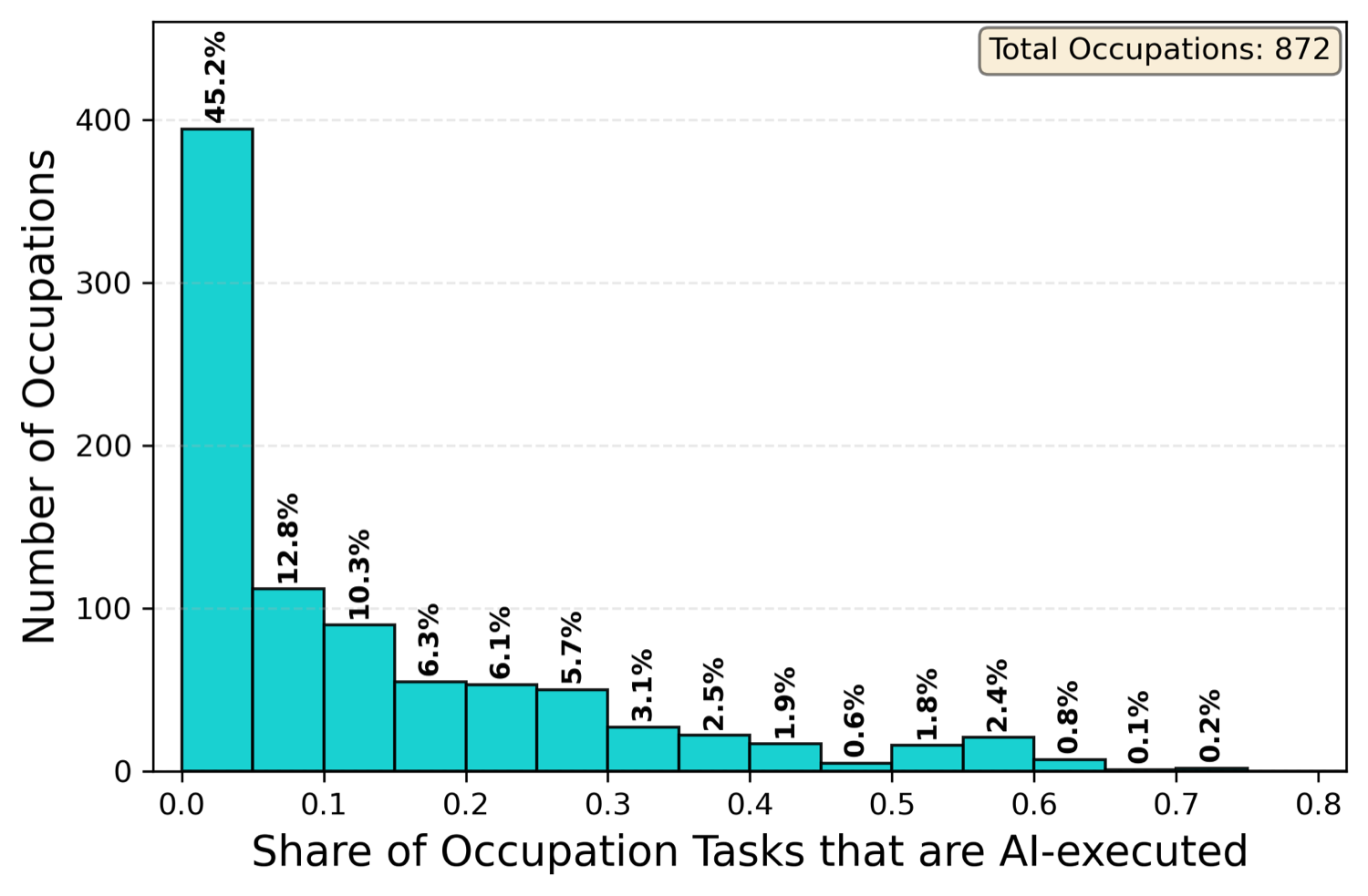}
    \end{subfigure}
    \par\vspace{\fignotesskip}
\begin{minipage}{1\linewidth}
\begin{spacing}{0.2}
{\fontsize{9.5pt}{10pt}\selectfont Notes: Percentages above bars denote fraction of observations corresponding to that bar out of all 872 occupations.}
\end{spacing}
\end{minipage}
\end{figure}

As the final input into our analyses, we generate a workflow sequence for every O*NET occupation. 
Specifically, we prompt GPT-5-mini through Expected Parrot \citep{Horton2024EDSL} with the complete set of tasks within each occupation and ask it to determine the most reasonable sequential ordering in which these tasks would be performed in practice.\footnote{The prompt used to generate task sequences is given in Appendix \ref{app:prompts}.
}
Figures \ref{fig:computer_programmers}, \ref{fig:public_relations}, and \ref{fig:electronic_equipments} show the resulting task sequences for three occupations: \textit{Computer Programmers} with a high share of its tasks executed by AI, \textit{Public Relations Specialists} with a moderate share, and \textit{Electronic Equipment Installers and Repairers, Motor Vehicles} with no tasks executed by AI.
We validate our approach in working with O*NET occupation task sequences in a variety of ways, discussed in Section~\ref{sec:empirics}.

Before closing, it is worth emphasizing a few points about the dataset we construct.
First, note that we use a conservative measure of AI involvement in task execution.
Recall that we label tasks with fully filtered conversations as manual, even though we know that all of them were executed in some way using Claude.
These tasks account for roughly one third of all Anthropic-mapped tasks, and excluding them limits our statistical power.
Nevertheless, the remaining two thirds prove sufficient for revealing the patterns predicted by our framework, suggesting that what we report should be viewed as a conservative signal of potentially much stronger effects.

A second concern is that workers may use different AI tools for different purposes, for example Claude for writing and ChatGPT for coding, which could potentially undermine our results.
While we cannot directly rule out such heterogeneous tool use given that we only observe Claude usage, our analyses are not tied to particular industries, occupations, or tasks, which makes them less vulnerable to missing isolated pockets of tool-specific usage.\footnote{
For our results to be entirely spurious, one would have to believe not only that no one used Claude for a large subset of tasks during Anthropic's sample period but also that this missing class of tasks has distinctive characteristics, such as workflow position or AI-suitability, that differ from other tasks for which Claude is used across the economy.
We do not have a reason to believe this is the case.
If anything, specialization across tools would imply that Anthropic's data might understate the true set of AI-executed tasks in the economy, which would primarily reduce our statistical power rather than generate artificial patterns.
}
Nevertheless, wherever applicable, we verify that our results are not driven by any narrow subset of occupations or tasks.
The patterns we document appear consistently across broad occupation groups, even if somewhat stronger in some groups than others, reflecting heterogeneous impacts of AI.
While we do not claim that these relationships are causal, taken together they provide strong suggestive evidence that the mechanisms highlighted by our model are relevant in practice and that our findings reflect broad patterns of AI execution rather than Claude-specific usage.

Third, the realized execution modes in the constructed dataset are broadly consistent with, but do not perfectly satisfy, the model's definition of an AI chain.
In the model an automated step can only be followed by another AI-executed step, either automated or augmented, since an automated step's output has to feed into a subsequent AI step rather than into a human.
In the data we occasionally observe an automated task followed by a manual one.
Task 5 in \textit{Computer Programmers} (Figure~\ref{fig:computer_programmers}) and Task 3 in \textit{Public Relations Specialists} (Figure~\ref{fig:public_relations}) are two such instances.

This is a definitional difference rather than a contradiction.
The automated versus augmented distinction in our model is defined by production sequencing and task dependencies, whereas the Anthropic classifications are based on task-level feasibility and are agnostic to workflow position.
A task is labeled automated if the conversations covering it are predominantly directive, regardless of what follows it in the occupation's workflow.
Tasks labeled automated in the data may therefore be followed by manual steps, producing apparent deviations from the model's sequencing restriction without contradicting its underlying logic.
For this reason the empirical analysis in Section~\ref{sec:empirics} does not distinguish tasks labeled automated from those labeled augmented, and treats both as indicating AI execution.

\begin{landscape}
\begin{table}[!t]
    \caption{Example O*NET Tasks Associated with Claude Conversations, and Calculated Augmentation vs.\ Automation Outcomes}
    \label{tab:example_anthropic_AI_tasks}
    \footnotesize
    \centering
    \resizebox{1.33\textwidth}{!}{
\begin{tabular}{|p{9.5cm}|cccccc|cc|c|}
\hline
\multirow{2}{*}{\textbf{O*NET Task}} &
\multicolumn{8}{c|}{\textbf{Share of Claude Conversations (\%)}} & \multirow{2}{*}{\textbf{Label}} \\ % <-- Added | after 8-column group
\cline{2-9} % <-- extended to include column 10 (Label)
& \textbf{Validation} &
\textbf{Task Iteration} &
\textbf{Learning} &
\textbf{Directive} &
\textbf{Feedback Loop} &
\textbf{Filtered} &
\textbf{Augmentation} &
\textbf{Automation} &
\\
\hline
\raggedright Provide road information to assist motorists. &
0.00 & 11.47 & 44.39 & 33.42 & 5.74 & 4.99 & \textbf{55.86} & 39.15 & Augmentation \\
\hline
\raggedright Calculate costs of orders, and charge or forward invoices to appropriate accounts. &
0.00 & 0.00 & 0.00 & 68.00 & 0.00 & 32.00 & 0.00 & \textbf{68.00} & Automation \\
\hline
\raggedright Process data for analysis, using computers. &
4.23 & 18.60 & 31.40 & 25.79 & 16.38 & 3.59 & \textbf{54.23} & 42.18 & Augmentation \\
\hline
\raggedright Adapt text to accommodate musical requirements of composers and singers. &
0.00 & 40.00 & 0.00 & 58.26 & 0.00 & 1.74 & 40.00 & \textbf{58.26} & Automation \\
\hline
\raggedright Provide system design and integration recommendations. &
2.13 & 29.75 & 39.98 & 21.23 & 5.80 & 1.11 & \textbf{71.87} & 27.02 & Augmentation \\
\hline
\raggedright Communicate traffic and crossing rules and other information to students and adults. &
0.00 & 0.00 & 0.00 & 65.52 & 0.00 & 34.48 & 0.00 & \textbf{65.52} & Automation \\
\hline
\raggedright Prepare, manipulate, and manage extensive databases. &
0.00 & 25.89 & 22.34 & 23.40 & 24.47 & 3.90 & \textbf{48.23} & 47.87 & Augmentation \\
\hline
\raggedright Prepare, administer, and grade tests and assignments to evaluate students' progress. &
7.64 & 20.94 & 7.39 & 58.13 & 3.94 & 1.97 & 35.96 & \textbf{62.07} & Automation \\
\hline
\raggedright Conduct statistical analyses to quantify risk using statistical analysis software or econometric models. &
0.00 & 25.23 & 31.53 & 17.12 & 18.92 & 7.21 & \textbf{56.76} & 36.04 & Augmentation \\
\hline
\raggedright Assemble, typeset, scan and produce digital camera-ready art or film negatives and printer's proofs. &
0.00 & 21.88 & 0.00 & 77.34 & 0.00 & 0.78 & 21.88 & \textbf{77.34} & Automation \\
\hline
\raggedright Conduct searches to find needed information, using such sources as the internet. &
1.31 & 5.53 & 58.30 & 23.14 & 3.57 & 8.15 & \textbf{65.14} & 26.71 & Augmentation \\
\hline
\raggedright Compile data and create statistical reports on library usage. &
0.00 & 0.00 & 0.00 & 64.86 & 0.00 & 35.14 & 0.00 & \textbf{64.86} & Automation \\
\hline
\raggedright Create custom illustrations or other graphic elements. &
0.00 & 48.17 & 3.17 & 43.82 & 2.74 & 2.10 & \textbf{51.34} & 46.56 & Augmentation \\
\hline
\raggedright Confer with clients to obtain and provide information when claims are made on a policy. &
0.00 & 0.00 & 45.45 & 0.00 & 0.00 & 54.55 & \textbf{45.45} & 0.00 & Augmentation \\
\hline
\end{tabular}
}
    \begin{minipage}{1\linewidth}
\begin{spacing}{0.2}
{\fontsize{9.5pt}{10pt}\selectfont Notes: Augmentation percentage is the sum of Validation, Task Iteration, and Learning types of conversations whereas the automation percentage is the sum of Directive and Feedback Loop types of conversations.
Whichever mode between Augmentation and Automation has a higher share is selected as the task's mode of AI execution, ignoring the Filtered share.}
\end{spacing}
\end{minipage}
\end{table}
\end{landscape}

% The three task-sequence figures are gathered here, at the end of the appendix and after
% all of its prose, rather than floating to the tops of the pages that discuss them. Each
% is [p] so it claims a dedicated float page: at 303pt, 316pt and 234pt of graphic plus a
% caption and a three-line Notes block, no two of them fit on one page together.
\begin{figure}[p]
    \caption{Task Sequence of Computer Programmers Occupation}
    \label{fig:computer_programmers}
    \centering
    \resizebox{\textwidth}{!}{%
    \includegraphics[width=\linewidth]{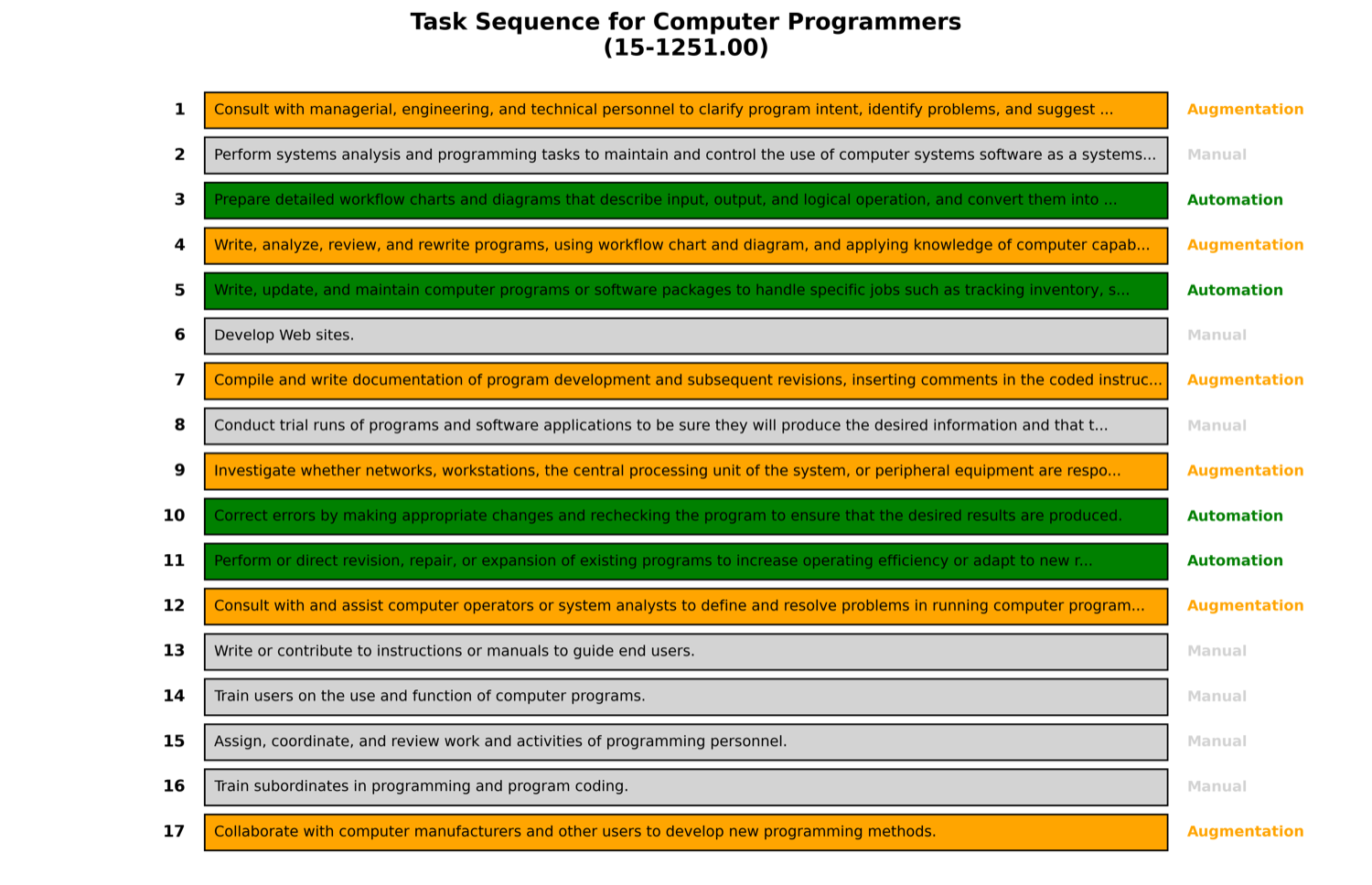}%
    }%
    \par\vspace{\fignotesskip}
\begin{minipage}{1\linewidth}
\begin{spacing}{0.2}
{\fontsize{9.5pt}{10pt}\selectfont Notes: This figure shows the task sequence for an occupation with a high share of its tasks (about two thirds) executed by AI.
The task ordering is generated by GPT-5-mini.
The execution labels come from the Anthropic Economic Index.
The O*NET code for this occupation is 15-1251.}
\end{spacing}
\end{minipage}
\end{figure}

\begin{figure}[p]
    \caption{Task Sequence of Public Relations Specialists Occupation}
    \label{fig:public_relations}
    \centering
    \resizebox{\textwidth}{!}{%
    \includegraphics[width=\linewidth]{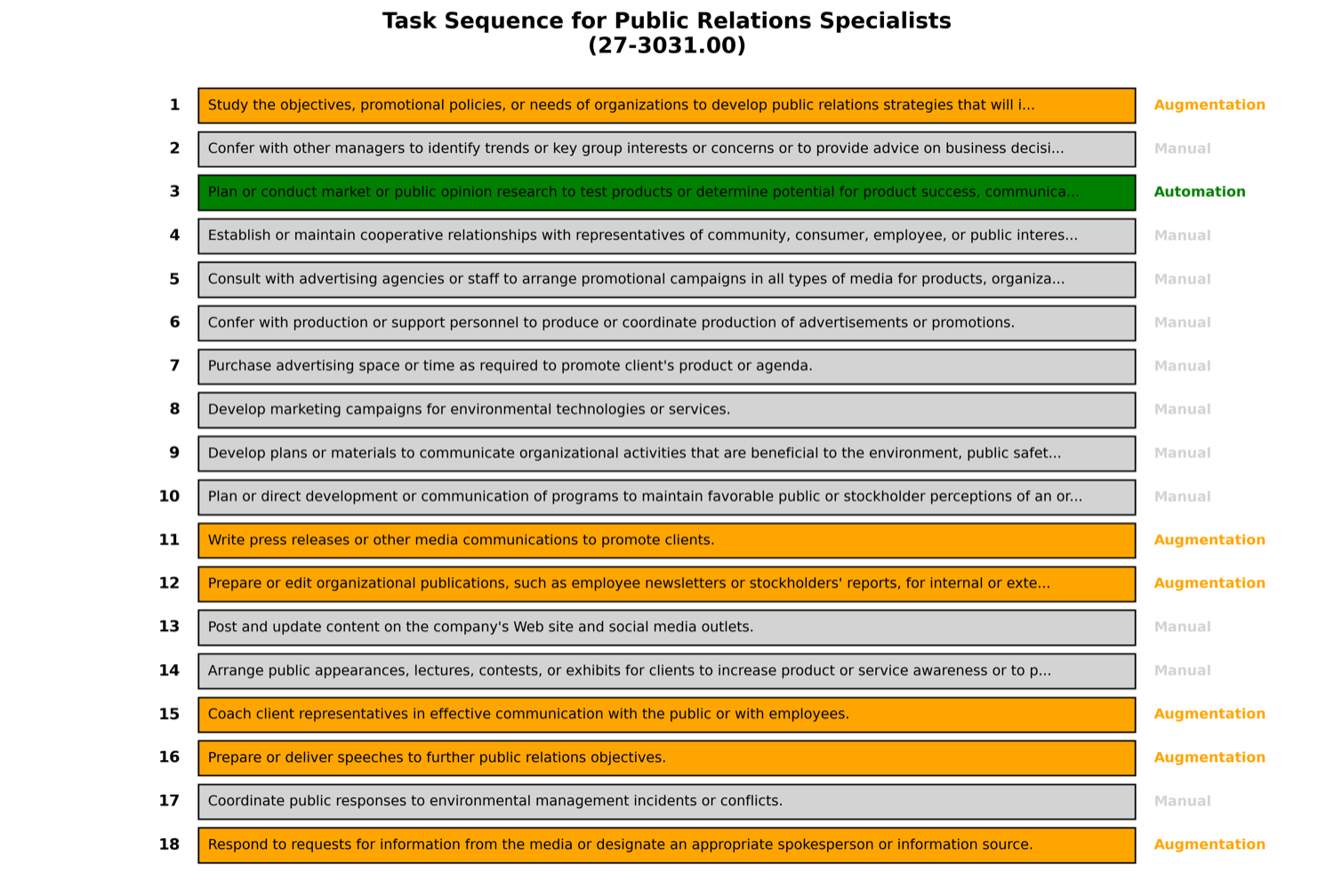}%
    }%
    \par\vspace{\fignotesskip}
\begin{minipage}{1\linewidth}
\begin{spacing}{0.2}
{\fontsize{9.5pt}{10pt}\selectfont Notes: This figure shows the task sequence for an occupation with a moderate share of its tasks (about one third) executed by AI.
Ordering and execution labels are as in Figure~\ref{fig:computer_programmers}.
The O*NET code for this occupation is 27-3031.}
\end{spacing}
\end{minipage}
\end{figure}

\begin{figure}[p]
    \caption{Task Sequence of Electronic Equipment Installers and Repairers, Motor Vehicles}
    \label{fig:electronic_equipments}
    \centering
    \resizebox{\textwidth}{!}{%
    \includegraphics[width=\linewidth]{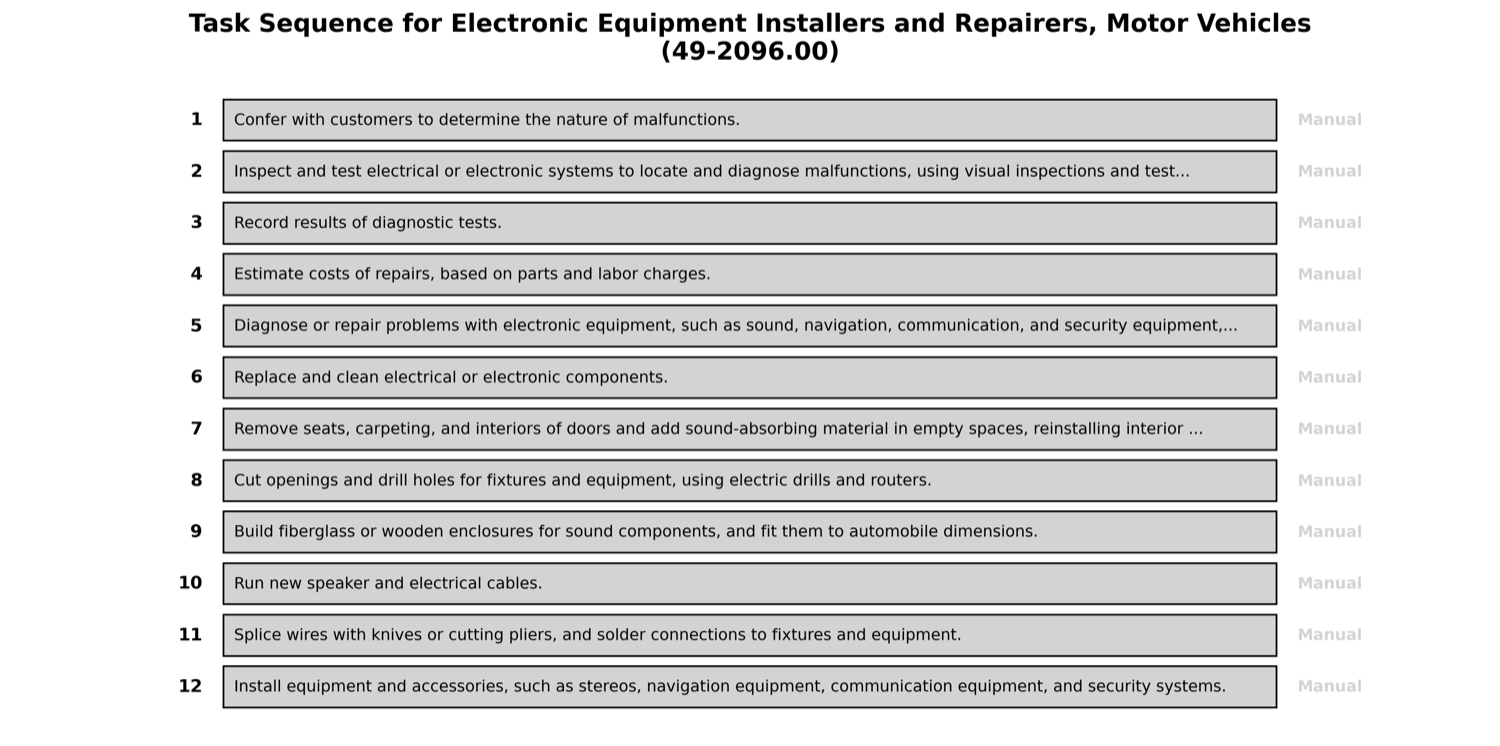}%
    }%
    \par\vspace{\fignotesskip}
\begin{minipage}{1\linewidth}
\begin{spacing}{0.2}
{\fontsize{9.5pt}{10pt}\selectfont Notes: This figure shows the task sequence for an occupation with none of its tasks executed by AI.
Ordering and execution labels are as in Figure~\ref{fig:computer_programmers}.
The O*NET code for this occupation is 49-2096.}
\end{spacing}
\end{minipage}
\end{figure}

\clearpage
% The title runs past one line of the 14pt section head and TeX hyphenates it as
% "Step Similar-ity", which reads badly in a heading. The optional argument keeps the
% contents list and the PDF bookmark on one unbroken string; the starred break puts the
% printed heading's turn at a comma instead of inside a word. \texorpdfstring is needed
% because a bare \\ is not allowed in a bookmark string.
\section{Alternative Definitions of Step Similarity and AI Execution}
\label{app:additional_robustness}
% The house format sets every float [!t]. On the first page of an appendix, that
% puts the leading table at the top of the page that carries the appendix heading,
% above the heading itself, so the table reads as though it belonged to the
% previous appendix. \suppressfloats[t] keeps the top of this page clear, and the
% leading table drops the ! so that it honours the suppression and starts at the
% top of the next page; the tables after it keep [!t] and follow in order behind it.
\suppressfloats[t]

This appendix re-estimates Prediction~\#2 of Section~\ref{sec:empirics} under alternative definitions of the variables it is built from, holding the task sequences fixed.
It tightens what it means for two tasks to be the same step across occupations, and replaces AI execution with the stricter outcome of AI automation.

Although DWAs are intended to group conceptually similar tasks, one might argue that even within a given DWA tasks may still somewhat differ in their objectives, execution nature, or required skills.
To address this concern, we re-estimate Equation~\eqref{eq:DWA_regression_ai} using a more restrictive definition of task similarity within DWAs.
Specifically, for each DWA, we ask GPT-5-mini to select the subset of tasks that are most similar in terms of skill requirements and execution complexity.\footnote{The prompt used to identify similar DWA tasks is provided in Appendix~\ref{app:prompts}.}
This procedure yields, for each DWA, a set of highly comparable tasks that appear across different occupational contexts.
Panel~(a) of Table~\ref{tab:DWA_tasks} presents example DWAs and their associated O*NET tasks across occupations, along with execution modes, while Panel~(b) reports the subset of tasks retained after applying the similarity filter.

In this robustness test we treat these tasks as instances of the same underlying step appearing across occupations and re-estimate Equation~\eqref{eq:DWA_regression_ai} on this restricted sample.
The resulting estimates, reported in Table~\ref{tab:DWA_regression_aiExecution_GPTsample} and Figure~\ref{fig:DWA_regression_aiExecution_GPTsample}, closely mirror the baseline patterns of the main specification in both sign and magnitude.

Finally, we examine a stricter outcome that is closer to the model's notion of automation, under which automated steps can only precede other AI-executed steps, whether automated or augmented.
Although the Anthropic automation label is not defined by workflow position (Appendix~\ref{app:data_prep}), we repeat the analysis focusing specifically on AI automation rather than AI execution more broadly, replacing the dependent variable in Equation~\eqref{eq:DWA_regression_ai} with an indicator for whether step $k$ is AI-automated ($\text{is\_automated}_{k}$) instead of AI-executed ($\text{is\_ai}_{k}$).
The corresponding estimates for the O*NET sample are reported in Table~\ref{tab:DWA_regression_aiAutomation_mainSample} and Figure~\ref{fig:DWA_regression_aiAutomation_mainSample}, while results for the GPT-filtered sample are reported in Table~\ref{tab:DWA_regression_aiAutomation_GPTsample} and Figure~\ref{fig:DWA_regression_aiAutomation_GPTsample}.\footnote{
Because relatively few tasks in the data are AI-automated, and because we further restrict attention to highly similar tasks across occupations, we retain tasks associated with multiple DWAs in the original O*NET dataset and randomly assign each such task to one of its associated DWAs for these exercises.
}
Under this stricter specification, the estimated effects are smaller in magnitude and often statistically insignificant, reflecting the more limited prevalence of AI automation in the data. Nevertheless, the signs and qualitative patterns of the estimates remain consistent with those obtained for AI execution.

\newpage
\begin{table}[!t]
\centering
\caption{Neighboring Tasks' AI Execution Status and a Task's AI Execution (GPT-filtered Sample)}
\label{tab:DWA_regression_aiExecution_GPTsample}
    \footnotesize
\resizebox{\textwidth}{!}{
% --- LaTeX Table for filtered_0 ---
\begin{tabular}{lcccccc}
\toprule
 & \multicolumn{6}{c}{Probability that Focal Task ($k$) is AI-executed} \\
\cmidrule(lr){2-7}
 & (1) & (2) & (3) & (4) & (5) & (6) \\
\midrule
Task ($k-2$) is AI-executed & 0.04 & 0.01 & 0.00 & 0.04 & 0.04 & 0.04 \\
 & (0.03) & (0.02) & (0.02) & (0.03) & (0.04) & (0.03) \\
\addlinespace
Task ($k-1$) is AI-executed & 0.11*** & 0.06** & 0.06** & 0.06** & 0.10*** & 0.06** \\
 & (0.04) & (0.03) & (0.03) & (0.03) & (0.04) & (0.03) \\
\addlinespace
Task ($k+1$) is AI-executed & 0.11*** & 0.08*** & 0.07*** & 0.09*** & 0.10*** & 0.09*** \\
 & (0.03) & (0.03) & (0.02) & (0.03) & (0.03) & (0.03) \\
\addlinespace
Task ($k+2$) is AI-executed & 0.01 & -0.02 & -0.03 & -0.00 & 0.00 & -0.00 \\
 & (0.03) & (0.02) & (0.02) & (0.03) & (0.04) & (0.03) \\
\addlinespace
\midrule
Pseudo $R^2$ & 0.032 & 0.056 & 0.065 & 0.199 & 0.020 & 0.199 \\
Observations & 3,689 & 3,689 & 3,567 & 2,544 & 2,544 & 2,544 \\
SOC Group FE &  & Major & Minor &  &  &  \\
DWA FE &  &  &  & \checkmark &  & \checkmark \\
NumTasks in DWA-Occupation Control &  &  &  &  & \checkmark & \checkmark \\
\bottomrule
\end{tabular}

}
\begin{minipage}{1\linewidth}
\begin{spacing}{0.2}
{\fontsize{9.5pt}{10pt}\selectfont Notes: Clustered standard errors in parentheses. *** p$<$0.01, ** p$<$0.05, * p$<$0.1.
This table reports average marginal effects from estimating Equation~\eqref{eq:DWA_regression_ai}.
The dependent variable in all specifications is an indicator for whether task $k$ is AI-executed ($\text{is\_ai}_{k}$).
The sample is restricted to similar tasks across occupations, identified by GPT-5-mini as tasks belonging to the same DWA with similar execution nature and skill characteristics.
Controls and standard errors are as in Table~\ref{tab:DWA_regression_aiExecution_mainSample}.}
\end{spacing}
\end{minipage}\end{table}

\begin{table}[!t]
\centering
\caption{Neighboring Tasks' AI Execution Status and a Task's AI Automation}
\label{tab:DWA_regression_aiAutomation_mainSample}
    \footnotesize
\resizebox{\textwidth}{!}{
% --- LaTeX Table for full_0 ---
\begin{tabular}{lcccccc}
\toprule
 & \multicolumn{6}{c}{Probability that Focal Task ($k$) is AI-automated} \\
\cmidrule(lr){2-7}
 & (1) & (2) & (3) & (4) & (5) & (6) \\
\midrule
Task ($k-2$) is AI-executed & 0.02** & -0.01 & -0.01 & -0.01 & 0.02 & -0.01 \\
 & (0.01) & (0.01) & (0.01) & (0.02) & (0.03) & (0.02) \\
\addlinespace
Task ($k-1$) is AI-executed & 0.04*** & 0.01* & 0.01 & 0.01 & 0.08** & 0.02 \\
 & (0.02) & (0.01) & (0.01) & (0.02) & (0.03) & (0.02) \\
\addlinespace
Task ($k+1$) is AI-executed & 0.05*** & 0.02*** & 0.02** & 0.01 & 0.09** & 0.01 \\
 & (0.01) & (0.01) & (0.01) & (0.02) & (0.03) & (0.02) \\
\addlinespace
Task ($k+2$) is AI-executed & 0.03*** & 0.00 & -0.00 & 0.02 & 0.06* & 0.02 \\
 & (0.01) & (0.01) & (0.01) & (0.02) & (0.03) & (0.02) \\
\addlinespace
\midrule
Pseudo $R^2$ & 0.102 & 0.192 & 0.173 & 0.242 & 0.036 & 0.242 \\
Observations & 13,786 & 12,777 & 10,390 & 2,815 & 2,815 & 2,815 \\
SOC Group FE &  & Major & Minor &  &  &  \\
DWA FE &  &  &  & \checkmark &  & \checkmark \\
NumTasks in DWA-Occupation Control &  &  &  &  & \checkmark & \checkmark \\
\bottomrule
\end{tabular}

}
\begin{minipage}{1\linewidth}
\begin{spacing}{0.2}
{\fontsize{9.5pt}{10pt}\selectfont Notes: Clustered standard errors in parentheses. *** p$<$0.01, ** p$<$0.05, * p$<$0.1.
This table reports average marginal effects from estimating Equation~\eqref{eq:DWA_regression_ai} but with dependent variable ($\text{is\_automated}_{k}$) instead of ($\text{is\_ai}_{k}$).
The sample is restricted to similar tasks across occupations, identified as tasks belonging to the same DWA in the O*NET dataset.
Tasks that are mapped to multiple DWAs are retained, with one associated DWA assigned randomly for estimation to preserve sufficient variation in the data.
Controls and standard errors are as in Table~\ref{tab:DWA_regression_aiExecution_mainSample}.}
\end{spacing}
\end{minipage}\end{table}

\begin{table}[!t]
\centering
\caption{Neighboring Tasks' AI Execution Status and a Task's AI Automation (GPT-filtered Sample)}
\label{tab:DWA_regression_aiAutomation_GPTsample}
    \footnotesize
\resizebox{\textwidth}{!}{
% --- LaTeX Table for filtered_0 ---
\begin{tabular}{lcccccc}
\toprule
 & \multicolumn{6}{c}{Probability that Focal Task ($k$) is AI-automated} \\
\cmidrule(lr){2-7}
 & (1) & (2) & (3) & (4) & (5) & (6) \\
\midrule
Task ($k-2$) is AI-executed & 0.01 & -0.02 & -0.02 & 0.00 & 0.01 & 0.00 \\
 & (0.01) & (0.01) & (0.02) & (0.03) & (0.03) & (0.03) \\
\addlinespace
Task ($k-1$) is AI-executed & 0.06** & 0.02 & 0.02 & 0.03 & 0.11** & 0.03 \\
 & (0.03) & (0.01) & (0.01) & (0.03) & (0.04) & (0.03) \\
\addlinespace
Task ($k+1$) is AI-executed & 0.06** & 0.03* & 0.03 & 0.02 & 0.10** & 0.02 \\
 & (0.03) & (0.02) & (0.02) & (0.03) & (0.05) & (0.03) \\
\addlinespace
Task ($k+2$) is AI-executed & 0.01 & -0.02 & -0.02 & 0.01 & 0.02 & 0.01 \\
 & (0.02) & (0.01) & (0.02) & (0.02) & (0.04) & (0.03) \\
\addlinespace
\midrule
Pseudo $R^2$ & 0.042 & 0.125 & 0.118 & 0.244 & 0.033 & 0.244 \\
Observations & 5,156 & 5,077 & 4,490 & 1,771 & 1,771 & 1,771 \\
SOC Group FE &  & Major & Minor &  &  &  \\
DWA FE &  &  &  & \checkmark &  & \checkmark \\
NumTasks in DWA-Occupation Control &  &  &  &  & \checkmark & \checkmark \\
\bottomrule
\end{tabular}

}
\begin{minipage}{1\linewidth}
\begin{spacing}{0.2}
{\fontsize{9.5pt}{10pt}\selectfont Notes: Clustered standard errors in parentheses. *** p$<$0.01, ** p$<$0.05, * p$<$0.1.
This table reports average marginal effects from estimating Equation~\eqref{eq:DWA_regression_ai} but with dependent variable ($\text{is\_automated}_{k}$) instead of ($\text{is\_ai}_{k}$).
The sample is restricted to similar tasks across occupations, identified by GPT-5-mini as tasks belonging to the same DWA with similar execution nature and skill characteristics.
Tasks that are mapped to multiple DWAs are retained, with one associated DWA assigned randomly for estimation to preserve sufficient variation in the data.
Controls and standard errors are as in Table~\ref{tab:DWA_regression_aiExecution_mainSample}.}
\end{spacing}
\end{minipage}\end{table}

\begin{figure}[!t]
  \centering
  \caption{Neighboring Tasks' AI Execution Status and a Task's AI Execution (GPT-filtered Sample): Placebo Distributions} 
  \label{fig:DWA_regression_aiExecution_GPTsample}
  \begin{subfigure}[b]{\textwidth}
    
    \subcaption{No Fixed Effects}
    \includegraphics[width=\textwidth]{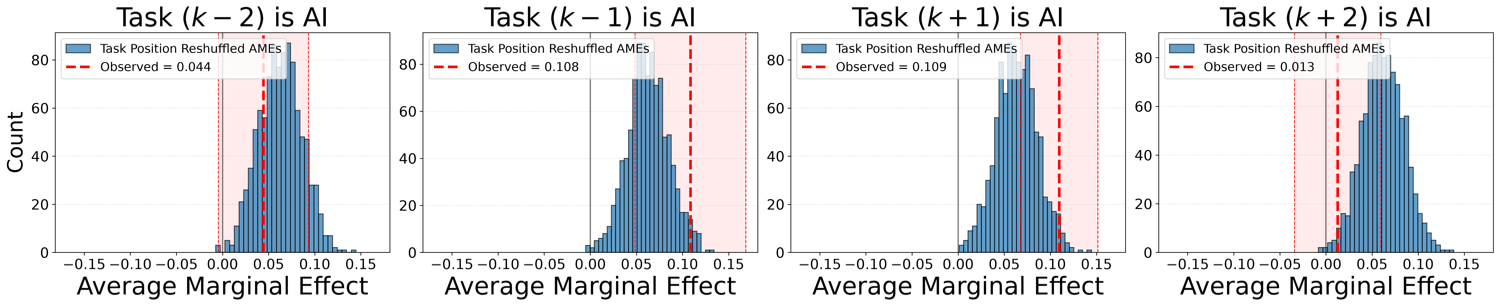}
  \end{subfigure}
  
  \vspace{1em}
  
  \begin{subfigure}[b]{\textwidth}
    
    \subcaption{SOC Major Group Fixed Effects}
    \includegraphics[width=\textwidth]{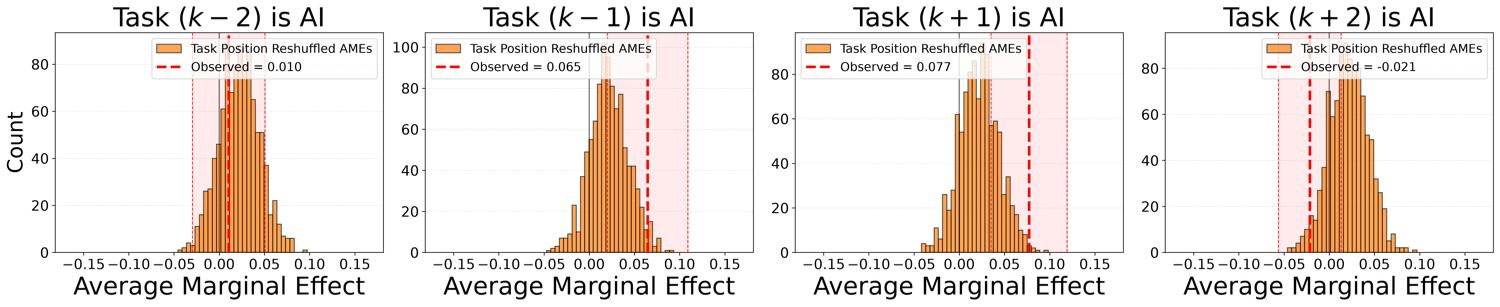}
  \end{subfigure}
  
  \vspace{1em}
  
  \begin{subfigure}[b]{\textwidth}
    
    \subcaption{SOC Minor Group Fixed Effects}
    \includegraphics[width=\textwidth]{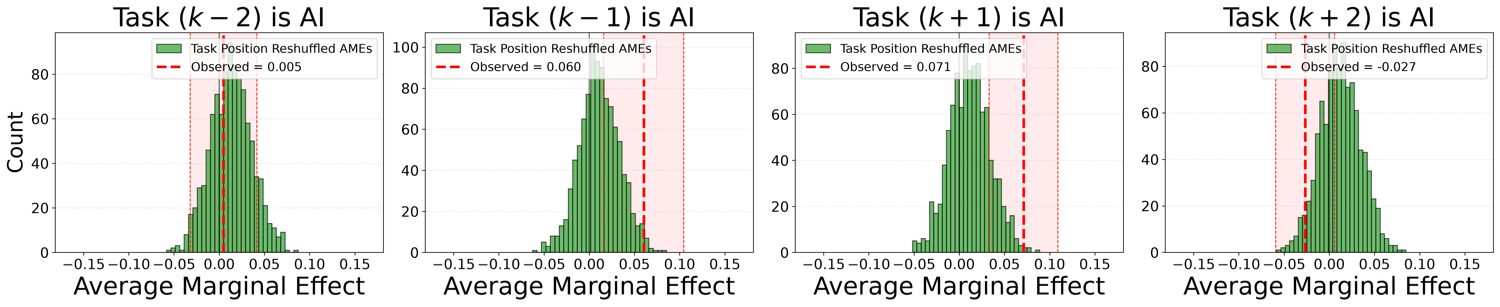}
  \end{subfigure}

  \vspace{1em}

  \begin{subfigure}[b]{\textwidth}
    
    \subcaption{Detailed Work Activity Fixed Effects}
    \includegraphics[width=\textwidth]{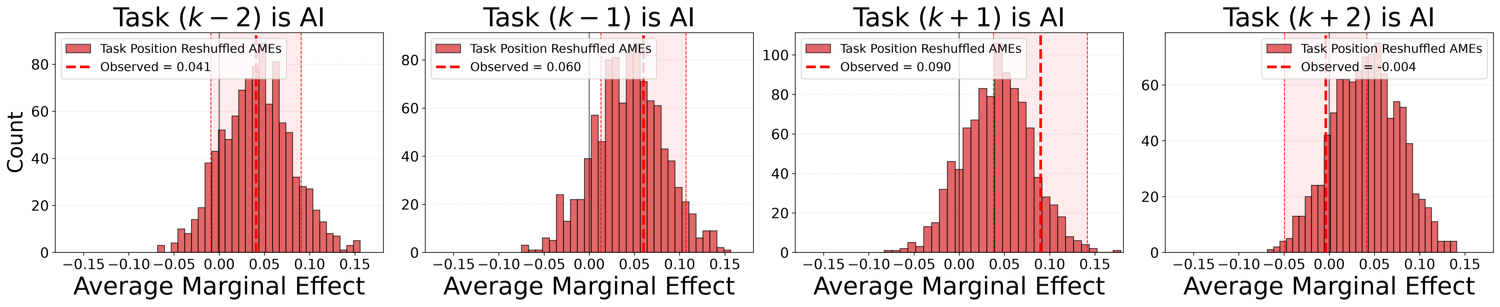}
  \end{subfigure}
  \par\vspace{\fignotesskip}
\begin{minipage}{1\linewidth}
\begin{spacing}{0.2}
{\fontsize{9.5pt}{10pt}\selectfont Notes: These graphs show that, among similar tasks appearing in multiple occupations, those sitting in a more AI-executed local context exhibit higher probabilities of being executed by AI, with the association concentrated on the steps immediately adjacent to the task.
The dashed line, shaded band, and histograms are constructed as in Figure~\ref{fig:DWA_regression_aiExecution_mainSample}, with the observed estimates taken from Table~\ref{tab:DWA_regression_aiExecution_GPTsample}.
Panel~(a) reports results from Equation~\eqref{eq:DWA_regression_ai}, while Panels~(b), (c), and (d) augment the baseline specification with SOC major group, SOC minor group, and DWA fixed effects, respectively.}
\end{spacing}
\end{minipage}
\end{figure}

\begin{figure}[!t]
  \centering
  \caption{Neighboring Tasks' AI Execution Status and a Task's AI Automation: Placebo Distributions} 
  \label{fig:DWA_regression_aiAutomation_mainSample}
  \begin{subfigure}[b]{\textwidth}
    
    \subcaption{No Fixed Effects}
    \includegraphics[width=\textwidth]{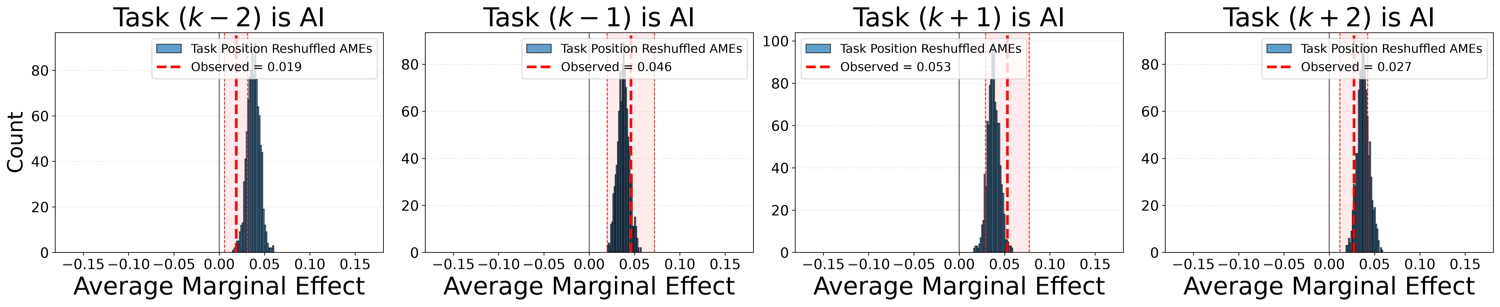}
  \end{subfigure}
  
  \vspace{0.5em}
  
  \begin{subfigure}[b]{\textwidth}
    
    \subcaption{SOC Major Group Fixed Effects}
    \includegraphics[width=\textwidth]{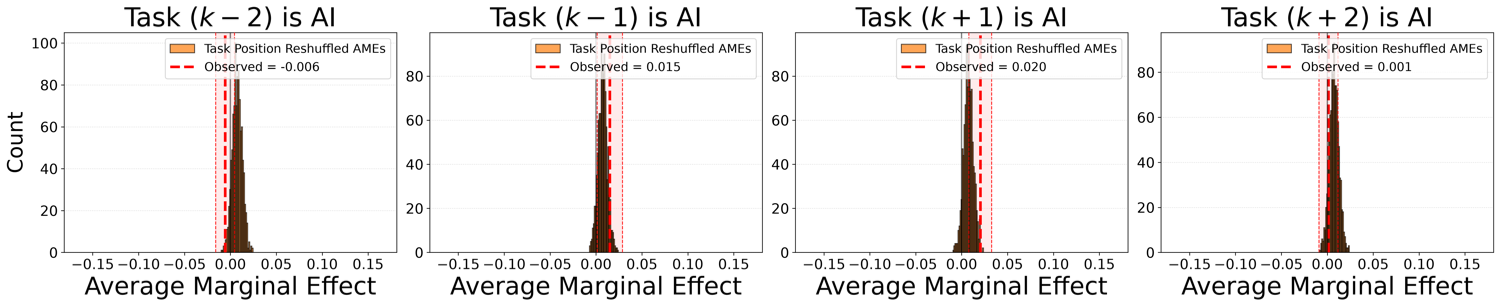}
  \end{subfigure}
  
  \vspace{0.5em}
  
  \begin{subfigure}[b]{\textwidth}
    
    \subcaption{SOC Minor Group Fixed Effects}
    \includegraphics[width=\textwidth]{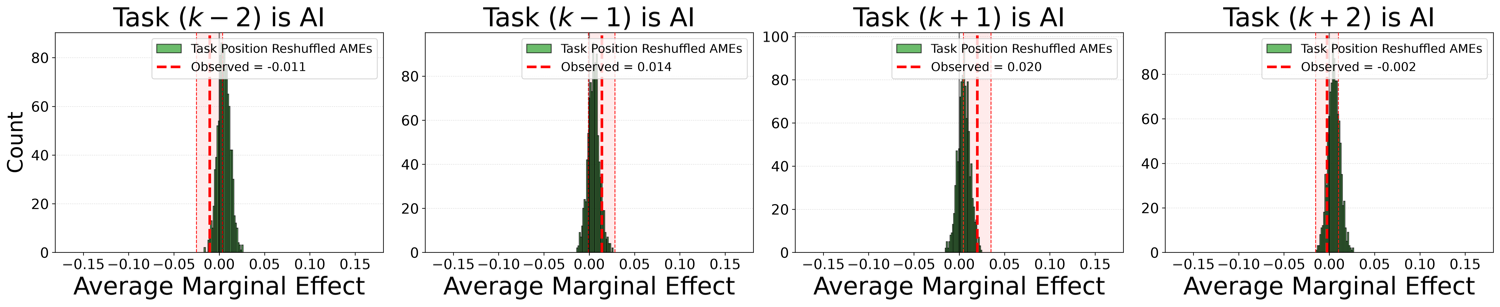}
  \end{subfigure}

  \vspace{0.5em}

  \begin{subfigure}[b]{\textwidth}
    
    \subcaption{Detailed Work Activity Fixed Effects}
    \includegraphics[width=\textwidth]{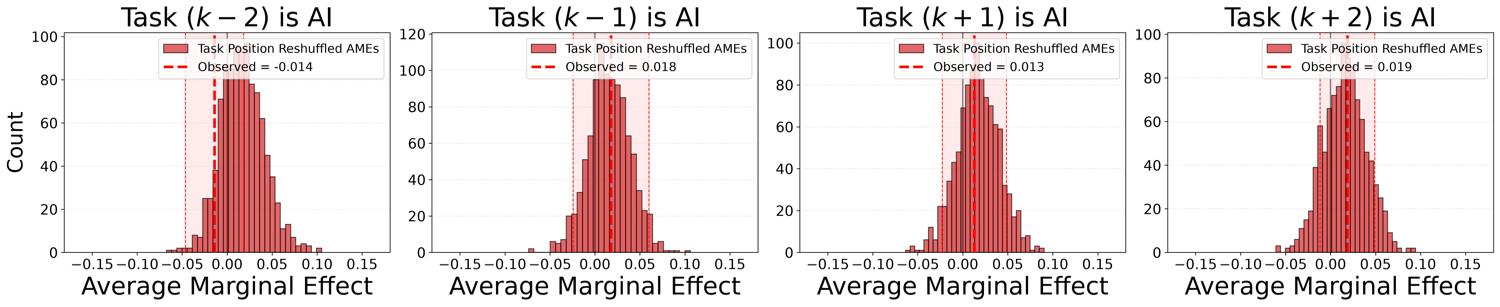}
  \end{subfigure}
  \par\vspace{\fignotesskip}
\begin{minipage}{1\linewidth}
\begin{spacing}{0.2}
{\fontsize{9.5pt}{10pt}\selectfont Notes: These graphs show that, among similar tasks appearing in multiple occupations, those sitting in a more AI-executed local context exhibit higher probabilities of being automated by AI, with the association concentrated on the steps immediately adjacent to the task.
The dashed line, shaded band, and histograms are constructed as in Figure~\ref{fig:DWA_regression_aiExecution_mainSample}, with the observed estimates taken from Table~\ref{tab:DWA_regression_aiAutomation_mainSample}.
Panel~(a) reports results from Equation~\eqref{eq:DWA_regression_ai} but with dependent variable ($\text{is\_automated}_{k}$) instead of ($\text{is\_ai}_{k}$), while Panels~(b), (c), and (d) augment the baseline specification with SOC major group, SOC minor group, and DWA fixed effects, respectively.}
\end{spacing}
\end{minipage}
\end{figure}

\begin{figure}[!t]
  \centering
  \caption{Neighboring Tasks' AI Execution Status and a Task's AI Automation (GPT-filtered Sample): Placebo Distributions} 
  \label{fig:DWA_regression_aiAutomation_GPTsample}
  \begin{subfigure}[b]{\textwidth}
    
    \subcaption{No Fixed Effects}
    \includegraphics[width=\textwidth]{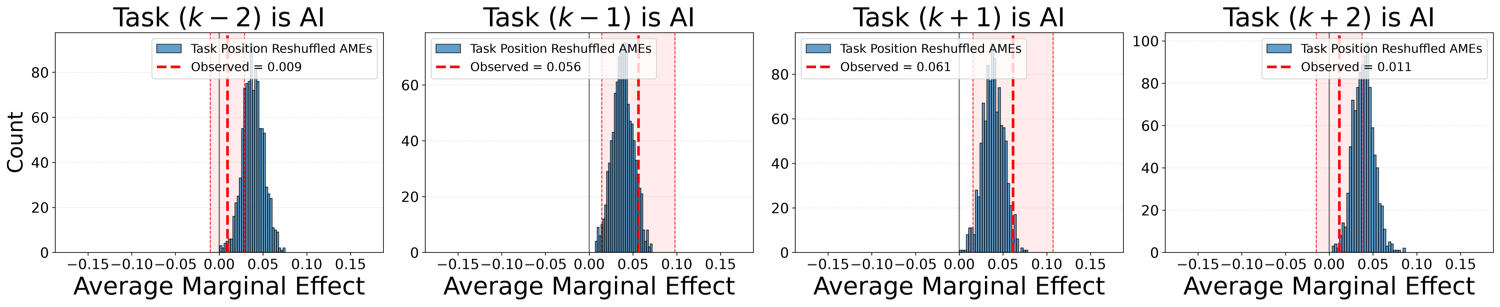}
  \end{subfigure}
  
  \vspace{0.5em}
  
  \begin{subfigure}[b]{\textwidth}
    
    \subcaption{SOC Major Group Fixed Effects}
    \includegraphics[width=\textwidth]{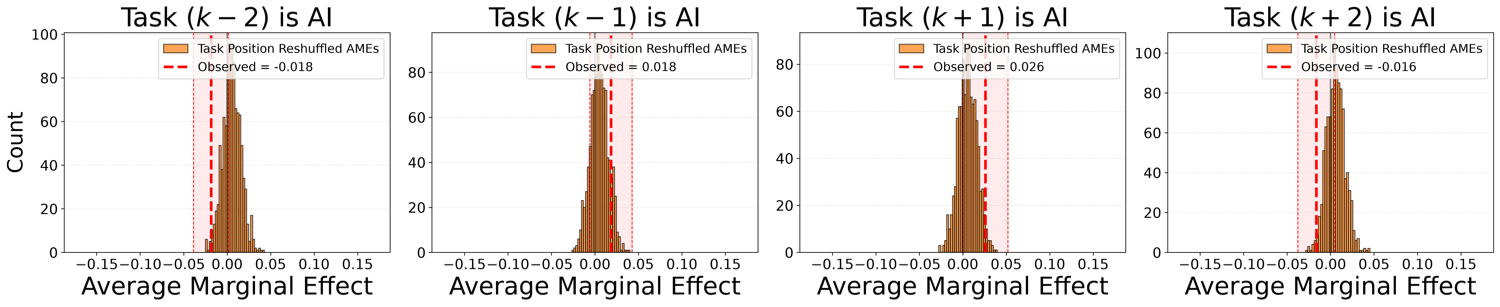}
  \end{subfigure}
  
  \vspace{0.5em}
  
  \begin{subfigure}[b]{\textwidth}
    
    \subcaption{SOC Minor Group Fixed Effects}
    \includegraphics[width=\textwidth]{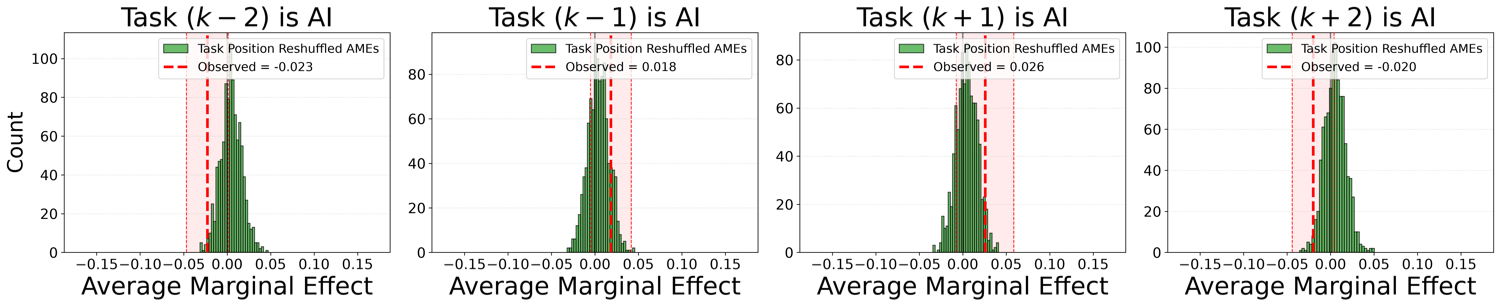}
  \end{subfigure}

  \vspace{0.5em}

  \begin{subfigure}[b]{\textwidth}
    
    \subcaption{Detailed Work Activity Fixed Effects}
    \includegraphics[width=\textwidth]{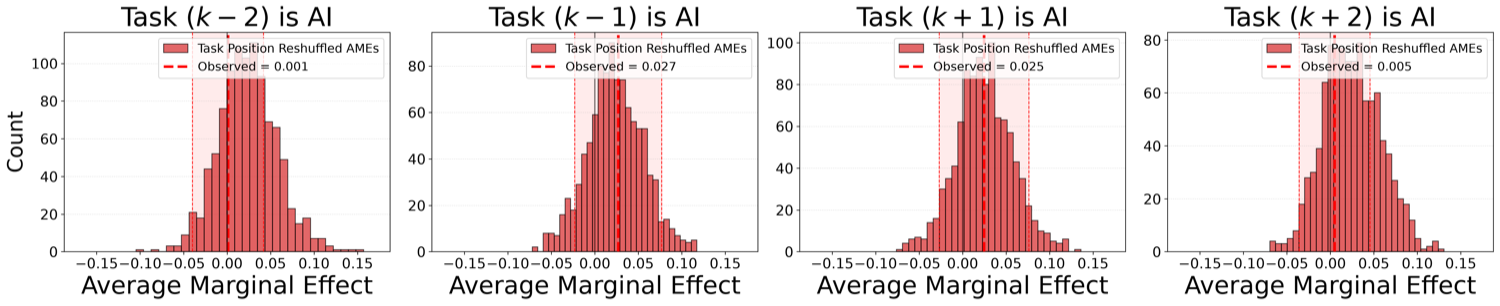}
  \end{subfigure}
  \par\vspace{\fignotesskip}
\begin{minipage}{1\linewidth}
\begin{spacing}{0.2}
{\fontsize{9.5pt}{10pt}\selectfont Notes: The pattern is the same as in Figure~\ref{fig:DWA_regression_aiAutomation_mainSample}, now on the GPT-filtered sample.
The dashed line, shaded band, and histograms are constructed as in Figure~\ref{fig:DWA_regression_aiExecution_mainSample}, with the observed estimates taken from Table~\ref{tab:DWA_regression_aiAutomation_GPTsample}.
Panel~(a) reports results from Equation~\eqref{eq:DWA_regression_ai} but with dependent variable ($\text{is\_automated}_{k}$) instead of ($\text{is\_ai}_{k}$), while Panels~(b), (c), and (d) augment the baseline specification with SOC major group, SOC minor group, and DWA fixed effects, respectively.}
\end{spacing}
\end{minipage}
\end{figure}

\begin{landscape}
% Full-page landscape table: cannot sit at the top of a portrait page, so it keeps
% a float-page specifier rather than the [!t] used for every other exhibit.
\begin{table}[p]
\caption{Example DWAs with Tasks in Multiple Occupations and the Tasks Surviving the ``Finding Similar Tasks'' Procedure}
\label{tab:DWA_tasks}

\resizebox{0.545\textheight}{!}{%
\begin{minipage}{\textheight}

\centering

\vspace{-0.1cm}

% \caption*, not \captionsetup{labelformat=empty} + \caption: the latter hides the
% "Table N:" label but still steps the table counter, so these two panel headings were
% consuming E.6 and E.7 and adding two captionless rows to the draft List of Tables.
\caption*{\hfill \Large{Panel (a)}}
\begin{tabular}{|l|p{15cm}|l|c|}
\hline
\textbf{O*NET DWA Title} & \textbf{O*NET Task Title} & \textbf{O*NET Occupation Title} & \textbf{Execution Label} \\
\hline

% ============================
% DWA 1
% ============================
\multirow{6}{*}{Analyze operational or research data.} &
Analyze or manipulate bioinformatics data using software packages, statistical applications, or data mining techniques. &
\multirow{2}{*}{Bioinformatics Technicians} & Augmentation \\ \cline{2-2}\cline{4-4}

& \textcolor{red}{Conduct quality analyses of data inputs and resulting analyses or predictions.} &
& \textcolor{red}{Augmentation} \\ \cline{2-4}

& Analyze research data to determine its significance, using computers. &
Astronomers & Augmentation \\ \cline{2-4}

& \textcolor{red}{Collect and analyze data, such as studying old records, tallying the number of outpatients entering each day or week, or participating in federal, state, or local population surveys as a Census Enumerator.} &
\textcolor{red}{Interviewers, Except Eligibility and Loan} & \textcolor{red}{Manual} \\ \cline{2-4}

& \textcolor{red}{Analyze data to determine answers to questions from customers or members of the public.} &
\textcolor{red}{Receptionists and Information Clerks} & \textcolor{red}{Augmentation} \\ \cline{2-4}

& Compute and analyze data, using statistical formulas and computers or calculators. &
Statistical Assistants & Augmentation \\ \hline

% ============================
% DWA 2
% ============================
\multirow{8}{*}{Prepare research or technical reports on environmental issues.} &
\textcolor{red}{Write reports or articles for Web sites or newsletters related to environmental engineering issues.} &
\textcolor{red}{Environmental Engineers} & \textcolor{red}{Manual} \\ \cline{2-4}

& Prepare technical and research reports, such as environmental impact reports, and communicate the results to individuals in industry, government, or the general public. &
Biologists & Manual \\ \cline{2-4}

& Prepare scientific atmospheric or climate reports, articles, or texts. &
Atmospheric and Space Scientists & Automation \\ \cline{2-4}

& Prepare charts or graphs from data samples, providing summary information on the environmental relevance of the data. &
Environmental Scientists and Specialists, Including Health & Augmentation \\ \cline{2-4}

& \textcolor{red}{Write reports or academic papers to communicate findings of climate-related studies.} &
\multirow{2}{*}{Climate Change Policy Analysts} & \textcolor{red}{Manual} \\ \cline{2-2}\cline{4-4}

& Prepare study reports, memoranda, briefs, testimonies, or other written materials to inform government or environmental groups on environmental issues, such as climate change. &
 & Manual \\ \cline{2-4}

& Prepare technical and research reports, such as environmental impact reports, and communicate the results to individuals in industry, government, or the general public. &
Industrial Ecologists & Manual \\ \cline{2-4}

& Produce environmental documents, such as environmental assessments or environmental impact statements. &
Transportation Planners & Manual \\ \hline

% ============================
% DWA 3
% ============================
\multirow{4}{*}{Document events or evidence using photographic or audiovisual equipment.} &
Take photographs and motion pictures for use in lectures and publications and to develop displays. &
Park Naturalists & Manual \\ \cline{2-4}

& Create data records for use in describing and analyzing social patterns and processes, using photography, videography, and audio recordings. &
Anthropologists and Archeologists & Automation \\ \cline{2-4}

& Create photographic recordings of information, using equipment. &
Geological Technicians, Except Hydrologic Technicians & Automation \\ \cline{2-4}

& \textcolor{red}{Use photographic or video equipment to document evidence or crime scenes.} &
\textcolor{red}{Forensic Science Technicians} & \textcolor{red}{Manual} \\ \hline

\end{tabular}

\vspace{0.3cm}

\caption*{\hfill \Large{Panel (b)}}
    \footnotesize
\begin{tabular}{|l|p{15cm}|l|c|}
\hline
\textbf{O*NET DWA Title} & \textbf{O*NET Task Title} & \textbf{O*NET Occupation Title} & \textbf{Execution Label} \\
\hline

% ============================
% DWA 1
% ============================
\multirow{3}{*}{Analyze operational or research data.} &
Analyze or manipulate bioinformatics data using software packages, statistical applications, or data mining techniques. &
Bioinformatics Technicians & Augmentation \\ \cline{2-4}

& Analyze research data to determine its significance, using computers. &
Astronomers & Augmentation \\ \cline{2-4}

& Compute and analyze data, using statistical formulas and computers or calculators. &
Statistical Assistants & Augmentation \\ \hline

% ============================
% DWA 2
% ============================
\multirow{6}{*}{Prepare research or technical reports on environmental issues.} & Prepare technical and research reports, such as environmental impact reports, and communicate the results to individuals in industry, government, or the general public. &
Biologists & Manual \\ \cline{2-4}

& Prepare scientific atmospheric or climate reports, articles, or texts. &
Atmospheric and Space Scientists & Automation \\ \cline{2-4}

& Prepare charts or graphs from data samples, providing summary information on the environmental relevance of the data. &
Environmental Scientists and Specialists, Including Health & Augmentation \\ \cline{2-4}

& Prepare study reports, memoranda, briefs, testimonies, or other written materials to inform government or environmental groups on environmental issues, such as climate change. &
Climate Change Policy Analysts & Manual \\ \cline{2-4}

& Prepare technical and research reports, such as environmental impact reports, and communicate the results to individuals in industry, government, or the general public. &
Industrial Ecologists & Manual \\ \cline{2-4}

& Produce environmental documents, such as environmental assessments or environmental impact statements. &
Transportation Planners & Manual \\ \hline

% ============================
% DWA 3
% ============================
\multirow{3}{*}{Document events or evidence using photographic or audiovisual equipment.} &
Take photographs and motion pictures for use in lectures and publications and to develop displays. &
Park Naturalists & Manual \\ \cline{2-4}

& Create data records for use in describing and analyzing social patterns and processes, using photography, videography, and audio recordings. &
Anthropologists and Archeologists & Automation \\ \cline{2-4}

& Create photographic recordings of information, using equipment. &
Geological Technicians, Except Hydrologic Technicians & Automation \\ \hline

\end{tabular}
\end{minipage}
}

\begin{minipage}{1\linewidth}
\begin{spacing}{0.2}
{\fontsize{9.5pt}{10pt}\selectfont Notes: Panel (a) shows three example detailed work activities (DWAs) and tasks associated with them in the O*NET dataset.
Entries highlighted with a red font are those that are dropped in the ``finding similar tasks'' procedure described in text.
Panel (b) shows the set of tasks for DWAs of Panel (a) that successfully survive the ``finding similar tasks'' procedure.}
\end{spacing}
\end{minipage}
\end{table}
\end{landscape}

\clearpage
\section{GPT-5-mini Prompts}
\label{app:prompts}

This Appendix provides the GPT-5-mini prompts used to generate a task sequence for each occupation (Prompt 1), and for finding similar tasks across occupations for each DWA (Prompt 2). \\

\begin{tcolorbox}[colback=gray!5, colframe=black!60, title={Prompt \#1: GPT-5-mini Prompt for Ordering of Tasks in Occupation}]
\footnotesize
\begin{verbatim}
You are an expert in workflow analysis for the occupation: {{ occupation }}. 
Below is a list of {{ num_tasks }} tasks that are part of this occupation:
{{ tasks_list }}

Provide the typical sequential order in which these tasks are performed in a 
real-world workflow. 

Return your answer as a JSON array where each element has:
- "Task Position": the sequence number (1, 2, 3, etc.).
- "Task Title": the exact task text from the list above.
Format: [{"Task Position": 1, "Task Title": "..."}, 
         {"Task Position": 2, "Task Title": "..."},
         ...]
Only return the JSON array, nothing else.
\end{verbatim}
\end{tcolorbox}

\newpage
\begin{tcolorbox}[colback=gray!5, colframe=black!60, title={Prompt \#2: GPT-5-mini Prompt for Finding Similar Tasks in Each DWA}]
\footnotesize
\begin{verbatim}
You are an expert in workflow analysis for the detailed work activity: 
{{ detailed_work_activity }}.

Below is a list of {{ num_tasks }} task IDs and titles that belong to this 
detailed work activity and appear across similar or different occupations 
(tasks and occupations are ordered such that the first task belongs to the 
first occupation, the second task belongs to the second occupation, etc.).
Tasks IDs: {{ tasks_ids }}
Tasks list: {{ tasks_list }}
Occupations list: {{ occupations_list }}
Occupation Codes list: {{ occupation_codes_list }}

Determine which tasks are similar in nature and in terms of their objectives, 
methods, or required skills. There may be more than one task associated with 
an occupation. Return only the most relevant task for every occupation. 
Only look for tasks that are actually similar. Do not feel obliged to return 
all occupations.

Return the task-occupation pairs you determine as similar as a JSON array 
where each element has:
- "Task ID": the exact task ID from the list of task IDs above
- "Task Title": the exact task text from the list of tasks above
- "O*NET-SOC Code": the exact occupation code text from the list of occupation
   codes above
- "Occupation Title": the exact occupation text from the list of occupations 
   above
Format: [{"Task ID": 1234, "Task Title": "...", 
          "O*NET-SOC Code": "...", "Occupation Title": "..."}, 
          {"Task ID": 5678, "Task Title": "...", 
          "O*NET-SOC Code": "...", "Occupation Title": "..."}, 
          ...]
Only return the JSON array, nothing else.
\end{verbatim}
\end{tcolorbox}

\clearpage
\section{Robustness to Alternative GPT Prompts}
\label{app:gptPrompts_robustness}

In this appendix, we assess whether our results are sensitive to the specific GPT prompt used to order tasks within occupations.  
Although alternative prompts can generate different task sequences, we show that these differences do not exhibit systematic patterns across subsets of occupations, and that our headline results are robust to prompt formulation.

We begin by generating task orderings for each occupation using 10 alternative prompts in addition to the baseline prompt.
We use the same structure as the first prompt in Appendix~\ref{app:prompts}, but only change the sentence starting with ``\textit{Provide the typical sequential...}'' with an alternative from the list below:
\begin{enumerate}
    \item \small \textit{Imagine a typical workday for this occupation. As the day unfolds, tasks arise and are completed as needed. Order the tasks in the sequence they most naturally occur.}
    \item \small \textit{For each task, consider its inputs and outputs. Order tasks so outputs of earlier tasks plausibly feed into later tasks. If tasks are parallel, place the more upstream task first.}
    \item \small \textit{Order tasks to minimize rework, waiting, and unnecessary handoffs. Assume an experienced worker executing the workflow efficiently.}
    \item \small \textit{Think about what must ultimately be produced in this occupation and what needs to happen before that. Use this reasoning to produce a natural forward sequence of tasks.}
    \item \small \textit{Identify which tasks logically depend on others, then order the tasks in a single sequence consistent with those dependencies and typical practice.}
    \item \small \textit{Order tasks according to how information is generated, transformed, and used over the course of the work.}
    \item \small \textit{Order tasks based on when mistakes would be most costly, placing tasks that prevent or constrain downstream errors earlier.}
    \item \small \textit{Order tasks so that tasks informing important decisions tend to occur before tasks that rely on those decisions.}
    \item \small \textit{Order the tasks as an experienced practitioner would intuitively carry them out, without explicitly planning or formalizing the workflow.}
    \item \small \textit{Order the tasks to reflect how the work is most commonly carried out in practice, rather than how it is formally described.}
\end{enumerate}
Notice that these prompts are not simple re-wordings of the same instruction.  
Instead, they approach the task ordering problem from different perspectives, so that, in principle, the resulting task sequences could differ.  
Our aim is to quantify how much variation these alternative prompt formulations induce in practice, and whether such variation poses a concern for our purposes.

To measure the degree of overlap between task orderings, we compute pairwise Kendall's $\tau$ within each occupation across all pairs of prompts.
Across all occupations and prompt pairs, the average Kendall's $\tau$ is $0.6$, with the full distribution shown in Panel (a) of Figure~\ref{fig:kendall_tau_dist}.
\begin{figure}[!t]
\centering
\caption{Kendall's $\tau$ Distribution Across GPT Prompt Task Orderings}
\label{fig:kendall_tau_dist}

\begin{subfigure}[t]{0.49\linewidth}
    \centering
    \subcaption{Full Sample}
    \includegraphics[width=\linewidth]{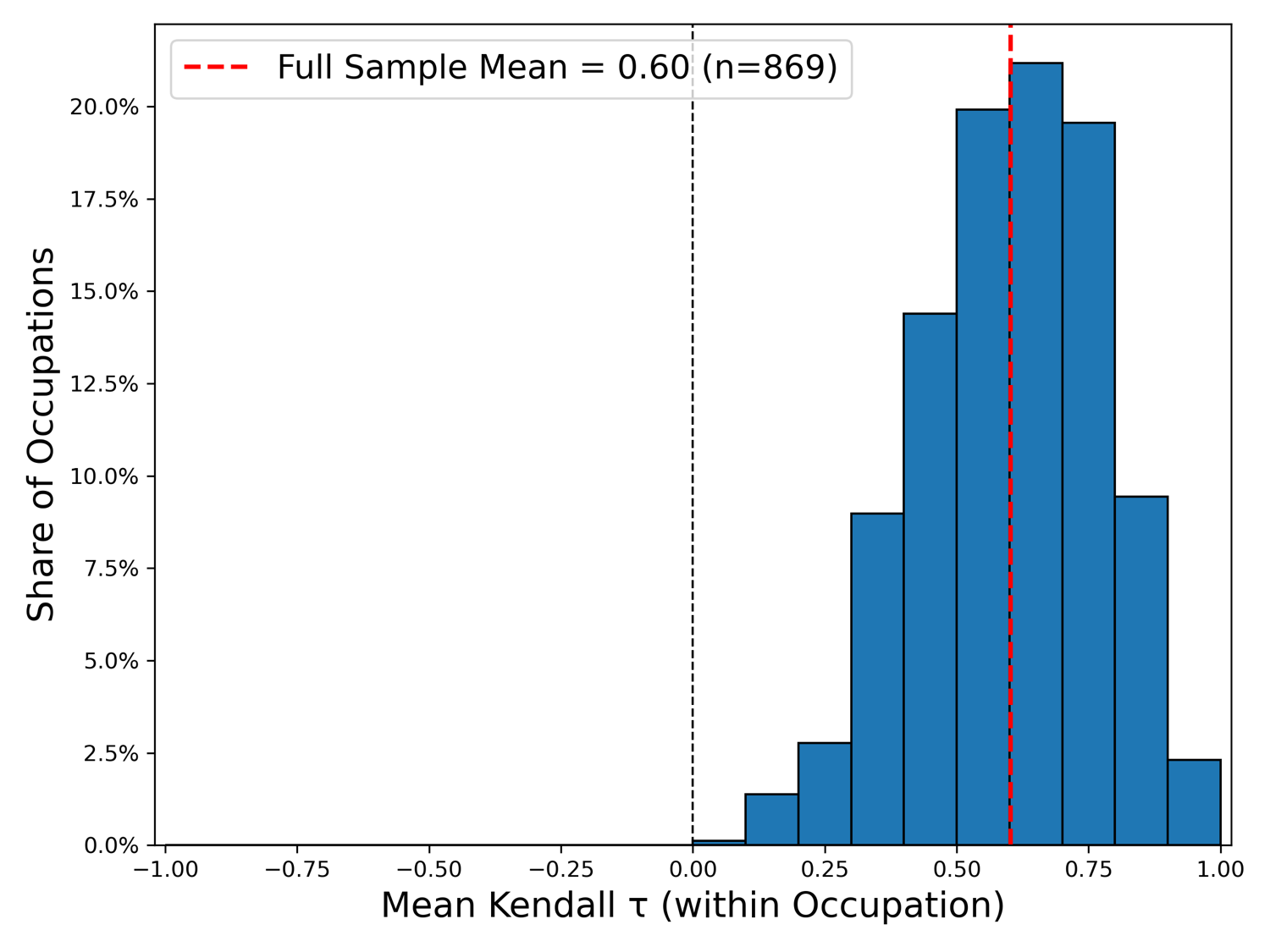}
\end{subfigure}
\hfill
\begin{subfigure}[t]{0.49\linewidth}
    \centering
    \subcaption{Empirical Fragmentation Index Split}
    \includegraphics[width=\linewidth]{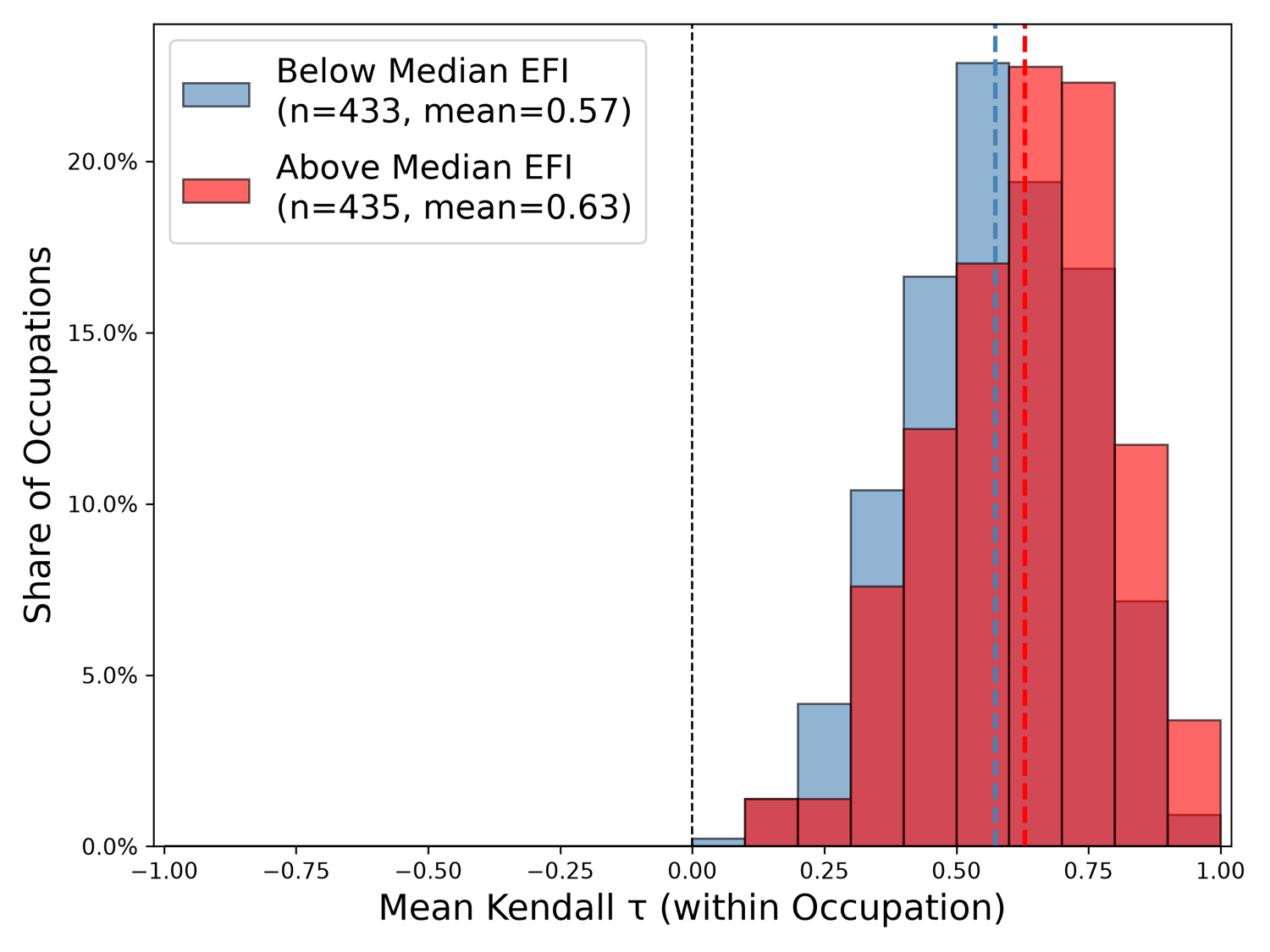}
\end{subfigure}

\vspace{0.3cm}

\begin{subfigure}[t]{0.49\linewidth}
    \centering
    \subcaption{AI Exposure (E1 or E2) Split}
    \includegraphics[width=\linewidth]{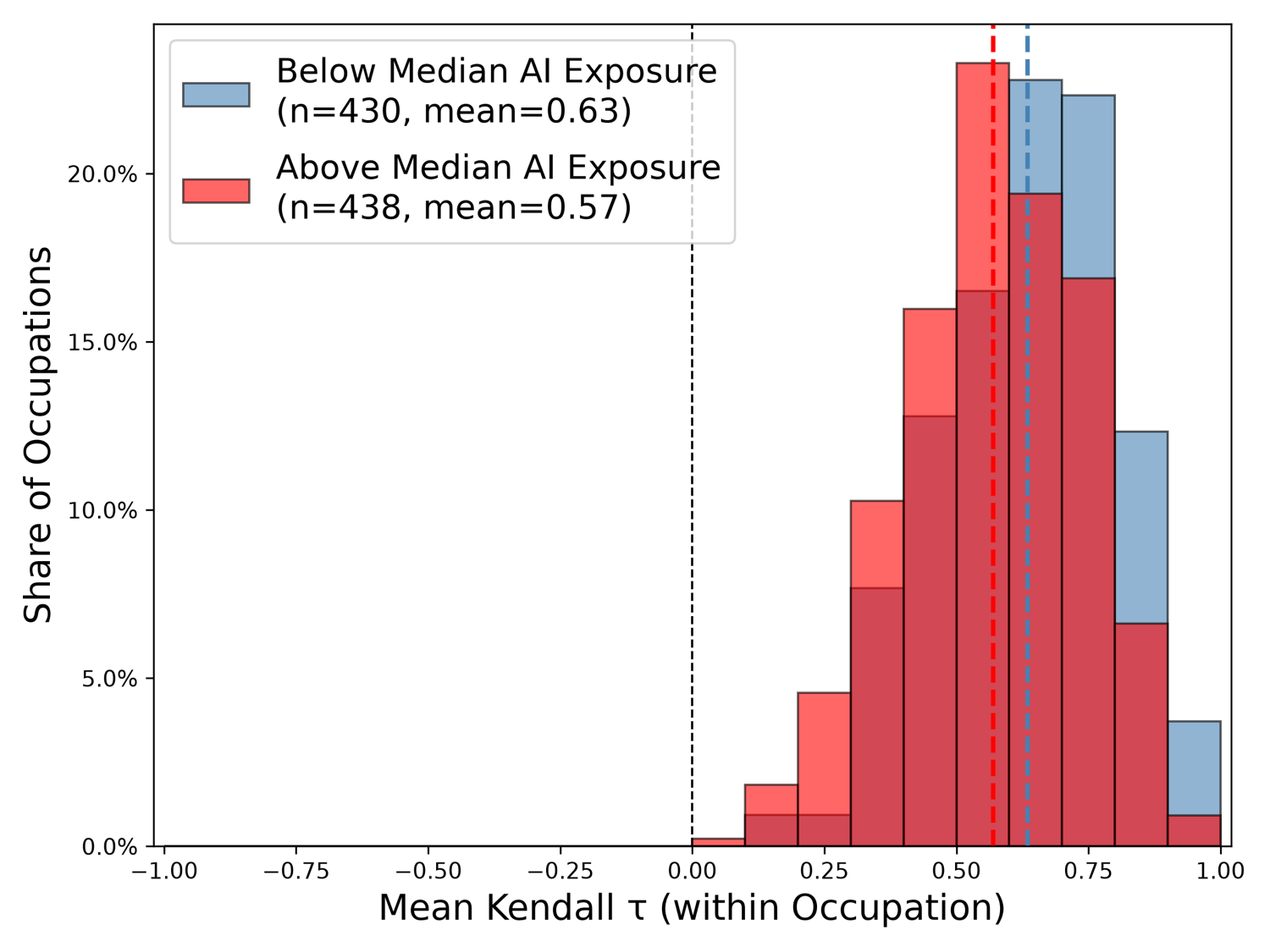}
\end{subfigure}
\hfill
\begin{subfigure}[t]{0.49\linewidth}
    \centering
    \subcaption{AI Execution Split}
    \includegraphics[width=\linewidth]{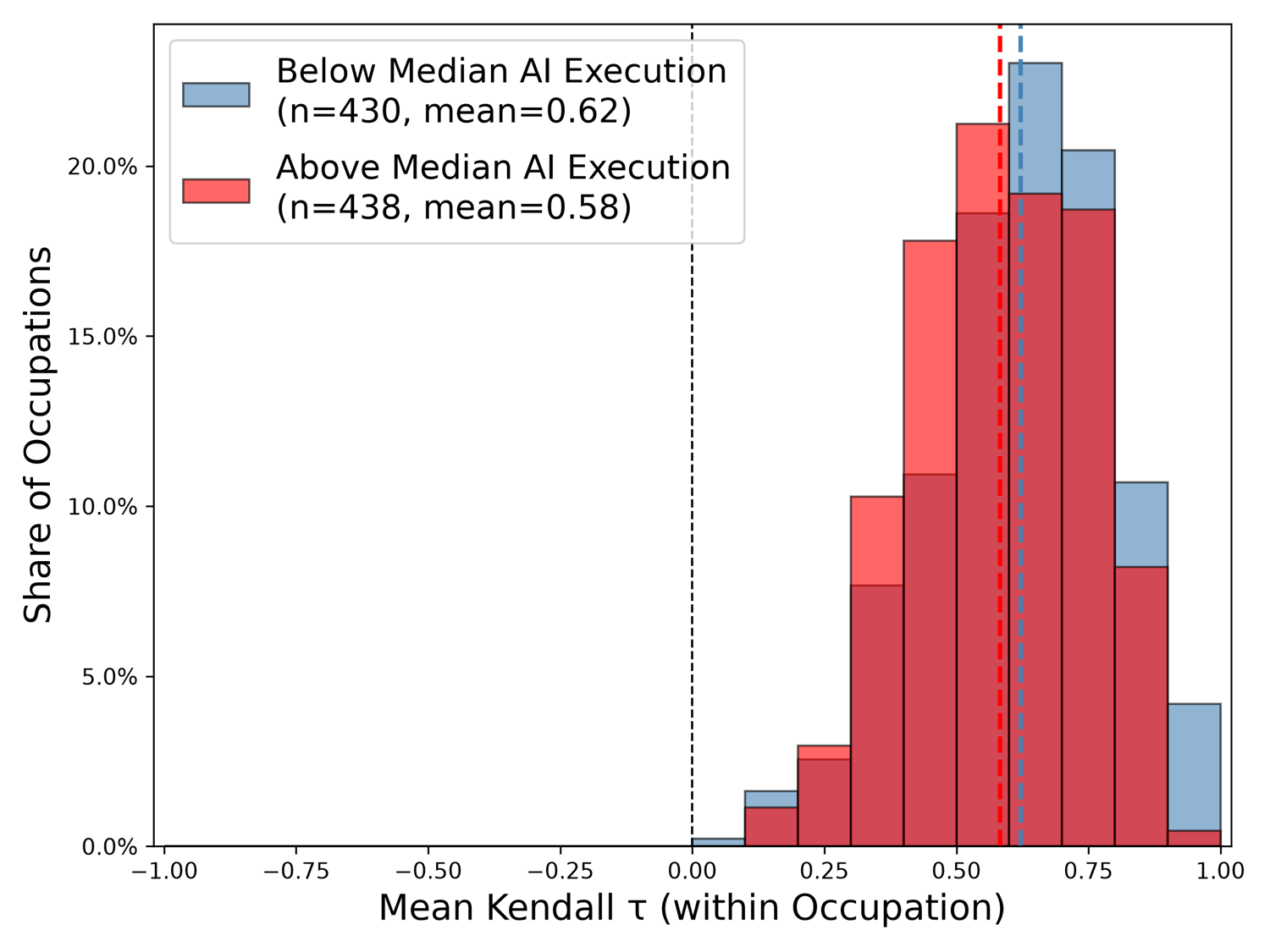}
\end{subfigure}

\par\vspace{\fignotesskip}
\begin{minipage}{1\linewidth}
\begin{spacing}{0.2}
{\fontsize{9.5pt}{10pt}\selectfont Notes: Each panel shows the distribution of mean Kendall's $\tau$ computed across 11 prompts (main + ten alternatives) within occupations.
Split figures in Panels (b)\textendash (d) compare occupations above (red) and below (blue) the median of the indicated measure, computed in the sample generated by the main prompt.}
\end{spacing}
\end{minipage}
\end{figure}
This value implies that task orderings are not identical across prompts, but that there is a reasonably high degree of overlap given the diversity of prompt formulations considered.
More importantly for our analyses, we find no evidence that GPT generates systematically different task orderings across different types of occupations.
Specifically, we split occupations based on whether they fall above or below the median in the \emph{main prompt sample} along three dimensions discussed in Subsection~\ref{sec:fragmentation_prediction}: the empirical fragmentation index, the share of occupation tasks exposed to AI (E1 or E2), and the share of tasks executed by AI.
Panels (b)\textendash(d) of Figure~\ref{fig:kendall_tau_dist} show that the mean Kendall's $\tau$ is very similar for above- and below-median occupations in all cases.

In the following Subsections, we show that we obtain the same results for each of our headline analyses using these alternative prompts.

\subsection{Robustness of Prediction \#1 Results}
\label{app:gptPrompts_robustness_pred1}

First, we examine how sensitive the average AI chain length statistic is to the choice of GPT prompt.
Figure~\ref{fig:ai_chains_length_count_prompt_robustness} reports the average AI chain length computed under each prompt.
Across prompts, the statistic remains tightly clustered around the baseline value implied by the main prompt.
In particular, the mean across all 11 prompts (the dashed colored line) matches the O*NET sample estimate from the main prompt, which lies in the extreme tail of the placebo distributions in Figure~\ref{fig:aiChains_graphs_def1}.
\begin{figure}[!t]
\caption{Robustness of AI Chain Length to GPT Prompts}
\label{fig:ai_chains_length_count_prompt_robustness}
\centering
\includegraphics[width=0.6\textwidth]{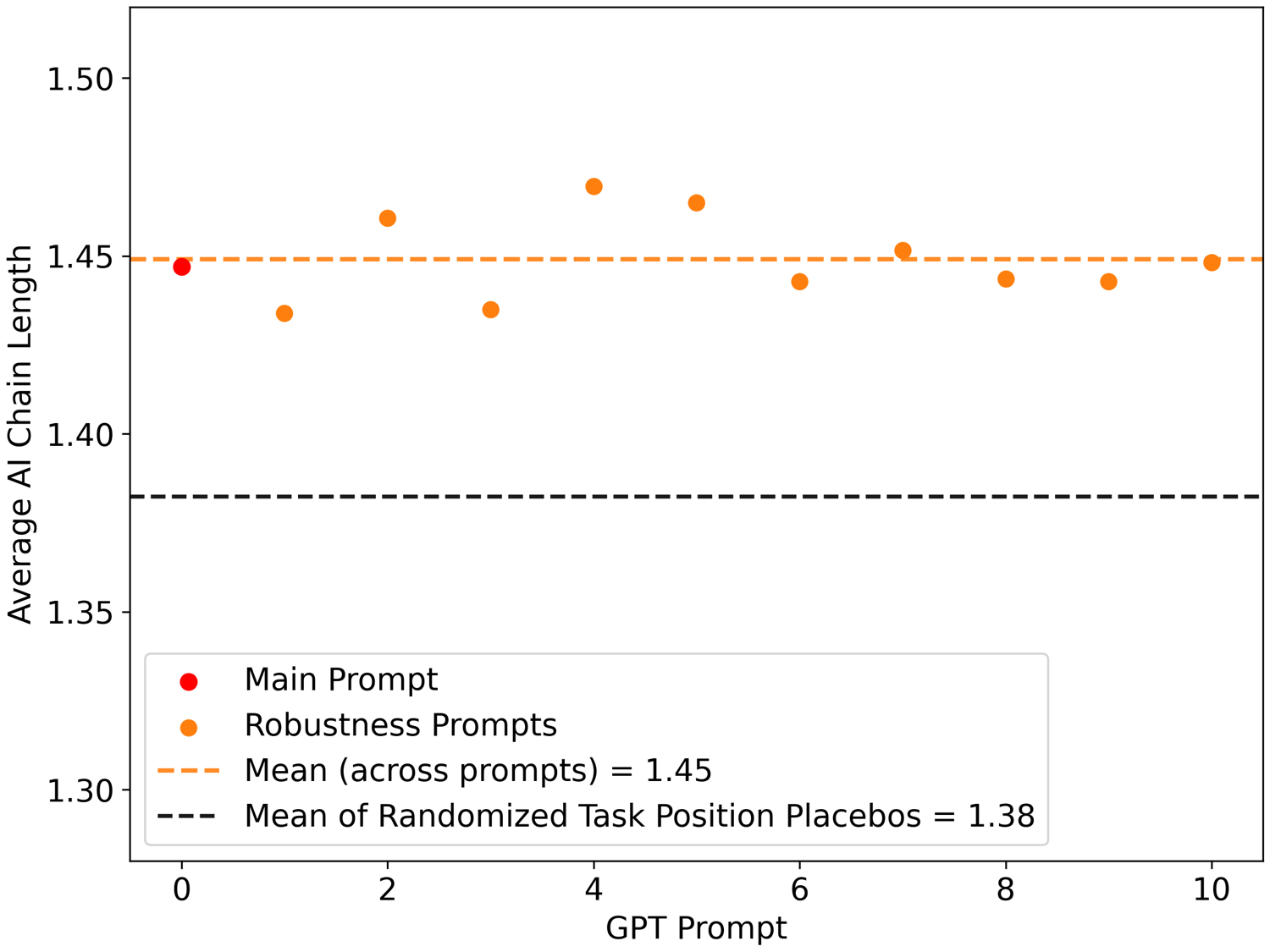}
\par\vspace{\fignotesskip}
\begin{minipage}{1\linewidth}
\begin{spacing}{0.2}
{\fontsize{9.5pt}{10pt}\selectfont Notes: This figure shows that the average AI chain length measure is robust to alternative prompt formulations.
The figure reports the average AI chain length computed separately under each prompt.
Prompt~0 corresponds to the baseline prompt used in the main text.
The dashed black reference line is the mean of the same statistic across 1{,}000 randomized task-position reshuffle placebos shown in Figure~\ref{fig:aiChains_graphs_def1}.
The dashed colored line reports the mean of the statistic across all 11 prompts.}
\end{spacing}
\end{minipage}
\end{figure}

\subsection{Robustness of Prediction \#2 Results}
\label{app:gptPrompts_robustness_pred3}

Second, we assess the robustness of the association between neighboring tasks' AI execution and the focal task's likelihood of AI execution.

For each alternative prompt, we re-estimate Equation~\eqref{eq:DWA_regression_ai} on the corresponding O*NET sample under four specifications: baseline, SOC major group fixed effects, SOC minor group fixed effects, and DWA fixed effects.
Figure~\ref{fig:DWA_regression_aiExecution_mainSample_robustness} reports the equivalent of average marginal effects reported in columns (1)\textendash(4) of Table~\ref{tab:DWA_regression_aiExecution_mainSample} for all prompts.
The layout and color coding of specifications in this figure mirror those in Figure~\ref{fig:DWA_regression_aiExecution_mainSample} from the reshuffled task positions robustness exercise.

We find that even though the average effect of steps two positions away is slightly larger across alternative prompts than under the main prompt, the substantive conclusion remains unchanged.
The two adjacent steps (middle columns) still have larger and more precisely estimated effects than the steps two positions away (side columns), and the latter attenuate and become less distinguishable from zero as additional fixed effects are included in Panels~(b)\textendash (d).

\subsection{Robustness of Prediction \#3 Results}
\label{app:gptPrompts_robustness_pred2}

Finally, we re-estimate Equation~\eqref{eq:fragmentation_index_regression} using task orderings generated by the alternative prompts and report the analogs of the estimates in Table~\ref{tab:fragmentation_index_regression_exposure}.

Figure~\ref{fig:efi_regression_exposure_robustness} plots the equivalents of estimated coefficients reported in Table~\ref{tab:fragmentation_index_regression_exposure} for the main variables across the three specifications (baseline, SOC major group fixed effects, and SOC minor group fixed effects) for the exposure-based empirical fragmentation index.
Across all specifications the estimates obtained under alternative prompts fall inside the confidence interval of the main prompt estimate.
The estimates on occupation-level AI exposure are positive and statistically different from zero throughout.
The EFI coefficients keep the predicted negative sign in 30 of the 33 cells, but none of them is significant at the 5\% level and only two clear the 10\% level, mirroring the null reported in Table~\ref{tab:fragmentation_index_regression_exposure}.
The dispersion of the EFI coefficient across prompts is small next to the sampling noise in any one estimate, so what the prompts deliver is a stable null rather than a fragile one.
\begin{figure}[!t]
  \centering
  \caption{Robustness of Table~\ref{tab:fragmentation_index_regression_exposure} Estimates to GPT Prompts} 
  \label{fig:efi_regression_exposure_robustness}
  
  \includegraphics[width=\textwidth]{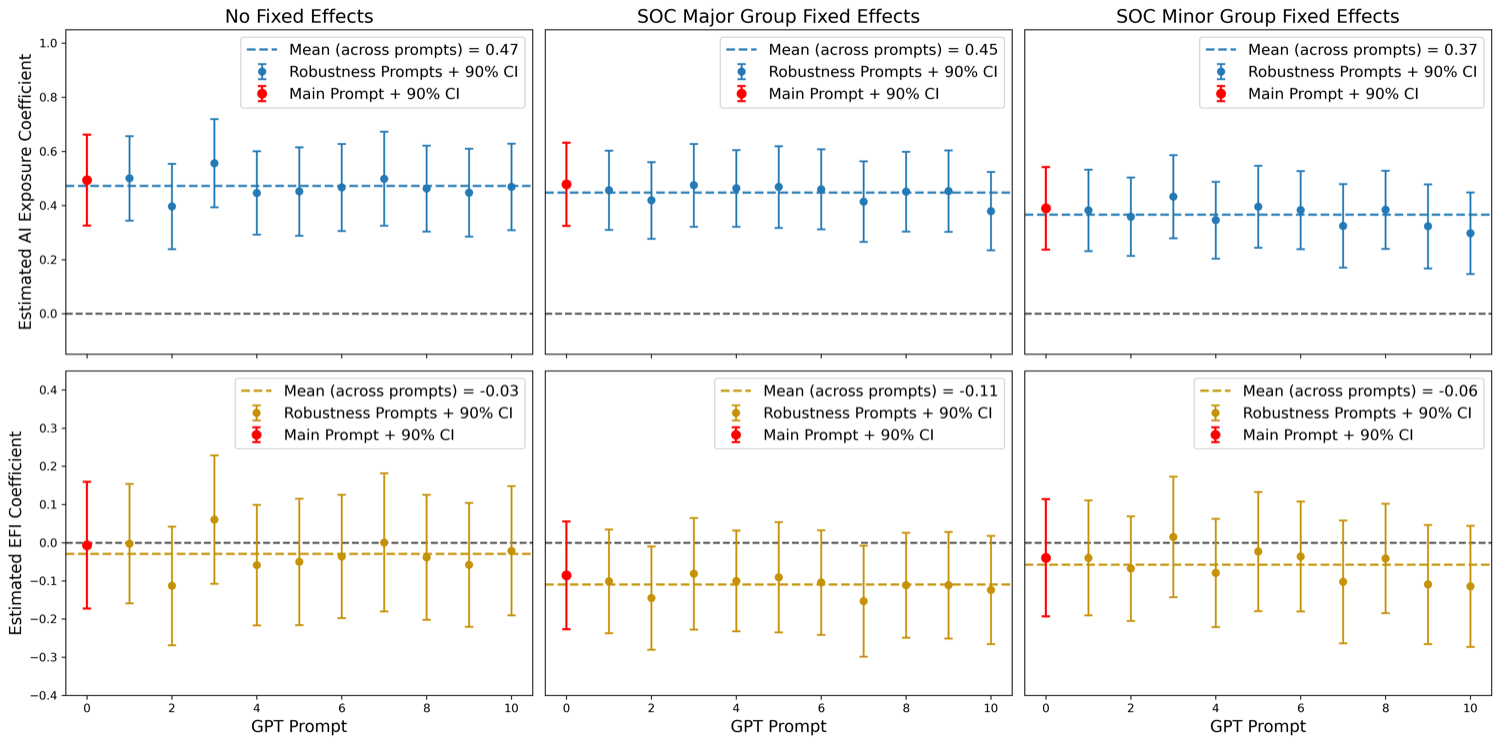}
  \par\vspace{\fignotesskip}
\begin{minipage}{1\linewidth}
\begin{spacing}{0.2}
{\fontsize{9.5pt}{10pt}\selectfont Notes: This figure shows that the estimated relationships between occupation-level AI execution and both AI exposure and the empirical fragmentation index are robust to alternative prompt formulations.
The graphs report coefficients on AI exposure (top row, blue) and EFI (bottom row, gold) for the exposure-based EFI.
Error bars denote 90\% confidence intervals.
The dashed colored lines report the mean across all 11 prompts.
In each graph, the Prompt~0 estimate is highlighted in red and corresponds to the estimate reported in Table~\ref{tab:fragmentation_index_regression_exposure} using the main prompt.
The graphs correspond to columns (1)\textendash(3) of Table~\ref{tab:fragmentation_index_regression_exposure}.}
\end{spacing}
\end{minipage}
\end{figure}

\newpage
\begin{figure}[!t]
  \centering
  \caption{Robustness of Table~\ref{tab:DWA_regression_aiExecution_mainSample} Estimates to GPT Prompts} 
  \label{fig:DWA_regression_aiExecution_mainSample_robustness}
  \begin{subfigure}[b]{\textwidth}
    
    \subcaption{No Fixed Effects}
    \includegraphics[width=\textwidth]{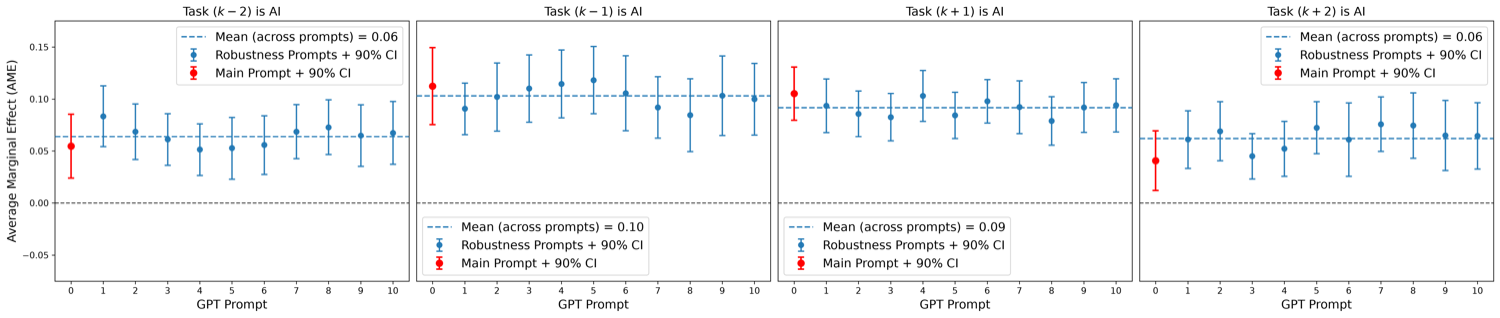}
  \end{subfigure}
  
  \vspace{1em}
  
  \begin{subfigure}[b]{\textwidth}
    
    \subcaption{SOC Major Group Fixed Effects}
    \includegraphics[width=\textwidth]{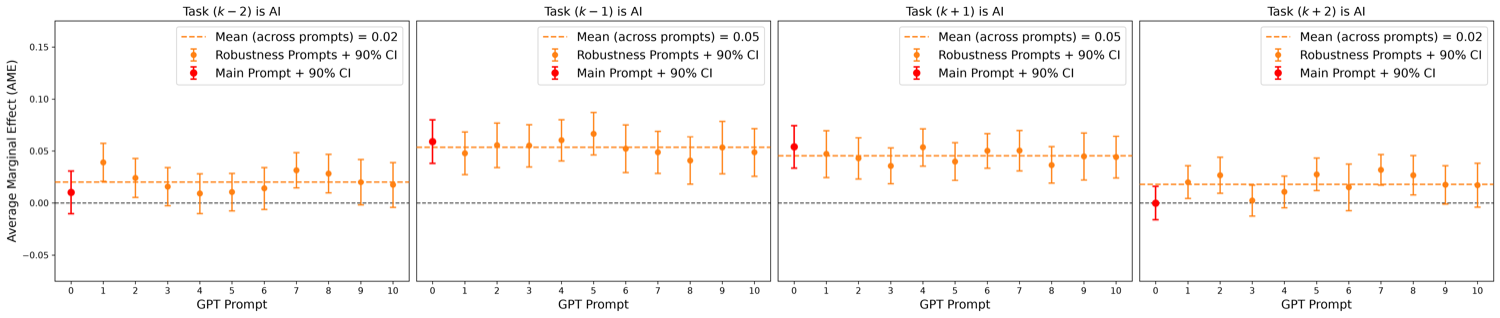}
  \end{subfigure}
  
  \vspace{1em}
  
  \begin{subfigure}[b]{\textwidth}
    
    \subcaption{SOC Minor Group Fixed Effects}
    \includegraphics[width=\textwidth]{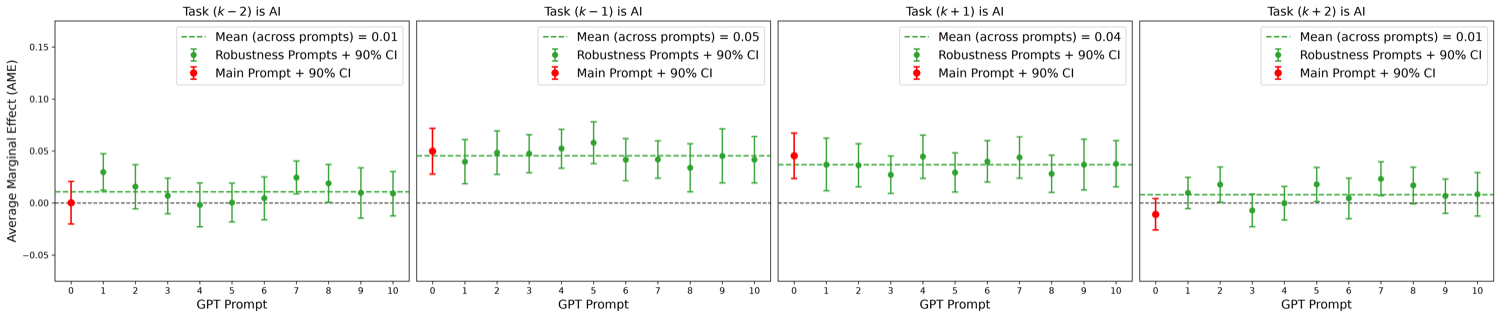}
  \end{subfigure}

  \vspace{1em}

  \begin{subfigure}[b]{\textwidth}
    
    \subcaption{Detailed Work Activity Fixed Effects}
    \includegraphics[width=\textwidth]{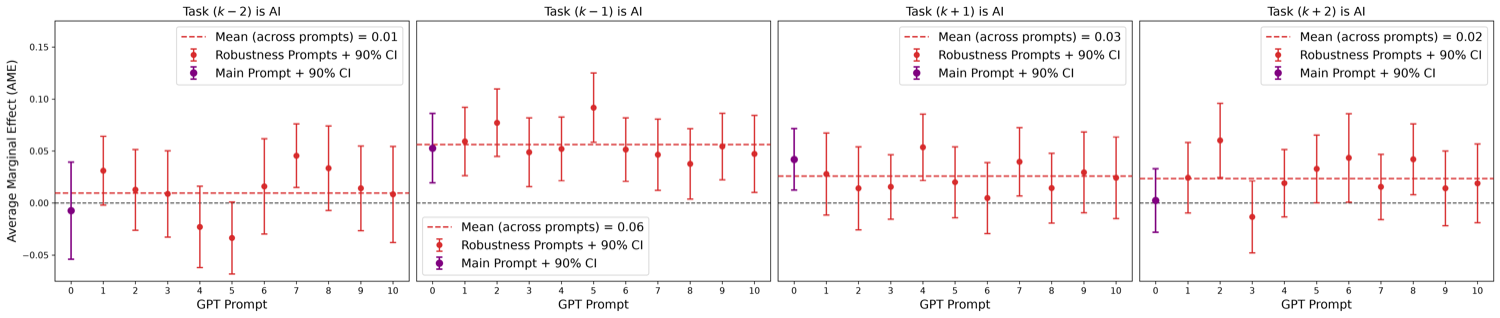}
  \end{subfigure}
  \par\vspace{\fignotesskip}
\begin{minipage}{1\linewidth}
\begin{spacing}{0.2}
{\fontsize{9.5pt}{10pt}\selectfont Notes: This figure shows that the association between neighboring tasks' AI execution and the focal task's AI execution is robust to alternative prompt formulations.
Across prompts, the two adjacent steps (middle columns) have larger and more precisely estimated effects than the steps two positions away (side columns).
The effects of steps two positions away attenuate toward zero and become less distinguishable from zero as additional fixed effects are included.
Panels~(a)\textendash(d) report the prompt-specific analogs of the estimates in columns (1)\textendash(4) of Table~\ref{tab:DWA_regression_aiExecution_mainSample}.
Error bars denote 90\% confidence intervals.
The dashed colored lines report the mean across all 11 prompts.
In each graph, the Prompt~0 estimate corresponds to the estimate reported in Table~\ref{tab:DWA_regression_aiExecution_mainSample} using the main prompt.}
\end{spacing}
\end{minipage}
\end{figure}

\clearpage
\section{Robustness to Frequently-Executed Tasks Sample Restriction}
\label{app:frequency_robustness}

The task sequences underlying our empirical analysis combine, for each occupation, the full set of O*NET tasks regardless of how often a task is actually performed on the job.
This raises the concern that rarely-executed tasks, which a worker might perform only rarely, lengthen the workflow and insert ``filler'' steps between the tasks that structure day-to-day production.
If they do, the average length of the AI chains we document in Prediction~\#1, the adjacency between AI-executed steps we exploit in Prediction~\#2, and the dispersion of AI-exposed steps we measure in Prediction~\#3 may partly reflect infrequent tasks rather than the operative structure of work.
This appendix re-examines all three predictions on samples restricted to the tasks that workers report performing frequently.
The chain-length and neighbor results keep their sign under pruning, though they are estimated less precisely as the samples shrink, and the fragmentation null of Subsection~\ref{sec:fragmentation_prediction} is likewise unchanged.

Rather than re-querying the language model to re-sequence each occupation's surviving tasks, we prune the workflows already generated for the main analysis.
Within each occupation we drop the tasks that fail the frequency filter (discussed below), keep the surviving tasks in their original order, and recompute the average AI chain length, the empirical fragmentation index, and the neighbor-execution flags on the pruned sequence before re-estimating the corresponding statistics. 
Pruning the existing sequences rather than re-prompting for a fresh ordering isolates the question of interest, namely whether infrequent tasks contaminate our measures, while holding the LLM's sequencing logic fixed.

To keep the sample definition consistent across the three predictions, we restrict each cut to the occupations whose frequency-pruned workflow retains at least five tasks, and we evaluate all three predictions on this common set of occupations.
The five-task floor is applied to the pruned workflow, that is, to the steps that survive the frequency filter.
For Predictions~\#1 and~\#3 the chain length and fragmentation index are computed over each occupation's full pruned workflow; for Prediction~\#2 the neighbor regression uses the eligible tasks of the same occupations, namely those with two neighbors on either side, which a five-task workflow always admits, and among these only the tasks whose DWA appears in more than one occupation, as in Section~\ref{sec:DWA_prediction}, so fewer occupations contribute to the neighbor estimates than to the other two.\footnote{For comparability across cuts, the fragmentation regressions z-score all three variables (the AI-execution share, the empirical fragmentation index, and AI exposure) within each cut, so the reported coefficients are standardized betas, where a one-standard-deviation change in fragmentation maps to a $\beta$-standard-deviation change in the AI-execution share. For the neighbor logits we report average marginal effects with analytic robust standard errors clustered at the DWA level, rather than the bootstrap of Table~\ref{tab:DWA_regression_aiExecution_mainSample}. Point estimates for the ``all tasks'' specification of Predictions~\#1 and~\#2 are identical to those reported in the main text. For Prediction~\#3, the five-task floor drops one of the $872$ occupations, which moves the SOC minor-group estimate from $-0.04$ to $-0.05$.}

\subsection{Constructing the Frequent-Task Samples}
\label{app:frequency_robustness_construction}

We measure how often a task is performed using the O*NET \emph{Frequency of Task} (FT) module.
O*NET asks the incumbent workers it surveys in each occupation how often they perform each of the occupation's tasks, offering seven answer categories ordered from least to most frequent: \textit{``yearly or less,'' ``more than yearly,'' ``more than monthly,'' ``more than weekly,'' ``daily,'' ``several times daily,''} and \textit{``hourly or more,''} and it reports the percentage of respondents choosing each category.
A task thus comes not with a single frequency but with a distribution of workers' responses over the seven categories, so restricting the sample to frequently-performed tasks requires two choices: (1) which categories count as ``frequent,'' and (2) how large the share of workers reporting the task in those categories must be.
We keep a task when the combined share of respondents in the designated frequent categories meets a threshold.

We consider three increasingly demanding notions of ``frequent.''
\emph{Daily$+$} counts tasks performed daily, several times daily, or hourly or more; \emph{SeveralDaily$+$} counts those performed several times daily or hourly or more; and \emph{Hourly$+$} counts only those performed hourly or more.
Each notion is evaluated at four thresholds: 20\%, 35\%, 50\%, and 65\% of incumbent worker responses.
A task is retained under the ``Hourly$+$ $\geq 20\%$'' cut, for example, if at least one in five workers reports performing it hourly or more.
To see how the cuts relate to one another, consider a task that 60\% of surveyed workers report performing several times daily, with the remaining responses spread over less frequent categories.
The task survives the SeveralDaily$+$ cuts at the 20\%, 35\%, and 50\% thresholds, and because responses in a stricter category also count toward every looser logic, it survives the corresponding Daily$+$ cuts as well.
Nothing stricter retains it, in either direction. No worker reports the task hourly or more, so it fails every Hourly$+$ cut, and its frequent share falls short of 65\%, so it fails the strictest threshold under both remaining logics.
Each logic-threshold pair therefore defines its own subsample of tasks and occupations.
We estimate every prediction separately on each subsample and present the results for all twelve cuts together in a grid, three logics by four thresholds, for ease of comparison.

Figure~\ref{fig:frequency_pruning_example} makes the pruning procedure concrete for a single occupation, Allergists and Immunologists (SOC 29-1229.01).
\begin{figure}[!t]
\caption{Frequency Pruning of a Single Occupation's Workflow: Allergists and Immunologists}
\label{fig:frequency_pruning_example}
\centering
\begin{subfigure}[b]{\textwidth}
    \centering
    \subcaption{Full Workflow (All Tasks)}
    \includegraphics[width=0.86\textwidth]{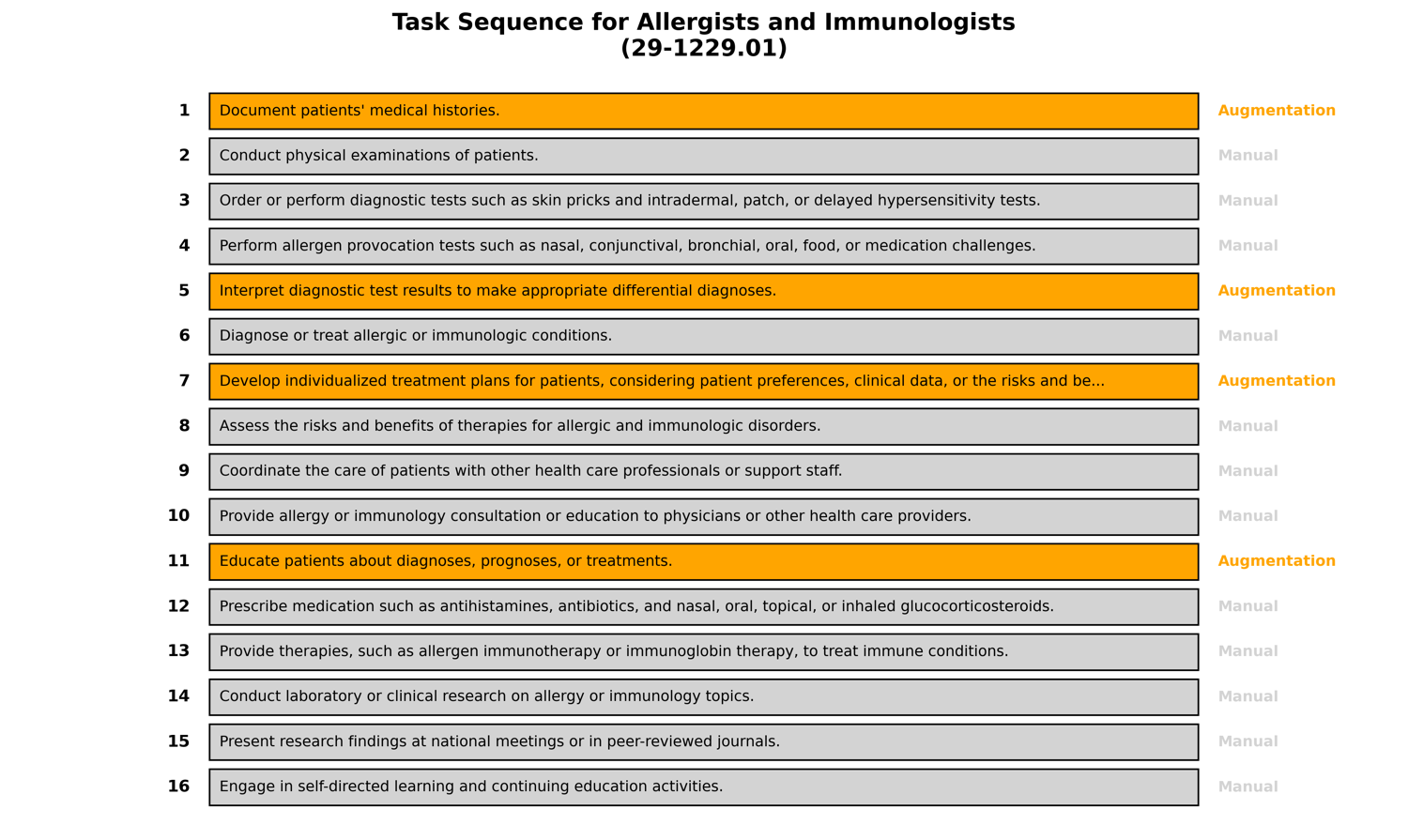}
\end{subfigure}

\vspace{0.6em}

\begin{subfigure}[b]{\textwidth}
    \centering
    \subcaption{Pruned Workflow (Hourly$+$ $\geq 50\%$)}
    \includegraphics[width=0.86\textwidth]{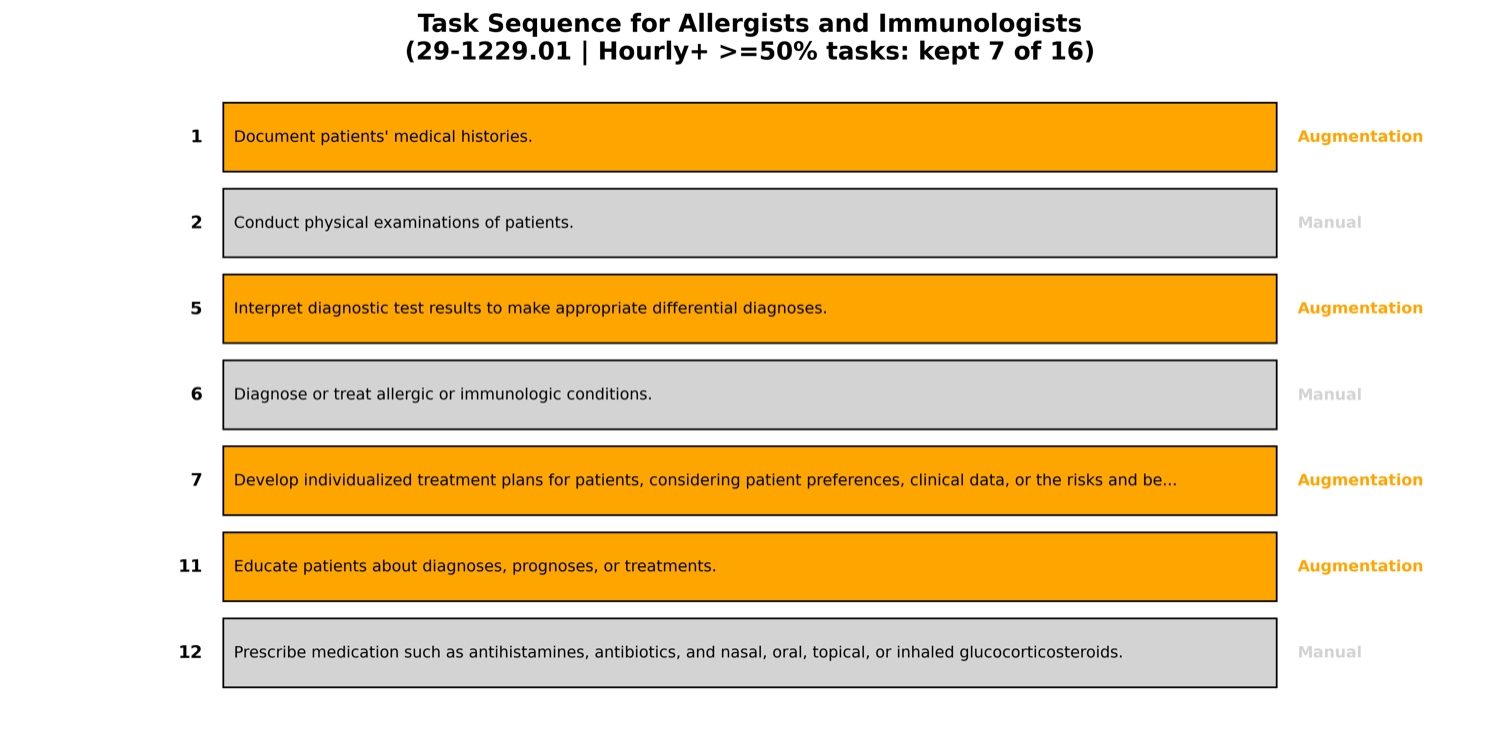}
\end{subfigure}
\par\vspace{\fignotesskip}
\begin{minipage}{1\linewidth}
\begin{spacing}{0.2}
{\fontsize{9.5pt}{10pt}\selectfont Notes: Panel~(a) reports the full sequenced workflow for Allergists and Immunologists (SOC 29-1229.01), the object used in the main analysis: the occupation's sixteen O*NET tasks in the order assigned by the language model.
Panel~(b) reports the workflow that survives the Hourly$+$ $\geq 50\%$ frequency cut, which retains a task only if at least half of the surveyed incumbent workers report performing it hourly or more.
Surviving tasks keep their original relative order and their position in the full workflow, so the numbers skipped in Panel~(b) are the pruned steps.
Shaded rows are AI-executed steps (augmentation or automation in the Anthropic Economic Index); unshaded rows are manual.
Pruning keeps all four AI-executed steps but drops the manual steps separating positions~7 and~11, merging two isolated AI-executed steps into a chain of length two.}
\end{spacing}
\end{minipage}
\end{figure}
Panel~(a) shows the occupation's full sequenced workflow, the object used throughout the main analysis: sixteen O*NET tasks in the order the language model places them, with the AI-executed steps shaded.
Panel~(b) shows what survives the Hourly$+$ $\geq 50\%$ cut, one of the most demanding in the grid: seven of the sixteen tasks, retained in their original relative order and still carrying their original position numbers so that the dropped steps can be read off directly.
What the filter removes is what one would expect it to remove.
The specialized diagnostic procedures (steps~3 and~4), the coordination and consultation activities (steps~9 and~10), and the research and continuing-education tasks (steps~14 through~16) are performed episodically rather than within the hour, and they drop out; what remains is the recurring clinical loop of documenting histories, examining patients, interpreting results, planning treatment, educating patients, and prescribing.
The figure also shows why pruning is not innocuous for our measures, which is precisely the concern this appendix addresses.
All four AI-executed steps survive the filter, but in the full workflow each is isolated between manual steps, so the occupation contributes four chains of length one and an average AI chain length of $1.00$; once the intervening manual steps~8 through~10 are dropped, steps~7 and~11 become adjacent and the same four AI-executed tasks now form three chains of lengths $1$, $1$, and $2$, raising the occupation's average to $1.33$.
Pruning can thus both shorten workflows and concentrate the surviving AI-executed steps, and the exercises that follow ask whether our three results survive once this is done for every occupation.
We take the three predictions in turn.

\subsection{Robustness of Prediction \#1 Results}
\label{app:frequency_robustness_chains}

Prediction~\#1 concerns the formation of AI chains.
We measure it, as in Section~\ref{sec:chainLength_prediction}, by the average length of a maximal run of contiguous AI-executed steps, where a step is AI-executed if it is labeled augmented or automated, pooled across all chains in the sample.
A longer-than-random average chain could, however, be delivered mechanically by the composition of each pruned sample.
We therefore benchmark every cut against a within-occupation task position-reshuffle placebo, where we redraw $1{,}000$ reshuffles of task positions within occupations, re-apply the frequency filter, and recompute the statistic, producing a null distribution that holds the composition of the sample fixed while destroying the workflow ordering.
Figure~\ref{fig:chainlength_frequency} reports the observed average against this null as a forest plot: each row is a frequency cut, with the ``all tasks'' baseline analysis in the main manuscript on top; the grey bar spans the 10th to 90th percentiles of the reshuffle null, the vertical tick marks the null mean, and the dot is the observed estimate, red when it falls outside the null band and blue when inside, with the observed percentile printed in the right margin.
\begin{figure}[!t]
\caption{Average AI Chain Length Versus Its Task Position-Reshuffle Null Across Frequency Cuts}
\label{fig:chainlength_frequency}
\centering
\includegraphics[width=0.82\textwidth]{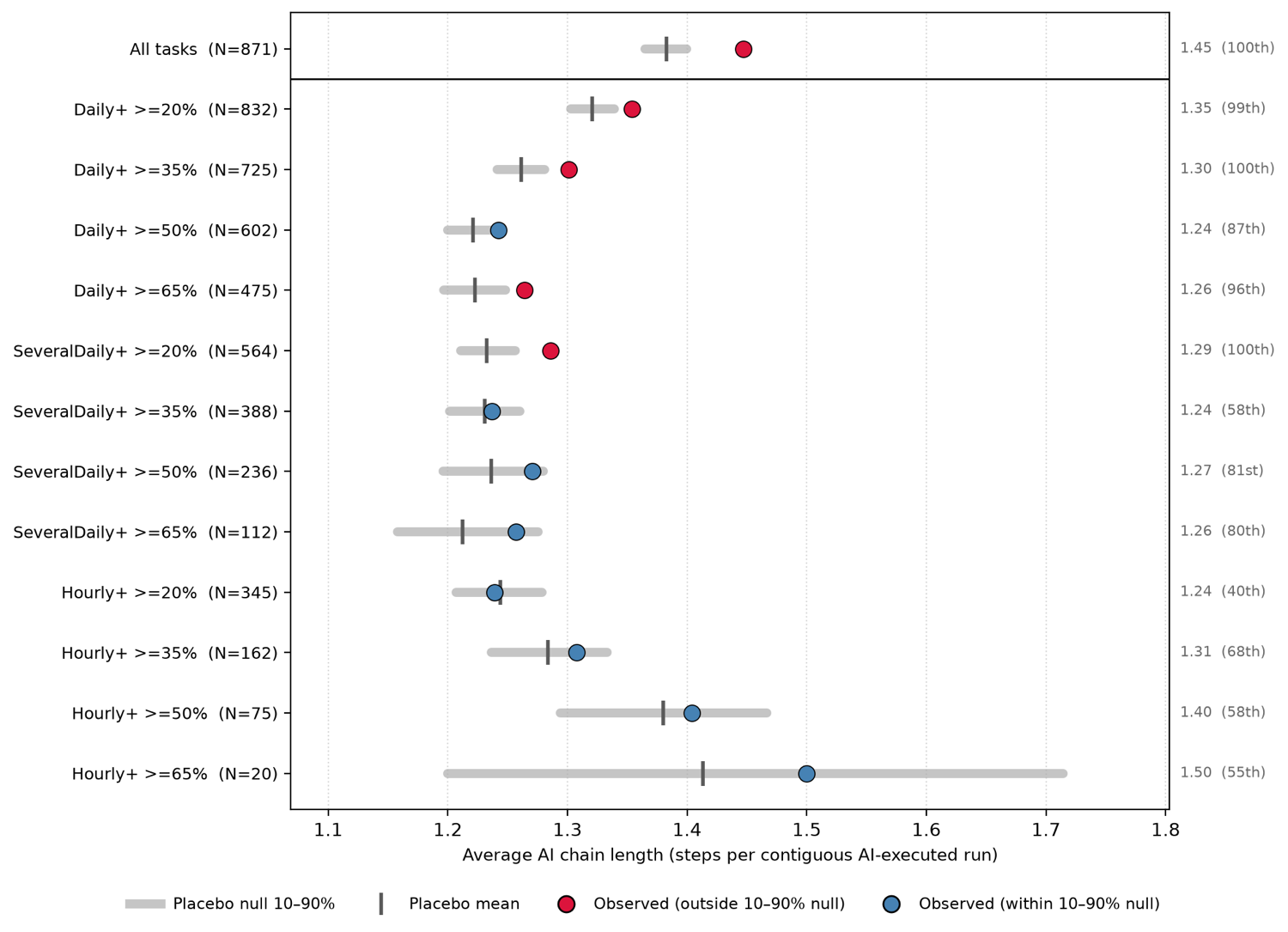}
\par\vspace{\fignotesskip}
\begin{minipage}{1\linewidth}
\begin{spacing}{0.2}
{\fontsize{9.5pt}{10pt}\selectfont Notes: Each row is a frequency cut, with the ``All tasks'' baseline from the main manuscript on top.
The left labels report the number of occupations in each cut.
The dot is the observed average AI chain length for that cut: the mean length of a maximal run of contiguous AI-executed (augmented or automated) steps.
The grey bar spans the 10th to 90th percentiles of the within-occupation task position-reshuffle null ($1{,}000$ draws), and the vertical tick marks the null mean.
The dot is red if it falls outside the null band and blue if inside.
The right margin reports the observed value and its percentile in the null.}
\end{spacing}
\end{minipage}
\end{figure}
The top all-tasks row reproduces the full-sample value of $1.45$.
Chains shorten modestly as the sample is restricted to frequent tasks, to between $1.24$ and $1.35$ in every cut except the two sparsest Hourly$+$ corners, where they tick back up ($1.40$ and $1.50$, with $75$ and $20$ occupations).
In the full sample the observed value sits at the 100th percentile of its null. AI-executed tasks form longer contiguous chains than random reordering would produce, so the chaining pattern is not a compositional artifact.
Pruning preserves the direction of the result but not its precision.
The observed chain length lies above its null mean in eleven of the twelve pruned cuts and clears the 10-to-90 null band in four of them, all under the inclusive logics, namely Daily$+$ at the $20\%$, $35\%$, and $65\%$ thresholds (the 99th, 100th, and 96th percentiles) and SeveralDaily$+$ at $20\%$ (the 100th).
In the remaining cuts it sits inside the band, at the 58th to 87th percentiles under Daily$+$ and SeveralDaily$+$ and at the 40th to 68th under Hourly$+$, where the samples are smallest and the null bands widest.

\subsection{Robustness of Prediction \#2 Results}
\label{app:frequency_robustness_neighbor}

Figure~\ref{fig:neighbor_frequency_heatmap} repeats the exercise for the neighbor prediction.
\begin{figure}[!t]
\caption{Neighbor AI Execution and Focal-Task AI Execution Across Frequency Cuts}
\label{fig:neighbor_frequency_heatmap}
\centering
\begin{subfigure}[b]{\textwidth}
    \centering
    \subcaption{Previous Task ($k{-}1$)}
    \includegraphics[width=\textwidth]{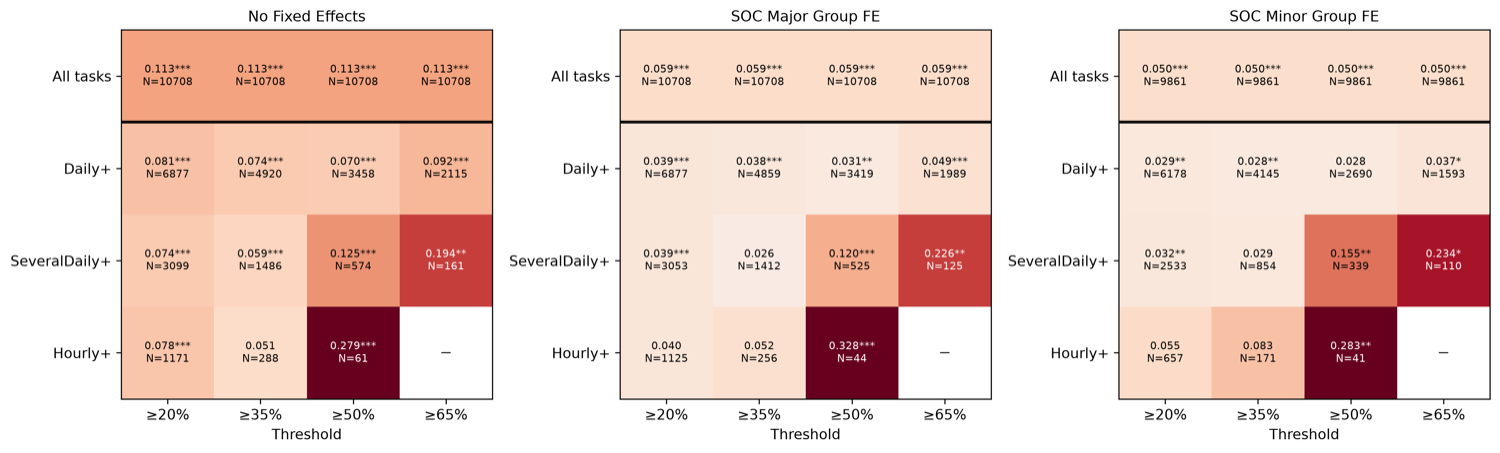}
\end{subfigure}

\vspace{0.8em}

\begin{subfigure}[b]{\textwidth}
    \centering
    \subcaption{Next Task ($k{+}1$)}
    \includegraphics[width=\textwidth]{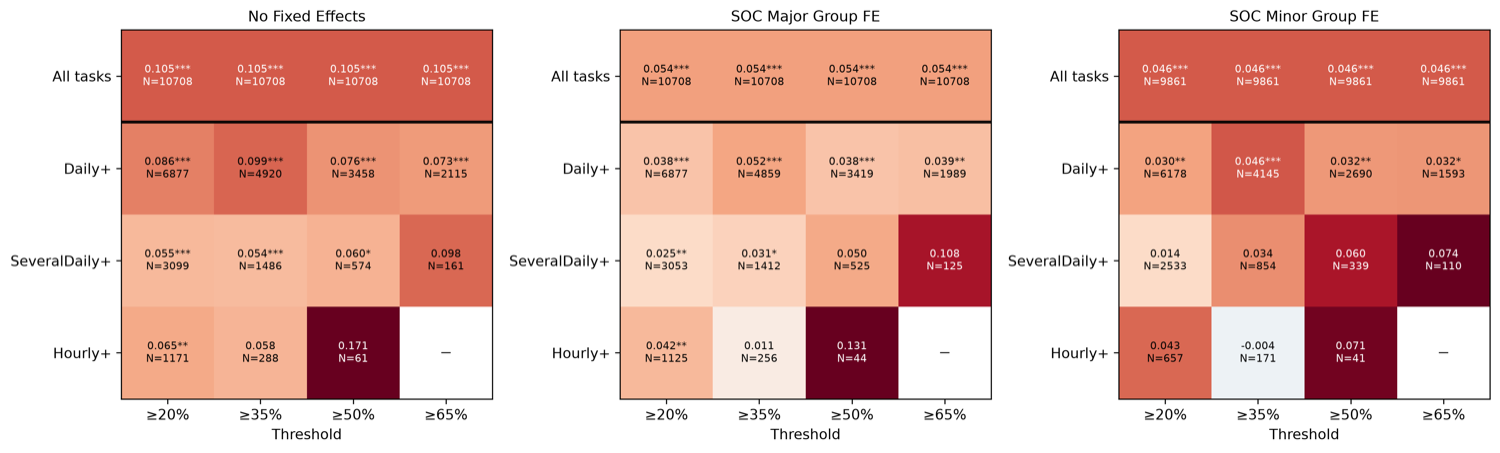}
\end{subfigure}
\par\vspace{\fignotesskip}
\begin{minipage}{1\linewidth}
\begin{spacing}{0.2}
{\fontsize{9.5pt}{10pt}\selectfont Notes: Each cell reports the average marginal effect of having the immediately previous task ($k{-}1$, Panel~(a)) or next task ($k{+}1$, Panel~(b)) be AI-executed on the focal task's probability of AI execution, from re-estimating Equation~\eqref{eq:DWA_regression_ai} on the corresponding frequency cut.
The heatmap layout is as in Figure~\ref{fig:frag_frequency}, except that the number beneath each estimate is the number of task observations and the top ``All tasks'' row is the unpruned specification of Section~\ref{sec:DWA_prediction}.
Standard errors are clustered at the DWA level, and significance stars come from the clustered coefficient test. *** p$<$0.01, ** p$<$0.05, * p$<$0.1.
Red denotes a positive effect (the predicted sign) and blue a negative one.
A ``$-$'' marks a cut with too few observations to estimate.}
\end{spacing}
\end{minipage}
\end{figure}
Each cell reports the average marginal effect, from re-estimating Equation~\eqref{eq:DWA_regression_ai}, of having the immediately preceding ($k{-}1$, Panel~(a)) or following ($k{+}1$, Panel~(b)) task be AI-executed on the focal task's probability of AI execution.
A red cell denotes a positive estimate, the sign found in the full sample.
The top row reproduces the full-sample adjacent-step estimates of $+0.11$ for both the previous and the next task under no fixed effects, which attenuate to between $+0.05$ and $+0.06$ once SOC fixed effects absorb cross-family differences.
The estimate is positive in 65 of the 66 estimable pruned cells.
Under the inclusive Daily$+$ logic it is somewhat smaller than in the full sample, at $+0.07$ to $+0.10$ without fixed effects and $+0.03$ to $+0.05$ with SOC fixed effects, and it remains significant in 23 of the 24 cells.
Under the stricter SeveralDaily$+$ and Hourly$+$ logics the estimates are noisier.
The previous-task estimate is largest in the sparsest cuts (SeveralDaily$+$ $\geq 50\%$ and $\geq 65\%$, and Hourly$+$ $\geq 50\%$, with $41$ to $574$ task observations), at $+0.12$ to $+0.33$ and significant throughout, while the next-task estimate in the same cells ranges from $+0.05$ to $+0.17$ and is mostly insignificant.
Outside those cuts, significance is lost mainly in the intermediate SeveralDaily$+$ and Hourly$+$ cuts, and most often with SOC minor-group fixed effects.
The effects of steps two positions away remain small and, under fixed effects, are statistically indistinguishable from zero across the grid, in line with the main text. The part of a step's local context that predicts its execution mode is the part a chain of average length $1.45$ can actually span.\footnote{ To save space, we do not report the results on steps two position away from the focal step and only focus on the outcomes for immediate neighbors.}

Figure~\ref{fig:neighbor_placebo_forest} benchmarks the previous-task ($k{-}1$, Panel~(a)) and next-task ($k{+}1$, Panel~(b)) effects against their position-reshuffle nulls, one column per specification.
\begin{figure}[!t]
\caption{Immediate-Neighbor Associations: Observed Average Marginal Effects Versus Their Position-Reshuffle Nulls Across Frequency Cuts}
\label{fig:neighbor_placebo_forest}
\centering
\begin{subfigure}[b]{\textwidth}
    \centering
    \subcaption{Previous Task ($k{-}1$)}
    \includegraphics[width=\textwidth]{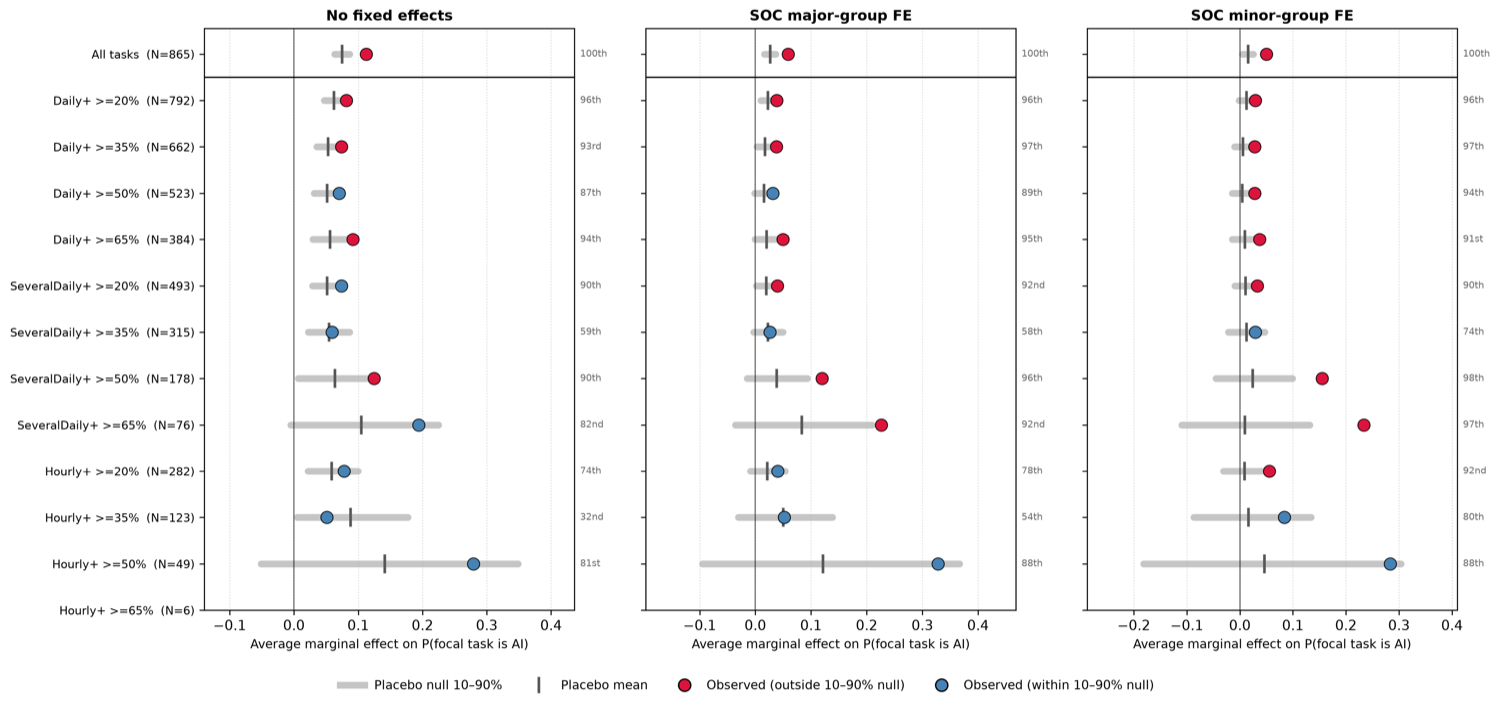}
\end{subfigure}

\vspace{0.8em}

\begin{subfigure}[b]{\textwidth}
    \centering
    \subcaption{Next Task ($k{+}1$)}
    \includegraphics[width=\textwidth]{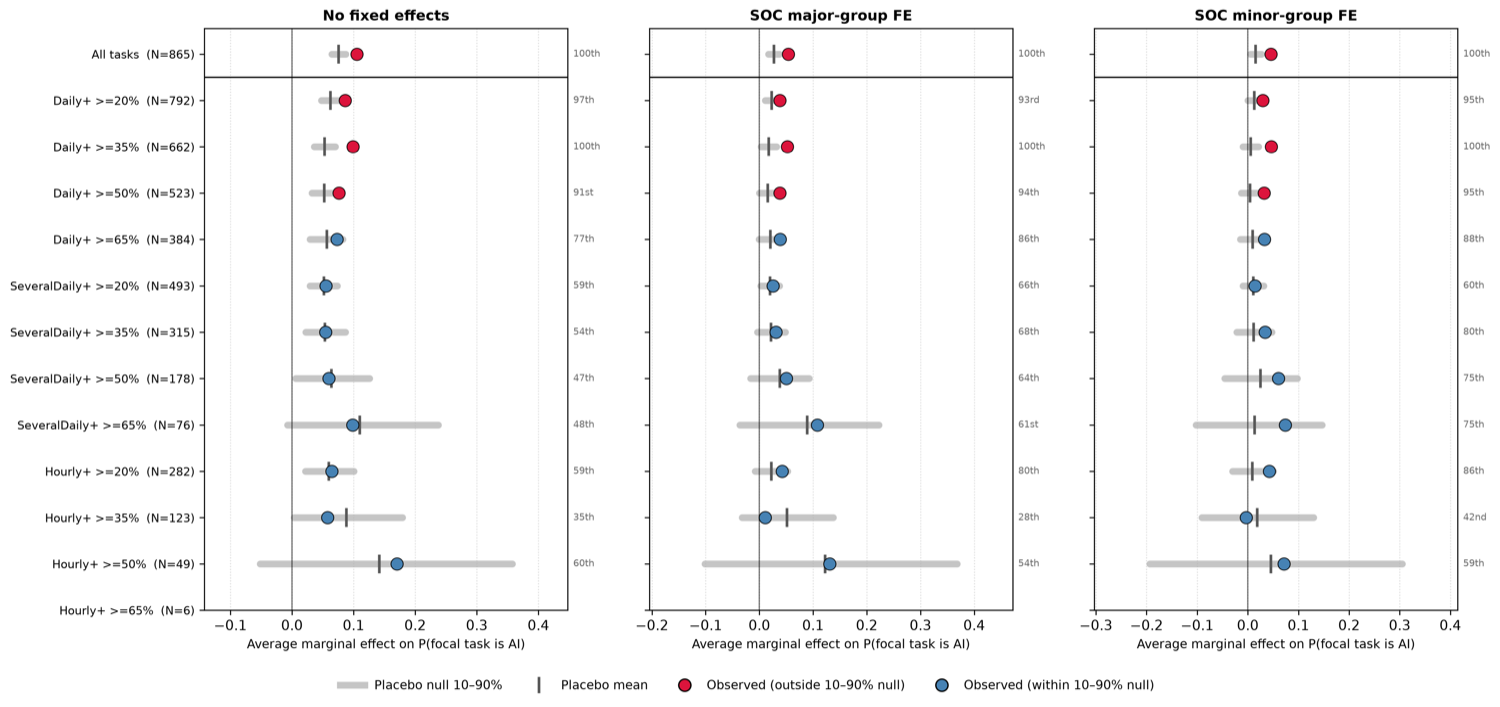}
\end{subfigure}
\par\vspace{\fignotesskip}
\begin{minipage}{1\linewidth}
\begin{spacing}{0.2}
{\fontsize{9.5pt}{10pt}\selectfont Notes: Panel~(a) reports the average marginal effect of having the previous task ($k{-}1$) be AI-executed, and Panel~(b) the corresponding effect for the next task ($k{+}1$).
Rows, the null band, the dot coloring, and the margin labels are as in Figure~\ref{fig:chainlength_frequency}, with the dot now the observed average marginal effect; the three columns within each panel are the no-fixed-effects, SOC major-group, and SOC minor-group specifications.
A blank row denotes a cut with too few neighbored task observations to estimate.}
\end{spacing}
\end{minipage}
\end{figure}
Both estimates sit far in the upper tail of their nulls in the full sample, at the 100th percentile in all three specifications, and mostly stay there under the Daily$+$ logic.
The previous-task estimate lies outside its 10-to-90 null band in ten of the twelve Daily$+$ cells and the next-task estimate in nine, at the 87th to 97th and the 77th to 100th percentiles across those cells, respectively.
Under the stricter logics the two separate.
The previous-task estimate still escapes its null in seven of the twelve SeveralDaily$+$ cells, including all three at the $\geq 50\%$ threshold, but in only one of the nine Hourly$+$ cells, while the next-task estimate stays inside its null band in every SeveralDaily$+$ and Hourly$+$ cell, at the 28th to 86th percentiles.
The positive association between adjacent-step AI execution and the focal step's AI execution thus persists under the inclusive cuts and loses precision, but rarely its sign, under the strictest thresholds, where the surviving task observations number in the low hundreds and the reshuffle bands are correspondingly wide.

\subsection{Robustness of Prediction \#3 Results}
\label{app:frequency_robustness_frag}

Figure~\ref{fig:frag_frequency} summarizes the fragmentation regressions as a heatmap over the logic-by-threshold grid, with the unpruned ``all tasks'' specification from the main analysis in the top row, so that the effect of progressively dropping infrequent tasks can be read down the rows, which impose a more demanding frequency logic, and across the columns, which raise the threshold.
\begin{figure}[!t]
\caption{Empirical Fragmentation Index and AI Execution Across Frequency Cuts}
\label{fig:frag_frequency}
\centering
\includegraphics[width=\textwidth]{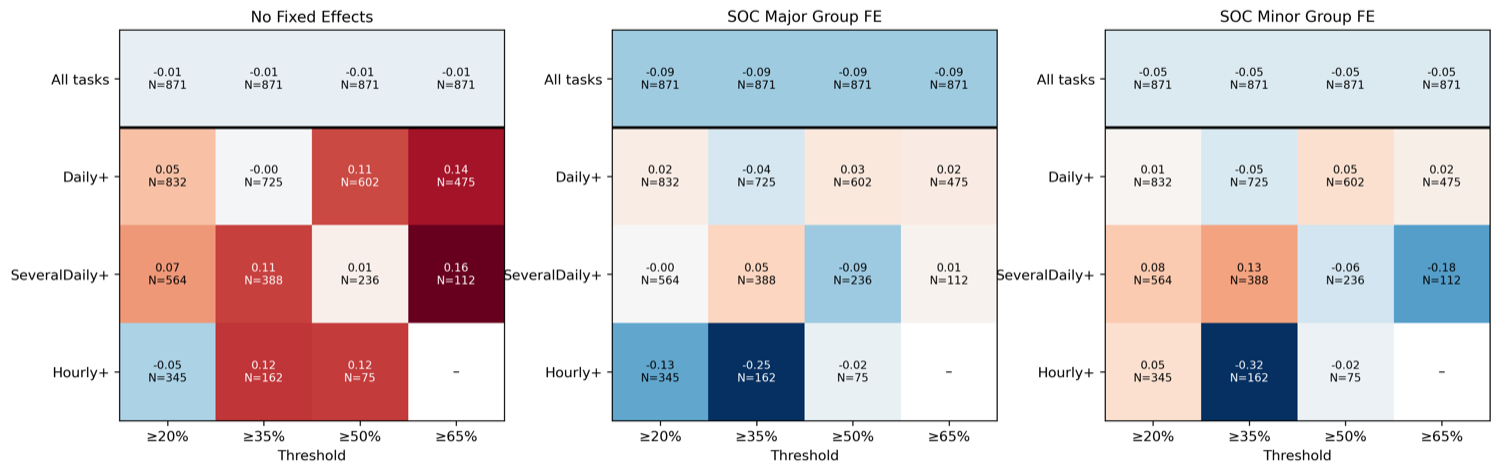}
\par\vspace{\fignotesskip}
\begin{minipage}{1\linewidth}
\begin{spacing}{0.2}
{\fontsize{9.5pt}{10pt}\selectfont Notes: Each cell reports the standardized coefficient on the empirical fragmentation index from re-estimating Equation~\eqref{eq:fragmentation_index_regression} on the corresponding frequency cut.
The index is the preferred exposure-based measure.
The number of occupations appears beneath each estimate. No cell in the grid is significant at the 10\% level, so no significance stars are shown. Following the neighbor heatmaps of Subsection~\ref{app:frequency_robustness_neighbor}, the Hourly$+$ and $\geq 65\%$ cut is left blank, since the $20$ occupations it retains are too few for this specification to estimate.
The three heatmaps report the no-fixed-effects, SOC major-group, and SOC minor-group specifications.
Within each heatmap, rows correspond to the three frequency logics and columns to the four thresholds.
The top ``All tasks'' row is the unpruned specification of Section~\ref{sec:fragmentation_prediction}.
Blue denotes a negative coefficient (the predicted sign) and red a positive one. Each panel is scaled to its own largest absolute coefficient, so shading is comparable within a panel but not across panels.}
\end{spacing}
\end{minipage}
\end{figure}
Each cell reports the standardized coefficient on the empirical fragmentation index from re-estimating Equation~\eqref{eq:fragmentation_index_regression} for the corresponding cut, with the number of occupations printed beneath (no cell is significant at the $10\%$ level, so no stars appear); as in the main text, every specification controls for the number of AI-exposed steps in the occupation.
We use our preferred exposure-based measure throughout (as defined in Section~\ref{sec:fragmentation_prediction}, under which both E1- and E2-exposed tasks may form AI chains), and the three heatmaps correspond to the no-fixed-effects, SOC major-group, and SOC minor-group specifications.
A blue cell denotes a negative coefficient, the sign the model predicts, and a red cell a positive one; darker shading indicates a larger magnitude.
The top all-tasks row gives the full-sample standardized estimates of $-0.01$, $-0.09$, and $-0.05$, none of them distinguishable from zero, here estimated on the $871$ occupations that retain at least five tasks and closely matching the headline estimates of Table~\ref{tab:fragmentation_index_regression_exposure}.
Pruning does not recover the relationship.
Across the $33$ pruned cells the coefficient is negative in $13$, the median absolute estimate is $0.05$, and no cell in the grid reaches significance at the $10\%$ level under any of the three specifications.
We do not read the pruned cells as ruling the channel out either, since $20$ of the $33$ retain confidence intervals wide enough to admit a moderate negative effect.
What the grid establishes is that the null of Subsection~\ref{sec:fragmentation_prediction} is a feature of the O*NET sample rather than an artifact of counting rarely-performed tasks as steps.

In short, the chain-length and neighbor patterns keep their sign among frequently performed tasks, and neither they nor the fragmentation null of the O*NET sample appear to be artifacts of the rarely performed activities catalogued in O*NET.

\clearpage
\section{External Validation of the Sequencing Approach and Results}
\label{app:external_validation}

The O*NET task sequences underlying our empirical analysis are not observed.
O*NET catalogues the tasks that belong to an occupation but not the order in which they are performed, so we
construct that order by prompting a language model (GPT-5-mini).
Every adjacency-based object we compute on the O*NET sample rests on that imputation, which raises two concerns.

The first is that the task orderings may be artifacts of how LLMs reason rather than reconstructions of
how work is done.
An LLM asked to arrange a list of tasks may place semantically similar items next to one another, encoding
topical similarity instead of production logic.
Since AI exposure and execution are themselves correlated with what a task is about, such an ordering could manufacture the
contiguity we document without any underlying sequence and contaminate our other exercises.
Appendix~\ref{app:gptPrompts_robustness} showed our orderings are stable across prompts, but stability is
not accuracy, since an instrument can be consistently wrong.

The second is that not all O*NET tasks may be performed as part of a single workflow sequence.
O*NET is a catalogue of what an occupation involves, not a process map, and if its task lists do not describe
work that proceeds in an order, no ordering procedure, however accurate, could deliver the objects our predictions require.
Appendix~\ref{app:frequency_robustness} pruned the O*NET list of tasks to those that are highly frequently performed, which makes
the resulting task sequences more workflow-like, but the caveat still applies.

This appendix addresses both concerns with external data whose work sequences a language model did not impute.
Two corpora serve.
The first is the \emph{Process Classification Framework} (PCF) of the American Productivity \& Quality Center
(APQC), whose member working groups have recorded since 2004 the activities and tasks that make up a firm's
business processes, in the order practitioners perform them \citep{apqc_pcf}.
The second is a collection of process-mining event logs from the 4TU.ResearchData repository \citep{roadfines_log, sepsis_log, bpi2017_log, bpi2020_log}, which
record case by case the activities organizations actually executed and when, so that the order is observed
rather than asserted.

The bulk of the appendix establishes that the LLM can recover an ordering of activities at all, against corpora
whose order is already known.
Appendix~\ref{app:apqc_validation} shuffles the steps of 345 documented process branches and asks the LLM to
order them in a sequence, under the main analysis's prompt and each of its ten alternatives.
It recovers the practitioners' order to a substantial degree, and does so despite the presence of branches whose activities
are mutually exclusive options for which no true order exists.
Appendix~\ref{app:eventlog_validation} repeats the exercise against the event logs, where the data reveal
which precedences are real and the comparison can be restricted to them.
The two benchmarks agree closely about how much sequence structure the instrument recovers.

Subsection~\ref{app:apqc_fragmentation_top} then turns the PCF corpus on the second concern.
It documents how AI exposure and AI execution labels are carried onto PCF steps with natural-language
embeddings, which is the construction behind columns~(4)\textendash(6) of
Table~\ref{tab:fragmentation_index_regression_exposure}, where the PCF estimates of the fragmentation channel
of Prediction~\#3 are reported alongside the O*NET sample and discussed in
Subsection~\ref{sec:fragmentation_prediction}.
It also gives in full the Prediction~\#1 exercise on those documented sequences that
Subsection~\ref{sec:chainLength_prediction} summarizes, where AI-executed steps again form chains longer
than chance allows.
Prediction~\#2 is the one test this corpus cannot carry, for the reason given in
Subsection~\ref{sec:DWA_prediction} and with the counts reported below.

Taken together, the exercises conducted in this appendix indicate that our findings are artifacts neither of how LLMs
order a list of tasks in a sequence nor of how O*NET catalogues tasks in occupations.

\subsection{Concern \#1: LLM Recovering Real-World Task Orderings}
\label{app:sequencing_validation}

We benchmark the sequencing procedure twice, against two corpora that differ in where their order comes from.
The first records the order practitioners say they work in; the second records the order in which work was
actually carried out.

\subsubsection{Benchmark \#1: Practitioner-Documented Orderings (PCF)}
\label{app:apqc_validation}

We apply our task-ordering approach, exactly as used in the main analysis, to the APQC corpus, and measure how much of the
order practitioners recorded it recovers.
The benchmark is the \emph{Process Classification Framework} (PCF) of APQC in its Cross-Industry version 8.0.
The PCF decomposes work into five nested levels, from 13 Categories through 74 Process Groups, 362 Processes, 1,353 Activities, and 215 Tasks, for 2,017 process elements in total, each with a persistent identifier, a name, and a short description.
What makes it useful here is that within a parent element the children are listed in the order in which they are performed whenever applicable. Ordering steps in ``Order materials and services,'' for instance, runs from \textit{``process/review requisitions''} through \textit{``approve requisitions,''} \textit{``solicit supplier quotes,''} \textit{``create/distribute purchase orders,''} \textit{``expedite orders and satisfy inquiries,''} \textit{``reconcile purchase orders,''} and finally \textit{``research/resolve order exceptions.''}
APQC also publishes industry-specific versions of the framework, which Subsection~\ref{app:apqc_fragmentation_top} draws on; Table~\ref{tab:apqc_process_examples} shows one documented process group from each of two of them.\footnote{Our results are not particular to the cross-industry framework. While there is some degree of heterogeneity in how the LLM recovers sequences from business processes in different industries, the magnitudes of the difference are not large and we do not report them in detail here for brevity.}
\begin{table}[!t]
    \caption{Example Process Groups From Two Industry Frameworks}
    \label{tab:apqc_process_examples}
    \footnotesize
    \centering
    \adjustbox{max width=\textwidth}{\setlength{\tabcolsep}{6pt}
\footnotesize
\begin{tabular}{cll}
\toprule
 & \textbf{Aerospace and Defense} & \textbf{Automotive} \\
Step & Manage product recalls and regulatory audits & Manage parts \\
\midrule
1 & Initiate recall & Manage inventory after sale \\
2 & Assess the likelihood and consequences of occurrence of any hazards & Manage electronic parts catalog \\
3 & Manage recall related communications & Exchange parts/locate parts \\
4 & Submit regulatory reports & Manage returns \\
5 & Monitor and audit recall effectiveness & Rebuild part \\
6 & Manage recall termination & Manage parts retail operations \\
\bottomrule
\multicolumn{3}{l}{\footnotesize Process elements appear in the order APQC's working groups document them.}\\
\end{tabular}
}
    \begin{minipage}{1\linewidth}
\begin{spacing}{0.2}
{\fontsize{9.5pt}{10pt}\selectfont Notes: Each column shows one PCF Process Group with its constituent process elements, in the order APQC documents them. The groups are Hierarchy ID~6.4 of the Aerospace and Defense framework (PCF ID 20110) and Hierarchy ID~6.11 of the Automotive framework (PCF ID 12685).}
\end{spacing}
\end{minipage}\end{table}

We take as our unit of analysis a \emph{process branch}: a parent element together with its ordered children.
Restricting attention to branches with at least three children leaves 345 branches covering 1,917 process elements, with a median of five steps per branch.
Relative to O*NET, the corpus is small and is weighted toward enterprise process work such as finance, procurement, human resources, information technology, and supply chain.

\paragraph{Practitioner orderings are not immutable ground truth.}
\label{app:apqc_validation_caveat}
Before using APQC's listings as a benchmark, we register a caveat about what they are.
The PCF is a documentation convention produced by committees, not a record of observed execution, and for some branches the listed order is one defensible arrangement among several rather than a statement about precedence.

Two examples make the point.
The branch ``Analyze deal options'' lists four children: \textit{``evaluate acquisition options,''} \textit{``evaluate merger options,''} \textit{``evaluate de-merger options,''} and \textit{``evaluate divestiture options.''}
These are mutually exclusive alternatives that a firm considers not necessarily in the presented order.
Reversing the order, from divestiture to acquisition, would be an equally coherent way to write the same branch down.
Similarly, ``Process accounts payable and expense reimbursements'' lists \textit{``process accounts payable,''} then \textit{``process expense reimbursements,''} then \textit{``manage corporate credit cards.''}
These are three distinct disbursement streams that can be run concurrently, and nothing in the process dictates which is documented first.
For branches of this kind any ordering an LLM produces will disagree with APQC roughly half the time by construction, and disagreement carries no information about whether the LLM understands sequence.

We make no attempt to purge such branches from the sample, since doing so would require the very judgment about true sequence that the exercise is meant to test.
We proceed with the full set of branches, which makes the measured overlap a conservative one.

\paragraph{Design.}
\label{app:apqc_validation_design}
For each of the 345 branches we take APQC's published child order as the target, shuffle the children under a fixed seed so that the LLM never observes the published sequence, and ask it to order them.
The prompt is the one used to sequence O*NET tasks in the main analysis (Prompt~\#1 of Appendix~\ref{app:prompts}), carried over verbatim with a single mechanical substitution that maps the occupation domain onto the process domain, replacing ``occupation'' with ``business process'' and ``task'' with ``step.''
We repeat the exercise for each of the ten alternative prompts of Appendix~\ref{app:gptPrompts_robustness}, applying the identical substitution to all eleven so that no prompt is privileged by hand-tuning, and we use the same model and settings as in the main analysis (GPT-5-mini, temperature $0$).

We measure overlap between the LLM's ordering and APQC's with Kendall's $\tau$, computed within each branch and averaged across branches.
We also report the share of APQC-adjacent step pairs that the model places in the correct direction, whose null value is $0.5$, and the share of branches whose full sequence the model reproduces exactly.
Statistical significance is assessed against a permutation null of $1{,}000$ random orderings per branch, which centers $\tau$ at zero.
Of the $3{,}795$ orderings requested, $3{,}793$ returned a complete and parseable permutation of the branch.%\footnote{The two exceptions occur under one alternative prompt, where the model paraphrased a step title rather than returning it verbatim; we require exact matches before accepting an ordering, so those two branches are dropped for that prompt alone.}

\paragraph{Results.}
\label{app:apqc_validation_results}
Table~\ref{tab:apqc_pcf_ordering_validation} reports the main results and Figure~\ref{fig:apqc_tau_by_size} the underlying distributions.
\begin{table}[!t]
    \caption{Overlap Between GPT-Generated and Practitioner-Documented Process Orderings}
    \label{tab:apqc_pcf_ordering_validation}
    \footnotesize
    \centering
    \setlength{\tabcolsep}{8pt}
\begin{tabular}{lc}
\toprule
\addlinespace
\multicolumn{2}{l}{\textit{Mean Kendall $\tau$ against APQC's published process order}} \\
\addlinespace
\quad Main prompt (Prompt 0) & +0.53 \\
\quad Average of 10 alternative prompts & +0.51 \\
\addlinespace
\hline\\[-1.25em]
\multicolumn{2}{l}{\textit{Other overlap measures}} \\
\addlinespace
\quad Adjacent-pair agreement, main prompt (null 0.5) & 0.70 \\
\quad Adjacent-pair agreement, 10-prompt average & 0.70 \\
\quad Exact-order recovery, main prompt & 0.24 \\
\quad Share of branches with $\tau > 0$, main prompt & 0.86 \\
\quad Share of branches with $\tau > 0$, 10-prompt average & 0.90 \\
\hline\\[-1.25em]
Process branches & 345 \\
Process elements ordered & 1,917 \\
\bottomrule
\end{tabular}

    \begin{minipage}{1\linewidth}
\begin{spacing}{0.2}
{\fontsize{9.5pt}{10pt}\selectfont Notes: This table reports the overlap between the orderings GPT-5-mini produces for APQC process branches and the orderings APQC's practitioner working groups published, pooled across all 345 branches.
``Adjacent-pair agreement'' is the share of APQC-adjacent step pairs the model places in the correct direction, with a null value of $0.5$.
``Exact-order recovery'' is the share of branches whose full published sequence the model reproduces.}
\end{spacing}
\end{minipage}\end{table}
\begin{figure}[!t]
\caption{Overlap between GPT-orderings and Original APQC-PCF Order by Branch Length}
\label{fig:apqc_tau_by_size}
\centering
\begin{subfigure}[b]{0.49\textwidth}
    \subcaption{Main prompt}
    \label{fig:apqc_tau_size_main}
    \includegraphics[width=\linewidth]{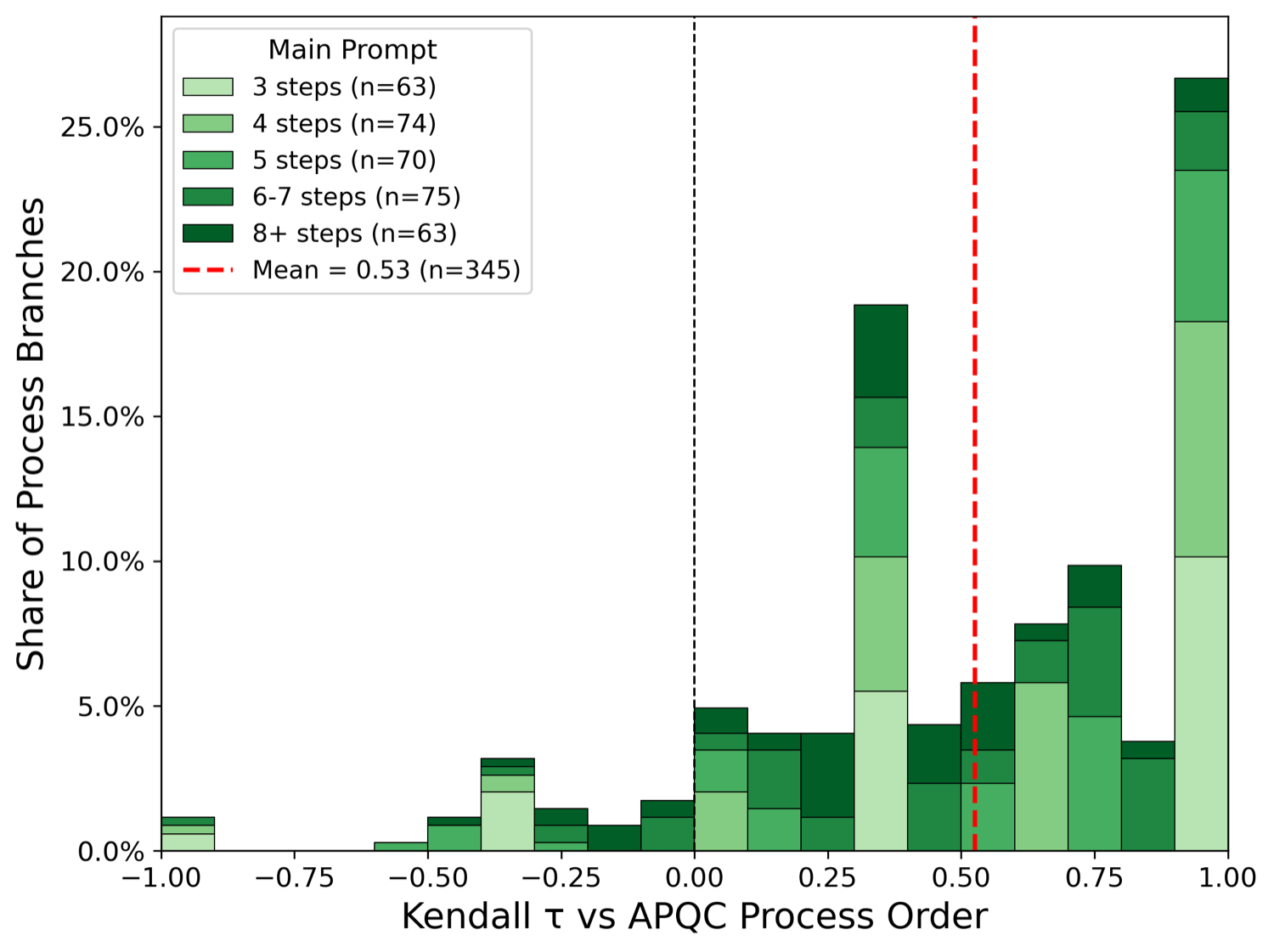}
\end{subfigure}
\hfill
\begin{subfigure}[b]{0.49\textwidth}
    \subcaption{Average of ten alternative prompts}
    \label{fig:apqc_tau_size_alt}
    \includegraphics[width=\linewidth]{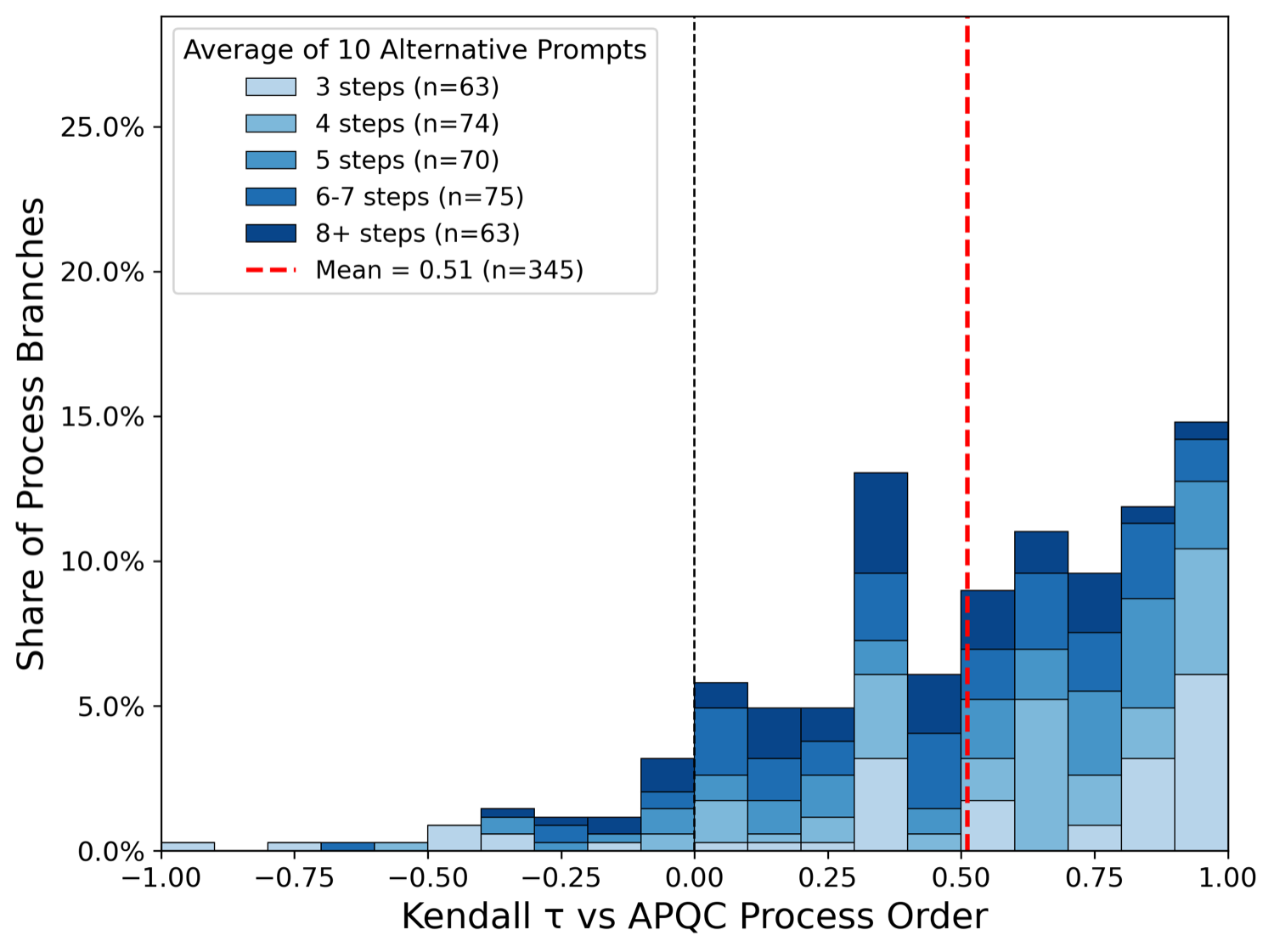}
\end{subfigure}
\par\vspace{\fignotesskip}
\begin{minipage}{1\linewidth}
\begin{spacing}{0.2}
{\fontsize{9.5pt}{10pt}\selectfont Notes: Each panel plots the distribution across the 345 APQC process branches of the within-branch Kendall's $\tau$ between GPT-5-mini's ordering and APQC's published ordering, with each bin decomposed by the number of steps in the branch; darker shading denotes longer branches.
Panel~(a) uses the main prompt and panel~(b) averages, within each branch, the $\tau$ obtained under the ten alternative prompts of Appendix~\ref{app:gptPrompts_robustness}.
The red dashed line marks the sample mean and the black dashed line the value a random ordering would deliver in expectation.
Kendall's $\tau$ is coarse on short branches, which is why the extreme bins at both ends are dominated by branches of three and four steps.}
\end{spacing}
\end{minipage}
\end{figure}
Under the main prompt the average Kendall's $\tau$ between the LLM's ordering and APQC's is $0.53$, with a median of $0.60$; $86\%$ of branches have positive $\tau$.
The LLM places $70\%$ of APQC-adjacent step pairs in the correct direction against a null of $50\%$, and reproduces the practitioner sequence \emph{exactly} in $24\%$ of branches.
The last of these is the most demanding, because exact recovery by chance is one in $120$ for a five-step branch and one in $40{,}320$ for an eight-step branch.
Against the permutation null the average $\tau$ corresponds to $z = 22.6$ ($p = 0.001$).
The LLM is not reproducing the order in which the steps were shown to it. The average $\tau$ between its output and the shuffled input list is $0.11$.

Running the ten alternative prompts leaves this essentially unchanged.
Prompt-specific averages span a narrow range, from $0.49$ to $0.53$, so the result is not an artifact of the particular phrasing used in the main prompt.
Averaging the ten alternatives within each branch, which removes the sampling noise of any single formulation, leaves the mean at $0.51$ while raising the share of branches with positive $\tau$ from $86\%$ to $90\%$ and lowering the cross-branch standard deviation from $0.42$ to $0.37$; twelve of the branches that score negative under the main prompt turn positive once averaged.
Prompt choice therefore affects individual branches but not the aggregate picture.

Overlap is higher where the benchmark is a stronger claim about sequence.
Across PCF categories the average $\tau$ ranges from $0.66$ in ``Deliver Services,'' $0.69$ in ``Manage Customer Service,'' and $0.65$ in ``Manage Supply Chain for Physical Products,'' which are transactional and genuinely sequential, down to $0.25$ in ``Develop and Manage Products and Services'' and $0.27$ in ``Manage External Relationships,'' whose children are largely mutually exclusive options of the kind discussed above.

\paragraph{Reading the left tail.}
\label{app:apqc_validation_tail}
A tenth of branches record a negative $\tau$ under the main prompt, and it is worth being clear about what they are, since the sign of $\tau$ overstates the disagreement on short branches.
Kendall's $\tau$ is coarse when there are few steps to order. On a three-step branch it can only take the values $\{-1, -1/3, +1/3, +1\}$, so a single inverted pair already registers as a large negative number, whereas on a ten-step branch the same inversion costs $0.02$ and leaves $\tau$ near $+1$.
Figure~\ref{fig:apqc_tau_by_size} decomposes the distribution by branch length and shows the consequence directly. The mass at $\tau = -1$ consists almost entirely of three- and four-step branches (three of its four branches), and so does most of the spike at $\tau = +1$, where $56\%$ of three-step branches sit against $10\%$ of the branches with five steps or more.

The tail is not, however, purely an artifact of short branches.
Negative values are about as common among the longest branches as among the shortest, $14\%$ in both cases, and what differs is their magnitude, in that short branches produce reversals while long branches produce values between $-0.03$ and $-0.20$, that is, a handful of misordered pairs within an otherwise correct sequence.
Inspecting the branches that remain negative after averaging over all ten alternative prompts, the recurring case is the one anticipated above: parents whose children are parallel streams or mutually exclusive alternatives, for which no true ordering exists and the sign of $\tau$ is close to a coin flip.

\subsubsection{Benchmark \#2: Observed Process Execution (4TU.ResearchData)}
\label{app:eventlog_validation}

The benchmark above has a limitation it cannot overcome. A documented ordering is an editorial convention, so where a process has no true sequence, disagreement says nothing about the instrument.
Moving from documented to \emph{observed} order removes it.

A process-mining event log records, case by case, the activities an organization actually executed and when.
Two things follow.
First, the ordering is a fact about production rather than an assertion about it.
Second, and more useful here, the data reveal \emph{which precedences are real}. If activity $x$ precedes $y$ in $97\%$ of cases the precedence is genuine, whereas a $50$/$50$ split means the two run concurrently and no ordering could be correct.
We can therefore condition on the pairs that have a true order before asking whether the model recovers it, which is the step the documented benchmark does not permit.

These logs record processes carried out by several people, with the case passing between them, and that property is what makes them the right benchmark here rather than the wrong one.
Read against the short run alone it would be a mismatch, since Section~\ref{sec:shortrun.production} assigns the whole sequence to a single worker.
But the long-run problem of Section~\ref{sec:longrun.production} is about exactly this object: a workflow partitioned into jobs, with a hand-off incurred wherever the work passes from one worker to the next.
What the model holds fixed across both horizons is the order of the steps, not who performs them, so an event log observes the sequence our imputation is meant to recover, at the level of aggregation the long run works with.

\paragraph{Data.}
\label{app:eventlog_data}
We use eight publicly available event logs from the 4TU.ResearchData repository \citep{roadfines_log, sepsis_log, bpi2017_log, bpi2020_log}, listed in Table~\ref{tab:eventlog_overview}: fine collection by a local police force, sepsis care in a hospital emergency department, consumer lending at a financial institution, and five university travel and reimbursement processes.
Together they cover $215{,}715$ cases and $2{,}048{,}949$ recorded events.
Their activity sets are considerably larger than the APQC process branches, between $8$ and $28$ activities against a median of five, so they also test the instrument on long sequences.

From each log we recover one trace per case, the sequence of activity labels in timestamp order, and keep the activities appearing in at least $2\%$ of cases so that the target is not driven by rare exceptions.
We use each activity's \emph{first} occurrence within a case, which handles the loops present in every real log such as repeated laboratory tests or repeated approval rounds.
This yields the two objects we need.
The \emph{observed ordering} ranks activities by their mean normalized first-occurrence position across cases; this is the analogue of APQC's published order.
The \emph{pairwise precedence matrix} records, for every pair, the share of cases containing both in which one precedes the other.
We call a pair \emph{determinate} when one direction holds in at least $80\%$ of cases.
Across the eight logs, $1{,}124$ of $1{,}178$ activity pairs ($95.4\%$) are determinate, so the benchmark is a statement about sequence for almost all of the comparisons we make.

\begin{table}[!t]
    \caption{Event Logs Used for Validation}
    \label{tab:eventlog_overview}
    \begin{center}
    \resizebox{\textwidth}{!}{%
    \begin{tabular}{llrrrr}
    \toprule
    Log & Process & Cases & Events & Activities & Determinate pairs \\
    \midrule
    Road traffic fines & Fine management, local police & 150,370 & 561,470 & 8 & 27 of 28 \\
    Sepsis & Emergency department sepsis care & 1,050 & 15,214 & 15 & 91 of 98 \\
    BPI 2017 & Loan application handling & 31,509 & 1,202,267 & 23 & 238 of 245 \\
    BPI 2020 & Domestic travel declarations & 10,366 & 56,303 & 10 & 41 of 44 \\
    BPI 2020 & International travel declarations & 6,449 & 72,151 & 22 & 218 of 224 \\
    BPI 2020 & Requests for payment & 6,814 & 36,724 & 10 & 39 of 42 \\
    BPI 2020 & Prepaid travel costs & 2,092 & 18,239 & 17 & 128 of 129 \\
    BPI 2020 & Travel permits and declarations & 7,065 & 86,581 & 28 & 342 of 368 \\
    \midrule
    Total & & 215,715 & 2,048,949 & 133 & 1,124 of 1,178 \\
    \bottomrule
    \end{tabular}}
    \end{center}
    \par\vspace{\fignotesskip}
\begin{minipage}{1\linewidth}
\begin{spacing}{0.2}
{\fontsize{9.5pt}{10pt}\selectfont Notes: All logs are openly available from the 4TU.ResearchData repository \citep{bpi2017_log, bpi2020_log, roadfines_log, sepsis_log}.
``Activities'' counts those appearing in at least $2\%$ of a log's cases, which is the set handed to the model.
A pair of activities is ``determinate'' when one of the two orders holds in at least $80\%$ of the cases containing both.}
\end{spacing}
\end{minipage}
\end{table}

\paragraph{Design.}
\label{app:eventlog_design}
The procedure mirrors that of Appendix~\ref{app:apqc_validation}.
For each log we shuffle the activity labels under a fixed seed, so the model never sees the observed sequence, and ask it to order them, supplying a one-line description of the process in place of an occupation title.
We use all eleven prompts, the main specification of Appendix~\ref{app:prompts} and the ten alternatives of Appendix~\ref{app:gptPrompts_robustness}, carried over with the same mechanical substitution that maps the occupation domain onto the process domain.
No log is told which activity starts or ends a case.

We report two levels of analysis.
At the \emph{log level} we compute Kendall's $\tau$ between the model's ordering and the observed ordering, which is directly comparable to the APQC numbers, benchmarked against a permutation null of $1{,}000$ random orderings per log.
At the \emph{pair level} we ask only whether the model places a pair in the direction the traces show.
The pair level is the more informative of the two here, both because it is where the determinacy conditioning bites and because it is closer to what the paper's predictions actually use, since Predictions~\#1 and~\#2 depend on adjacency rather than on the global ordering.

\paragraph{Results.}
\label{app:eventlog_results}
Table~\ref{tab:eventlog_ordering_validation} reports the results and Figure~\ref{fig:eventlog_activity_accuracy} the underlying distribution.

Against a null of $0.5$, the model places $78.8\%$ of all activity pairs in the direction the traces show.
Restricting to the determinate pairs, those the data confirm have a real order, gives $78.9\%$, and weighting each pair by how determinate it is gives $78.8\%$.
That these three numbers coincide is itself informative. The model does no better on the pairs that a real organization orders strictly than on pairs in general, so the aggregate is not being propped up by easy cases.
At the log level the mean Kendall's $\tau$ is $0.48$ with a median of $0.56$, against a permutation null centered at $0.00$ ($z = 6.8$, $p = 0.001$).
The model is not echoing the list it was shown. The average $\tau$ between its output and the shuffled input is $0.007$.

The ten alternative prompts again leave the picture unchanged.
Determinate-pair accuracy under individual prompts spans $0.767$ to $0.817$ and log-level mean $\tau$ spans $0.433$ to $0.559$, so no conclusion here depends on the phrasing used in the main analysis.

The result that matters for the paper is the comparison with Appendix~\ref{app:apqc_validation}.
Restricting both exercises to \emph{adjacent} pairs, the statistic they share, the documented benchmark gives $0.700$ under the main prompt and $0.697$ averaged over the alternatives, while the observed benchmark gives $0.679$ and $0.699$.
Two independent benchmarks, one written down by practitioners and one recovered from what organizations actually did, agree to within roughly two percentage points about how much sequence structure the instrument recovers.

\begin{table}[!t]
    \caption{Overlap Between GPT-Generated Orderings and Observed Process Execution}
    \label{tab:eventlog_ordering_validation}
    \begin{center}
    \setlength{\tabcolsep}{8pt}
\small
\begin{tabular}{lcc}
\toprule
 & Main prompt & 10 alternative prompts \\
\cmidrule(lr){2-3}
 & (1) & (2) \\
\midrule
\addlinespace
\multicolumn{3}{l}{\textit{Pairwise precedence accuracy against observed execution (null 0.5)}} \\
\addlinespace
\quad All activity pairs & 0.79 & 0.79 \\
\quad Determinate pairs only & 0.79 & 0.79 \\
\quad Weighted by determinacy & 0.79 & 0.79 \\
\quad Adjacent pairs only & 0.68 & 0.70 \\
\addlinespace
\hline\\[-1.25em]
\multicolumn{3}{l}{\textit{Log-level Kendall $\tau$ against the observed ordering}} \\
\addlinespace
\quad Mean across logs & +0.48 & +0.49 \\
\quad Median across logs & +0.56 & +0.54 \\
\quad Random ordering (null) & 0.00 & 0.00 \\
\hline\\[-1.25em]
Event logs & 8 & 8 \\
Activity pairs & 1,178 & 11,780 \\
Determinate pairs & 1,124 & 11,240 \\
Cases & 215,715 & 215,715 \\
\bottomrule
\end{tabular}

    \end{center}
    \par\vspace{\fignotesskip}
\begin{minipage}{1\linewidth}
\begin{spacing}{0.2}
{\fontsize{9.5pt}{10pt}\selectfont Notes: Pairwise precedence accuracy is the share of activity pairs the model places in the direction the event traces show, with a null value of $0.5$.
``Determinate pairs'' are those for which one order holds in at least $80\%$ of the cases containing both activities.
``Adjacent pairs'' are consecutive in the observed ordering and are the statistic directly comparable to the adjacent-pair agreement of Table~\ref{tab:apqc_pcf_ordering_validation}.
Column~(2) pools the ten alternative prompts of Appendix~\ref{app:gptPrompts_robustness}, so its pair counts are ten times those of column~(1).
The random-ordering row is the mean of $1{,}000$ random permutations per log.}
\end{spacing}
\end{minipage}
\end{table}

\begin{figure}[!t]
\caption{Distribution of Overlap With Observed Execution Across Activities}
\label{fig:eventlog_activity_accuracy}
\begin{center}
\begin{subfigure}[b]{0.49\textwidth}
    \subcaption{Main prompt}
    \label{fig:eventlog_acc_main}
    \includegraphics[width=\linewidth]{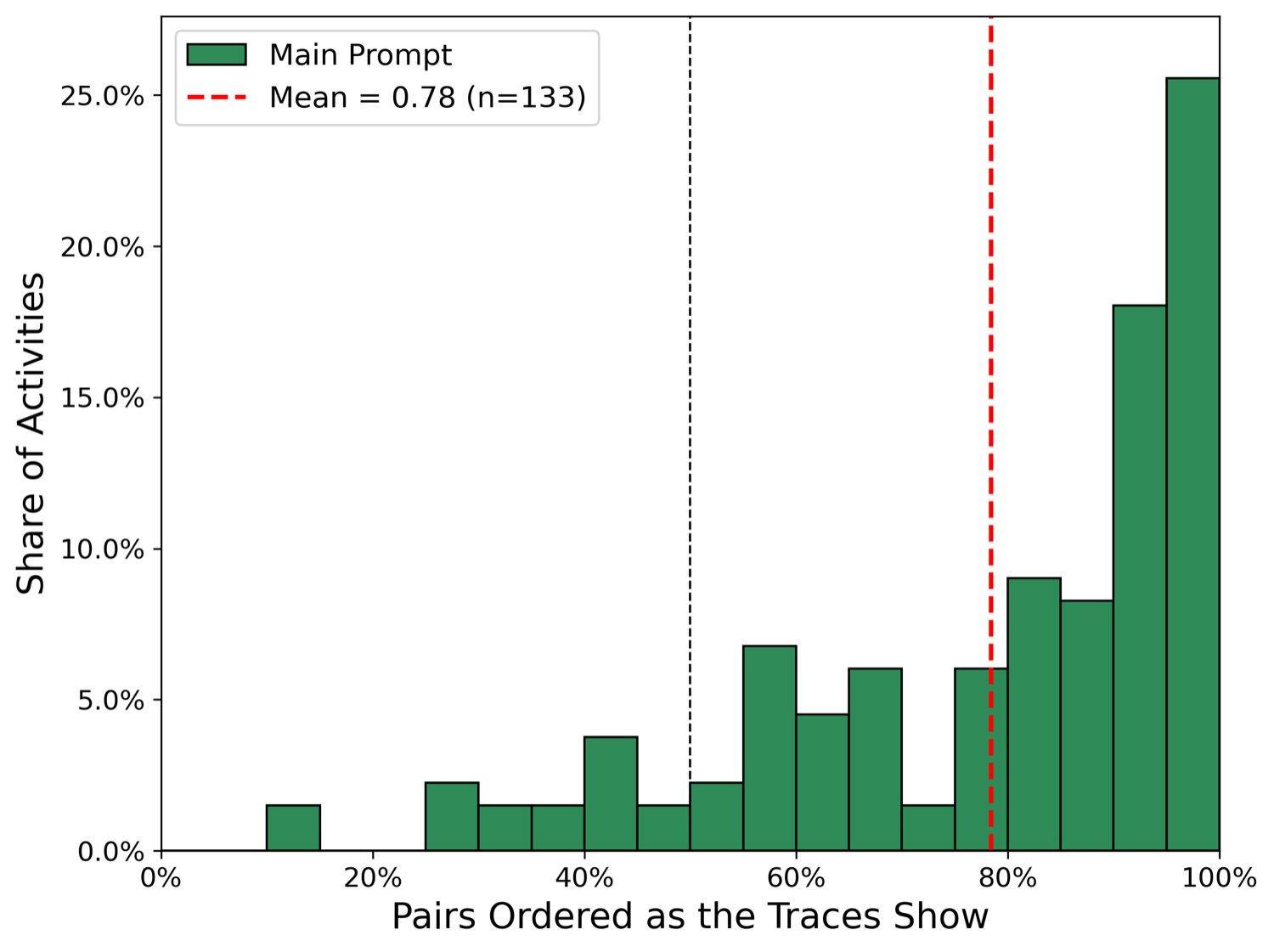}
\end{subfigure}
\hfill
\begin{subfigure}[b]{0.49\textwidth}
    \subcaption{Average of ten alternative prompts}
    \label{fig:eventlog_acc_alt}
    \includegraphics[width=\linewidth]{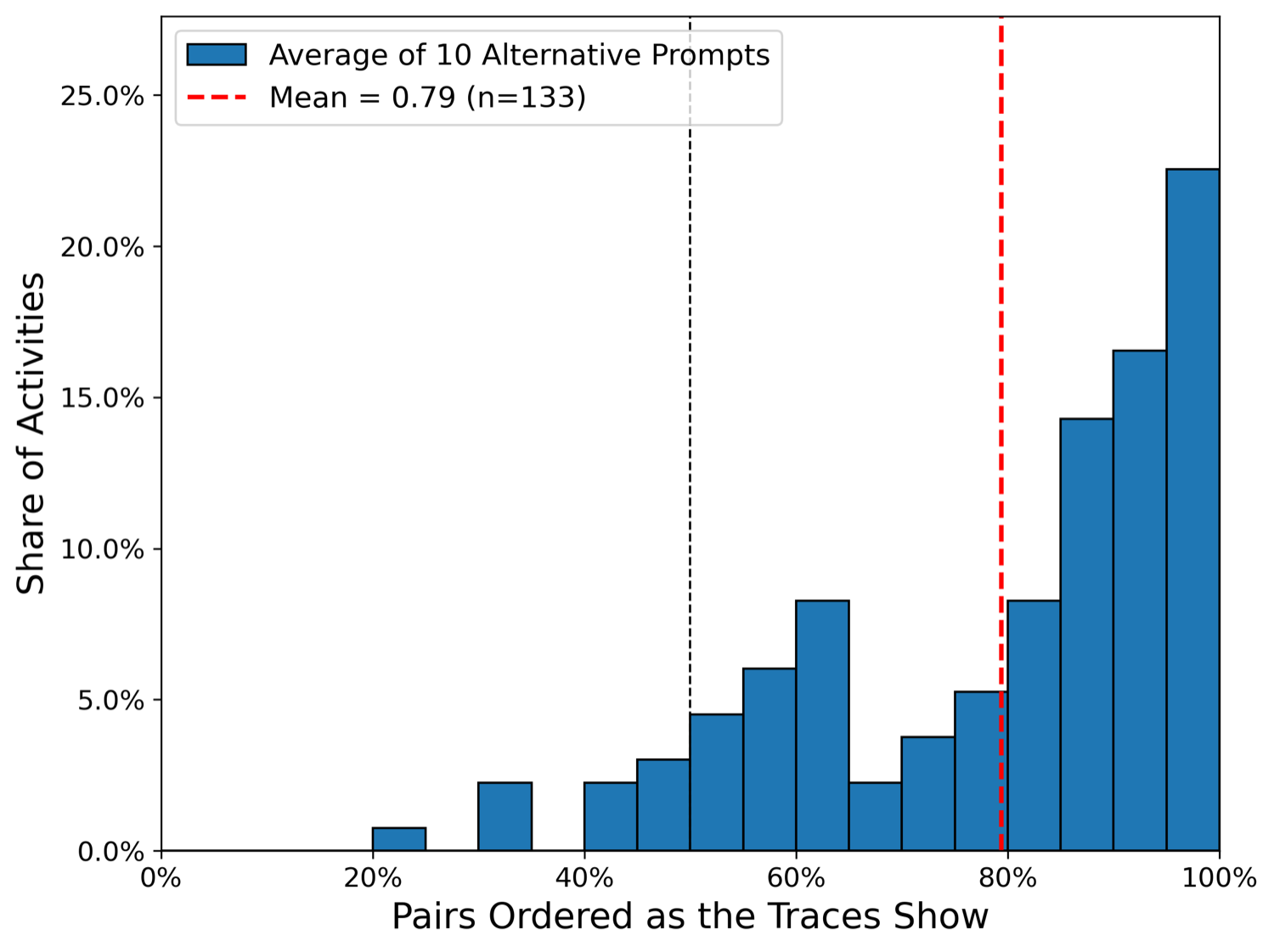}
\end{subfigure}
\end{center}
\par\vspace{\fignotesskip}
\begin{minipage}{1\linewidth}
\begin{spacing}{0.2}
{\fontsize{9.5pt}{10pt}\selectfont Notes: The unit is an activity ($133$ across the eight logs).
For each activity we take the determinate pairs it belongs to and plot the share the model ordered as the traces show.
The black dashed line marks the coin flip and the red dashed line the mean.
Averaging over prompts raises the share of activities above the coin flip from $86.5\%$ to $91.7\%$ while leaving the mean almost unchanged, the same pattern as in Figure~\ref{fig:apqc_tau_by_size}.
Because activities within a log share one ordering and each pair contributes to two activities, this distribution is descriptive; the inferential statistics are the pooled pair-level accuracy and the log-level permutation test.}
\end{spacing}
\end{minipage}
\end{figure}

\paragraph{Where the instrument fails, and why it understates its own accuracy.}
\label{app:eventlog_failures}
Performance varies sharply across logs, and the variation is diagnostic.
Determinate-pair accuracy is $0.96$ for road traffic fines, $0.95$ for sepsis care, and $0.91$ for international declarations, but $0.56$ for requests for payment and $0.49$ for domestic declarations.
Inspecting the errors, the failures are overwhelmingly about \emph{reading the labels}, not about understanding sequence.

Three examples carry most of the damage.
In the domestic declaration log the model ranks \texttt{Request Payment} first, reading it as the employee's initiating request; in the data it is the second-to-last step, the finance system requesting disbursement.
In the loan log the model treats \texttt{A\_Accepted} as terminal, whereas in the data it is an early state meaning the application entered processing, preceding validation, offer creation, and completion; seven of the twelve most determinate errors are this single activity.
In the sepsis log the model places \texttt{Return ER} mid-visit, while in the data it denotes a patient returning to the emergency department after discharge and is genuinely last.

This matters for how the number should be read.
Activity labels in an event log are the strings an information system emits, not descriptions of work.
Some are terse (\texttt{A\_Concept}), and many encode the actor rather than the action (\texttt{Declaration APPROVED by ADMINISTRATION}).
The model is therefore working with substantially less context than an O*NET task statement provides, and the errors we observe are of a kind that richer task descriptions would largely prevent.
Accuracy here is, in that specific sense, a conservative estimate of what the procedure achieves on the task statements the paper actually uses.

\subsection{Concern \#2: Transferring the Labels onto Documented Sequences}
\label{app:apqc_fragmentation_top}

Estimating anything on APQC sequences requires transferring AI exposure and execution labels, since those are measured over O*NET tasks and have no counterpart defined on PCF elements.
We pool the Cross-Industry framework with APQC's seventeen industry-specific frameworks, discard duplicate
process groups, embed each leaf element and each labeled O*NET task with the language embedding model \texttt{all-mpnet-base-v2}, and match every step to its nearest task by cosine similarity.
Matching process steps to task statements is lossy, so we set a similarity threshold of $0.71$ and carry a
label across only when the match clears it.\footnote{We chose the $0.71$ similarity threshold after manual inspection of some of the matches. The paragraph on the similarity floor below reports how the results move as it is varied.}
A step whose best match falls below the threshold is not discarded.
Doing so would shorten the very sequences the exercise is about, and the shortening would not be random, since
the steps that fail to match are concentrated in the operational and industry-specific parts of the frameworks.
We instead keep every step in place and code the sub-threshold ones as neither AI-exposed nor AI-executed,
which preserves the documented ordering intact and treats an unverifiable label as an absent one.
Process groups with fewer than five steps are dropped, leaving $13{,}482$ steps in $525$ process groups and
$26$ steps per group, against $20$ tasks per occupation in the O*NET sample.

\paragraph{What survives the threshold.}
The $0.71$ cosine similarity threshold is demanding, and it should be read as such before any estimate is.
Only $2{,}067$ steps, $15\%$ of the corpus, clear it, at a mean cosine similarity of $0.75$.
Because every step below it is coded manual and unexposed, the resulting sample is very thin in AI content, where
$12\%$ of steps are AI-exposed and $4\%$ are AI-executed, against $44\%$ and $13\%$ in the O*NET task
universe.
Anything estimated on these sequences therefore asks whether a prediction holds on a corpus carrying roughly a
third of the O*NET sample's AI density, with every label that does survive verified at high similarity.
The caveat applies to the fragmentation estimates in columns~(4)\textendash(6) of Table~\ref{tab:fragmentation_index_regression_exposure} as much as to the chain-length test below.

\paragraph{How much the similarity floor matters.}
The floor is a judgement call, so we vary it from $0.65$ to $0.75$ in steps of $0.01$ and re-run both exercises at each value.
The $525$ process groups are held fixed throughout, since the five-step minimum applies to the raw step count rather than to the labels, so raising the floor only relabels steps as unexposed and unexecuted.
Figure~\ref{fig:apqc_similarity_threshold} reports the result.
Both signs survive the whole range.
The average AI chain exceeds its reshuffle null at all eleven floors, and the fragmentation coefficient is negative at all eleven under each of the three fixed-effect specifications, reaching the $5\%$ level in eight of them without fixed effects.
The magnitude is not monotone, however, running from $-0.26$ at $0.65$ to $-0.63$ at $0.69$ and back to $-0.11$ at $0.72$, and the standard error roughly doubles across the range as the labels thin out.
Neighboring floors can therefore fall on opposite sides of a significance threshold without the underlying estimates differing, as at $0.72$ and $0.73$, whose point estimates each lie inside the other's confidence interval, so the range matters more than any single cell.
We therefore choose the floor on the density of the transferred labels rather than on the estimates it produces.
Raising it thins both the AI exposure and the AI execution content of the corpus, and exposure, the more common of the two labels, is the binding one.
At $0.71$ the corpus retains $12\%$ of its steps as AI-exposed and $4\%$ as AI-executed, and every floor above it drops the exposure share below $10\%$, which is why we settled there.
\begin{figure}[!t]
\caption{PCF Results as the Label-Transfer Similarity Floor Varies}
\label{fig:apqc_similarity_threshold}
\centering
\includegraphics[width=\textwidth]{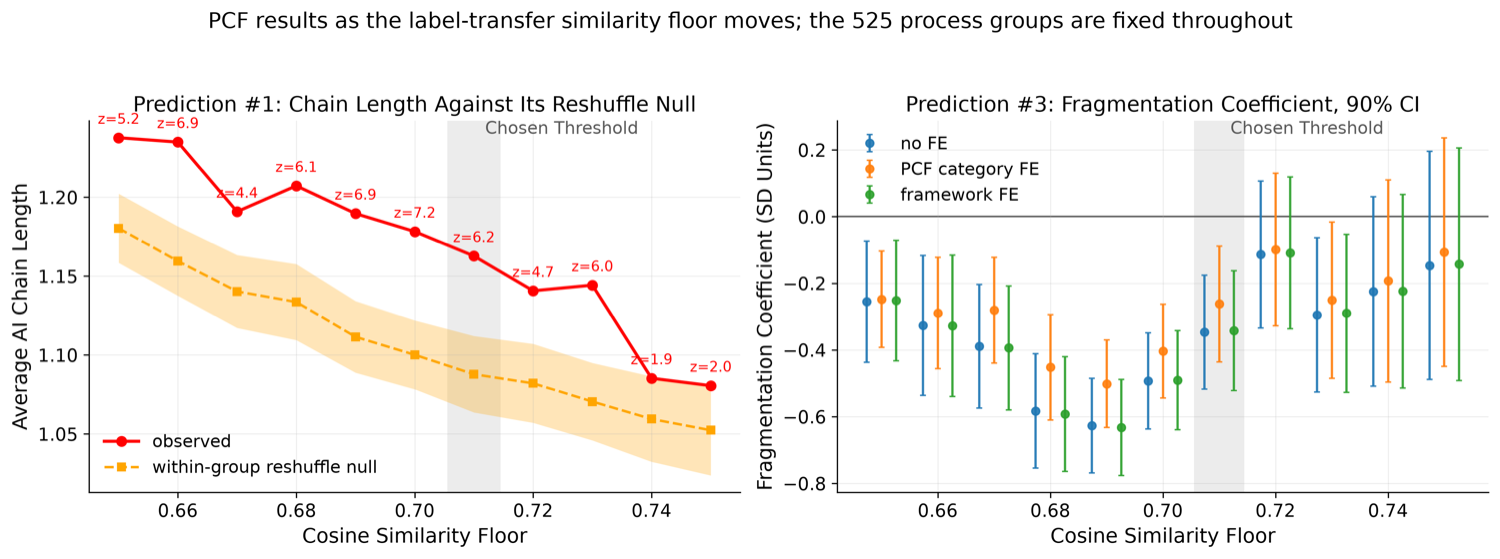}
\par\vspace{\fignotesskip}
\begin{minipage}{1\linewidth}
\begin{spacing}{0.2}
{\fontsize{9.5pt}{10pt}\selectfont Notes: Each point re-runs the label transfer at the indicated cosine similarity floor and re-estimates on the resulting labels.
The $525$ process groups are the same at every floor, since the five-step minimum is applied to the raw step count and not to the labels.
The left panel reports the average AI chain length in red against the mean of $1{,}000$ within-group step-order reshuffles in orange, with the orange band covering two standard deviations of that null and the annotation giving the $z$-statistic.
The right panel reports the coefficient on the empirical fragmentation index from Equation~\eqref{eq:fragmentation_index_regression} with $90\%$ confidence intervals, under each of the three fixed-effect specifications of columns~(4)\textendash(6) of Table~\ref{tab:fragmentation_index_regression_exposure}, which are drawn in the same order of colors as the placebo panels of Figure~\ref{fig:DWA_regression_aiExecution_mainSample} and are offset horizontally at each floor so that their intervals do not overlap.
The grey vertical band marks the $0.71$ floor used throughout the paper, and covers every estimate made at that floor.}
\end{spacing}
\end{minipage}
\end{figure}

\paragraph{Prediction \#1.}
Subsection~\ref{sec:chainLength_prediction} reports this exercise alongside the O*NET one; we give it here in full.
The average AI chain in the PCF dataset spans $1.16$ steps, against $1.45$ in the main text for O*NET.
The magnitude is small, and mechanically so, because with only $4\%$ of steps AI-executed, there is little
opportunity for long runs to form, and the ceiling on the statistic is far below the O*NET sample's.
What the exercise can still ask is whether the runs that do form are longer than the same labels arranged at
random along the same sequences, and they clearly are.
Benchmarked as in Section~\ref{sec:chainLength_prediction} against $1{,}000$ draws of a within-group
step-order reshuffle, the null averages $1.09$ with a standard deviation of $0.01$, and the observed value
sits at the 100th percentile ($z = 6.2$); against a within-category reassignment that preserves group size,
the null averages $1.06$ and the observed value again sits at the 100th percentile ($z = 10.2$).
The economic magnitude of AI chain length is therefore modest, but the co-occurrence of AI-executed content behind it is unambiguous.
Since no language model placed these steps in order, the contiguity is a property of the documented workflow
rather than of the ordering procedure.

\paragraph{Prediction \#2.}
We do not attempt a counterpart of the neighbor-execution regression, since that test requires the same step
to appear in several workflows with \emph{different} execution outcomes, while transferred labels are a
deterministic function of the matched task.
Restricting to steps that clear the threshold and whose matched task maps to a single DWA appearing in more
than one process group leaves $981$ observations across $137$ DWAs, of which only $9$ contain both an
AI-executed and a non-executed step, against $10{,}708$ tasks across $1{,}748$ DWAs in the baseline specification of
Table~\ref{tab:DWA_regression_aiExecution_mainSample}, so the matched corpus is underpowered for that exercise
by roughly an order of magnitude.

\paragraph{Reading the PCF fragmentation coefficient.}
The fragmentation regression of Equation~\eqref{eq:fragmentation_index_regression} is estimated on this corpus in columns~(4)\textendash(6) of Table~\ref{tab:fragmentation_index_regression_exposure} and discussed in Subsection~\ref{sec:fragmentation_prediction}, which also discusses why the workflow-level channel is detected on this corpus and not on the O*NET sample.
One feature of the transfer belongs here rather than in the main text.
All variables are z-scored against their mean and standard deviation within the PCF sample, and the index varies far less across PCF process groups than across O*NET occupations, with a standard deviation of $0.06$ against $0.26$.
A one-standard-deviation move is thus a much smaller move in the index here, and the slope per unit of the index is correspondingly steeper.

\subsection{Summary and Takeaways}
\label{app:external_validation_discussion}

The benchmarks answer the objection that motivated them.
An instrument that merely clustered semantically similar items would have no reason to place seven in ten
consecutive step pairs in the practitioners' direction, or to reproduce a full documented sequence one time in
four, on a corpus that working groups wrote down with no language model involved at any stage.
Nor would it place four in five activity pairs in the direction real organizations executed them, on pairs the
traces confirm have a true order.
That two independent benchmarks, one documented and one observed, agree to within roughly two percentage
points on the statistic they share is the strongest evidence here that what the instrument recovers is
sequence rather than topic.

The overlap is nonetheless imperfect, and two considerations bound how it should be read.
On one side, some of the residual disagreement is indeterminacy in the benchmark rather than error in the
language model's logic, so the measured overlap understates agreement with genuine workflow order.
On the other, the shortfall is not merely noise.
Across the same branches the model's orderings agree with \emph{each other} at an average pairwise $\tau$ of
$0.79$, well above the $0.53$ they achieve against APQC, so the instrument is considerably more
self-consistent than it is accurate.

Three limitations qualify what these exercises establish.
First, they validate the ordering \emph{procedure} rather than the specific O*NET orderings used in the paper,
since neither the PCF nor the event logs records the order in which O*NET occupations perform their tasks.
Second, the PCF is not a random sample of work, since it is weighted toward enterprise process work such as finance,
procurement, human resources, information technology, and supply chain.
Third, the PCF carries no AI exposure or execution labels of its own, since our AI exposure labels
come from \citet{eloundou2023gpts} and our AI execution labels from the Anthropic Economic Index
\citep{handa2025economictasksperformedai}, both defined over O*NET tasks.
The labels we use on the PCF sample are transferred from the nearest O*NET task
(Appendix~\ref{app:apqc_fragmentation_top}), and the lossy match leaves that sample thin in AI content, which is why
we treat the two samples as complements and report them side by side rather than relying on either alone.
Working with O*NET keeps our measures comparable to the rest of the literature that uses these now-standard
objects, at the cost of not observing workflow order directly; this appendix is our attempt to bound that
cost.

\clearpage
\renewcommand{\refname}{\large References for Supplementary Appendix}
\setlength{\bibsep}{0ex}
\begin{spacing}{0.98}
\putbib[rubin]
\end{spacing}
\end{bibunit}

\end{document}